%% file: qefs.tex
\def\cmntsoff{}
\documentclass[aps,pra,superscriptaddress,12pt,longbibliography]{revtex4-2}

\makeatletter
\renewcommand{\p@subsection}{}
\renewcommand{\p@subsubsection}{}
\makeatother

\usepackage{amsmath,amsfonts,amssymb,amsthm}
\usepackage{mathtools}
\usepackage{setspace}
\usepackage[normalem]{ulem} 
\usepackage{stmaryrd}       
\usepackage{dsfont}         
\usepackage{expdlist}
\usepackage{physics}
\usepackage{verbatim}
\usepackage{xspace}

\usepackage[nolists,heads,nomarkers]{endfloat}

\usepackage{graphicx,xcolor}
\usepackage{hyperref}
\hypersetup{pdfstartview=}
\usepackage{url}

\usepackage{verbatim}
\usepackage{alltt}
\usepackage{moreverb}

\usepackage{xr} 

\usepackage[ruled,lined]{algorithm2e}
\SetAlgorithmName{Protocol}{List of Protocols}

\allowdisplaybreaks

\input{qefs_defs}

\begin{document} 

\title{Quantum Probability Estimation for Randomness with Quantum Side Information}

\author{Emanuel Knill}
\affiliation{National Institute of Standards and Technology, Boulder, Colorado 80305, USA}
\affiliation{Center for Theory of Quantum Matter, University of Colorado, Boulder, Colorado 80309, USA}
\author{Yanbao Zhang}
\affiliation{NTT Basic Research Laboratories, NTT Corporation, 3-1
Morinosato-Wakamiya, Atsugi, Kanagawa 243-0198, Japan}
\affiliation{NTT Research Center for Theoretical Quantum Physics, NTT Corporation, 
3-1 Morinosato-Wakamiya, Atsugi, Kanagawa 243-0198, Japan}
\author{Honghao Fu}
\affiliation{Joint Institute for Quantum Information and Computer Science, University of Maryland, College Park, Maryland 20740, USA}

\begin{abstract} 
  We develop a quantum version of the probability estimation framework [arXiv:1709.06159] for randomness generation with quantum side information. We show that most of the properties of probability estimation hold for quantum probability estimation (QPE). This includes asymptotic optimality at constant error and randomness expansion with logarithmic input entropy. QPE is implemented by constructing model-dependent quantum estimation factors (QEFs), which yield statistical confidence upper bounds on data-conditional normalized R\'enyi powers. This leads to conditional min-entropy estimates for randomness generation.  The bounds are valid for relevant models of sequences of experimental trials without requiring independent and identical or stationary behavior.  QEFs may be adapted to changing conditions during the sequence and trials can be stopped any time, such as when the results so far are satisfactory.  QEFs can be constructed from entropy estimators to improve the bounds for conditional min-entropy of classical-quantum states from the entropy accumulation framework [Dupuis, Fawzi and Renner, arXiv:1607.01796]. QEFs are applicable to a larger class of models, including models permitting  experimental devices with super-quantum but non-signaling behaviors and semi-device dependent models.  The improved bounds are relevant for finite data or error bounds of the form $e^{-\kappa s}$, where $s$ is the number of random bits produced.  We give a general construction of entropy estimators based on maximum probability estimators, which exist for many configurations.  For the class of $(k,2,2)$ Bell-test configurations we provide schemas for directly optimizing QEFs to overcome the limitations of entropy-estimator-based constructions. We obtain and apply QEFs for examples involving the $(2,2,2)$ Bell-test configuration to demonstrate substantial improvements in finite-data efficiency.
\end{abstract}

\maketitle
\tableofcontents

\SetKwInOut{Given}{Given}
\SetKwInOut{Input}{Input}
\SetKwInOut{Output}{Output}
\SetInd{0.5em}{1em}

\section{Overview}

\subsection{Introduction}

For a relevant overview of the problem of device-independent
randomness generation and expansion and how probability estimation
(PE) solves this problem for classical side-information, see
Ref.~\cite{knill:qc2017a}. Here we establish the mathematical
foundations for quantum probability estimation (QPE), which implements
most of the features of PE from Ref.~\cite{knill:qc2017a} for quantum
side-information. The features implemented include: (i) Sound
conditional min-entropy estimation for general models covering 
device-independent and device-dependent configurations without
assuming stationarity or independence of trials. (ii) Forward
adaptability to changing experimental conditions and the ability to
stop acquiring trials early when satisfied. (iii) Asymptotically
optimal rates at constant error bounds. (iv) Uncomplicated and clean
exponential expansion with highly biased inputs. (v) Accessible
constructions for available experimental configurations.  We did not
implement a generalization to ``soft'' estimators that would allow use
of information not intended to be part of the extractor input.  In
addition, while we have general effective methods for PE optimization,
effective methods for unrestricted QPE optimization presently exist
only for special configurations, which include standard Bell-test
configurations.

The first insight of the PE framework is that it is possible to
directly estimate the data-dependent side information and input
conditional probabilities for a sequence of trials. The estimate is a
traditional statistical one, giving confidence upper bounds on the
conditional probability of the data. The second insight is that these
estimates can be used to estimate conditional min-entropy for use with
classical-proof strong randomness extractors to produce near-uniform
random bits, or directly to prove soundness of bits extracted with
arbitrary strong randomness extractors. The third insight is that
probability estimates can be obtained by martingale methods from
probability estimation factors (PEFs) that are computed for each trial.

In the presence of quantum side information, instead of estimating
conditional probabilities, we estimate conditional R\'enyi powers for
the observed data given the inputs and the side information.  The
conditional R\'enyi powers are non-commutative generalizations of the
conditional probabilities estimated in PE.  R\'enyi entropies have
played major roles in previous works showing that it is possible to
generate randomness in a device-independent way with $\Pfnt{E}$
holding quantum side
information~\cite{miller_c:qc2014a,miller_c:qc2014b,dupuis:qc2016a,arnon-friedman:qc2018a}.
Most of the properties of R\'enyi entropies rest on properties
established for R\'enyi powers, so estimating the latter may be viewed
as more fundamental.  The conditional R\'enyi powers are estimated via
quantum estimation factors (\QEFs), replacing PEFs in PE.  We prove
that chaining \QEFs by multiplying them for a sequence of trials
yields \QEFs for the sequence as a whole.  As a result, \QEFs (more
precisely, their inverses) may be seen as accumulating
conditional R\'enyi power estimates, so the framework could
alternatively be called ``R\'enyi power accumulation''.  The trials,
their models and the \QEFs in a chain can depend arbitrarily on data
from previous trials, as a result of which it is also possible to stop
trials whenever sufficient R\'enyi power has been accumulated.  Other
approaches to randomness generation have not explicitly developed
these capabilities to the same extent. Because the R\'enyi power
estimates depend on the specific data observed, they imply but are
separate from any entropy estimates for the state as a whole. A main
result is that like PEFs for PE, \QEFs yield a conditional min-entropy
estimate that can be used directly with quantum-proof strong
extractors.

The conceptual principles of QPE rest on statistical estimates of
probabilities rather then entropic analyses, and the proofs of the
mathematical results characterizing \QEFs and establishing their
chainability reflect these principles.  However, given the common
goals of the entropy accumulation framework~\cite{dupuis:qc2016a} and
QPE, it is not surprising that there are connections between the two.
Every \QEF yields an entropy estimator, which is equivalent to an
instance of affine min-tradeoff functions as defined in the entropy
accumulation framework.  Conversely, \QEFs can be constructed from
entropy estimators. However, the construction is not reversible in the
sense that \QEFs obtained from entropy estimators belong to a
restricted class of \QEFs with strictly worse performance than the
original \QEFs from which the entropy estimator was derived.  In the
examples of Sect.~\ref{subsec:examples}, the performance is
substantially worse.

Our construction of \QEFs from entropy estimators and its consequences
for conditional min-entropy estimation parallel the corresponding
results in Ref.~\cite{dupuis:qc2016a}.  A corollary of our
construction is an improved version of the entropy accumulation
theorem (EAT, Thm.~4.4 of Ref.~\cite{dupuis:qc2016a}) for the case of
conditional min-entropy of classical-quantum states.  The EAT is
formulated for quantum-quantum states, but for randomness generation
there is no need to estimate conditional min-entropy for such states,
so we do not pursue this generalization here. Neither do we consider
extensions to estimating smooth max-entropy, which is another
capability of the EAT.  Unlike the original EAT, our construction leads to
exponential randomness expansion without protocol complications, where
the input entropy is a simple logarithm of the output entropy.
We remark that there is now a refinement of the EAT which
yields ``second-order'' improvements similar to ours and also
achieves exponential randomness expansion~\cite{dupuis:qc2018a}.

The QPE framework has more flexibility for models of the quantum side
information. In particular, we can obtain randomness secure against
any non-signaling devices, quantum or otherwise, provided the
side information is still quantum.  At the time of writing, there are
few min-tradeoff functions suitable for use with the EAT. We provide a
large family of entropy estimators from which \QEFs can be constructed
and optimized. In general, we prefer to optimize \QEFs directly
whenever possible, and we show that the optimization problem can be
solved numerically for the important class of $(k,2,2)$-Bell-test
configurations.

Like entropy accumulation, QPE is asymptotically optimal at constant
error bounds. This does not imply optimality for finite data, for
randomness expansion, or when error bounds decrease exponentially with
the randomness produced.  For this regime, we do not know what the
optimal rates are, but like PE for classical side information, QPE
performs substantially better than other methods developed so far for
quantum side information. For this, we consider two closely related
problems. Suppose we are given a model for the side information after
any sequence of trials, and we anticipate a particular distribution
for the results from each trial. The first problem is to determine the
minimum number of trials $n$ required to obtain $k$ random bits at a
given error bound $\epsilon$. The second is to determine the
asymptotic rate of random bits that can be produced given that the
error bound is of the form $e^{-\kappa n}$. For the EAT and QPE, 
the solutions of the two problems are
essentially equivalent, but the second problem has the advantage of a
clear asymptotic formulation not affected by finite $n$.  
For $\kappa=0$, the maximum rate is determined by the asymptotic
equipartition property~\cite{tomamichel:qc2009a}.

The problems of the previous paragraph are motivated by relevant
applications such as randomness beacons~\cite{fischer:qc2011a} or
low-latency randomness generation. In these cases, a fixed-size block
of random bits, uniform within a given error bound, needs to be
produced within a short time.  This is typically far from an
asymptotic regime, where the amount of randomness generated is much
larger than the log-error bound and there is a long delay
from protocol initiation to randomness availability.  A relevant
finite problem for benchmarking purposes is to produce $512$ random
bits certified to be within $2^{-64}$ of uniform. The performance of a
particular protocol is determined by the resources required. We
usually fix the observed trial distribution, assume that it is
independent and identical, then ask for trade-off curves for
the number of trials and the number of initial random bits
required. The initial random bits are needed for input choices and for
the extractor seed.  Under many circumstances, the initial random bits
may come from a public source.  Here, the assumptions on the trial
distribution are a completeness property, where in an ideal setup we
expect to be able configure the experiment so that overall frequencies
approach the assumed ones. Soundness of the protocols does not depend
on the specific distributions, only on the model.

Since the completion of this preprint, parts of this work have been
published. Ref.~\cite{zhang_y:qc2020a} covers the basic theory of QEFs
for randomness generation and Ref.~\cite{zhang_y:qc2018a} describes an
experimental implementation for repeated and low-latency production of
blocks of $512$ random bits.

\subsection{Summary of Main Results}

The purpose of this manuscript is to provide the mathematical
foundations for quantum probability estimation. The technical
results in the manuscript may be difficult to interpret without having
worked through the parts leading up to them.  For accessibility, in
this section we summarize the main results without precise definitions.

We consider systems consisting of classical variables $C$ and $Z$ and
a quantum system containing the side information $\Pfnt{E}$.  For the
present purposes, these symbols may be treated as system labels. In
quantum terms, a joint state of the systems may be written as
$\rho_{CZ\Pfnt{E}}=\sum_{cz}\dyad{cz}\otimes\rho_{\Pfnt{E}}(cz)$ with respect
to the classical basis of $C$ and $Z$, where
$\sum_{cz}\tr(\rho_{\Pfnt{E}}(cz))=1$. We treat $Z$ as the input and
$C$ as the output system. In a typical Bell test, $Z$ is the sequence
of measurement settings choices (or inputs) and $C$ is the sequence of
measurement outcomes (or outputs), where the inputs and outputs may
contain choices and results from multiple devices. The joint state
given is the final state after the experiment, which consists of a
sequence of trials generating results $C_{i}Z_{i}$ so that
$C=(C_{i})_{i=1}^{n}$ and $Z=(Z_{i})_{i=1}^{n}$.  A model for the
experiment is the set of final states that can occur and is normally
constructed by chaining models for each trial. The models must be
chained while satisfying a Markov condition on the inputs similar to
the Markov condition required for EAT channel
chains~\cite{arnon-friedman:qc2018a}. To avoid the Markov condition
one can drop the use of explicit inputs by including them in $C$. For
example, see Protocol~\ref{prot:condimplicit}, which requires that the
conditional min-entropy witnessed exceeds the number of bits required for the
inputs.  
For Bell tests, the trial models are constrained by
non-signaling conditions and, for quantum devices, by the requirement 
that the results can be achieved with measurements of separate quantum
systems according to the configuration. We develop a general framework
for models and their construction in Sect.~\ref{sec:qmodels}.  We
explain how models capture standard configurations for device-dependent and device-independent randomness generation in
Sect.~\ref{sec:models:examples}.  Configurations modeled with explicit
quantum systems and quantum processes producing the data are readily
accounted for, as are scenarios where the devices may exhibit
unspecified super-quantum behaviors, as long as the side information
is still quantum.

Let $\alpha>1$ and $\beta=\alpha-1$. Given $\rho_{CZ\Pfnt{E}}$ as
above, define $\rho(z)=\sum_{c}\rho(cz)$, where we omit the $\Pfnt{E}$
system label when this is the only quantum system in play.  For a
given state $\rho_{CZ\Pfnt{E}}$, the normalized, sandwiched,
conditional $\alpha$-R\'enyi power for value $cz$ of $CZ$ is given by
\begin{equation}
  \hatRpow{\alpha}{\rho(cz)}{\rho(z)} =
  \frac{1}{\tr(\rho(cz))} \tr( (\rho(z)^{-\beta/(2\alpha)}\rho(cz)\rho(z)^{-\beta/(2\alpha)})^{\alpha}).
\end{equation}
If $\Pfnt{E}$ is one-dimensional, then $\mu(cz)\doteq\rho(cz)$ is a
probability distribution and the conditional R\'enyi power becomes
$(\mu(cz)/\mu(z))^{\beta}$, a power of the probability of $c$
conditional on $z$. In the probability estimation
framework~\cite{knill:qc2017a}, the main goal is to estimate such
conditional probabilities. Here, the non-commutative generalization is
to estimate the conditional R\'enyi powers.

The success of probability estimation framework rests on the
construction of probability estimation factors (PEFs) which yield
probability estimates via a martingale analysis. Quantum estimation
factors (\QEFs) with power $\beta$ are functions $F:cz\mapsto
F(cz)\geq 0$ such that for all states $\rho_{CZ\Pfnt{E}}$ in the
model, $F$ satisfies the \QEF inequality
\begin{equation}
  \sum_{cz} \tr(\rho(cz)) F(cz) \hatRpow{\alpha}{\rho(cz)}{\rho(z)}\leq 1.
\end{equation}
We do not use an explicit martingale analysis for \QEFs. Instead we
show directly that \QEFs for the trial models can be multiplied to
yield \QEFs for the sequence of trials. \QEFs for later trials may
depend on data from earlier trials, so we refer to this procedure as
\QEF chaining.  \QEFs and their variations are defined in
Sect.~\ref{sec:qefdefs}.  That they can be chained is
Thm.~\ref{thm:qefchainmain}.  It appears that the sandwiched R\'enyi
powers are particularly well suited for chaining. We have not succeeded
in chaining other quantities that yield conditional min-entropy
estimates.

The main result for \QEFs is that they
yield confidence upper bounds on the conditional R\'enyi powers:
\begin{theorem*}
  \emph{(Thm.~\ref{thm:qefs_estimate})} If $F$ is a \QEF with power
  $\beta$ for a model, and $\rho_{CZ\Pfnt{E}}$ is a state in the
  model, then $[0,1/(\epsilon F(cz))]$ is a significance-level
  $\epsilon$ confidence interval for
  $\hatRpow{\alpha}{\rho(cz)}{\rho(z)}$ with respect to the
  probability distribution $\tr(\rho(cz))$ induced on $CZ$ by
  $\rho_{CZ\Pfnt{E}}$.
\end{theorem*}

For randomness generation, \QEFs are used to estimate conditional
min-entropy with an error bound. If the estimate is larger than a
protocol threshold, a quantum-proof strong extractor can be applied to
the outputs to obtain a string of nearly uniform random bits.  The
number of bits is somewhat less than the estimate in order to take into account
extractor constraints.  Let $H^{\epsilon}_{\infty}(C|Z\Pfnt{E},\Phi')$
denote the smooth quantum conditional min-entropy for the state of
$CZ\Pfnt{E}$ conditional on the event $\Phi'$ defined as a set of
values $cz$ of $CZ$.  The smoothness parameter $\epsilon$ is an error
bound that chains directly with error bounds of extractors. It is
defined with respect to purified distance, but may be interpreted
as total variation distance for chaining with protocols whose
error bounds use the latter distance.  The conditional R\'enyi power estimate
provided by a \QEF implies a conditional min-entropy estimate suitable
for randomness generation protocols:
\begin{theorem*}
  \emph{(Thm.~\ref{thm:bnds_from_qef})}
  Suppose that $F$ is a \QEF with power $\beta$ for a model, and
  $\rho_{CZ\Pfnt{E}}$ is a state in the model.  Fix $1\geq p>0$ and
  $\epsilon>0$ and write $\Phi=\{cz:F(cz)\geq
  1/(p^{\beta}(\epsilon^{2}/2))\}$.  Let $\Phi'\subseteq\Phi$ and let
  $\kappa=\sum_{cz\in\Phi'}\tr(\rho(cz))$ be the probability of the
  event $\Phi'$ according to the state. Then
  $H^{\epsilon}_{\infty}(C|Z\Pfnt{E},\Phi') \geq -\log(p) +
  \frac{\alpha}{\beta}\log(\kappa)$.
\end{theorem*}
Here we used the convention $\log(0)=-\infty$.  We formulated the
theorem to parallel the statements of the EAT and the
  propositions that lead to the EAT in Ref.~\cite{dupuis:qc2016a}.  If
$CZ$ is generated by a sequence of trials chained with identical
models and $F$ is obtained by multiplying identical trial-wise \QEFs
$F_{0}$, then we can define a rate $h$ by $h \doteq -\log(p)/n$. The event
$\Phi$ can alternatively be expressed as
$\Phi=\{cz:\sum_{i}\log(F_{0}(c_{i}z_{i}))/\beta \geq
nh-2\log(\epsilon/\sqrt{2})/\beta\}$. This identifies $h$ as the
targeted conditional min-entropy rate, and we can interpret
$\log(F_{0}(c_{i}z_{i}))/\beta$ as the trial-wise contributions to the
final conditional min-entropy.  When configuring an experiment, the
goal is therefore to maximize the expected values of
$\log(F_{0}(c_{i}z_{i}))/\beta$. Comparing the bounds to the
corresponding ones for PEFs in Ref.~\cite{knill:qc2017a}, the main
difference is the change in the threshold requirement replacing the
term $\epsilon$ by $\epsilon^{2}/2$. An interpretation is that for the
same witnessed rate and for a positive conditional min-entropy bound,
twice as many trials are required to satisfy the error bound with
quantum side information than with classical side information.  A
similar phenomenon occurs when comparing parameters of quantum-proof
to classical-proof strong extractors, for example, see
Ref.~\cite{mauerer:2012}.

The QPE framework was motivated and developed as a generalization of
the PE framework~\cite{knill:qc2017a} to quantum side-information, which 
in turn arose from a program~\cite{bierhorst:qc2017a, bierhorst:qc2018a} 
for randomness generation based on test supermartingales~\cite{shafer:qc2009a} 
constructed from trial-wise test factors~\cite{zhang:2011}.  This led to 
the development of conditional R\'enyi power estimates. To obtain conditional
min-entropy estimates suitable for randomness generation we take
advantage of the connection between R\'enyi relative entropy and
conditional min-entropy~\cite{tomamichel:qc2009a}, which is also used
to prove the EAT from its prequel.

Explicit protocols for randomness generation that compose the
conditional min-entropy estimate with quantum-proof randomness
extractors are given in Sect.~\ref{sec:qef_protocols}.  For the
soundness of the protocols, the power $\beta$, the smoothness
$\epsilon$ and the target entropy $-\log(p)$ must be chosen before the
protocol, in particular before or at least independently of the data
being generated by the experiment. For \QEFs, it is possible to
optimize and update trial-wise \QEFs (with $\beta$ fixed in advance)
before each trial, but after the data is obtained no further
optimization is possible. These considerations apply to all randomness
generation protocols. For example, to apply the EAT, the number of
trials, the target conditional min-entropy rate $h$ and the affine
min-tradeoff function are fixed before the protocol and temptation to
optimize them after the protocol in view of the trial results must be
resisted.

In Ref.~\cite{knill:qc2017a}, effective 
algorithms for optimizing PEFs are described and implemented. We do
not have such algorithms for \QEFs but offer two general theoretical
constructions and a schema for optimizing \QEFs for Bell-test
configurations with two input choices and two possible outputs for
each station. The first construction is based on a relationship
between \QEFs and entropy estimators.  The function $K:cz\mapsto
K(cz)\in\rls$ is an entropy estimator for a model if for all states
$\rho_{CZ\Pfnt{E}}$ of the model,
\begin{equation}
  \sum_{cz}K(cz)\tr(\rho(cz)) \leq H_{1}(C|Z\Pfnt{E}),
\end{equation}    
where $H_{1}(C|Z\Pfnt{E})$ is the quantum conditional entropy of the
state. Every \QEF yields an entropy estimator.
\begin{theorem*}
  \emph{(Thm.~\ref{thm:qef_to_ee})}
  Suppose that $F$ is a \QEF with power $\beta$ for a model.
  Then $K:cz\mapsto \log(F(cz))/\beta$ is an entropy estimator for the model.
\end{theorem*}
In the examples of Sect.~\ref{subsec:examples}, the entropy estimators
so obtained can have comparable performance to existing min-tradeoff
functions when used with EAT, but only at small powers.  We infer that
the \QEF and entropy-estimator or min-tradeoff-function
optimization problems are not well matched.

It is possible to obtain \QEFs from entropy estimators:
\begin{theorem*}
  \emph{(Thm.~\ref{thm:ee_to_qef})} Let $K$ be an entropy estimator
  for a model. Then there exists $\tilde c:\beta\in (0,1/2]\mapsto
  \tilde c(\beta)\in(0,u]$ such that $F:cz \mapsto e^{\beta
    K(cz)}/(1+\tilde c(\beta)\beta^{2}/2)$ is a \QEF with power
  $\beta$ for the model. The upper bound $u$ depends on the model and the
  image of $K$.
\end{theorem*}
The \QEFs so obtained belong to the special class of Petz \QEFs
(\QEFPs). Because the construction is essentially model-agnostic, it
does not yield optimal \QEFs. In particular, the strategy of
optimizing entropy estimators and then determining \QEFs accordingly does
not yield good \QEFs for finite data.  A function $\tilde c$ is
explicitly obtained in Thm.~\ref{thm:ee_to_qef}.  This theorem can
substitute for the EAT prequel, Prop.~4.5 of
Ref.~\cite{dupuis:qc2016a} to obtain improvements on the EAT bounds
for conditional min-entropy. 
(Similar improvements are also obtained in Ref.~\cite{dupuis:qc2018a}.)
For this we optimize $\beta$ given the
number of trials and a targeted conditional min-entropy rate, see the
handicapped comparison in Sect.~\ref{subsec:handeat}. We also include
examples that demonstrates the broad applicability of 
\QEFs and the significant improvements achievable by direct \QEF 
construction, see Sect.~\ref{subsec:examples}.

The connection between entropy estimators and \QEFPs relies on a
R\'enyi relative entropy bounding technique from
Ref.~\cite{tomamichel:qc2009a} that is also used for the connection
between R\'enyi relative entropy and min-tradeoff functions that is
needed for the proof of the EAT in Ref.~\cite{dupuis:qc2016a}.  This
suggests the view that the EAT fundamentally rests on QPE via \QEFPs.
Our work makes this connection explicit, thereby enabling
extensions, improvements and broader applicability of the results.

An application of entropy estimators and their \QEFPs is a proof that
asymptotically optimal conditional min-entropy rates are achieved with
\QEFPs. As suggested in Ref.~\cite{arnon-friedman:qc2016a}, this
follows from the quantum asymptotic equipartition
property~\cite{tomamichel:qc2009a}. We provide the necessary convexity
arguments to determine entropy estimators that witness achievability
of optimal rates. 

To remedy the lack of availability of general entropy estimators, we
show how entropy estimators can be obtained from max-prob estimators.
The function $B:cz\mapsto B(cz)\in\rls$ is a max-prob estimator for a
model if for all states $\rho_{CZ\Pfnt{E}}$ of the model,
$\sum_{cz}\tr(\rho(cz))B(cz)\geq
\max_{cz}(\tr(\rho(cz))/\tr(\rho(z)))$. Note that the definition
depends only on the classical probability distributions of $CZ$ that
are allowed by the model and can therefore be designed for general
non-signaling distributions. In particular, it is of foundational
interest that they can be used for sound and complete randomness
generation assuming only non-signaling constraints on the experimental
devices, which may have super-quantum capabilities. However, if
super-quantum devices are reused in subsequent protocols,
composability may be compromised in ways that are not accounted for by a quantum analysis.

Max-prob estimators are used in probability estimation to directly
construct PEFs for exponential randomness expansion. Non-trivial
max-prob estimators exist for Bell-test configurations. For \QEFs, the
direct construction from max-prob estimators fails, but it is possible
to obtain entropy estimators by a similar method. The \QEFPs then
derived from these entropy estimators can be used for exponential
randomness expansion.
\begin{theorem*}
  \emph{(Thm.~\ref{thm:expexpansion} and its proof)} Suppose that $B$
  is a max-prob estimator for a trial model with $Z$ uniformly
  distributed such that there exists $\rho_{CZ\Pfnt{E}}$ in the model
  satisfying $\sum_{cz}\tr(\rho(cz)) B(cz)<1$.  Then there is a
  configuration with highly biased probability distributions 
  of independent and identical trial inputs and \QEFPs for this configuration
  such that for $n$ trials, the conditional min-entropy witnessed is
  at least $ng$ and the input entropy is $\log(n)g'$ for some
  constants $g,g'>0$.  The bias of the input distribution
  depends on $n$.
\end{theorem*}

In Sect.~\ref{sec:k22configs} we consider the standard
$(k,2,2)$-Bell-test configurations involving $k$ stations, two input
choices at each station and two possible outputs for each input.  It
is well-known that the quantum devices in such configurations can be
reduced to devices measuring one qubit in each station. For $k=2$ the reduction is
well explained in~\cite{pironio:qc2009a}, Sect.~2.4.1, where the main
mathematical results needed are from Ref.~\cite{tsirelson:qc1993a} and
Ref.~\cite{masanes:qc2006a}. We establish a general form of this
observation for arbitrary $k$ and suitable for use with \QEF
optimization. As a result, the \QEF optimization problem for
$(k,2,2)$-Bell-test configurations can be effectively solved by
numerical methods, after exploiting concavity and convexity properties
of the relevant quantities. 

Finally, in Sect.~\ref{subsec:examples} we construct \QEFs from PEFs
for examples involving
$(2,2,2)$-Bell-test configurations.  We apply \QEFs to the data from
the first demonstration of certified conditional min-entropy with
respect to classical side information~\cite{pironio:qc2010a}.  Our
analysis shows that \QEFs would have yielded more bits while being
secure against quantum side information. To illustrate the excellent
finite-data performance of \QEFs, we consider the minimum number of
trials required for three families of standard quantum states of the
devices to show orders of magnitude improvement over EAT.  We
highlight the improvement by determining the number of trials required
for the reference example of $512$ bits with error bound $2^{-64}$
with the distributions observed in the loophole-free Bell test used
previously for randomness generation with classical side-information
in Ref.~\cite{bierhorst:qc2018a}.

\section{Preliminaries}
\label{sec:preliminaries}

\subsection{Basics}

Let $\cH$ be a finite dimensional Hilbert space.  $B(\cH)$ is the set
of operators on $\cH$, $A(\cH)$ the subset of self-adjoint
(equivalently, Hermitian) operators, $S(\cH)$ the subset of Hermitian,
positive semidefinite operators, $S_{1}(\cH)=\{A\in S(\cH):\tr(A)=1\}$
the set of density operators, and $S_{\leq 1}(\cH)=\{A\in
S(\cH):\tr(A)\leq 1\}$.  For vectors $\ket{\psi}\in\cH$, we abbreviate
$\hat\psi=\dyad{\psi}$.  If $\ket{\psi}$ is normalized, then
$\hat\psi$ is the projector onto the one-dimensional subspace spanned
by $\ket{\psi}$.  For $\sigma\in B(\cH)$ we write $\sigma\ge 0$ if
$\sigma\in S(\cH)$.  The comparison $\sigma\geq\tau$ is equivalent to
$\sigma-\tau\geq 0$.  For $\sigma\in A(\cH)$, the support of $\sigma$
is the span of the eigenvectors of $\sigma$ with non-zero eigenvalues.
The support of $\sigma$ is denoted by $\Supp(\sigma)$, and the
projector onto the support of $\sigma$ is denoted by
$\knuth{\sigma\not=0}$.  For $\sigma,\tau\in S(\cH)$, we write
$\sigma\ll\tau$ if $\Supp(\sigma)\subseteq\Supp(\tau)$. Equivalently,
$\sigma\ll\tau$ iff there exists $\lambda>0$ such that
$\sigma<\lambda\tau$.  For Hermitian $\sigma$, the spectrum
$\Spec(\sigma)$ is the family of eigenvalues of $\sigma$ accounting
for multiplicity.  To be specific, we treat the spectrum as a vector
of real numbers in descending order.  We use the fact that
$\Spec(A^{\dagger}A)= \Spec(AA^{\dagger})$.  For $\sigma\in A(\cH)$
without full support, we define $\sigma^{-1}$ as the relative
inverse. That is, given a spectral decomposition of $\sigma$ in the
form $\sigma= \sum_{j}\lambda_{j}\hat j$ with $\tr(\hat j \hat i) =
\delta_{i,j}$ and $\lambda_{j}\not= 0$, we have
$\sigma^{-1}=\sum_{j}\lambda_{j}^{-1}\hat j$.  If $\Pi$ is the
projector onto the support of $\sigma$, then $
\sigma\sigma^{-1}=\Pi\sigma\sigma^{-1}\Pi = \Pi$ and
$\Supp(\sigma^{-1})=\Supp(\sigma)$.  We define $\log(\sigma)$ in the
same relative way.  For $\sigma\in A(\cH)$ with spectral
decomposition $\sigma=\sum_j \lambda_j\hat j$, the positive part of
$\sigma$ is defined as $\pospart{\sigma}=\sum_{j:\lambda_{j}>0}
\lambda_{j}\hat j$.  The absolute value is $|\sigma| =
\sum_{j}|\lambda_{j}|\hat j = \pospart{\sigma}+\pospart{-\sigma}$.
For arbitrary $A\in B(\cH)$, define $|A|=\sqrt{A^{\dagger}A}$.  The
projector onto the support of $\pospart{\sigma}$ is denoted by
$\suppproj{\sigma}$.  We need two properties of positive parts:

\begin{lemma}\label{lem:trpospart}
  $\tr(\pospart{\sigma})$ is monotone in $\sigma$,
  and for $\sigma\ge 0$, $\tau\ge 0$ we have
  $\tr(\sigma\suppproj{\sigma-\tau})\geq \tr(\pospart{\sigma-\tau})$.
\end{lemma}

\begin{proof}
  Since $\pospart{\sigma} = f(\sigma)$ with $f(x)=(|x|+x)/2$ and $f$
  is continuous and monotone increasing, $\tr(\pospart{\sigma})$ is
  monotone in $\sigma$ according to Ref.~\cite{carlen:qc2009a},
  Thm.~2.10.  Let $\ket{i}$ be an orthonormal basis of eigenvectors of
  $\sigma-\tau$ with $(\sigma-\tau)\ket{i} = \lambda_{i}\ket{i}$.
  Write $\sigma_{ii}=\tr(\sigma\hat i)$ and $\tau_{ii}=\tr(\tau\hat
  i)$.  Then
  \begin{align}
    \tr(\pospart{\sigma-\tau})
     &= \sum_{i:\lambda_{i}>0}\sigma_{ii}-\tau_{ii}\notag\\
     &\leq \sum_{i:\lambda_{i}>0}\sigma_{ii}\notag\\
     & = \sum_{i:\lambda_{i}>0}\tr(\sigma\hat i)\notag\\
     &= \tr(\sigma \sum_{i:\lambda_{i}>0}\hat i)\notag\\
     &= \tr(\sigma\suppproj{\sigma-\tau}).
  \end{align}
\end{proof}

A linear map $\cE:B(\cH)\rightarrow B(\cH')$ is positive if
$\cE(S(\cH))\subseteq S(\cH')$.  The map $\cE$ is a pure completely
positive map ($\pCP$ map) if it is of the form $\cE(\rho)=A\rho
A^{\dagger}$ for some $A\in B(\cH)$.  A completely positive map ($\CP$
map) is a positive linear combination of $\pCP$ maps. A $\CP$ map
$\cE:B(\cH)\rightarrow B(\cH')$ can be expressed non-uniquely in the
form $\cE(\rho)=\sum_{i}A_{i}\rho A_{i}^{\dagger}$.  $\cE$ is
trace-preserving if $\tr(\cE(\rho))=\tr(\rho)$ or equivalently,
$\sum_{i}A_{i}^{\dagger}A_{i}=\one$. A quantum operation is a $\CP$
map that is trace preserving. Quantum operations are also referred to
as $\CPTP$ maps.

For $n\in\nats$, $[n]=\{k\in\nats:1\leq k\leq n\}$.  For maps
$f:X\rightarrow Y$, we extend $f$ to subsets $\cX$ of $X$ according to
$f(\cX)=\{f(x):x\in \cX\}$. For a formula $\phi$ with free variables,
the expression $\knuth{\phi}$ is a function from the set of values of
the free variables to $\{0,1\}$ defined as $\knuth{\phi}=1$ for values
of the variables where $\phi$ is true and $\knuth{\phi}=0$ otherwise.
There should be no confusion with the case where $\knuth{\ldots}$ is
applied to a comparison of a given Hermitian operator and a real
number to define a projector.

A subset $\cC$ of a vector space is convex if
$\sum_{i=1}^{k}\lambda_{i}c_{i}\in\cC$ whenever $c_{i}\in\cC$,
$\lambda_{i}\geq 0$ for all $i\in[k]$ and $\sum_{i=1}^{k}\lambda_{i} =
1$. Vectors $\sum_{i=1}^{k}\lambda_{i}c_{i}$ with $\lambda_{i}\geq 0$
and $\sum_{i=1}^{k}\lambda_{i}=1$ are referred to as convex
combinations of the $c_{i}$.  For any $\cC$, the convex closure
$\Cvx(\cC)$ of $\cC$ is the set of all convex combinations of members
of $\cC$. We write $\Cone(\cC)=[0,\infty)\Cvx(\cC)$ for the convex
cone generated by $\cC$. The set of extreme points of $\cC$ is denoted
by $\Xtrm(\cC)$.

\subsection{Systems}

We distinguish between systems and their state spaces.  We denote and
label quantum systems with $\Pfnt{A,B,\ldots, E,\ldots,U,V,W,X,Y,Z}$.
In this work, $\Pfnt{E}$ plays a distinguished role as a universal
quantum system for defining models or as the system carrying the
quantum side information. We often use $\Pfnt{U,V,W}$ to denote generic
quantum systems.  For a quantum system $\Pfnt{U}$, its Hilbert space
is $\cH(\Pfnt{U})$ with dimension $\vdim(\Pfnt{U})$. $\cS(\Pfnt{U})$
is the set of positive semidefinite operators on $\cH(\Pfnt{U})$, and
$\cS_{1}(\Pfnt{U})=\{\rho\in\cS(\Pfnt{U}): \tr(\rho)=1\}$ is the set
of density operators of $\Pfnt{U}$.  Members of $\cS_{1}(\Pfnt{U})$
are referred to as the states of $\Pfnt{U}$. States $\rho$ are
considered to be normalized by the condition $\tr(\rho)=1$, and
general members of $\cS(\Pfnt{U})$ are referred to as
unnormalized states. We abbreviate
$B(\cH(\Pfnt{U}))=\cB(\Pfnt{U})$ and $A(\cH(\Pfnt{U}))=\cA(\Pfnt{U})$.
If $\vdim(\Pfnt{U})=1$, we call $\Pfnt{U}$ trivial and
$\cS(U)=[0,\infty)$.  The set of systems in play has a joint state. We
use juxtaposition to combine systems, so $\Pfnt{UV}$ combines systems
$\Pfnt{U}$ and $\Pfnt{V}$. Its Hilbert space is
$\cH(\Pfnt{UV})=\cH(\Pfnt{U})\otimes\cH(\Pfnt{V})$.

We need to refer to subsystem factorizations of quantum state spaces.
For a Hilbert space $\cH$, a factorization of $\cH$ is a
representation of $\cH$ in the form
$\cH=\bigoplus_{k}\cH_{k}\otimes\cC_{k}\;\oplus \cR$.  Technically,
such factorizations are realized by an isomorphism, but we freely
identify the two sides without making this isomorphism explicit.
Given this factorization, states of $\cH_{k}\otimes\cC_{k}$ are also
states of $\cH$, and we construct unnormalized states of the form
$\sigma_{k}\otimes\tau_{k}$ accordingly with $\sigma_{k}\in
S(\cH_{k})$ and $\tau_{k}\in S(\cC_{k})$. The state space membership
may be left implicit when the factorization is clear and the index
sets match, here by using the same index-symbol $k$ with implicit
index set $K$.

We identify classical systems with classical variables (CVs).
Notationally and operationally we treat CVs as random variables
(RVs) without specified probability distributions.  CVs are
denoted by capital letters $A,B,C,\ldots,U,V,W,X,Y,Z,\Omega$.  In this
work, $A,B,C,X,Y,Z$ play a distinguished role, and $U,V,W$ are often
used as generic CVs.  Like RVs, as mathematical objects CVs are
functions from an underlying set $\Omega$, which we assume is
finite. Accordingly, a CV $U$ has an associated space of values denoted 
by $\Rng(U)$ with cardinality $|\Rng(U)|$.  Values of CVs
are denoted by the corresponding lower case letter.  Thus the symbol
$u$ denotes values of $U$. This implies that in a CV context, the
symbol $u$ is typed and always refers to a member of $\Rng(U)$. This
simplifies notation. For example, $\sum_{u}\ldots =
\sum_{u\in\Rng(U)}\ldots$ and $\{u:\ldots\} = \{u\in
\Rng(U):\ldots\}$.  If we need distinct symbols of this type we use
primed symbols such as $u'$ or explicitly specify the symbols'
membership. In a context where a CV $U$ has an associated state,
possibly joint with other CVs and quantum systems, we refer to the
process of obtaining a value $u$ of $U$ as instantiating $U$, with the
connotation that the value was not available for inspection before it
was instantiated.

The CV $U$ is trivial if $|\Rng(U)|=1$.  We freely construct
CVs by concatenation denoted by juxtaposition.  For example, if $U$
and $V$ are CVs, then $UV$ is a CV with values $uv$. If $u$ and $v$
are strings or sequences, then $uv$ is the concatenation of the two strings
or sequences.  Otherwise, $uv$ may be interpreted as the pair or
two-element sequence with first element $u$ and 
second element $v$. Any 
of the typical mathematical realizations of these concepts may be used.

The CV $F$ is determined by the CV $U$ if for some function $\cF$ on
$\Rng(U)$, for all $\omega\in\Omega$, $F(\omega)=\cF(U(\omega))$.  We
introduce such determined CVs as $F(U)$, which specifies that $F$ is a
CV determined by $U$ as well as a function $u\mapsto F(u)$.  This
overloads the symbol $F$. Its meaning is determined by the type of the
argument.  The special expression $F(U)$ may be considered to refer to
both meanings while emphasizing the type of the argument of $F$ as a
value of $U$.  Thus, given an expression $\cF(u)$, we may define
$F(U)$ by specifying a function $F:u\mapsto\cF(u)$ and call $F(U)$ a
function of $U$, or we may specify $F(U)$ by an identity of the form
$F(U)=\cF(U)$, which we also consider equivalent to the statement
$\forall u: F(u)=\cF(u)$.  We remark that in expressions such as $F(U)$  
or $F(U)=\cF(U)$, the symbol $U$ plays the role of a free variable with
arbitrary values in $\Rng(U)$. We may introduce objects such as
$\rho(U)$ that are primarily functions of CV values and not intended
to be interpreted as determined by CVs themselves.

When considering sequences of trials for randomness generation, the
final state involves a CV consisting of a sequence of individual trial
CVs. We use boldface to distinguish such CVs. A sequence CV $\Sfnt{U}$
is defined in terms of the trial CVs $U_{i}$ by
$\Sfnt{U}=U_{1}U_{2}\ldots U_{N}$ and has values
$\Sfnt{u}=u_{1}u_{2}\ldots u_{N}$. Here, $N$ is an absolute upper
bound on the number of trials that might be considered before a
protocol stops. We always assume that such an upper bound exists. The actual
number of trials considered is denoted by $n$.  To refer to initial
and final segments of $\Sfnt{U}$ we use the notation $\Sfnt{U}_{\leq
  k}=U_{1}\ldots U_{k}$ and similarly for $\Sfnt{U}_{<k}$,
$\Sfnt{U}_{\geq k}$ and $\Sfnt{U}_{>k}$.  The length of $\Sfnt{U}$ is
denoted by $|\Sfnt{U}|$. Similarly, if $U$ is a string, the number of
letters in $U$ is denoted by $|U|=\log_{l}(|\Rng(U)|)$, where $l$ is
the size of the alphabet of the string.  We may treat string CVs as
sequence CVs without using the explicit boldface.

A CV's state is a probability distribution on its values.  $\cS(U)$ is
the set of unnormalized, non-negative distributions on $U$, and
$\cS_{1}(U)$ is the set of probability distributions on $U$.  If $U$
is a CV, then $\Pfnt{U}$ is its quantization. The Hilbert space
of $\Pfnt{U}$ has a classical basis whose members are
$\ket{u}$. Probability distributions $\mu(U)$ of $U$ are associated
with the corresponding states $\sum_{u}\mu(u)\hat u$ diagonal in the
classical basis.  Probabilities and expectations with respect to the
probability distribution $\mu(U)$ are expressed as
$\Prob_{\mu(U)}(\phi)=\sum_{u}\mu(u) \knuth{\phi}$ and
$\Exp_{\mu(U)}(G(U)) = \sum_{u}\mu(u) G(u)$.

\subsection{Classical-Quantum States}

We study joint states of classical-quantum systems. For a CV $U$ and a
quantum system $\Pfnt{V}$, $U\Pfnt{V}$ is the joint system.  We define
the set of $\cS(\Pfnt{V})$-valued distributions of $U$ as
\begin{equation}
  \cS(U\Pfnt{V}) = \left\{\rho:
    u\mapsto \rho(u)\in\cS(\Pfnt{V})\right\}.
\end{equation}
The members of $\cS(U\Pfnt{V})$ may be considered as CVs with values
in $\cS(\Pfnt{V})$, so we denote these members by $\rho(U)$.  If
$\Pfnt{V}$ is clear from context or generic, we refer to $\rho(U)$ as
a state-valued distribution, or just a distribution of $U$ or a state
of $U\Pfnt{V}$, although the values are unnormalized states of $\Pfnt{V}$.

For the purpose of universality, we may consider $\Pfnt{V}$ with
infinite-dimensional $\cH(\Pfnt{V})$. However, by default we assume
that the values $\rho(u)$ of distributions are finite rank.  A
$\cS(\Pfnt{V})$-valued distribution $\rho(U)$ is normalized if
$\tr(\sum_{u}\rho(u))=1$.  The set of normalized distributions of $U$
is denoted by $\cS_{1}(U\Pfnt{V})$.  The set of sub-normalized
distributions is $\cS_{\leq 1}(U\Pfnt{V})=\{\rho(U)\in\cS(U\Pfnt{V}):
\tr(\sum_{u}\rho(u))\leq 1\}$.  The set $\cS_{1}(U\Pfnt{V})$ is
the set of states of $U\Pfnt{V}$.  For finite-dimensional
$\cH(\Pfnt{V})$, it is consistent with the conventional, quantized
definition of the set of classical-quantum states of $\Pfnt{UV}$ as the set of density
operators of the form $\sum_{u}\hat u\otimes\rho(u)$.  If $\Pfnt{V}$
is trivial and $\rho(U)$ is normalized, then $\rho(U)$ is a
probability distribution.  Our notational choices are
designed to be compatible with those in Ref.~\cite{knill:qc2017a} when
specialized to trivial $\Pfnt{V}$ for handling classical side
information.  We use symbols such as $\rho,\sigma,\tau,\chi,\zeta,\xi$
for general states and $\mu,\nu$ for probability distributions.

In this work we normally consider finite CVs and density operators
with finite support. The soundness of randomness generation protocols
is relative to a model, which is a set of state-valued distributions,
see Sect.~\ref{sec:qmodels}.  Some models are most conveniently
formulated with states in an infinite-dimensional Hilbert space, but
we define them so that the relevant state-valued distributions have
finite support in the Hilbert space.  The support of a distribution
$\rho(U)$ is the linear span of the supports of the $\rho(u)$. The
projector onto the support is the smallest projector $\Pi$ such that
for all $u$, $\Pi\rho(u)=\rho(u)$. While the technical arguments are
restricted to effectively finite dimensional situations, in most cases
the consequences for randomness generation extend to
countable-dimension side information.  To verify this requires
approximating a model's infinite-support trace-class states by model
states with finite-dimensional support.  

A positive map $\cE:B(\cH(\Pfnt{V}))\rightarrow B(\cH(\Pfnt{W}))$
induces a map $\cS(U\Pfnt{V})\rightarrow\cS(V\Pfnt{W})$ defined by
$\rho(U)\mapsto (\cE\circ \rho)(U)=\cE(\rho(U))$. If $\cE$ is
trace-preserving, then the map restricts to
$\cS_{1}(U\Pfnt{V})\rightarrow\cS_{1}(U\Pfnt{W})$.

We adapt RV and probability distribution conventions to denote and
manipulate state-valued distributions.  If $\rho(UV)$ is a
$\cS(\Pfnt{W})$-valued distribution, then $\rho(uv)$ refers to the
value of the distribution at $uv$.  According to marginalization
conventions, $\rho(U)$ is the marginal state-valued distribution of
$U$ and defined as $\rho(U)=\sum_{v}\rho(Uv)$.  With this, $\rho() =
\sum_{uv}\rho(uv)$ is the marginal state of $\Pfnt{W}$.   We
abbreviate $\rho=\rho()$ whenever the meaning is clear from context.
Conventions for events apply: If $\cX,\cY\subseteq\Rng(UV)$, then
$\rho(\cX)=\sum_{uv\in\cX}\rho(uv)$ and
$\rho(\cX,\cY)=\rho(\cX\cap\cY)$. We can specify subsets using logical
expressions in the CVs.  If $\phi(U,V)$ is such a logical formula with
free variables $U$ and $V$, we define $\{\phi\} = 
\{\phi(U,V)\}=\{uv:\phi(u,v)\}$.  In arguments of a distribution, the
curly brackets are normally omitted.  With this, we have the 
identities $\rho(u)=\rho(U=u)=\rho(\{U=u\})$.  Thus, our conventions 
imply that the expression $\rho(V,U=u)$ defines a distribution $\sigma(V)$ 
depending on $V$ only, but since this can be confusing we circumvent 
such expressions whenever possible.

We also adapt the usual conventions for conditioning.  We define
conditioning on a CV event according to the states obtained
conditionally on observing the event. If $\rho(UV)\in\cS(UV\Pfnt{W})$
and $\phi(U,V)$ is a formula with free variables $U$ and $V$, then
$\rho(UV|\phi)=\knuth{\phi(U,V)}\rho(UV)/\tr(\rho(\phi))$.  We define
$\rho(uv|\phi)=0$ if $\tr(\rho(\phi))=0$.  Note that if
$\tr(\rho(\phi))\ne 0$, then $\tr(\rho(UV|\phi)) = 1$ and therefore
$\rho(UV|\phi)\in\cS_{1}(UV)$.  In view of conventions for point
events, the expression $\rho(U|v)$ is interpreted as
$\rho(U|v)=\big(u\mapsto\rho(u|V=v)=\rho(uv)/\tr(\rho(v))\big)$.

For chaining purposes, we distinguish distributions $\rho(UV)$ for
  which $\rho(V)=\mu(V)\rho$ for a probability distribution
  $\mu(V)$. In this case $\rho(|v)=\rho$ is independent of $v$, that is,
  the systems $V$ and $\Pfnt{E}$ are independent.  We
  define $\cS((U|V)\Pfnt{E})=\{\tau(UV):\textrm{$\tau(|V)=\tau$
    independent of $V$}\}$.  Members of this set of distributions may
  be written as $\sigma(U|V)\in\cS((U|V)\Pfnt{E})$, the idea being
  that up to normalization, $\sigma(U|V)$ could have been obtained by
  conditioning some \(\sigma(UV)\) on $V$, where \(\sigma(|V)\) is
  independent of $V$. In this situation $\sigma(UV)$ is
  unspecified until we provide the probability distribution $\mu(V)$,
  at which point we can define $\sigma(UV)= \mu(V)\sigma(U|V)$. 

If $\rho(X)\in\cS(X\Pfnt{U})$ and $\rho(Y)\in\cS(Y\Pfnt{V})$, then
$\rho(X)\otimes\rho(Y)\in\cS(XY\Pfnt{UV})$.  If $F(U)$ is a function
of $U$, then $F$ pushes distributions forward according to
$(F_{*}\rho)(f)=\rho(F(U)=f)$. For clarity, the marginalization
conventions do not apply when distributions are expressed in terms of
compound constructions such as $\cE(\rho(UV))$, $\cP(\rho(X))$ or
$\fkC(\rho(U);\ldots)$ without an explicit final CV argument of the
form $\ldots(UV\ldots)$. The CV arguments of the proper construction 
are bound variables and not intended to be substituted by values. 
The
construction's expression refers to a distribution with CVs determined
by the specific expression.

We occasionally define state-valued distributions using anonymous
mapping notation, which includes the equivalence $\rho(UV) =
(uv\mapsto \rho(uv))$. For example, the expression $u\mapsto
\rho/|\Rng(U)|$ defines the uniform distribution on $U$ independent of
$\Pfnt{E}$ with the reduced density matrix of $\Pfnt{E}$ the state
$\rho$. In quantized terms this is the joint state
$\one_{\Pfnt{U}}/|\Rng(U)|\otimes\rho$, a notation with similar
complexity. The uniform probability distribution of $V$ is defined as
$\Unif(V): v\mapsto 1/|\Rng(V)|$ or equivalently $\Unif(V)= \big(v\mapsto
1/|\Rng(V)|\big)$.  Here, the quantum system is trivial.

We define POVMs of $\Pfnt{V}$ with outcomes $U$ as linear maps
$\cP:\cS(\Pfnt{V})\rightarrow\cS(U)$ of the form
$\cP(\rho)(U)=\tr(P_{U}\rho)$ with $P_{u}\in\cS(\Pfnt{V})$ for all $u$
and $\sum_{u}P_{u}=\one_{\Pfnt{V}}$.  Without confusion and
  following tradition, we refer to families of operators
  $P_{U}=(P_{u})_{u}$ satisfying these conditions as POVMs.
 The term ``POVM'' is an abbreviation for ``positive, operator-valued
  measure''.  We can naturally apply $\cP$ to members of
$\cS(X\Pfnt{VW})$ by defining $\cP(\rho(X))(XU)\in\cS(XU\Pfnt{W})$
according to
\begin{equation}
  \cP(\rho(X))(xu)=
  \tr_{\Pfnt{V}}((P_{u}\otimes\one_{\Pfnt{W}})\rho(x)).
\end{equation}
POVMs defined in this way remove the quantum system being
measured. POVMs do not specify what happens to the measured system, so
if we want to retain the measured system, we need to consider quantum
operations with classical outputs.

For the purpose of explicit conditioning on inputs, we make use of the
concept of short quantum Markov chains~\cite{hayden:qc2004b}.  We
define these chains for the class of states used here. For the general
definition, see the references.

\begin{definition} The distribution $\rho(UVW)\in\cS(UVW\Pfnt{E})$ is a
  \emph{short quantum Markov chain over $W\Pfnt{E}$}, written as $\rho(UVW)\in
  U\leftrightarrow W\Pfnt{E}\leftrightarrow V$, if for all $w$, there
  is a factorization $\cH(\Pfnt{E})=\bigoplus_{k}\cU_{w,k}\otimes
  \cV_{w,k}\;\oplus\cR$ such that $\rho(UVw)=\bigoplus_{k}
  \sigma_{w,k}(U)\otimes \tau_{w,k}(V)$.
\end{definition}
The definition is symmetric in $U$ and $V$.  That is,
$\rho(UVW)\in U\leftrightarrow W\Pfnt{E}\leftrightarrow V$ iff
$\rho(UVW)\in V\leftrightarrow W\Pfnt{E}\leftrightarrow U$.

\subsection{Distances}

We use the half trace distance as the extension of total variation ($\TV$) distance from
probability distributions to states for compatibility with classical
protocols and conventions.  Purified distance is  more
natural when dealing with quantum side information, partly because it is well-behaved with respect to extension to previously traced-out quantum systems, see Ref.~\cite{tomamichel:qc2012a}, Cor.~3.6, Pg.~52. Since purified
distance is an upper bound on half trace distance, this usually does not
complicate comparisons.

\begin{definition} \label{def:tv_distance}
  Let $\rho(U),\sigma(U)\in\cS_{1}(U\Pfnt{W})$. The \emph{$\TV$ distance
  between $\rho(U)$ and $\sigma(U)$} is given by
  \begin{equation}
    \trdist{\rho(U)}{\sigma(U)}=\frac{1}{2}\sum_{u}\tr(|\rho(u)-\sigma(u)|).
  \end{equation}
\end{definition}

We remark that the $\TV$ distance between $\rho(U)$ and $\sigma(U)$ is 
the same as that between the two quantized states
$\sum_{u}\hat u\otimes\rho(u)$ and $\sum_{u}\hat u\otimes\sigma(u)$.
The $\TV$ distance is $1/2$ of the conventional trace distance. We use
the name and the factor of $1/2$ for consistency with the conventions
for probability distributions and the treatment of randomness
generation in the presence of classical side information. It ensures
that the results of Ref.~\cite{knill:qc2017a} are directly comparable
to the results in this manuscript and that there are no discrepancies
when interpreting protocol soundness.  In works emphasizing general
quantum states, it is extended to trace-class operators and called the
generalized trace distance~\cite{tomamichel:qc2015a}.  Composition
with other classical protocols behaves as expected since the $\TV$
distance satisfies the triangle inequality (as it should) and the
data-processing inequality, see Ref.~\cite{nielsen:qc2001a},
Sect.~9.2.1 or the extensions in Ref.~\cite{tomamichel:qc2015a},
Sect.~3.2.  The next lemmas establish basic properties of $\TV$
distance needed later. Versions of these lemmas can be found in the
cited literature.

\begin{lemma}\label{lem:trdist_pospart}
  Let $\rho(U),\sigma(U)\in\cS_{1}(U\Pfnt{W})$. Then
  \begin{equation}
    \trdist{\rho(U)}{\sigma(U)} = \sum_{u}\tr(\pospart{\rho(u)-\sigma(u)}).
  \end{equation}
\end{lemma}

\begin{proof}
  In general
  $|\xi-\chi|=\pospart{\xi-\chi}+\pospart{\chi-\xi}$ and
  $\tr(\xi)-\tr(\chi)=\tr(\xi-\chi)=\tr(\pospart{\xi-\chi})-\tr(\pospart{\chi-\xi})$.
  Since $\sum_{u}\tr(\rho(u))=\sum_{u}\tr(\sigma(u))$,   
  we find that $\sum_{u}\tr(\pospart{\rho(u)-\sigma(u)})
  =\sum_{u}\tr(\pospart{\sigma(u)-\rho(u)})$ and
  \begin{align}
    \trdist{\rho(U)}{\sigma(U)} &=
    \sum_{u}\frac{1}{2}\tr(|\rho(u)-\sigma(u)|)\notag\\
    &=\sum_{u}\frac{1}{2}\tr(\pospart{\rho(u)-\sigma(u)}
    +\pospart{\sigma(u)-\rho(u)})\notag\\
    &=\frac{1}{2}\sum_{u}\tr(\pospart{\rho(u)-\sigma(u)})
    +\frac{1}{2}\sum_{u}\tr(\pospart{\sigma(u)-\rho(u)})\notag\\
    &=\sum_{u}\tr(\pospart{\rho(u)-\sigma(u)}).
  \end{align}
\end{proof}

\begin{lemma}\label{lem:trdist_weight}
  Let $\rho(U),\sigma(U)\in\cS_{1}(U\Pfnt{W})$. If there exists
  $\tau(U)\in\cS(U\Pfnt{W})$ with $\tr(\tau)\geq 1-\epsilon$,
  $\tau(U)\leq\rho(U)$ and $\tau(U)\leq\sigma(U)$, then
  $\trdist{\rho(U)}{\sigma(U)}\leq \epsilon$.
\end{lemma}

\begin{proof}
  Suppose that $\tau(U)$ has the given properties.  Then the $\TV$ distance is
  \begin{align}
    \trdist{\rho(U)}{\sigma(U)}
    &= \sum_{u}\tr(\pospart{\rho(u)-\sigma(u)})\notag\\
    &=\sum_{u}\tr(\pospart{(\rho(u)-\tau(u))-(\sigma(u)-\tau(u))})\notag\\
    &\leq\sum_{u}\tr(\pospart{\rho(u)-\tau(u)})\notag\\
    &=\sum_{u}\tr(\rho(u)-\tau(u))\notag\\
    &\leq \epsilon.
  \end{align}
  For the inequality of the third line, we have that for all $u$,
  $(\sigma(u)-\tau(u))\geq 0$, so we can apply the first part of
  Lem.~\ref{lem:trpospart}. 
\end{proof}

\begin{lemma}\label{lem:pmax_fillin}
  Let $\tau(UV)\in\cS_{\leq 1}(UV\Pfnt{W})$ 
  and $\sigma(V)\in\cS_{1}(V\Pfnt{W})$ with
  $\tau(UV)\leq p\sigma(V)$ and $p|\Rng(U)|\geq 1$.  Then there exists
  $\rho(UV)\in\cS_{1}(UV\Pfnt{W})$ such that $\tau(UV)\leq\rho(UV)\leq
  p\sigma(V)$.
\end{lemma}

\begin{proof}
  Let $\epsilon =  1-\tr(\tau)$ and
  $\delta = \sum_{uv}\tr(p\sigma(v)-\tau(uv))=|\Rng(U)|
  p-(1-\epsilon)\geq \epsilon$. Let
  $\xi(UV) = (\epsilon/\delta)(p\sigma(V)-\tau(UV))$.  Then
  $\tr(\xi)=\epsilon$ and $0\leq \xi(UV)\leq
  p\sigma(V)-\tau(UV)$. Define $\rho(UV)=\tau(UV)+\xi(UV)$.
  Then $\rho(UV)$ satisfies the desired conditions. 
\end{proof}

\begin{definition}\label{def:pd_distance}
  For $\sigma\in S_{1}(\cH)$ and $\tau\in S_{\leq 1}(\cH)$,
  the \emph{purified distance
  between $\sigma$ and $\tau$} is given by
  \begin{equation}
    \purdist{\sigma}{\tau}=
      \sqrt{1-\big(\tr(|\sqrt{\sigma}\sqrt{\tau}|)\big)^{2}}.
  \end{equation}
  For $\sigma(U)\in \cS_{1}(U\Pfnt{W})$
  and $\tau(U)\in \cS_{\leq 1}(V\Pfnt{W})$,
  \begin{equation}
    \purdist{\sigma(U)}{\tau(U)}
     = \sqrt{1-\left(\sum_{u}\tr(|\sqrt{\sigma(u)}\sqrt{\tau(u)}|)\right)^{2}}.
  \end{equation}
   The \emph{fidelity between $\sigma(U)$ and $\tau(U)$} is
   $F(\sigma(U),\tau(U))=\sum_{u}\tr(|\sqrt{\sigma(u)}\sqrt{\tau(u)}|)$.
\end{definition}

The definition of purified distance can be extended to $S_{\leq
  1}(\cH)$ in the first argument, but the expression becomes more
involved. We do not need the extension. The relevant properties of
purified distance can be determined from Tbl.~3.1, Pg.~48 in
Ref.~\cite{tomamichel:qc2012a} and the subsequent sections, given the
definition of purified distance in terms of fidelity (Def.~3.3,
Pg.~49).  We remark that the extension of purified distance to
distributions is consistent with the definition of purified distance
for the quantization of the distributions, see property (vi) in the
referenced table. That is, the purified distance between
  $\rho(U)\in \cS_{1}(U\Pfnt{W})$ and $\sigma(U)\in
  \cS_{1}(U\Pfnt{W})$ is the same as that between the quantized states
  $\sum_{u}\hat u\otimes\rho(u)$ and $\sum_{u}\hat
  u\otimes\sigma(u)$.  

The purified
distance satisfies the triangle inequality (as it should) and the
data-processing inequality, see  Ref.~\cite{tomamichel:qc2012a},
Prop.~3.2, Pg.~50 and Thm.~3.4, Pg.~51.
We also need the following relationships:

\begin{lemma}
  \label{lem:purdistprops}
  If $\rho(U),\sigma(U)\in \cS_{1}(U\Pfnt{W})$ and
  $\tau(U)\in \cS_{\leq 1}(U\Pfnt{W})$ such
  that $\tau(U)\leq \sigma(U)$, then
  $\purdist{\rho(U)}{\sigma(U)}\leq \purdist{\rho(U)}{\tau(U)}$
  and $\trdist{\rho(U)}{\sigma(U)}\leq\purdist{\rho(U)}{\sigma(U)}
  \leq \sqrt{2\trdist{\rho(U)}{\sigma(U)}}$.
\end{lemma}

\begin{proof}
  The first statement follows from property (v) of Tbl.~3.1, Pg.~48,
  and the second from Prop.~3.3, Pg.~50 of
  Ref.~\cite{tomamichel:qc2012a}, in view of the two remarks
  after Defs.~\ref{def:tv_distance} and~\ref{def:pd_distance}.
\end{proof}

\subsection{R\'enyi Powers}

We adopt the convention that the trace has higher priority than power
so that $\tr(A)^{\alpha}=(\tr(A))^{\alpha}$. Since many works have the
opposite convention, we often use the additional parentheses to
disambiguate.

\begin{definition}\label{def:sandwiched_Renyi}
  Let $0\leq \rho\ll\sigma$, $\alpha>1$ and $\beta=\alpha-1$. The
  \emph{sandwiched R\'enyi power of order $\alpha$ of $\rho$
    conditional on $\sigma$} is defined as
  \begin{equation}
    \Rpow{\alpha}{\rho}{\sigma} =
    \tr( \left(\sigma^{-\beta/(2\alpha)}\rho\sigma^{-\beta/(2\alpha)}\right)^{\alpha} ).
  \end{equation}
  The \emph{Petz R\'enyi power of order $\alpha$
    of $\rho$ conditional on $\sigma$} is defined as
  \begin{equation}
    \Ppow{\alpha}{\rho}{\sigma} =
    \tr( \rho^{\alpha}\sigma^{-\beta}).
  \end{equation}
  Both R\'enyi powers are defined to be identically $0$ if both
  $\rho=0$ and $\sigma=0$.  

  The \emph{normalized R\'enyi powers} are defined by
  \begin{align}
    \hatRpow{\alpha}{\rho}{\sigma} &=
      \frac{1}{\tr(\rho)}\Rpow{\alpha}{\rho}{\sigma},\notag\\
    \hatPpow{\alpha}{\rho}{\sigma} &=
      \frac{1}{\tr(\rho)}\Ppow{\alpha}{\rho}{\sigma},\notag\\
  \end{align}
  for $\tr(\rho)>0$. For $\tr(\rho)=0$ they are defined to be
  identically $1$.
\end{definition}
Throughout this work, we use the convention that the symbols $\alpha$
and $\beta$ satisfy $\alpha>1$ and $\beta=\alpha-1>0$. We normally do
not reiterate these constraints on $\alpha$ and $\beta$.  For Petz
R\'enyi powers we generally also assume $\alpha\leq 2$.  By default,
R\'enyi powers are sandwiched.  We only consider R\'enyi powers of
order $\alpha>1$, but they are well-defined and useful for
$0<\alpha<1$.  A pedagogical introduction to R\'enyi powers and their
properties is in Ref.~\cite{tomamichel:qc2015a}. See Sect.~4.3 for the
sandwiched R\'enyi powers and Sect.~4.4 for the Petz R\'enyi powers.
The focus in Ref.~\cite{tomamichel:qc2015a} and most other references
is on R\'enyi divergences, which are entropic quantities obtained from
the R\'enyi powers, although many of the fundamental properties are
derived by an analysis of the latter.  The divergences share a set of
properties given in Sect.~4.1.1 and~4.1.2 of
Ref.~\cite{tomamichel:qc2015a} and labeled (I)-(X). The next lemmas
give properties of R\'enyi powers that we need. The Roman numerals in
the headings refer to the labels used in
Ref.~\cite{tomamichel:qc2015a} for related properties of R\'enyi
divergences.

\begin{lemma}\label{lem:petz_sandwiched}
  We have
  $\Ppow{\alpha}{\rho}{\sigma}\geq \Rpow{\alpha}{\rho}{\sigma}$.
\end{lemma}

\begin{proof}
  This follows from the Araki-Lieb-Thirring inequality
  $\tr(B^{\gamma}A^{\gamma}B^{\gamma})\geq \tr((BAB)^{\gamma})$ for all $\gamma\geq 1$,
  $A\geq 0$ and $B\geq 0$, where we set $\gamma=\alpha$, $A=\rho$ and $B=\sigma^{-\beta/(2\alpha)}$.
  See Ref.~\cite{bhatia:qc1997a}, Pg.~258.
\end{proof}

\begin{lemma}\label{lem:rp_continuity}
  \emph{(I) Continuity of R\'enyi powers.}  Suppose that
  $0<\rho\ll\sigma$.  The R\'enyi powers
  $\Rpow{\alpha}{\rho'}{\sigma'}$ and $\Ppow{\alpha}{\rho'}{\sigma'}$
  are continuous at $\rho'=\rho,\sigma'=\sigma$ in each of $\rho'$ and
  $\sigma'$.
\end{lemma}

Given appropriate conditions on the support of $\sigma$, joint
continuity also holds.

\begin{proof}
  For the sandwiched R\'enyi entropy this is shown in
  Ref.~\cite{mueller-lennert:qc2013a}, Sect.~IV.B.  For the Petz R\'enyi
  entropy, this is stated as an exercise at the end of Sect.~4.4.1 in
  Ref.~\cite{tomamichel:qc2015a}. 
\end{proof}

\begin{lemma}\label{lem:conditionmonotone_rp}
  \emph{(X) Dominance property of R\'enyi powers.}
  For $0\leq\rho\ll\sigma\leq\sigma'$,
  $\Rpow{\alpha}{\rho}{\sigma'}\leq \Rpow{\alpha}{\rho}{\sigma}$.
  If $\alpha\leq 2$, then 
  $\Ppow{\alpha}{\rho}{\sigma'}\leq \Ppow{\alpha}{\rho}{\sigma}$.
\end{lemma}

\begin{proof}
  The relevant arguments can be found in Sects.~4.3 and~4.4 of
  Ref.~\cite{tomamichel:qc2015a}.  \Pc{Here are the details. Since
    \begin{align}
      \Spec\left(\sigma^{-\beta/(2\alpha)}\rho\sigma^{-\beta/(2\alpha)}\right)
      &=\Spec\left(\sigma^{-\beta/(2\alpha)}\rho^{1/2}
        \left(\sigma^{-\beta/(2\alpha)}\rho^{1/2}\right)^{\dagger}\right)\notag\\
      &=\Spec\left(\left(\sigma^{-\beta/(2\alpha)}\rho^{1/2}\right)^{\dagger}
        \sigma^{-\beta/(2\alpha)}\rho^{1/2}\right)\notag\\
      &=\Spec\left(\rho^{1/2}\sigma^{-\beta/\alpha}\rho^{1/2}\right)
    \end{align}
    and $\tr(\xi^{\alpha})=\sum \Spec(\xi^{\alpha})=\sum
    \left(\Spec\xi\right)^{\alpha}$, we can write
    \begin{equation}
      \Rpow{\alpha}{\rho}{\sigma} =
      \tr(\left(\rho^{1/2}\sigma^{-\beta/\alpha}\rho^{1/2}\right)^{\alpha}).
    \end{equation}
    Since $0<\beta/\alpha<1$, the function $A\mapsto -A^{-\beta/\alpha}$
    is operator monotone for $A>0$ (Ref.~\cite{bhatia:qc1997a},
    Prop.~V.1.6 and Thm.~ V.1.9), as is $B\mapsto XBX^{\dagger}$
    (Ref.~\cite{bhatia:qc1997a}, Lem.~V.1.5).  For all $\alpha\geq 0$,
    $C\mapsto \tr(C^{\alpha})$ is monotone in $C$ for $C\geq 0$.
    (Ref.~\cite{carlen:qc2009a}, Thm.~2.10). When $\sigma$ has full
    support, the monotonicity property of the sandwiched R\'enyi power
    follows by composing these monotonicity properties with $A=\sigma$,
    $B=A^{-\beta/\alpha}$ $X=\rho^{1/2}$ and
    $C=\rho^{1/2}\sigma^{-\beta/\alpha}\rho^{1/2}$.  To deal with the
    case where $\sigma$ does not have full support, and $\rho\not=0$, we
    can use the trick of replacing $\sigma$ with $\sigma+\epsilon\one$
    and $\sigma'$ with $\sigma'+\epsilon\one$, let $\epsilon>0$ go to
    zero and invoke continuity of the R\'enyi powers.  If $\rho=0$,
    $\Rpow{\alpha}{\rho}{\sigma} = \Rpow{\alpha}{\rho}{\sigma'} = 0$.
    For the Petz R\'enyi power, write $\Ppow{\alpha}{\rho}{\sigma} =
    \tr(\rho^{\alpha/2}\sigma^{-\beta}\rho^{\alpha/2})$ and apply
    monotonicity of $A\mapsto -A^{-\beta}$ for $0<\beta\leq 1$.}
\end{proof}

\begin{lemma}\label{lem:rp_sumineq}
  Let $0\leq \rho_{i}$ and $\rho=\sum_{i}\rho_{i}$. Then
  \begin{equation}
    \sum_{i}\Rpow{\alpha}{\rho_{i}}{\rho}\leq \tr(\rho).
  \end{equation}
  If $\alpha\leq 2$, then
  \begin{equation}
    \sum_{i}\Ppow{\alpha}{\rho_{i}}{\rho}\leq \tr(\rho).
  \end{equation}
\end{lemma}

\begin{proof}
  By Lem.~\ref{lem:conditionmonotone_rp}, we have
  \begin{equation}
    \sum_{i}\Rpow{\alpha}{\rho_{i}}{\rho}
      \leq     \sum_{i}\Rpow{\alpha}{\rho_{i}}{\rho_{i}}
      = \sum_{i}\tr(\rho_{i})=\tr(\rho),
  \end{equation}
  and similarly for the Petz R\'enyi power when $\alpha\leq 2$.
\end{proof}

\begin{lemma}\label{lem:convex_rp}
  \emph{Log-convexity of R\'enyi powers:} For $0\leq\rho\ll\sigma$ the
  function $\alpha\mapsto \log(\Rpow{\alpha}{\rho}{\sigma})$ is
  convex, and so is
  $\alpha\mapsto\log(\Ppow{\alpha}{\rho}{\sigma})$.
\end{lemma}

\begin{proof}
  These are the first halves of Cor.~4.2, Pg.~56 (sandwiched R\'enyi
  power) and of Cor.~4.3, Pg.~62 (Petz R\'enyi power) of
  Ref.~\cite{tomamichel:qc2015a}.
\end{proof}

\begin{lemma}\label{lem:monotone_rp}
  \emph{Monotonicity of R\'enyi powers:} For $0\leq\rho\ll\sigma$ the
  function $\alpha\mapsto \hatRpow{\alpha}{\rho}{\sigma}^{1/\beta}$ is
  non-decreasing, and so is
  $\alpha\mapsto\hatPpow{\alpha}{\rho}{\sigma}^{1/\beta}$.
\end{lemma}

\begin{proof}
  These are the second halves of Cor.~4.2, Pg.~56 (sandwiched R\'enyi
  power) and Cor.~4.3, Pg.~62 (Petz R\'enyi power) of
  Ref.~\cite{tomamichel:qc2015a}.
\end{proof}

\begin{lemma}\label{lem:jconvex_rp}
  \emph{Joint convexity of R\'enyi powers:}
  The function $\rho,\sigma\mapsto \Rpow{\alpha}{\rho}{\sigma}$
  is jointly convex in $\rho$ and $\sigma$ on its domain, 
  and similarly for the Petz R\'enyi powers when $\alpha\leq 2$.
\end{lemma}

\begin{proof}
  For the sandwiched R\'enyi powers, see Prop.~3 of Ref.~\cite{frank:qc2013a}.
  For the Petz R\'enyi powers, this is Prop.~4.8, Pg.~61 in
  Ref.~\cite{tomamichel:qc2015a}.
\end{proof}

\begin{lemma}\label{lem:dataprocessing_rp}
  \emph{(VIII) Data-processing inequality for R\'enyi powers:} Let
  $\cE$ be a quantum operation and $0\leq\rho\ll\sigma$. Then
  $\Rpow{\alpha}{\cE(\rho)}{\cE(\sigma)}\leq
  \Rpow{\alpha}{\rho}{\sigma}$ and similarly for the Petz R\'enyi
  powers when $\alpha\leq 2$.
\end{lemma}

\begin{proof}
  For the sandwiched R\'enyi powers, see
  Ref.~\cite{frank:qc2013a,beigi:qc2013a}.  For the Petz R\'enyi
  powers, see Sect.~4.4.1 of Ref.~\cite{tomamichel:qc2015a}.
\end{proof}

\subsection{Quantum Relative Entropy}

Most of this work concerns estimation of R\'enyi powers so R\'enyi
entropies and divergences play a secondary role.  However, according
to the quantum asymptotic equipartition
property~\cite{tomamichel:qc2009a}, the asymptotic rate for randomness
generation is determined by quantum relative entropies. The quantum
relative entropy arises naturally as a limit of R\'enyi divergences.

Throughout this work, logarithms are base $e$ and entropies are
expressed in nits (the natural units of information) unless explicitly
specified otherwise. This simplifies calculus; conversion is only
needed when composing with extractors to specify the relationships
between certified conditional min-entropy and lengths of bit
strings. For results mentioning entropies, the conversion between nits
and bits usually just requires replacing log base $e$ with log base
$2$. Exceptions are the theorems of Sect.~\ref{subsec:handeat} stating
EAT and \QEFP bounds, which are not intended to be used in
applications.

\begin{definition}
  Let $0\leq \rho\ll\sigma$ and $\alpha>1$.  The \emph{sandwiched R\'enyi
  divergence of order $\alpha$ for $\rho$ given $\sigma$} is
  \begin{equation}
    \tildeDrel{\alpha}{\rho}{\sigma}
    = \frac{1}{\beta}\log(\hatRpow{\alpha}{\rho}{\sigma}).
  \end{equation}
  (This is Def.~4.3, Pg.~53 in Ref.~\cite{tomamichel:qc2015a}.)
\end{definition}

\begin{lemma}\label{lem:drellimit1} 
  Let $0< \rho\ll\sigma$. The limit of $\tildeDrel{\alpha}{\rho}{\sigma}$ as $\alpha\searrow 1$ exists
  and satisfies
  \begin{equation}
    \tildeDrel{1}{\rho}{\sigma} \defeq
    \lim_{\alpha\searrow 1}\tildeDrel{\alpha}{\rho}{\sigma}=
    \tr(\rho(\log(\rho)-\log(\sigma)))/\tr(\rho),
    \label{eq:relentlim1}
  \end{equation}
  which is the quantum relative entropy.  
\end{lemma}

\begin{proof}
  This is Prop.~4.5, Pg.~57 of Ref.~\cite{tomamichel:qc2015a}.
\end{proof}

\subsection{Min-Entropy}

Quantum min-entropy characterizes the randomness that is
available in a given system. We define the relevant quantities for the
family of classical-quantum states treated in this work, where $C$ is
the output CV, $Z$ is the input CV and $\Pfnt{E}$ is the system
containing the quantum side information.  We can instantiate these
variables in each context as we wish.  For example, we can consider
the situation where we let $Z$ be a trivial CV, which is equivalent to
just leaving it out.

\begin{definition}
  Let $\rho(CZ)\in\cS_{\leq 1}(CZ\Pfnt{E})$.  Then \emph{$\rho(CZ)$ has max-prob
    $p$ given $Z\Pfnt{E}$} if there exists $\sigma(Z)\in
  \cS_{1}(Z\Pfnt{E})$ such that $\rho(CZ)\leq p\sigma(Z)$.  The
  \emph{exact max-prob of $\rho(CZ)$ given $Z\Pfnt{E}$} is
  \begin{equation}
    P_{\max}(\rho(CZ)|Z\Pfnt{E}) = \inf\left\{p:\textrm{there exists  
      $\sigma(Z)\in\cS_{1}(Z\Pfnt{E})$ such that
      $\rho(CZ)\leq p\sigma(Z)$}\right\}.
  \end{equation}
  The quantity $H_{\infty}(\rho(CZ)|Z\Pfnt{E})=-\log(P_{\max}(\rho(CZ)|Z\Pfnt{E}))$
  is called the \emph{conditional min-entropy of $\rho(CZ)$ given $Z\Pfnt{E}$}.
\end{definition}
When writing conditional quantities like $P_{\max}$, we put the state
with its CV arguments first. The conditioned systems are always
classical and consist of every CV that does not occur in the
conditioner.

We need a lemma to switch between conditioning on a CV and conditioning
on its quantization.

\begin{lemma}\label{lem:pmaxcz}
  Let $\rho(CZ)\in\cS_{1}(CZ\Pfnt{E})$ and
  define $\tau(C)=\sum_{z}\hat z\otimes \rho(Cz)\in\cS_{1}(C\Pfnt{ZE})$.
  Then $P_{\max}(\rho(CZ)|Z\Pfnt{E})=P_{\max}(\tau(C)|\Pfnt{ZE})$.
\end{lemma}

\begin{proof}
  For $\sigma(Z)\in\cS_{1}(Z\Pfnt{E})$ such that $\rho(CZ)\leq
  p\sigma(Z)$, we have that $\tau(C)\leq p \sum_{z}\hat z\otimes\sigma(z)$. This
  implies that $P_{\max}(\rho(CZ)|Z\Pfnt{E})\geq
  P_{\max}(\tau(C)|\Pfnt{ZE})$.  For the reverse inequality, consider
  $\sigma'\in\cS_{1}(\Pfnt{ZE})$ such that $\tau(C)\leq p\sigma'$.
  Since the map $\xi\mapsto \hat z\xi\hat z$ is positive, it preserves
  operator ordering and $\hat z\otimes \rho(Cz)=\hat z\tau(C)\hat z\leq
  p\hat z\sigma'\hat z$.  With $\sigma(Z)$ defined
  by $\sigma(z)=\tr_{\Pfnt{Z}}\hat
  z\sigma'\hat z$, it follows that $\rho(CZ)\leq p\sigma(Z)$.
\end{proof}

\begin{definition}
  Let $\rho(CZ)\in\cS_{1}(CZ\Pfnt{E})$.  System $Z\Pfnt{E}$'s
  \emph{guessing probability for $\rho(CZ)$} is
  \begin{equation}
    G_{\max}(\rho(CZ)|Z\Pfnt{E}) =
    \sup\left\{\sum_{cz}\tr(P_{c|z}\rho(cz)):\textrm{for all $z$
      $(P_{c|z})_{c}$ is a POVM}\right\}.
  \end{equation}
\end{definition}

\begin{lemma}\label{lem:maxprob_guess}
  Let $\rho(CZ)\in\cS_{1}(CZ\Pfnt{E})$. Then
  \begin{equation}
    G_{\max}(\rho(CZ)|Z\Pfnt{E})=P_{\max}(\rho(CZ)|Z\Pfnt{E}).
  \end{equation}
\end{lemma}

\begin{proof}
  Let $\tau(C)=\sum_{z}\hat z\otimes\rho(Cz)$.  According to
  Ref.~\cite{koenig:qc2009a}, Thm.~1,
  $P_{\max}(\tau(C)|\Pfnt{Z}\Pfnt{E})=G_{\max}(\tau(C)|\Pfnt{Z}\Pfnt{E})$.
  According to Lem.~\ref{lem:pmaxcz} it suffices to show that
  $G_{\max}(\tau(C)|\Pfnt{Z}\Pfnt{E})=G_{\max}(\rho(CZ)|Z\Pfnt{E})$.
  Let $(P_{c|z})_{c}$ be $z$-indexed POVMs. Then 
  $P_{c}=(\sum_{z}\hat z\otimes P_{c|z})_{c}$
  is a POVM, and 
  \begin{equation}
    \sum_{cz}\tr(P_{c|z}\rho(cz))=\sum_{c}\tr(\tr_{\Pfnt{Z}}(P_{c}\tau(c)))=
    \sum_{c}\tr(P_{c}\tau(c)),
  \end{equation}
  from which it follows that $G_{\max}(\tau(C)|\Pfnt{Z}\Pfnt{E})\geq
  G_{\max}(\rho(CZ)|Z\Pfnt{E})$.
  For the reverse inequality, let
  $(P_{c})_{c}$ be a POVM on $\Pfnt{ZE}$. We have
  \begin{align}
    \sum_{c}\tr(P_{c}\tau(c)) =
      \sum_{cz}\tr(P_{c}(\hat z\otimes\rho(cz)))
      = \sum_{cz}\tr(\tr_{\Pfnt{Z}}(
      P_{c}(\hat z\otimes\one_{\Pfnt{E}}))\rho(cz)).\label{eq:lem:maxprob_guess:1}
  \end{align}
  Let $P_{c|z}=\tr_{\Pfnt{Z}}( P_{c}(\hat z\otimes\one_{\Pfnt{E}}))
  =\tr_{\Pfnt{Z}}((\hat z\otimes\one_{\Pfnt{E}}) P_{c}(\hat
  z\otimes\one_{\Pfnt{E}}))$. Then $(P_{c|z})_{c}$ is a POVM for each
  $z$, and Eq.~\ref{eq:lem:maxprob_guess:1} and arbitrariness of
  $(P_{c})_{c}$ implies that $G_{\max}(\tau(C)|\Pfnt{Z}\Pfnt{E})\leq
  G_{\max}(\rho(CZ)|Z\Pfnt{E})$.
\end{proof}

\begin{definition}
  Let $\rho(CZ)\in\cS_{1}(CZ\Pfnt{E})$.  The \emph{conditional entropy of
  $\rho(CZ)$ given $Z\Pfnt{E}$} is
  \begin{equation}
    H_{1}(\rho(CZ)|Z\Pfnt{E})=
    -\sum_{cz}\tr\left(\rho(cz)\big(\log(\rho(cz))-\log(\rho(z))\big)\right)
    = -\sum_{cz}\tr(\rho(cz))\tildeDrel{1}{\rho(cz)}{\rho(z)}.
  \end{equation}
\end{definition}

\begin{lemma}\label{lem:minent_ent}
  $H_{1}(\rho(CZ)|Z\Pfnt{E})\geq H_{\infty}(\rho(CZ)|Z\Pfnt{E})$.
\end{lemma}

\begin{proof}
  Define $\tilde H_{\alpha}(\rho(CZ)|Z\Pfnt{E})
  =-\sum_{cz}\tr(\rho(cz))\tildeDrel{\alpha}{\rho(cz)}{\rho(z)}$.  The
  lemma follows from $H_{1}=\lim_{\alpha\searrow 1}\tilde H_{\alpha}$,
  $H_{\infty}=\lim_{\alpha\nearrow \infty}\tilde H_{\alpha}$ and monotonicity of
  $\tilde H_{\alpha}$ in $\alpha$.  These facts can
  be found in Ref.~\cite{tomamichel:qc2015a}.  The first limit is an
  application of Lem.~\ref{lem:drellimit1}.  For the second limit, see
  Ref.~\cite{tomamichel:qc2015a}, Def.~4.2, Pg.~52 and the comment at the
  beginning of Sect.~4.3.2.  That $\tilde H_{\alpha}$ is
  non-increasing in $\alpha$ follows from
  Lem.~\ref{lem:monotone_rp}. 
\end{proof}

\subsection{Smooth Min-Entropy}

\begin{definition}\label{def:smooth_max_prob}
  Let $\rho(CZ)\in\cS_{1}(CZ\Pfnt{E})$. Then
  \emph{$\rho(CZ)$ has $\epsilon$-smooth max-prob $p$ given
    $Z\Pfnt{E}$}
  if there exists a $\rho'(CZ)\in\cS_{\leq 1}(CZ\Pfnt{E})$  with
  $\purdist{\rho(CZ)}{\rho'(CZ)}\leq \epsilon$ and
  $P_{\max}(\rho'(CZ)|Z\Pfnt{E})\leq p$.
  The \emph{exact $\epsilon$-smooth max-prob of $C$ 
  given $Z\Pfnt{E}$ at $\rho(CZ)$} is
  \begin{equation}
    P^{\epsilon}_{\max}(\rho(CZ)|Z\Pfnt{E}) =
    \inf\{P_{\max}(\rho'(CZ)|Z\Pfnt{E}): \rho'(CZ)\in\cS_{\leq 1}(CZ\Pfnt{E}),
    \purdist{\rho'(CZ)}{\rho(CZ)}\leq\epsilon\}.
  \end{equation}
  The quantity $H_{\infty}^{\epsilon}(\rho(CZ)|Z\Pfnt{E}) = -\log(
  P^{\epsilon}_{\max}(\rho(CZ)|Z\Pfnt{E}))$ is called the \emph{smooth
    conditional min-entropy of $\rho(CZ)$ given $Z\Pfnt{E}$}. Here,
  the smoothing is with respect to the purified distance, as in
  Refs.~\cite{koenig:qc2009a, tomamichel:qc2012a}.
\end{definition}

For relevant cases, the witnesss $\rho'(CZ)$ in the definition of
$\epsilon$-smooth max-prob can be assumed to be normalized
states. This observation is formalized by the next lemma.

\begin{lemma}\label{lem:purdist_normalized}
  Suppose that $\rho(CZ)\in\cS_{1}(CZ\Pfnt{E})$ has $\epsilon$-smooth max-prob
  $p$ given $Z\Pfnt{E}$ with $p|\Rng(C)|\geq 1$. Then
  there exists $\rho''(CZ)\in\cS_{1}(CZ\Pfnt{E})$ such that
  $\purdist{\rho(CZ)}{\rho''(CZ)}\leq \epsilon$ and
  $P_{\max}(\rho''(CZ)|Z\Pfnt{E})\leq p$.
\end{lemma}

\begin{proof}
  Let $\rho'(CZ)\in\cS_{\leq 1}(CZ\Pfnt{E})$ and
  $\sigma(Z)\in\cS_{1}(Z\Pfnt{E})$ such that
  $\purdist{\rho(CZ)}{\rho'(CZ)}\leq \epsilon$ and $\rho'(CZ)\leq
  p\sigma(Z)$.  By Lem.~\ref{lem:pmax_fillin}, there exists
  $\rho''(CZ)\in\cS_{1}(CZ\Pfnt{E})$ such that
  $\rho'(CZ)\leq\rho''(CZ)$ and $\rho''(CZ)\leq p\sigma(Z)$. 
  So $P_{\max}(\rho''(CZ)|Z\Pfnt{E})\leq p$, and by
  Lem.~\ref{lem:purdistprops}, 
  \begin{equation}
    \purdist{\rho(CZ)}{\rho''(CZ)}\leq
    \purdist{\rho(CZ)}{\rho'(CZ)}\leq\epsilon.
  \end{equation}  
\end{proof}

\begin{definition}\label{def:tvsmooth_max_prob}
Let $\rho(CZ)\in\cS_{1}(CZ\Pfnt{E})$.  Then \emph{$\rho(CZ)$ has
    $\TV{:}\epsilon$-smooth max-prob $p$ given $Z\Pfnt{E}$} if there
  exists a $\rho'(CZ)\in\cS_{1}(CZ\Pfnt{E})$ with
  $\trdist{\rho(CZ)}{\rho'(CZ)}\leq \epsilon$ and
  $P_{\max}(\rho'(CZ)|Z\Pfnt{E})\leq p$.  The \emph{$\TV$:exact
    $\epsilon$-smooth max-prob of $C$ given $Z\Pfnt{E}$ 
  at $\rho(CZ)$} is
  \begin{equation}
    P^{\TV{:}\epsilon}_{\max}(\rho(CZ)|Z\Pfnt{E}) =
    \inf\{P_{\max}(\rho'(CZ)|Z\Pfnt{E}): \rho'(CZ)\in\cS_{1}(CZ\Pfnt{E}),
    \trdist{\rho'(CZ)}{\rho(CZ)}\leq\epsilon\}.
  \end{equation}
  The quantity $H_{\infty}^{\TV{:}\epsilon}(\rho(CZ)|Z\Pfnt{E}) =
  -\log( P^{\TV{:}\epsilon}_{\max}(\rho(CZ)|Z\Pfnt{E}))$ is called the
  \emph{TV:smooth conditional min-entropy of $\rho(CZ)$ given
    $Z\Pfnt{E}$}. Here, the smoothing is with respect to the TV
  distance, as first proposed in Ref.~\cite{renner:qc2005}.
\end{definition}
We remark that the definitions are monotonic in the smoothness
parameter $\epsilon$.  For example, if
$P^{\epsilon}_{\max}(\rho(CZ)|Z\Pfnt{E})\leq p$ and
$\epsilon'>\epsilon$, then
$P^{\epsilon'}_{\max}(\rho(CZ)|Z\Pfnt{E})\leq p$. Besides using $\TV$
distance instead of purified distance, the second definition requires
that the state being compared is normalized.  This is unproblematic
for max-prob bounds greater than $1/|\Rng(C)|$, and smaller bounds are
generally not helpful, see the next lemma.  As explained in
Ref.~\cite{tomamichel:qc2012a}, when dealing with quantum information,
purified distance is preferred and the fact that it exceeds TV
distance means that there are few complications when chaining with
classical protocols or extractors, or for interpreting results in
familiar probabilistic terms.

We can readily switch from smoothing with purified distance to
smoothing with $\TV$ distance by applying the next lemma. Switching in
the other direction involves a square-root increase of smoothing
parameter; we do not consider this switch here.

\begin{lemma}
  \label{lem:pmax_tv_from_pur}
  Let $\rho(CZ)\in\cS_{1}(CZ\Pfnt{E})$ have $\epsilon$-smooth max-prob
  $p$ given $Z\Pfnt{E}$ with $p|\Rng(C)|\geq 1$. Then
  $P^{\TV{:}\epsilon}_{\max}(\rho(CZ)|Z\Pfnt{E})\leq p$.
  It follows that
  \begin{equation}
    P_{\max}^{\TV{:}\epsilon}(\rho(CZ)|Z\Pfnt{E})\leq
    \max(P_{\max}^{\epsilon}(\rho(CZ)|Z\Pfnt{E}),1/|\Rng(C)|).\label{eq:hinfpur_tv}
  \end{equation}
\end{lemma}

\begin{proof}
  By Lem.~\ref{lem:purdist_normalized},
  there exists $\rho''(CZ)\in\cS_{1}(CZ\Pfnt{E})$ such that
  $\purdist{\rho(CZ)}{\rho''(CZ)}\leq \epsilon$ and
  $P_{\max}(\rho''(CZ)|Z\Pfnt{E})\leq p$.
  By
  Lem.~\ref{lem:purdistprops}, 
  \begin{equation}
    \trdist{\rho(CZ)}{\rho''(CZ)}\leq
      \purdist{\rho(CZ)}{\rho''(CZ)}.
  \end{equation}
  Hence $P^{\TV{:}\epsilon}_{\max}(\rho(CZ)|Z\Pfnt{E})\leq p$. 
  For Eq.~\ref{eq:hinfpur_tv}, we set
  $p=\max(1/|\Rng(C)|,P_{\max}^{\epsilon}(\rho(CZ)|Z\Pfnt{E}))$
  and apply the result just proven.
\end{proof}

\begin{lemma}\label{lem:addcondition}
  Let $\rho(CZ)\in\cS_{1}(CZ\Pfnt{E})$. Then
  \begin{equation}
    P^{\epsilon}_{\max}(\rho(CZ)|Z\Pfnt{E})\leq
    |\Rng(Z)|\;P^{\epsilon}_{\max}(\rho(CZ)|\Pfnt{E}).
  \end{equation}
\end{lemma}

\begin{proof}
  This is Lem.~6.8, Pg. 95 of Ref.~\cite{tomamichel:qc2015a}.  Consider
  an arbitrary $p>P_{\max}^{\epsilon}(\rho(CZ)|\Pfnt{E})$. Then there
  exist $\rho'(CZ)\in \cS_{\leq 1}(CZ\Pfnt{E})$ and
  $\tau\in\cS_{1}(\Pfnt{E})$ such that
  $\purdist{\rho(CZ)}{\rho'(CZ)}\leq\epsilon$ and $\rho'(CZ)\leq
  p\tau$, which we can rewrite as $\rho'(CZ) \leq p|\Rng(Z)|\;
  \tau/|\Rng(Z)|$. Define $\sigma(Z)=\tau/|\Rng(Z)|$, which is in
  $\cS_{1}(Z\Pfnt{E})$.  Therefore $\rho'(CZ)$ and $\sigma(Z)$ witness
  that $P_{\max}^{\epsilon}(\rho(CZ)|Z\Pfnt{E})\leq |\Rng(Z)|p$.
  Letting $p\searrow P_{\max}^{\epsilon}(\rho(CZ)|Z\Pfnt{E})$ proves
  the lemma.
\end{proof}

\begin{lemma}\label{lem:pmaxdetermined}
  Let $\rho(CZH)\in\cS_{1}(CZH\Pfnt{E})$, and suppose that $Z=Z(H)$ is
  determined by $H$, then $P^{\epsilon}_{\max}(\rho(CH)|\Pfnt{E}) \leq
  P^{\epsilon}_{\max}(\rho(CZ)|\Pfnt{E})$ and
  $P^{\epsilon}_{\max}(\rho(CZ)|Z\Pfnt{E})\leq
  P^{\epsilon}_{\max}(\rho(CH)|H\Pfnt{E})$.
\end{lemma}

\begin{proof}
  These are instances of data-processing inequalities for smooth
  conditional min-entropy.  Since $Z$ is determined by $H$, the first
  statement is a consequence of Prop.~6.4, Pg.~96 of
  Ref.~\cite{tomamichel:qc2015a}, according to which applying a
  function to a classical system does not increase the
  $\epsilon$-smooth conditional min-entropy of the system conditional on other
  systems.  For the second, the transformation $h\mapsto Z(h)$ can be
  considered as a $\CPTP$ map of system $H$ to the system $Z$, where
  these are the systems in the conditioners of the smooth max-probs
  being compared.  The inequality is therefore obtained from Thm.~6.2,
  Pg.~95 of Ref.~\cite{tomamichel:qc2015a}, according to which a
  $\CPTP$ process applied to the conditioning system does not decrease
  the $\epsilon$-smooth conditional min-entropy.
\end{proof}

\subsection{Extractors}

For randomness generation protocols, we assume that
a quantum-proof strong extractor $\cE$ is available.

\begin{definition} \label{def:extractor}
  Let $C$, $S$ and $R$ be CVs. Define $n=\log_{2}(|\Rng(C)|)$,
  $k_{s}=\log_{2}(|\Rng(S)|)$ and $k_{o}=\log_{2}(|\Rng(R)|)$.  Here
  $S$ is a seed CV with probability distribution $\mu(S)=\Unif(S)$ and
  independent of all other systems.  Consider a function $\cE:(C,S;n,
  k_{s},k_{o},k_{i},\epsx)\mapsto \Rng(R)$. Define
  $\bar\cE:c,s\mapsto\cE(c,s)s$, where the parameters are implicit.
  The function $\cE$ is a \emph{quantum-proof strong extractor with
    parameters $(n,k_{s},k_{o},k_{i},\epsx)$} if for every
  $\rho(CS)\in\cS_{1}(CS\Pfnt{E})$  of the form
  $\rho(CS)=\rho(C)\Unif(S)$ that satisfies
  $P_{\max}(\rho(C)|\Pfnt{E})\leq 2^{-k_{i}}$, the extractor and seed
  output $\bar\cE$ is close to uniform and independent of $\Pfnt{E}$
  with distance
  \begin{equation}
    \purdist{\rho(\bar\cE)}{\Unif(RS)\rho}\leq \epsx.
  \end{equation}
\end{definition}
This definition of quantum-proof extractors differs from
others such as Ref.~\cite{mauerer:2012} by requiring small purified
distance instead of small $\TV$ distance. With this change we can take
advantage of extensions to previously traced-out quantum systems.

In this work, we use the term \emph{extractor} to refer to a function
$\cE$ that is a quantum-proof strong extractor provided that the
parameters $(n,k_{s},k_{o},k_{i},\epsx)$ satisfy constraints that we
refer to as the \emph{extractor constraints}.  (The convention in this
manuscript for parameters and their ordering differs from that in
Ref.~\cite{knill:qc2017a}.)  We assume that the extractor constraints
include the conditions $1\leq k_{i}\leq n$, $k_{s}\geq 0$, $k_{o}\leq
k_{i}$, and $0<\epsx\leq 1$. We generally deal with bit strings $C$,
$S$ and $R$, so we also assume that $n$, $k_{s}$ and $k_{o}$ are
integers.

A specific quantum-proof strong extractor with reasonably low seed
requirements is the TMPS extractor based on Ref.~\cite{mauerer:2012},
which we applied in Ref.~\cite{bierhorst:qc2017a} using the
implementation available at
\url{https://github.com/usnistgov/libtrevisan}.  Simplified
constraints for this extractor include $2\leq k_{o}\leq k_{i}\leq n$
and
\begin{align}
  k_{o}+4\log_{2}(k_{o}) &\le k_{i} -  4\log_2(1/\delta_{x} )-6, \notag\\
  k_{s} &\ge 36 \log_{2}(k_{o}) (\log_2(4nk_{o}^2/\delta_{x} ^2))^2.
  \label{eq:TMPS_bounds}
\end{align}
Here, $\delta_{x}$ is the desired error in terms of $\TV$ distance.
To ensure that the purified distance is at most $\epsx$, we set
$\delta_{x}=\epsx^{2}/2$, see Lem.~\ref{lem:purdistprops}. 
See Ref.~\cite{bierhorst:qc2017a} for the smaller expression for $k_{s}$
in terms of $\delta_{x}$
used by the implementation. Better extractors exist in theory,
but full implementations are still rare.

\section{Models}
\label{sec:qmodels}

\subsection{Definitions}

\begin{definition}
  A \emph{model $\cC(U)$ for $U\Pfnt{E}$} is a subset of
  $\cS(U\Pfnt{E})$ closed under multiplication by non-negative real
  numbers.  The set of normalized distributions in $\cC(U)$ is
  $\cN(\cC(U))=\{\rho(U)\in\cC:\tr(\rho)=1\}$.  
  The model $\cC(U)$ is \emph{null}
  if its only member is the zero distribution
  given by $u\mapsto 0$. 
\end{definition} 
If $\cC(U)$ is not null, we can
reconstruct $\cC(U)$ from $\cN(\cC(U))$ by
$\cC(U)=[0,\infty)\cN(\cC(U))$.  We normally omit ``for $U\Pfnt{E}$''
when introducing a model. In this case, the default quantum system is
$\Pfnt{E}$.

Expressions of the form $\cC(U)$ with $U$ a CV are reserved for
models.  We may subscript $\cC$ to distinguish models in context.  The
notation $\cC(U)$ indicates the CV or CVs that the members of the
model depend on and does not indicate function application or a CV
construction.  We adapt the marginalization conventions for
state-valued distributions for models. Thus if $\cC(UV)$ is a model,
then $\cC(U)=\{\rho(U):\rho(UV)\in\cC(UV)\}$ and
$\cC=\{\rho:\rho(UV)\in\cC(UV)\}$. When a model is expressed in terms
a compound construction such as $\cE(\cC(UV))$, $\cM(\cC(UV);\ldots)$
or $\cC(U)\circ\cC_{U}(V)$ without a final CV argument, the
marginalization conventions do not apply.

A \emph{classical model} $\cC(U)$ is a model for $U$, which means that
the quantum system is trivial and the model consists of a set of
unnormalized distributions on $U$. In this case, $\cN(\cC(U))$
consists of probability distributions and is a standard statistical
model.  For any model $\cC(U)$, $\tr(\cC(U))$ is a classical model.

We consider several closure properties and operations on models.
First we define $V$-conditional quantum operations on $\cB(\Pfnt{E})$
as a family $\cE_{V}$ of $v$-dependent quantum operations $\cE_{v}$ on
$\cB(\Pfnt{E})$.  As an operation, $\cE_{V}$ transforms members of
$\cS(UV\Pfnt{E})$ according to $\cE_{V}:\rho(uv)\mapsto
\cE_{v}(\rho(uv))$.  Among the many closure properties that can be
satisfied by models, we distinguish the following:

\begin{definition}
  The model $\cC(U)$ is \emph{closed under the linear map
    $\cE:\cB(\Pfnt{E})\rightarrow \cB(\Pfnt{E})$} if $\cE(\cC(U))
  \subseteq \cC(U)$.  $\cC(U)$ is \emph{$\pCP$-closed} if it is
  closed under $\pCP$ maps, \emph{$\CP$-closed} if it is
  closed under $\CP$ maps and \emph{$\CPTP$-closed} if it is closed under
  trace-preserving $\CP$ maps.  The model $\cC(UV)$ is \emph{closed under
    $V$-conditional quantum operations} if $\cE_{V}(\cC(UV))\subseteq
  \cC(UV)$ for every $V$-conditional quantum operation $\cE_{V}$.
\end{definition}
In this work, many results are established under the condition that 
the model involved is $\pCP$-closed. As $\pCP$ maps are special $\CP$ maps
and closure under $\pCP$ maps is weaker than closure under $\CP$ maps,
these results automatically apply if the model is  $\CP$-closed.

We may also consider closedness under special families of $\CP$ maps,
for instance the family of $\CP$ maps that preserve the projectors of
a partition of unity.  For each closedness property in the definitions
above, there is a corresponding closure operation. We use suggestive
notation for closure operations.  For example $\Cvx(\cC(U))$ is
convex closure, $\pCP(\cC(U))$ is closure under $\pCP$ maps, and
$\CPTP_{V}(\cC(UV))$ is closure under $V$-conditional $\CPTP$ maps.

\subsection{General Constructions}

Models $\cC(U)$ arise from constraints on the physical processes that result in
the distributions $\rho(U)$ in $\cC(U)$.  It is possible to associate quantum models  
to classical models.
\begin{definition}\label{def:maxextension}
  Let $\cC(U)$ be a classical model. Then the
  \emph{maximal extension of $\cC(U)$ to $\Pfnt{E}$} is defined as
  \begin{equation}
    \cM(\cC(U);\Pfnt{E})=\{\rho(U):
    \tr(\sigma\rho(U))\in\cC(U)\textrm{\ for all $\sigma\in\cS(\Pfnt{E})$}\}.
    \label{eq:maxextmodelsdef}
  \end{equation}
\end{definition}
In this definition, if $\cC(U)$ is convex closed, one can restrict
$\sigma$ to pure states when verifying membership in $\cM(\cC(U);\Pfnt{E})$
according to Eq.~\ref{eq:maxextmodelsdef}. 

\begin{lemma}\label{lem:maxextend_is_CPclosed}
  If $\cC(U)$ is a classical model, then $\cM(\cC(U);\Pfnt{E})$ is $\CP$-closed.
\end{lemma}

\begin{proof}
  Let $\rho(U)\in\cM(\cC(U);\Pfnt{E})$ and let $\cE:\tau\mapsto
  \sum_{i}A_{i}\tau {A_{i}}^{\dagger}$ be a $\CP$ map. 
  Given $\sigma\in\cS(\Pfnt{E})$, let
  $\chi = \sum_{i}{A_{i}}^{\dagger}\sigma A_{i}\in\cS(\Pfnt{E})$ and evaluate
  \begin{align}
    \tr(\sigma\cE(\rho(U))) &=
    \tr(\sigma\sum_{i} A_{i}\rho(U) {A_{i}}^{\dagger})\notag\\
    &= \sum_{i}\tr(\sigma A_{i}\rho(U){A_{i}}^{\dagger})\notag\\
    &= \sum_{i}\tr({A_{i}}^{\dagger}\sigma A_{i}\rho(U))\notag\\
    &= \tr(\sum_{i}{A_{i}}^{\dagger}\sigma A_{i}\rho(U))\notag\\
    &= \tr(\chi\rho(U)) \in\cC(U).
  \end{align}
  Since $\sigma\in\cS(\Pfnt{E})$ is
  arbitrary, it follows that $\cE(\rho(U))\in\cM(\cC(U);\Pfnt{E})$.
\end{proof}

If $\cC(U)$ is the classical model arising from a Bell-test configuration with
only non-signaling assumptions and no additional quantum constraints,
then the maximal extension of $\cC(U)$ to $\Pfnt{E}$ makes no physical
assumptions on the protocol devices other than non-signaling and
therefore allows the devices to exhibit super-quantum
correlations. The models obtained when the devices and $\Pfnt{E}$ are
jointly quantum are more constrained. They arise from families of
POVMs as follows.
\begin{definition}\label{def:inducedmodel}
  Let $\fkP(U)$ be a family of POVMs of $\Pfnt{D}$ with outcomes $U$.
  The model for $U\Pfnt{E}$ \emph{induced by $\fkP(U)$} is defined by
  \begin{equation}
    \cM(\fkP(U);\Pfnt{E})=
    \{\cP(\sigma)(U):\sigma\in\cS(\Pfnt{DE}),\cP\in\fkP(U)\}.
    \label{eq:inducedmodelsdef}
  \end{equation}
\end{definition}
Expressions of the form $\fkP(U)$ with $U$ a CV are reserved for
families of POVMs.  The notation $\fkP(U)$ indicates the outcome CV of
the members and does not indicate function application or a CV
construction.  We may subscript $\fkP$ to distinguish families in
context.  If \(\cC(U)\) is an induced model, then the maximal extension
of \(\tr(\cC(U))\) contains \(\cC(U)\). On the other hand, 
for Bell-test configurations, adding all quantum constraints
to a classical non-signaling model $\cC(V)$ and constructing the
maximal extension of $\cC(V)$ need not be equivalent to inducing a
model from a suitably constrained set of POVMs.  Further research is
required to explore the relationships between maximal extensions and
induced models.

\begin{lemma}\label{lem:induced_cpclosed}
  For any family $\fkP(U)$ of POVMs of $\Pfnt{D}$ with outcomes $U$, 
  the induced model $\cM(\fkP)$ is $\CP$-closed.
\end{lemma}

\begin{proof}
  It suffices to observe that by definition, $\CP$ maps on $\cS(\Pfnt{E})$
  preserve $\cS(\Pfnt{E})$, and POVMs of $\Pfnt{D}$ with outcomes $U$
  commute with $\CP$ maps on $\cS(\Pfnt{E})$.
\end{proof}

For induced models, $\Pfnt{D}$ consists of the devices used by a
protocol and the POVMs can be constrained by partial trust in device
behavior. For example, in many situations, the trust involves
assumptions that $Z$ is an input with known probability distribution,
and that there exists a system decomposition of the devices according
to protocol parties, with the POVMs acting independently on the
subsystems.  In partially device-dependent applications, one may also
trust the form of the specific measurements or the dimensions of
the subsystems. It is possible to generalize the definition of induced
models by restricting the measured states of $\cS(\Pfnt{DE})$ to a
model of $\Pfnt{DE}$ or of $W\Pfnt{DE}$ for some CV $W$.

Both maximal extensions and induced models are defined uniformly,
independent of the dimension of $\Pfnt{E}$. We can take the state
space of $\Pfnt{E}$ to be an infinite dimensional Hilbert space, but
according to our finiteness assumptions, we restrict to states
with finite support.

\subsection{Chaining Models}

\begin{definition}\label{def:chaining}
  Let $\cC(U)$ be a model for $U\Pfnt{E}$ and for each $u$,
  let $\cC_{u}(V)$ be model for $V\Pfnt{E}$.  We write
  $\cC_{U}(V)$ for the $u$-indexed family of models
  consisting of the $\cC_{u}(V)$.
  The result of \emph{chaining $\cC(U)$ and $\cC_{U}(V)$} is the
  model for $UV\Pfnt{E}$ defined by
  \begin{equation}
    \cC(U)\circ\cC_{U}(V) = 
    \{\rho(UV):
    \textrm{$\rho(U)\in\cC(U)$ and for all $u$,
      $\rho(uV)\in\cC_{u}(V)$}\}.
  \end{equation}
\end{definition}
Chained models can be null unless $\cC(U)$ and $\cC_{U}(V)$ are
sufficiently rich.

The next lemma shows that quantum operations distribute over chaining.
\begin{lemma}
  Let $\cC(U)$ be a model for $U\Pfnt{E}$, $\cC_{U}(V)$ a family of models
  for $V\Pfnt{E}$ and $\cE:\cB(\Pfnt{E})\rightarrow\cB(\Pfnt{E})$ a
  positive linear map. Then $\cE(\cC(U)\circ\cC_{U}(V))\subseteq
  \cE(\cC(U))\circ\cE(\cC_{U}(V))$.  In particular, if $\cC(U)$ and the
  $\cC_{u}(V)$ are closed under $\cE$, then so is $\cC(U)\circ\cC_{U}(V)$.
\end{lemma}

\begin{proof}
  Let $\rho(UV)\in\cC(U)\circ\cC_{U}(V)$ and consider
  $\rho'(UV)=\cE(\rho(UV))$. Since $\rho'(U)=\cE(\rho(U))$, we have
  $\rho'(U)\in\cE(\cC(U)$).  Similarly, for each $u$,
  $\rho'(uV)=\cE(\rho(uV))\in\cE(\cC_{u}(V))$.  It follows that
  $\rho'(UV)\in\cE(\cC(U))\circ\cE(\cC_{U}(V))$.
\end{proof}

When the CV over which a model is defined consists of inputs and
outputs where we later condition on the inputs, we need to restrict
the composed models so that future inputs are effectively independent
of the past outputs given $\Pfnt{E}$ and the past inputs. Because
$\Pfnt{E}$ is quantum, this is formulated by means of a short quantum
Markov chain.  In the next definition, $\Sfnt{CZ}$ and $CZ$ are
separate CVs with no relationship assumed. In an experiment consisting
of a sequence of trials, $\Sfnt{CZ}$ are the outputs and inputs of the
trials so far, and $CZ$ is the output and input of the next trial.
\begin{definition}\label{def:condchaining}
  Let $\cC(\Sfnt{CZ})$ be a model for $\Sfnt{CZ}\Pfnt{E}$ and
  $\cC_{\Sfnt{CZ}}(CZ)$ a family of models for $CZ\Pfnt{E}$.  The set
  of models obtained by \emph{chaining $\cC(\Sfnt{CZ})$ and
    $\cC_{\Sfnt{CZ}}(CZ)$ with conditionally independent inputs} is
  written as $\cC(\Sfnt{CZ}) \circ_{Z|\Sfnt{Z}} \cC_{\Sfnt{CZ}}(CZ)$
  and consists of the members $\rho(\Sfnt{CZ}CZ)$ of $\cC(\Sfnt{CZ}) \circ
  \cC_{\Sfnt{CZ}}(CZ)$ such that $\rho(\Sfnt{CZ}Z)\in Z\leftrightarrow
  \Sfnt{Z}\Pfnt{E}\leftrightarrow \Sfnt{C}$.
\end{definition}

\subsection{Input-Output Models}

When considering models for $CZ\Pfnt{E}$, $Z$ is normally an input CV
that can be freely chosen in some sense. We may expect conditional
  distributions of $C$ given $Z=z$ are independent of $z$.  For
classical side information, this idea was captured with some
generality by models that are free for $Z$ in
Ref.~\cite{knill:qc2017a}.  For quantum side information, the
conditional constraints are captured by models for $(C|Z)\Pfnt{E}$
according to the next definition.

\begin{definition}
  $\cC(C|Z)$ is a \emph{model for $(C|Z)\Pfnt{E}$} if
  $\cC(C|Z)\subseteq \cS((C|Z)\Pfnt{E})$ and $\cC(C|Z)$ is closed
  under multiplication by non-negative real numbers.
\end{definition}
By default, the quantum system for $\cC(C|Z)$ is $\Pfnt{E}$
and we normally omit the phrase ``for $(C|Z)\Pfnt{E}$''.

If $\cC(Z)$ is a classical model for $Z$ and $\cC(C|Z)$ 
is a model for $(C|Z)\Pfnt{E}$, then we can formalize the idea 
that we freely choose inputs according to $\cC(Z)$ with the 
conditional distributions constrained by $\cC(C|Z)$ as follows:

\begin{definition}
  Let $\cC(Z)$ and $\cC(C|Z)$ be models where $\cC(Z)$ is classical.
  The \emph{free-for-$Z$ chaining
  of $\cC(Z)$ with $\cC(C|Z)$} is defined as
  \begin{align}
    \cC(Z)\ltimes\cC(C|Z)
    = \{\nu(Z)\rho(C|Z):\nu(Z)\in\cC(Z), \rho(C|Z)\in\cC(C|Z)\}.
  \end{align}
  If $\cC(Z)=[0,\infty)\mu(Z)$, we abbreviate
  $\cC(Z)\ltimes\cC(C|Z)=\mu(Z)\ltimes\cC(C|Z)$.
\end{definition}
Here is a more general form of free-for-$Z$ chaining that allows for
quantum side information on $Z$. 

\begin{definition}
  Let $\cC(Z)$ be a model for $Z\Pfnt{V}$ and $\cC(C|Z)$ a model for
  $(C|Z)\Pfnt{W}$. The \emph{free-for-$Z$ chaining of $\cC(Z)$ with
    $\cC(C|Z)$} is the model
  of $CZ\Pfnt{V}\Pfnt{W}$ 
  given by
  \begin{equation}
    \cC(Z)\ltimes\cC(C|Z) = \{\sigma(Z)\otimes \rho(C|Z):
      \sigma(Z)\in\cC(Z), \rho(C|Z)\in\cC(C|Z)\}.
  \end{equation}
\end{definition}

\subsection{Constructing Models for Experimental Configurations}
\label{sec:models:examples}

The models introduced above can represent all experimental
configurations involving quantum side information.  In particular,
they can represent configurations involving a sequence of trials with
devices that perform measurements based on random input choices. The
simplest case is where the side information is in a quantum system
$\Pfnt{E}$ that has no interaction with the experimental devices after
the experiment starts. If $\Pfnt{E}$ has independent dynamics during
the experiment and protocol, we can time-shift the dynamics to the
initial state and then treat $\Pfnt{E}$ as being static.  From the
point of view of the experimenter, the initial state of $\Pfnt{E}$ is
a density operator $\rho$. If the devices are quantum, then $\rho$ is
the marginal state of $\Pfnt{E}$ for the initial joint quantum state 
of the devices and $\Pfnt{E}$.  The joint state can depend on initial,
classical information that the experiment may depend on. We
condition on all such information and omit it from further
consideration.  By the end of the experiment classical data
$\Sfnt{CZ}$ is obtained, which includes the inputs $\Sfnt{Z}$ and
outputs $\Sfnt{C}$ of the devices.  The inputs come from a random
source, which must be modeled along with everything else, but is often
constrained to produce random bits independently of $\Pfnt{E}$ and the
devices. The relevant part of the final state is the joint state of
$\Sfnt{CZ}$ and $\Pfnt{E}$, which can be described by
$\rho(\Sfnt{CZ})$ and satisfies that
$\sum_{\Sfnt{cz}}\rho(\Sfnt{cz})=\rho$. The model must be formulated
so that any such final state that may be encountered is in the model.

We construct models by chaining individual trials.  Given that
$\Pfnt{E}$ does not interact with the results $\Sfnt{cz}$ of the
experiment so far, the (unnormalized) state of $\Pfnt{E}$ is
$\sigma=\rho(\Sfnt{cz})$, where $\rho(\Sfnt{CZ})$ is in the model
$\cC(\Sfnt{CZ})$ for the past.  The model $\cC_{\Sfnt{cz}}(CZ)$ for
the next trial may depend on the past and constrains on the results
$CZ$ of the next trial. The state of $\Pfnt{E}$ given the next
  trial results $cz$ and the past is $\sigma(cz)$, and we require
that $\sigma(CZ)$ is in $\cC_{\Sfnt{cz}}(CZ)$.  Thus chaining
$\cC(\Sfnt{CZ})$ with $\cC_{\Sfnt{CZ}}(CZ)$ according to
Def.~\ref{def:chaining} yields the model for the results including
$CZ$.

When chaining, the trial models are motivated by physical constraints
on the devices used.  For quantum experiments, the current
state $\rho_{\Pfnt{E}}=\rho(\Sfnt{cz})$
of $\Pfnt{E}$ must be related to a joint
state $\rho_{\Pfnt{ED}}$ of $\Pfnt{E}$ and the devices $\Pfnt{D}$ by
performing a  measurement on the devices $\Pfnt{D}$ and then 
tracing out $\Pfnt{D}$. We make no assumptions on the joint state and
its dependence on $\Sfnt{cz}$ other than the requirement that
$\rho(\Sfnt{CZ})$ is in the model for the past results.  The
experiment is constructed to constrain the way in which the devices
can use fresh random input $Z=z$ to perform a measurement during the
next trial. The constraints are typically described by constraints on
the $z$-dependent POVMs that are applied.  These may be modeled by a
single family $\fkP$ of POVMs, where the $z$-dependence is transferred
to structural constraints on the POVMs. For example, consider the
experimental configuration of a two-station, $l$-input, $m$-output
Bell test (the $(2,l,m)$-Bell-test configuration) with inputs $X,Y$
and outputs $A,B$ where the input distribution is uniform. In this
case, we have a factorization $\cV\otimes\cW$ of the devices' Hilbert
space for this trial and write the POVM in the form $P_{XA}\otimes
Q_{YB}$ where $\sum_{a}P_{xa}=\one/l$,
$\sum_{b}Q_{yb}=\one/l$. With $\fkP$ the set of all such
POVMs, the trial model becomes the model induced by $\fkP$ according
to Def.~\ref{def:inducedmodel}, and this model chains as desired with
the past.  See Sect.~\ref{sec:k22configs} for a detailed analysis of
$(k,2,2)$-Bell-test configurations.

In the trial model considered in the previous paragraph, the
observable probability distributions of the inputs and outputs form
the set of quantum-realizable distributions for this configuration,
which is a subset of non-signaling distributions. The distribution
$\mu(ABXY)$ is non-signaling if $\mu(A|XY)=\mu(A|X)$ and
$\mu(B|XY)=\mu(B|Y)$, so a station's observed output distribution does
not depend on the inputs of the other station. We can drop the
assumption that the devices are quantum and consider the trial model
where the only restriction is that conditional on $\Pfnt{E}$, the
observed probability distributions are non-signaling. This idea is
captured by the maximal extension of the non-signaling distributions
according to Def.~\ref{def:maxextension}.  While it is not realistic
at this time to think that super-quantum devices exist and can be
exploited by an otherwise quantum entity $\Pfnt{E}$, that randomness
can be generated for this model is of fundamental interest. Caution is
required when reusing super-quantum devices in multiple protcols as
composability may be compromised in ways that are not yet accounted for.

We remark that there is no restriction on the dynamics of the devices
between trials, nor is there any reason to explicitly represent this
dynamics.  The model keeps track only of the state of $\Pfnt{E}$,
and with the formulation of the trial models as maximal extensions or
induced models, any quantum systems or quantum operations that the
devices use over the course of the experiment are subsumed by the
trial models and the chaining constructions.

If the inputs are published or may become known to $\Pfnt{E}$, final
probabilities and entropies are conditioned on the inputs.  For
randomness generation, one option is to estimate the joint min-entropy
of inputs and outputs conditional on the side information and
eliminate the input entropy by subtracting the number of bits that
generated the inputs before applying an extractor, see
Protocol~\ref{prot:condimplicit}. For input distributions with low
entropy per trial, this is inefficient, so we need a direct method of
conditioning on inputs. Direct methods developed so far require that
model chaining is restricted to chaining with conditionally
independent inputs according to Def.~\ref{def:condchaining}, which
imposes an additional restriction on the relationship between the next
input and the past.  The conditional independence restriction is
satisfied if the input distribution is fixed and the inputs are
assumed to be independent of the devices and $\Pfnt{E}$. More
generally, it is satisfied if the source for the inputs has only
classical initial correlations with the devices and $\Pfnt{E}$, so
that given a classical part of $\Pfnt{E}$ the input distribution is
independent of the devices and the quantum part of $\Pfnt{E}$.

It is desirable to have models that can capture restricted
interactions between $\Pfnt{E}$ and the devices. Consider the case
where $\Pfnt{E}$ controls the source of the states used by the devices
for producing the outputs. We study the following two different types
of interactions. First,  we assume that the interaction is
representable by a strictly one-way communication, which means that
for a given trial, $\Pfnt{E}$ includes a subsystem $\Pfnt{S}$ that is
prepared and then transferred permanently to the devices.  All such
transfers can be time-shifted to before the protocol to return to the
situation of the strictly non-interacting $\Pfnt{E}$ already
discussed.  Second,  a more challenging and interesting situation we can study is where
$\Pfnt{E}$ learns the inputs of the past trials before preparing a
state and transferring it to the devices for the next trial. For this
situation we can start with the model for the past trials, close under
$Z$-conditional quantum operations, then use chaining, with
conditionally independent inputs if necessary.  The $Z$-conditional
quantum operations model the change of state of $\Pfnt{E}$ when
$\Pfnt{E}$ prepares a state in a source subsystem after having learned
the previous inputs and transfers the subsystem to the devices. 
In view of the QEF property presented as Lem.~\ref{lem:qef_cptpzclosure},
QEFs constructed under the first type of interaction works as well
under the second type of interaction. 

We finish this section with EAT models, which are the models that are
determined by EAT channel chains as required to apply the EAT for
randomness generation. The term ``EAT channel'' is from
Refs.~\cite{arnon-friedman:qc2016a,arnon-friedman:qc2018a}, but for an
authoritative definition and statement of the EAT, see
Ref.~\cite{dupuis:qc2016a}.  \Pc{The first two references forgot to
  specify that the min-tradeoff function has to be convex.}  An EAT
channel chain is a sequence of CPTP maps $\cN_{i}$ composed in a
specific way.  As defined in Ref.~\cite{arnon-friedman:qc2016a}
(Def. 5), $\cN_{i}$ is a CPTP map transforming system $\Pfnt{R}_{i-1}$
into $C_{i}Z_{i}\Pfnt{R}_{i}$, where $C_{i}$ here is $A_{i}B_{i}$
there and $Z_{i}$ here is $I_{i}$ there.  \Pc{The notation is
  different in Ref.~\cite{arnon-friedman:qc2018a}.} The systems
$\Pfnt{R}_{i}$ represent the devices used for trial $i$.  The
definition of EAT channels also includes a CV $X_{i}$ that is
determined by $C_{i}$ and $Z_{i}$.  Because it is determined, 
$X_{i}$ plays no role in our treatment.  For the EAT the CVs $X_{i}$
indirectly enable the possibility that the affine (or convex)
min-tradeoff function used in the EAT can quantify the final
conditional min-entropy in a way that depends on $i$.  This in turn
allows use of different types of trials in a single sequence, provided
that the type of the $i$'th trial is determined by information that
was or could have been public before the start of the trial. 
For QPE this is readily accounted for by the built-in option for
dependence on the past of both the models and the \QEFs.

The initial state of an EAT channel chain is a joint state of
$\Pfnt{R}_{0}\Pfnt{E}$. An experiment consists of applying the
$\cN_{i}$ sequentially to the system $\Pfnt{R}_{i-1}$ without touching
$\Pfnt{E}$ or the previously generated CVs. That is, for the $i$'th
trial, $\cN_{i}\otimes \one_{\Pfnt{E}}$ is applied to the quantum
systems.  The Markov chain condition applies at each step, namely for
the state after applying $\cN_{i}$ it is required that
\begin{equation}
  \Sfnt{C}_{<i} \leftrightarrow \Sfnt{Z}_{<i}\Pfnt{E}\leftrightarrow Z_{i}.
\end{equation}
Since after time-shifting one-way communications there is no interaction between $\Pfnt{E}$ and the devices (or
the CVs) after the initial state is determined, this fits the
non-interacting scenario introduced above. Each $\cN_{i}$ can be
expressed as a POVM $P^{(i)}_{C_{i}Z_{i}}$ of $\Pfnt{R_{i-1}}$ with
outcome $C_{i}Z_{i}$ followed by an outcome-conditional CPTP map to
transform $\Pfnt{R}_{i-1}$ into $\Pfnt{R}_{i}$. As far as the EAT is
concerned, the relevant properties are captured by associating with each
trial the model induced by $\fkP_{i}=\{P^{(i)}_{C_{i}Z_{i}}\}$ on
$C_{i}Z_{i}\Pfnt{E}$ in the sense that the EAT applies to chains of
these models.  After the experiment is formulated in terms of models
in our framework, the Markov chain condition for the EAT channel chain
is equivalent to the requirement that the model is chained with
conditionally independent inputs.

\section{Quantum Estimation Factors}
\label{sec:qefs}

\subsection{Definition and Equivalent Conditions}
\label{sec:qefdefs}

\begin{definition}
  The real-valued function $F(CZ)$ is a \emph{quantum estimation
    factor (\QEF) with power $\beta>0$ for $C|Z$ and the model
  $\cC(CZ)$} if $F(CZ)\geq 0$ and for all $\rho(CZ)\in\cC(CZ)$
  with $\rho\not=0$, $F(CZ)$
  satisfies the
  \emph{\QEF inequality with power $\beta$ at $\rho(CZ)$ for $C|Z$} given by
  \begin{equation}
    \sum_{cz}F(cz)\Rpow{\alpha}{\rho(cz)}{\rho(z)}
    \leq \Rpow{\alpha}{\rho}{\rho}=\tr(\rho).
    \label{eq:qefdef1}
  \end{equation}
  The real-valued function $F(CZ)$ is a \emph{Petz quantum estimation
    factor (\QEFP) with power $\beta>0$ for $C|Z$ and the model
    $\cC(CZ)$} if $F(CZ)\geq 0$ and for all $\rho(CZ)\in\cC(CZ)$ with
  $\rho\not=0$, $F(CZ)$ satisfies the \emph{\QEFP inequality with
    power $\beta$ at $\rho(CZ)$ for $C|Z$} given by
  \begin{equation}
    \sum_{cz}F(cz)\Ppow{\alpha}{\rho(cz)}{\rho(z)}
    \leq \Ppow{\alpha}{\rho}{\rho}=\tr(\rho).
    \label{eq:pefdef1}
  \end{equation}
\end{definition}
Both sides of the \QEF and \QEFP inequalities are positive
homogeneous of degree $1$ in $\rho(CZ)$. It follows that for $F(CZ)$
to be a \QEF (or \QEFP), it is necessary and sufficient that the \QEF
(or \QEFP) inequality holds for normalized distributions in $\cN(\cC(CZ))$.  For
normalized $\rho(CZ)$, the right-hand side of the \QEF and \QEFP inequalities
evaluate to $1$. We use \QEFPs primarily as a tool for constructing
\QEFs.

\begin{lemma}
  If $F(CZ)$ is a \QEFP with power $\beta\leq 1$, then $F(CZ)$ is a \QEF with power
  $\beta$.  This holds for all models.
\end{lemma}

\begin{proof}
  It suffices to apply the inequality $\Ppow{\alpha}{\sigma}{\tau}
  \geq \Rpow{\alpha}{\sigma}{\tau}$ (Lem.~\ref{lem:petz_sandwiched})
  to each summand of the \QEF inequality.
\end{proof}

The next lemmas give conditions for \QEFs that can be used when $\cC$
is closed under appropriate $\pCP$ maps.  The first is an alternative
form that may be useful for finding \QEFs, particularly for the
special cases in Sect.~\ref{sec:qefspec}.  The second is needed when
constructing \QEFs by \QEF chaining.  We remind that according to
  our marginalization convention, if $\rho(CZ)$ is a state of
  $CZ\Pfnt{E}$, we write the marginal state of $\Pfnt{E}$ as
  $\rho=\sum_{cz}\rho(cz)$.  

\begin{lemma}
  \label{lem:qefdef2}
  Let $\cC(CZ)$ be a model such that for all $\tau(CZ)\in\cC(CZ)$ we
  have the closure condition $\tau^{\gamma/2}\tau(CZ) \tau^{\gamma/2}\in\cC(CZ)$ 
  for  $\gamma=\beta$ and $\gamma=-\beta/\alpha$. Then $F(CZ)$ is a
  \QEF with power $\beta$ for $C|Z$ and $\cC(CZ)$ iff $F\geq 0$ and
  for all $\tau(CZ)\in\cC$,
  \begin{equation}
    \sum_{cz}F(cz)    \Rpow{\alpha}{\tau^{\beta/2}\tau(cz)\tau^{\beta/2}}{\tau^{\beta/2}\tau(z)\tau^{\beta/2}}
    \leq \Rpow{\alpha}{\tau}{\one}=\tr(\tau^{\alpha}).
    \label{eq:qefdef2}
\end{equation}
\end{lemma}

The closure condition in the lemma is satisfied if $\cC(CZ)$ is
$\pCP$-closed.

\begin{proof}
  Suppose that $F(CZ)$ is a \QEF with power $\beta$ for $C|Z$ and $\cC(CZ)$.
  Then $F(CZ)\geq 0$. For any $\tau(CZ)\in\cC$, define
  $\rho(CZ)=\tau^{\beta/2}\tau(CZ)\tau^{\beta/2}\in\cC(CZ)$.  Since
  $\rho=\tau^{\alpha}$, the right-hand side of Eq.~\ref{eq:qefdef2} is
  $\tr(\rho)$, matching the right-hand side of the \QEF inequality at
  $\rho(CZ)$.  Since $\rho(Z)=\tau^{\beta/2}\tau(Z)\tau^{\beta/2}$,
  the left-hand side of Eq.~\ref{eq:qefdef2} matches that of the \QEF
  inequality. Since the \QEF inequality at $\rho(CZ)$ is satisfied 
  by assumption, so is Eq.~\ref{eq:qefdef2}.

  Suppose that $F(CZ)$ satisfies the condition in the lemma. Then
  $F(CZ)\geq 0$. To show that $F(CZ)$ is a \QEF, consider any
  $\rho(CZ)\in\cC(CZ)$. To verify the \QEF inequality at $\rho(CZ)$, we
  reverse the transformation of the previous paragraph by defining
  $\tau(CZ)=\rho^{-\beta/(2\alpha)}\rho(CZ)\rho^{-\beta/(2\alpha)}\in\cC(CZ)$.
  We have $\tau=\rho^{1/\alpha}$, so
  $\tau^{\beta/2}\tau(CZ)\tau^{\beta/2}=\rho(CZ)$ and
  $\tau^{\beta/2}\tau(Z)\tau^{\beta/2}=\rho(Z)$. The expressions in
  Eq.~\ref{eq:qefdef2} are therefore identical to the corresponding
  ones in the \QEF inequality at $\rho(CZ)$, so the former implies the
  latter, as desired.
\end{proof}

\begin{lemma}\label{lem:qefdef3}
  Let $F(CZ)$ be a \QEF with power $\beta$ for $C|Z$ and $\cC(CZ)$.
  Consider $\sigma(CZ)\in\cC(CZ)$ and $\zeta(Z)\in\cS(Z\Pfnt{E})$ such
  that $\sigma(Z)\ll\zeta(Z)$ and define
  \begin{align}
    \xi(UZ)&=\sigma(Z)\knuth{U=0}+\zeta(Z)\knuth{U=1},\notag\\
    \chi&=\zeta^{-\beta/(2\alpha)}\sigma\zeta^{-\beta/(2\alpha)},\notag\\
    \rho(CZ)&=
    \chi^{\beta/2}
    \zeta^{-\beta/(2\alpha)}\sigma(CZ)\zeta^{-\beta/(2\alpha)}
    \chi^{\beta/2},
  \end{align}
  where $\Rng(U)=\{0,1\}$.
  If $\xi(UZ)\in U\leftrightarrow\Pfnt{E}\leftrightarrow Z$ and
  $\rho(CZ)\in\cC(CZ)$, then
  \begin{equation}
    \sum_{cz}F(cz)\Rpow{\alpha}{\sigma(cz)}{\zeta(z)} \leq
    \Rpow{\alpha}{\sigma}{\zeta}.
    \label{eq:qefdef3}
  \end{equation}
\end{lemma}

The condition $\rho(CZ)\in\cC(CZ)$ is satisfied if $\cC(CZ)$ is
$\pCP$-closed. The main purpose of the lemma is to enable a change in
the conditioner in the \QEF inequality from the marginal state to 
another one. This requires conditions on the relationship between the
two conditioners. The conditions are expressed by introducing the
auxiliary CV $U$ and state $\xi(UZ)$ and include the short Markov
chain condition in the lemma.  The lemma simplifies in the absence of
inputs or when the input distribution is fixed and known,
see the next section.

\begin{proof}
  By the definition of short quantum Markov chains, there is a
  factorization
  $\cH(\Pfnt{E})=\bigoplus_{i}\cU_{i}\otimes\cZ_{i}\;\oplus\cR$ such
  that $\sigma(Z)=\bigoplus_{i}\sigma_{i}\otimes\xi_{i}(Z)$ and
  $\zeta(Z)=\bigoplus_{i}\zeta_{i}\otimes\xi_{i}(Z)$, where $\sigma(Z)\ll\zeta(Z)$ implies
  $\sigma_{i}\ll\zeta_{i}$ for each $i$.  In order to
  derive the inequality in Eq.~\ref{eq:qefdef3} from the \QEF
  inequality, we can assure a match of the right-hand sides with
  \begin{align} 
    \rho&= 
    \left(\zeta^{-\beta/(2\alpha)}\sigma\zeta^{-\beta/(2\alpha)}\right)^{\alpha}\notag\\
    &=
    \left(\zeta^{-\beta/(2\alpha)}\sigma\zeta^{-\beta/(2\alpha)}\right)^{\beta/2}
    \zeta^{-\beta/(2\alpha)}\sigma\zeta^{-\beta/(2\alpha)}
    \left(\zeta^{-\beta/(2\alpha)}\sigma\zeta^{-\beta/(2\alpha)}\right)^{\beta/2}.
  \end{align}
  This motivates the definitions of $\chi$ and $\rho(CZ)$.
  The support assumptions ensure that the
  supports of $\sigma$ and $\sigma(CZ)$ are contained in that of
  $\zeta$.

  For a match of the left-hand sides of the target inequalities, we
  need to verify that
  $\Rpow{\alpha}{\rho(CZ)}{\rho(Z)}=\Rpow{\alpha}{\sigma(CZ)}{\zeta(Z)}$.
  For this it suffices that
  \begin{equation}
    \rho(Z)^{-\beta/(2\alpha)}\rho(CZ)\rho(Z)^{-\beta/(2\alpha)}
    \sim_{U} \zeta(Z)^{-\beta/(2\alpha)}\sigma(CZ)\zeta(Z)^{-\beta/(2\alpha)},
  \end{equation}
  where $\sim_{U}$ denotes equality up to conjugation by a unitary operator, or
  equivalently, that the two sides have the same spectrum with
  multiplicities.  The support assumptions ensure that the support of
  $\sigma(CZ)$ is contained in that of $\zeta(Z)$ for the right-hand
  side of the spectral equivalence.  Starting from the left-hand side, we
  get
  \begin{align}
    \rho(Z)^{-\beta/(2\alpha)}\rho(CZ)\rho(Z)^{-\beta/(2\alpha)} \hspace*{-1.5in}&\notag\\
    &=
    \rho(Z)^{-\beta/(2\alpha)}\chi^{\beta/2}
    \zeta^{-\beta/(2\alpha)}\sigma(CZ)\zeta^{-\beta/(2\alpha)}
    \chi^{\beta/2}\rho(Z)^{-\beta/(2\alpha)} \notag\\
    &\sim_{U} \sigma(CZ)^{1/2} \zeta^{-\beta/(2\alpha)}
    \chi^{\beta/2}\rho(Z)^{-\beta/\alpha}\chi^{\beta/2}
    \zeta^{-\beta/(2\alpha)} \sigma(CZ)^{1/2}\notag\\
    &= \sigma(CZ)^{1/2} \zeta^{-\beta/(2\alpha)} \chi^{\beta/2}\left(\chi^{\beta/2}
      \zeta^{-\beta/(2\alpha)}\sigma(Z)\zeta^{-\beta/(2\alpha)}
      \chi^{\beta/2}\right)^{-\beta/\alpha} \chi^{\beta/2}\zeta^{-\beta/(2\alpha)}
    \sigma(CZ)^{1/2},\label{eq:lem:qefdef3:A}
  \end{align}
  where the equivalence in the third line follows from $A^{\dagger}A
  \sim_{U} A A^{\dagger}$ for all operators $A$.  The expression
  between the two terms $\sigma(CZ)^{1/2}$ factors with respect to the
  representation of $\cH(\Pfnt{E})$, so we can compute each factor
  separately.  First determine
  \begin{equation}
    \chi = \bigoplus_{i}\zeta_{i}^{-\beta/(2\alpha)}\sigma_{i}\zeta_{i}^{-\beta/(2\alpha)}
      \otimes \xi_{i}^{1/\alpha}
  \end{equation}
  and define
  $\chi_{i}=\zeta_{i}^{-\beta/(2\alpha)}\sigma_{i}\zeta_{i}^{-\beta/(2\alpha)}$
  so that $\chi=\bigoplus_{i}\chi_{i}\otimes \xi_{i}^{1/\alpha}$.
  From this,
  \begin{equation}
    \chi^{\beta/2}\zeta^{-\beta/(2\alpha)} 
    =\bigoplus_{i}\chi_{i}^{\beta/2}\zeta_{i}^{-\beta/(2\alpha)}\otimes\one_{i},
  \end{equation}
  where $\one_{i}$ is the projector onto the support of $\xi_{i}$ in
  $\cZ_{i}$.  Since
  $\sigma(Z)=\bigoplus_{i}\sigma_{i}\otimes\xi_{i}(Z)$, we have for
  the inner expression on the right-hand side of
  Eq.~\ref{eq:lem:qefdef3:A}
  \begin{align}
    \left(\chi^{\beta/2}
      \zeta^{-\beta/(2\alpha)}\sigma(Z)\zeta^{-\beta/(2\alpha)}
      \chi^{\beta/2}\right)^{-\beta/\alpha} &=
    \bigoplus_{i}\left(\chi_{i}^{\beta/2}\zeta_{i}^{-\beta/(2\alpha)}
      \sigma_{i}\zeta_{i}^{-\beta/(2\alpha)}\chi_{i}^{\beta/2}\right)^{-\beta/\alpha}
    \otimes \xi_{i}(Z)^{-\beta/\alpha} \notag\\
    &= \bigoplus_{i}
    \left(\chi_{i}^{\beta/2}\chi_{i}\chi_{i}^{\beta/2}\right)^{-\beta/\alpha}\otimes
    \xi_{i}(Z)^{-\beta/\alpha}
    \notag\\
    &= \bigoplus_{i}\chi_{i}^{-\beta}\otimes\xi_{i}(Z)^{-\beta/\alpha}.
  \end{align}
  Define the support projectors $\Pi_{i}=\suppproj{\chi_{i}}$ and
  $\Pi=\suppproj{\chi}=\bigoplus_{i}\Pi_{i}\otimes\one_{i}$.  Substituting
  the identities obtained and continuing from the end of
  Eq.~\ref{eq:lem:qefdef3:A} we get
  \begin{align}
    \rho(Z)^{-\beta/(2\alpha)}\rho(CZ)\rho(Z)^{-\beta/(2\alpha)}
    \hspace*{-1.5in}&\notag\\
    &\sim_{U}\sigma(CZ)^{1/2}\left(
      \bigoplus_{i}\zeta_{i}^{-\beta/(2\alpha)}\chi_{i}^{\beta/2}
      \chi_{i}^{-\beta}
      \chi_{i}^{\beta/2}\zeta_{i}^{-\beta/(2\alpha)}\otimes
      \xi_{i}(Z)^{-\beta/\alpha}
    \right)\sigma(CZ)^{1/2}\notag\\
    &=\sigma(CZ)^{1/2}\left(\bigoplus_{i}
      \zeta_{i}^{-\beta/(2\alpha)}\Pi_{i}\zeta_{i}^{-\beta/(2\alpha)}\otimes
      \xi_{i}(Z)^{-\beta/\alpha}
    \right)\sigma(CZ)^{1/2}\notag\\
    &=\sigma(CZ)^{1/2}\left(\bigoplus_{i}
      \zeta_{i}^{-\beta/(2\alpha)}\Pi_{i}\zeta_{i}^{-\beta/(2\alpha)}\otimes
      \xi_{i}(Z)^{-\beta/(2\alpha)}\one_{i}\xi_{i}(Z)^{-\beta/(2\alpha)}
    \right)\sigma(CZ)^{1/2}\notag\\
    &=\sigma(CZ)^{1/2}\zeta(Z)^{-\beta/(2\alpha)}
    \Pi\zeta(Z)^{-\beta/(2\alpha)}\sigma(CZ)^{1/2}\notag\\
    &\sim_{U}\Pi\zeta(Z)^{-\beta/(2\alpha)}\sigma(CZ)\zeta(Z)^{-\beta/(2\alpha)}\Pi.
    \label{eq:lem:qefdef3:last}
  \end{align}
  The support of $\zeta(Z)^{-\beta/(2\alpha)}\sigma(CZ)\zeta(Z)^{-\beta/(2\alpha)}$
  is contained in that of 
  $\zeta(Z)^{-\beta/(2\alpha)}\sigma(Z)\zeta(Z)^{-\beta/(2\alpha)}$,
  which is the direct sum of the supports of 
  $\zeta_{i}^{-\beta/(2\alpha)}\sigma_{i}\zeta_{i}^{-\beta/(2\alpha)}\otimes \xi_{i}(Z)^{1/\alpha}$
  and therefore contained in the support of $\chi$.
  The support projector $\Pi$ can therefore be eliminated from the final
  expression in Eq.~\ref{eq:lem:qefdef3:last} to finish the proof of the lemma.
\end{proof}

\subsection{\QEF Conditions for Special Cases}
\label{sec:qefspec}

The conditions in Eqs.~\ref{eq:qefdef1},~\ref{eq:qefdef2}
and~\ref{eq:qefdef3} simplify when the probability distribution of $Z$
is given and independent of $\Pfnt{E}$.
\begin{lemma}
  \label{lem:qeffreeforz}
  Let $\mu(Z)$ be a probability distribution and
  $\cC(CZ)=\mu(Z)\ltimes\cC(C|Z)$.  Consider $F(CZ)\geq 0$.  Then
  $F(CZ)$ is a \QEF with power $\beta$ for $C|Z$ and $\cC(CZ)$ iff for all
  $\rho(CZ)\in\cC(CZ)$,
  \begin{equation}
    \sum_{cz}F(cz)\mu(z)\Rpow{\alpha}{\rho(c|z)}{\rho}
    \leq \tr(\rho).\label{eq:qefdef1_nu}
  \end{equation}
  If $\cC(C|Z)$ is $\pCP$-closed,
  then $F(CZ)$ is a \QEF with power $\beta$ for $C|Z$ and $\cC(CZ)$ iff for all
  $\tau(CZ)\in\cC(CZ)$,
  \begin{equation}
    \sum_{cz}F(cz)\mu(z)\tr(\tau(c|z)^{\alpha})
    \leq \tr(\tau^{\alpha}).\label{eq:qefdef2_nu}
  \end{equation}
  If $\cC(C|Z)$ is $\pCP$-closed and $F(CZ)$ is a \QEF with power $\beta$
  for $C|Z$ and $\cC(CZ)$, then for all $\sigma(CZ)\in\cC(CZ)$ and
  $\zeta\gg\sigma$,
  \begin{equation}
    \sum_{cz}F(cz)\mu(z)\Rpow{\alpha}{\sigma(c|z)}{\zeta} \leq
    \Rpow{\alpha}{\sigma}{\zeta}.
    \label{eq:qefdef3_nu}
  \end{equation}
\end{lemma}

\begin{proof}
  The first equivalence follows by substitution in the \QEF definition
  and the second by substitution in Lem.~\ref{lem:qefdef2}.  For the
  last claim, define $\zeta(Z)=\mu(Z)\zeta$.  The distribution
  $\xi(UZ)$ defined in Lem.~\ref{lem:qefdef3} can be written as
  $\xi(UZ)=\left(\sigma\knuth{U=0}+\zeta\knuth{U=1}\right)\mu(Z)$,
  which satisfies $\xi(UZ)\in U\leftrightarrow\Pfnt{E}\leftrightarrow
  Z$ with respect to the trivial factorization
  $\cH(\Pfnt{E})=\cH(\Pfnt{E})\otimes\cmplx$. The claim then follows
  by substitution in Eq.~\ref{eq:qefdef3}.
\end{proof}

The \QEF conditions further simplify in the absence of inputs, namely
when $Z$ is trivial and can be omitted.
\begin{lemma}
  \label{lem:qefnoz}
  Let $\cC(C)$ be a model and $F(C)\geq 0$. Then
  $F(C)$ is a \QEF with power $\beta$ for $\cC(C)$ iff for all
  $\rho(C)\in\cC(C)$,
  \begin{equation}
    \sum_{c}F(c)\Rpow{\alpha}{\rho(c)}{\rho}
    \leq \tr(\rho).\label{eq:qefdef1_noz}
  \end{equation}
  If $\cC$ is $\pCP$-closed,
  then $F(C)$ is a \QEF with power $\beta$ for $\cC(C)$ iff for all
  $\tau(C)\in\cC(C)$,
  \begin{equation}
    \sum_{c}F(c)\tr(\tau(c)^{\alpha})
    \leq \tr(\tau^{\alpha}).\label{eq:qefdef2_noz}
  \end{equation}
  If $\cC$ is $\pCP$-closed,
  and $F(C)$ is a \QEF with power $\beta$ for $\cC(C)$,
  then for all $\sigma(C)\in\cC(C)$ and $\zeta\gg\sigma$,
  \begin{equation}
    \sum_{c}F(c)\Rpow{\alpha}{\sigma(c)}{\zeta} \leq
    \Rpow{\alpha}{\sigma}{\zeta}.
    \label{eq:qefdef3_noz}
  \end{equation}
\end{lemma}

\begin{proof}
  Apply Lem.~\ref{lem:qeffreeforz} and simplify. 
\end{proof}

\subsection{\QEF Properties}
\label{sec:qefproperties}

\begin{lemma}\label{lem:trival_QEF}
  For $C|Z$ and all models, the function $F(CZ)=1$ is a \QEF with
    power $\beta$ for each $\beta>0$, and a \QEFP with power $\beta$
    for each $\beta\in (0,1]$. 
\end{lemma}

\begin{proof}
  It suffices to verify Eq.~\ref{eq:qefdef1}.
  \begin{align}
    \sum_{cz}F(cz)\Rpow{\alpha}{\rho(cz)}{\rho(z)} &=
    \sum_{cz}\Rpow{\alpha}{\rho(cz)}{\rho(z)}\notag\\
    &=\sum_{z} \sum_{c}\Rpow{\alpha}{\rho(cz)}{\rho(z)}\notag\\
    &\leq \sum_{z}\Rpow{\alpha}{\rho(z)}{\rho(z)}\notag\\
    &= \sum_{z}\tr(\rho(z))\notag\\
    &=\tr(\rho),
  \end{align}
  where we applied Lem.~\ref{lem:rp_sumineq} for the inequality in the
  third line. In this argument, we can replace the sandwiched by the
  Petz R\'enyi power provided $\beta\leq 1$.
\end{proof}

\begin{lemma}
  Let $F(CZ)$ be a \QEF with power $\beta$ for $C|Z$ and $\cC(CZ)$.
  Then for all $\beta'\geq \beta$,  $F(CZ)$ is a \QEF with power
  $\beta'$ for $C|Z$ and $\cC(CZ)$.
\end{lemma}

\begin{proof}
  Consider any $\rho(CZ)\in\cC(CZ)$. All expressions in the
  calculation below are homogeneous of the same degree, so we may
  assume that $\tr(\rho)=1$. If not, it suffices to rescale $\rho(CZ)$
  to ensure this condition. In view of the \QEF inequality, it
  suffices to show that the function $g_{cz}:
  \beta'\mapsto\Rpow{1+\beta'}{\rho(cz)}{\rho(z)}$ is non-increasing
  for all $cz$. According to Lem.~\ref{lem:trival_QEF},
  $\sum_{cz}\Rpow{1+\beta'}{\rho(cz)}{\rho(z)}\leq 1$, and since the
  summands are non-negative, for each $cz$ we have
  $\Rpow{1+\beta'}{\rho(cz)}{\rho(z)}\leq 1$.  For the $cz$ with
  $\rho(cz)=0$, $\Rpow{1+\beta'}{\rho(cz)}{\rho(z)}=0$ for 
  all $\beta'$ and $g_{cz}$ is non-increasing. 
  For the $cz$ with $\rho(cz)>0$ the function $\log(g_{cz})$
  is non-positive. Log-convexity of R\'enyi powers
  (Lem.~\ref{lem:convex_rp}) implies that the slope of $\log(g_{cz})$ 
  is non-decreasing. In view of $-\infty<\log(g_{cz})\leq 0$, 
  the slope of $\log(g_{cz})$ at any $\beta'$ cannot become positive, 
  otherwise when $\beta'\nearrow \infty$ the value of $\log(g_{cz})$ 
  would become positive. Thus $\log(g_{cz})$ is non-increasing and 
  since $x\mapsto\log(x)$ is order-preserving, 
  $g_{cz}$ is also non-increasing. 
\end{proof}

\begin{lemma}
  Let $F(CZ)$ be a \QEF with power $\beta$ for $C|Z$ and $\cC(CZ)$.
  Then for $0<\gamma\leq 1$, $F(CZ)^{\gamma}$ is a \QEF with power
  $\gamma\beta$ for $C|Z$ and $\cC(CZ)$. This also holds with
  ``\QEF'' replaced by ``\QEFP''.
\end{lemma}

The transformation $F\mapsto F^{\gamma}$ in the lemma is referred
to as \emph{power reduction by $\gamma$}.

\begin{proof}
  Consider any $\rho(CZ)\in\cC(CZ)$. All expressions in the
  calculation below are homogeneous of the same degree, so we may
  assume that $\tr(\rho)=1$.  Define the probability distribution
  $\mu(CZ)$ by $\mu(cz)=\tr(\rho(cz))$.  We check the \QEF inequality
  at $\rho(CZ)$:
  \begin{align}
    \sum_{cz}F(cz)^{\gamma} \Rpow{1+\gamma\beta}{\rho(cz)}{\rho(z)}&=
    \sum_{cz}F(cz)^{\gamma}\mu(cz)
    \hatRpow{1+\gamma\beta}{\rho(cz)}{\rho(z)}
    \notag\\
    &=\sum_{cz}\mu(cz)\left(F(cz)
      \hatRpow{1+\gamma\beta}
        {\rho(cz)}{\rho(z)}^{1/\gamma}\right)^{\gamma}\notag\\
    &\leq\left(\sum_{cz}\mu(cz)F(cz)
      \hatRpow{1+\gamma\beta}
        {\rho(cz)}{\rho(z)}^{1/\gamma}\right)^{\gamma},
  \end{align}
  since for $\gamma\in(0, 1]$ the function $x\mapsto x^{\gamma}$  is 
  concave and the sums are expectations with respect to $\mu(CZ)$.
  By monotonicity of R\'enyi powers (Lem.~\ref{lem:monotone_rp}),
  we have $\hatRpow{1+\gamma\beta}
        {\rho(cz)}{\rho(z)}^{1/(\beta\gamma)}
  \leq\hatRpow{1+\beta}
        {\rho(cz)}{\rho(z)}^{1/\beta}$, so we can continue
  where we left off to get
  \begin{align}
    \sum_{cz}F(cz)^{\gamma}\Rpow{1+\gamma\beta}{\rho(cz)}{\rho(z)} &\leq
    \left(\sum_{cz}F(cz)\mu(cz)
      \hatRpow{1+\beta}
        {\rho(cz)}{\rho(z)}\right)^{\gamma}\notag\\
      &=\left(\sum_{cz}F(cz)
      \Rpow{1+\beta}{\rho(cz)}{\rho(z)}\right)^{\gamma}\notag\\
    &\leq 1,
  \end{align}
  since $F$ is assumed to be a \QEF with power $\beta$.  The lemma
  follows by arbitrariness of $\rho(CZ)\in\cC(CZ)$.  In this
  argument, we can replace the sandwiched by the Petz R\'enyi power.
\end{proof}

Since the inequality in Eq.~\ref{eq:qefdef1} is linear in $F(CZ)$, the
set of \QEFs is convex. By positive homogeneity of the \QEF inequality
in $\rho(CZ)$, it suffices to check the trace-normalized
$\rho(CZ)\in\cN(\cC(CZ))$.  Further, as a consequence of the next
lemma, it suffices to check the \QEF inequalities on any subset of
$\cN(\cC(CZ))$ whose convex closure contains $\cN(\cC(CZ))$.

\begin{lemma}\label{lem:extremal_qef_constraints}
  $F(CZ)$ is a \QEF with power $\beta$ for $C|Z$ and $\cC(CZ)$ iff
  $F(CZ)$ is a \QEF with power $\beta$ for $C|Z$ and $\Cone(\cC(CZ))$.
  This also holds with ``\QEF'' replaced by ``\QEFP'' provided
  $\beta\leq 1$.
\end{lemma}

\begin{proof}
  It suffices to check that if the \QEF
  inequality holds at $\rho_{i}(CZ)\in\cC(CZ)$ for $i\in I$,
  then it holds at every convex combination
  $\rho(CZ)=\sum_{i}\lambda_{i}\rho_{i}(CZ)$.  By joint convexity of
  conditional R\'enyi powers (Lem.~\ref{lem:jconvex_rp}),
  \begin{equation}
    \Rpow{\alpha}{\rho(CZ)}{\rho(Z)}\leq
    \sum_{i}\lambda_{i}\Rpow{\alpha}{\rho_{i}(CZ)}{\rho_{i}(Z)}.
  \end{equation}
  Therefore
  \begin{align}
    \sum_{cz}F(cz)\Rpow{\alpha}{\rho(CZ)}{\rho(Z)} &\leq
    \sum_{cz}F(cz)\sum_{i}\lambda_{i}\Rpow{\alpha}{\rho_{i}(cz)}{\rho_{i}(z)}
    \notag\\
    &=\sum_{i}\lambda_{i}\sum_{cz}F(cz)\Rpow{\alpha}{\rho_{i}(cz)}{\rho_{i}(z)}
    \notag\\
    &\leq \sum_{i}\lambda_{i}\tr(\rho_{i})\notag\\
    &= \tr(\sum_{i}\lambda_{i}\rho_{i})\notag\\
    &= \tr(\rho).
  \end{align}
  In this argument, we can replace the sandwiched by the Petz R\'enyi
  power provided $\beta\leq 1$.
\end{proof}

It may be difficult to determine manageable subsets of $\cN(\cC(CZ))$
whose convex closure contains $\cN(\cC(CZ))$. If
$\Cone(\cC'(CZ))\supseteq\cC(CZ)$, then any \QEF for $\cC'(CZ)$ is a
\QEF for $\cC(CZ)$, so a strategy for constructing \QEFs is to find
better behaved models $\cC'(CZ)$ whose convex closure contains
$\cC(CZ)$.

According to the next lemma, \QEFs of a model are \QEFs of the closure of
the model under $Z$-conditional quantum operations.

\begin{lemma}\label{lem:qef_cptpzclosure}
  Let $\cC(CZ)$ be a model for $CZ\Pfnt{E}$ and let
  $\CPTP_{Z}(\cC(CZ))$ be the set of distributions that can be
  obtained by applying a $Z$-conditional quantum operation to members of
  $\cC(CZ)$.  Then $F(CZ)$ is a \QEF with power $\beta$ for $C|Z$ and
  $\cC(CZ)$ iff $F(CZ)$ is a \QEF with power $\beta$ for $C|Z$ and
  $\CPTP_{Z}(\cC(CZ))$.  This also holds with ``\QEF'' replaced by
  ``\QEFP'' provided $\beta\leq 1$.
\end{lemma}

\begin{proof}
  The lemma follows from the data-processing inequality for R\'enyi
  powers (Lem.~\ref{lem:dataprocessing_rp}).  It suffices to check
  that if the \QEF inequality holds at $\rho(CZ)\in\cC(CZ)$ and
  $\cE_{Z}$ is a $Z$-conditional quantum operation, then it holds at
  $\sigma(CZ) =\cE_{Z}(\rho(CZ))$:
  \begin{align}
    \sum_{cz}F(cz)\Rpow{\alpha}{\sigma(cz)}{\sigma(z)} &=
    \sum_{cz}F(cz)\Rpow{\alpha}{\cE_{z}(\rho(cz))}{\cE_{z}(\rho(z))}\notag\\
     &\leq \sum_{cz}F(cz)\Rpow{\alpha}{\rho(cz)}{\rho(z)}\notag\\
     &\leq\tr(\rho)\notag\\
     &=\sum_{z}\tr(\rho(z))\notag\\
     &=\sum_{z}\tr(\cE_{z}(\rho(z)))\notag\\
     &=\tr(\sigma),
  \end{align}
  since each $\cE_{z}$ is trace-preserving.  Again, in this argument, we can
  replace the sandwiched by the Petz R\'enyi power provided $\beta\leq
  1$.
\end{proof}

\subsection{Chaining \QEFs}

The next theorem shows that \QEFs can be chained with conditionally
independent inputs.  We do not know whether this is true for \QEFPs.

\begin{theorem}\label{thm:qefchainmain}
  Let $\cC(\Sfnt{CZ})$ be a model for $\Sfnt{CZ}\Pfnt{E}$ and for each
  $\Sfnt{cz}$, let $\cC_{\Sfnt{cz}}(CZ)$ be a $\pCP$-closed model for
  $CZ\Pfnt{E}$.  If $G$ is a \QEF with power $\beta$ for
  $\Sfnt{C}|\Sfnt{Z}$ and $\cC(\Sfnt{CZ})$, and for each $\Sfnt{cz}$, $F_{\Sfnt{cz}}$ is
  a \QEF with power $\beta$ for $C|Z$ and $\cC_{\Sfnt{cz}}(CZ)$, then
  $G(\Sfnt{CZ})F_{\Sfnt{CZ}}(CZ)$ is a \QEF with power $\beta$ for
  $\Sfnt{C}C|\Sfnt{Z}Z\Pfnt{E}$ and
  $\cC(\Sfnt{CZ})\circ_{Z|\Sfnt{Z}}\cC_{\Sfnt{CZ}}(CZ)$.
\end{theorem}

For the models constructed for experiments consisting of sequences of
trials discussed in Sect.~\ref{sec:models:examples}, the trial models
are maximal extensions or induced and therefore $\pCP$-closed since
$\pCP$ maps are special cases of $\CP$ maps
(Lems.~\ref{lem:maxextend_is_CPclosed}
and~\ref{lem:induced_cpclosed}).  The $\pCP$-closure condition can be
weakened by taking advantage of the specific membership condition in
Lem.~\ref{lem:qefdef3} as indicated in the proof.

\begin{proof}
  Consider any $\sigma(\Sfnt{CZ}CZ)\in
  \cC(\Sfnt{CZ})\circ_{Z|\Sfnt{Z}}\cC_{\Sfnt{CZ}}(CZ)$.  We show below
  that for each $\Sfnt{cz}$,
  \begin{equation}
    \sum_{cz}F_{\Sfnt{cz}}(cz)\Rpow{\alpha}{\sigma(\Sfnt{cz}cz)}{\sigma(\Sfnt{z}z)}
    \leq \Rpow{\alpha}{\sigma(\Sfnt{cz})}{\sigma(\Sfnt{z})}.
    \label{thm:qefchain:eq:1}
  \end{equation}
  Once this is shown, the theorem follows from
  \begin{align}
    \sum_{\Sfnt{cz}cz}G(\Sfnt{cz})F_{\Sfnt{cz}}(cz)
    \Rpow{\alpha}{\sigma(\Sfnt{cz}cz)}{\sigma(\Sfnt{z}z)}
    \hspace*{-2in}&\notag\\
    &=
    \sum_{\Sfnt{cz}}G(\Sfnt{cz})\sum_{cz}F_{\Sfnt{cz}}(cz)
    \Rpow{\alpha}{\sigma(\Sfnt{cz}cz)}{\sigma(\Sfnt{z}z)}\notag\\
    &\leq \sum_{\Sfnt{cz}}G(\Sfnt{cz})
    \Rpow{\alpha}{\sigma(\Sfnt{cz})}{\sigma(\Sfnt{z})}\notag\\
    &\leq \Rpow{\alpha}{\sigma}{\sigma},
  \end{align}
  where we applied Eq.~\ref{thm:qefchain:eq:1}, model chaining Def.~\ref{def:chaining}, and the assumption
  that $G$ is a \QEF for $\Sfnt{C}|\Sfnt{Z}$ and $\cC(\Sfnt{CZ})$. Thus $G(\Sfnt{CZ})F_{\Sfnt{CZ}}(CZ)$ is
  a \QEF as claimed.

  To show Eq.~\ref{thm:qefchain:eq:1}, we apply
  Lem.~\ref{lem:qefdef3} with $\sigma(CZ)$ there replaced by
  $\sigma(\Sfnt{cz}CZ)$ here, $\zeta(Z)$ there by $\sigma(\Sfnt{z}Z)$
  here, and $F(CZ)$ there by $F_{\Sfnt{cz}}(CZ)$ here.  By
  definition of chaining, $\sigma(\Sfnt{cz}CZ)\in\cC_{\Sfnt{cz}}(CZ)$.
  We verify that the Markov chain condition there follows from
  $\sigma(\Sfnt{CZ}Z)\in
  \Sfnt{C}\leftrightarrow\Sfnt{Z}\Pfnt{E}\leftrightarrow Z$
  according to the definition of chaining with conditionally
  independent inputs. For each $\Sfnt{z}$, there is a
  factorization $\cH(\Pfnt{E}) =
  \bigoplus_{i}\cD_{i}\otimes\cZ_{i}\;\oplus\cR$ for
  which $\sigma(\Sfnt{Cz}Z) =
  \bigoplus_{i}\sigma_{i}(\Sfnt{C})\otimes\zeta_{i}(Z)$ for some
  $\sigma_{i}(\Sfnt{C})$ and $\zeta_{i}(Z)$ that depend implicitly on
  $\Sfnt{z}$. This implies
  $\sigma(\Sfnt{z}Z)=\bigoplus_{i}\sigma_{i}\otimes\zeta_{i}(Z)$.  To
  verify the Markov chain condition of Lem.~\ref{lem:qefdef3}, we
  define $\xi(ZU)=\sigma(\Sfnt{cz}Z)\knuth{U=0}+
  \sigma(\Sfnt{z}Z)\knuth{U=1}$. Then
  \begin{equation}
    \xi(ZU) = \bigoplus_{i}\left(
      \sigma_{i}(\Sfnt{c})\knuth{U=0}+\sigma_{i}\knuth{U=1}
    \right)\otimes\zeta_{i}(Z),
  \end{equation}
  which implies $\xi(ZU)\in U\leftrightarrow \Pfnt{E}\leftrightarrow Z$.
  The membership condition of Lem.~\ref{lem:qefdef3} is satisfied
  since the $\cC_{\Sfnt{cz}}(CZ)$ are assumed to be $\pCP$-closed.
  For the purpose of weakening this condition the explicit distributions
  that need to be in $\cC_{\Sfnt{cz}}(CZ)$ are 
  \begin{equation}
    \rho(\Sfnt{cz}CZ)= 
    \chi(\Sfnt{cz})^{\beta/2}\sigma(\Sfnt{z})^{-\beta/(2\alpha)}
    \sigma(\Sfnt{cz}CZ)\sigma(\Sfnt{z})^{-\beta/(2\alpha)}\chi(\Sfnt{cz})^{\beta/2},
  \end{equation}
  where
  \begin{equation}
    \chi(\Sfnt{cz}) =
    \sigma(\Sfnt{z})^{-\beta/(2\alpha)}
    \sigma(\Sfnt{cz})\sigma(\Sfnt{z})^{-\beta/(2\alpha)}.
  \end{equation}
\end{proof}

Although it is an immediate consequence of the results so far, we give
the next corollary for emphasis, and so that we can use it explicitly when
discussing models relevant to experimental configurations.

\begin{corollary}\label{cor:qefclosechain}
  In Thm.~\ref{thm:qefchainmain}, we may close 
  $\cC(\Sfnt{CZ})$ under $\Sfnt{Z}$-conditional quantum operations and
  positive combinations before chaining.
\end{corollary}

\begin{proof}
  This follows from Thm.~\ref{thm:qefchainmain} after
  applying Lems.~\ref{lem:extremal_qef_constraints} and~\ref{lem:qef_cptpzclosure}.
\end{proof}

The $\Sfnt{Z}$-conditional quantum operations on $\cC(\Sfnt{CZ})$ may
affect the quantum Markov chain condition, but in chaining with
conditionally independent inputs, only cases where the condition
survives are passed on to the chained model. Since chaining is
monotone in the models being chained, no states are lost by closing
$\cC(\Sfnt{CZ})$ before chaining.

In Sect.~\ref{sec:models:examples} we mentioned some situations where
the quantum Markov chain condition applies, such as when the
distribution of the inputs is fixed and independent of $\Pfnt{E}$.
When such situations do not apply, we rely on physical constraints
satisfied by the experiments to make sure that the actual states after
the trials satisfy the quantum Markov chain condition.  Alternatively,
we use the strategy where input entropy is eliminated when the
extractor is applied and \QEFs are designed without conditioning on
inputs.

\subsection{\QEFs as Estimators}
\label{sec:qefest}

\QEFs and \QEFPs can be interpreted as estimators of normalized
R\'enyi powers. We formalize this interpretation for \QEFs.  Let $F(CZ)$
be a \QEF with power $\beta$ for $C|Z$ and $\cC(CZ)$.  Consider
$\rho(CZ)\in\cN(\cC(CZ))$.  We can interpret $1/(\epsilon F(CZ))$ as a
level-$\epsilon$ confidence upper bound on
$\hatRpow{\alpha}{\rho(CZ)}{\rho(Z)}$ in the following sense:

\begin{theorem}\label{thm:qefs_estimate}
  Let $F(CZ)$ be a \QEF with power $\beta$ for $C|Z$ and $\cC(CZ)$. Then for all $\rho(CZ)\in\cN(\cC(CZ))$,
  \begin{equation}
    \Prob_{\mu(CZ)}\left(1/(\epsilon F(CZ))<\hatRpow{\alpha}{\rho(CZ)}{\rho(Z)}
    \right)\leq \epsilon,
  \end{equation}
  where
  $\mu(CZ)=\tr(\rho(CZ))$. 
\end{theorem}

According to the theorem, the interval $[0,1/(\epsilon F(CZ))]$ has
coverage probability at least $1-\epsilon$ for
$\hatRpow{\alpha}{\rho(CZ)}{\rho(Z)}$ which is what is required of a
confidence interval at level $\epsilon$ (or confidence level
$1-\epsilon$).

\begin{proof}
  According to the \QEF inequality at $\rho(CZ)$,
  \begin{align}
    \Exp_{\mu(CZ)}\left(F(CZ)\hatRpow{\alpha}{\rho(CZ)}{\rho(Z)}\right)
    &= \sum_{cz}\mu(cz)F(cz)\hatRpow{\alpha}{\rho(cz)}{\rho(z)} \notag\\
    &= \sum_{cz}F(cz)\tr(\rho(cz))
        \hatRpow{\alpha}{\rho(cz)}{\rho(z)}\notag\\
    &=  \sum_{cz}F(cz)\Rpow{\alpha}{\rho(cz)}{\rho(z)}\notag\\
    &\leq 1.
  \end{align}
  Since $F(CZ)\hatRpow{\alpha}{\rho(CZ)}{\rho(Z)}\geq 0$ and
  by the Markov inequality,
  \begin{equation}
    \Prob_{\mu(CZ)}(F(CZ)\hatRpow{\alpha}{\rho(CZ)}{\rho(Z)} >1/\epsilon)\leq \epsilon.
  \end{equation}
  The theorem follows by rearranging the inequality defining
  the event in the probability on the left-hand side.
\end{proof}

We remark that the normalized $\alpha$-R\'enyi powers generalize the
$\beta$-power of conditional probabilities when $\Pfnt{E}$ is trivial.
This motivates our terminology and the description of the framework as
``quantum probability estimation''.

\begin{lemma}
  Let $0\leq\rho\ll\sigma$ and $p\geq 0$. Then
  \begin{equation}
    p^{\beta}\tr(\pospart{\rho-p\sigma})\leq
    p^{\beta}\tr(\rho\suppproj{\rho-p\sigma})\leq \Rpow{\alpha}{\rho}{\sigma}.
  \end{equation}
  \label{lem:renyimaxprob}
\end{lemma}

This lemma is one step in the proof of Prop.~6.2, Pg.~95 of
Ref.~\cite{tomamichel:qc2012a}, where it is applied with Petz R\'enyi
entropy in mind. That it works for sandwiched R\'enyi entropy is
established in the proof of Lem.~B.4., Ref.~\cite{dupuis:qc2016a}.

\begin{proof}
  The first inequality of the lemma follows from
  Lem.~\ref{lem:trpospart}.  For the second inequality, let
  $(\ket{i})_{i=1}^{k}$ be an eigenbasis of $\pospart{\rho-p\sigma}$
  ordered so that
  $\ket{i}$ has positive eigenvalue iff $i\in[l]$, where $l$
  is the number of positive eigenvalues of $\pospart{\rho-p\sigma}$
  counting multiplicity. 
  Write
  $\rho_{ii}=\bra{i}\rho\ket{i}$ and
  $\sigma_{ii}=\bra{i}\sigma\ket{i}$.  Because
  $\rho-p\sigma=\pospart{\rho-p\sigma}-\pospart{p\sigma-\rho}$, and since
  $\pospart{\rho-p\sigma}$ and $\pospart{p\sigma-\rho}$ have orthogonal
  supports, we have
  $\tr(\pospart{\rho-p\sigma})=\sum_{i=1}^{l}(\rho_{ii}-p\sigma_{ii})$
  and for each $i\in[l]$, $\rho_{ii}\geq p\sigma_{ii}$.  Since
  $\rho\ll\sigma$, $\rho_{ii}>0$ implies $\sigma_{ii}>0$.  From the
  data-processing inequality for R\'enyi powers
  (Lem.~\ref{lem:dataprocessing_rp}) with respect to decoherence in
  the $(\ket{i})_{i=1}^{k}$ basis,
  \begin{align}
    \Rpow{\alpha}{\rho}{\sigma}
    &\geq \Rpow{\alpha}{\sum_{i}\rho_{ii}\hat i}{\sum_{i}\sigma_{ii}\hat i}\notag\\
    &= \sum_{i=1}^{k}\Rpow{\alpha}{\rho_{ii}\hat i}{\sigma_{ii}\hat i}
    \notag\\
    &= \sum_{i=1}^{k}\rho_{ii}\frac{\rho_{ii}^{\beta}}{\sigma_{ii}^{\beta}}\notag\\
    &\geq \sum_{i=1}^{l}\rho_{ii}\frac{\rho_{ii}^{\beta}}{\sigma_{ii}^{\beta}}\notag\\
    &\geq \sum_{i=1}^{l}\rho_{ii} p^{\beta},
  \end{align}
  where the last inequality follows from $\rho_{ii}\geq p\sigma_{ii}$ for all $i\in[l]$. 
  Continuing
  \begin{align}
    \Rpow{\alpha}{\rho}{\sigma}&\geq
    p^{\beta}\sum_{i=1}^{l} \tr(\rho\hat i)\notag\\
    &= p^{\beta}\tr(\rho\sum_{i=1}^{l} \hat i)\notag\\
    &\geq p^{\beta}\tr(\rho\suppproj{\rho-p\sigma}).
  \end{align}
\end{proof}

The next theorem suggests another way in which \QEFs can be
interpreted as estimators.  The statement is not far from a
conditional min-entropy estimate.

\begin{theorem}\label{thm:qef_suppproj}
  Let $\rho(CZ)\in \cS_{1}(CZ\Pfnt{E})$ and
  suppose that $F(CZ)\geq 0$ satisfies the \QEF inequality with power $\beta$
  at $\rho(CZ)$ for $C|Z$.  Then for all $\epsilon>0$,
  \begin{equation}
    \sum_{cz}\tr(
    \pospart{\rho(cz)-\frac{1}{(\epsilon F(cz))^{1/\beta}}\rho(z)}) 
    \leq \sum_{cz}\tr(\rho(cz)
    \suppproj{\rho(cz)-\frac{1}{(\epsilon F(cz))^{1/\beta}}\rho(z)}) \leq  
    \epsilon.
  \end{equation}
\end{theorem}
The theorem does not require $F(CZ)$ to be a \QEF for a specific model.

\begin{proof}
  The first inequality is an application of Lem.~\ref{lem:trpospart}.
  For the second, we apply Lem.~\ref{lem:renyimaxprob} as follows:
  \begin{align}
    \sum_{cz}\tr(\rho(cz)
    \suppproj{\rho(cz)-\frac{1}{(\epsilon F(cz))^{1/\beta}}\rho(z)})
    \hspace*{-2in}&\notag\\
    &= \sum_{cz}\epsilon F(cz)\frac{1}{\epsilon F(cz)}\tr(\rho(cz)
    \suppproj{\rho(cz)-\frac{1}{(\epsilon F(cz))^{1/\beta}}\rho(z)})
    \notag\\
    &\leq
    \sum_{cz}\epsilon F(cz)\Rpow{\alpha}{\rho(cz)}{\rho(z)}\notag\\
    &\leq \epsilon,
  \end{align}
  according to the \QEF inequality and since $\tr(\rho)=1$.
\end{proof}

\subsection{Entropy Estimates From \QEFs}

\begin{theorem}
  \label{thm:qef_epmax}
  Let $\rho(CZ)\in\cS_{1}(CZ\Pfnt{E})$ and suppose that $F(CZ)\geq 0$
  satisfies the \QEF inequality with power $\beta$ at $\rho(CZ)$ for $C|Z$.  
  Fix $1\geq p>0$ and $\eps>0$ and write $\phi(CZ)=\left(F(CZ)\geq
    1/(p^{\beta}\epsilon)\right)$.  Let $\phi'(CZ)$ satisfy
  $\{\phi'(CZ)\}\subseteq\{\phi(CZ)\}$, and define
  $\kappa=\tr(\rho(\phi'))$. Then
  \begin{equation}
    \kappa \sum_{cz:\phi'(cz)} \tr(
    \pospart{\rho(cz|\phi')-
     \frac{p}{\kappa}\rho(z)})\leq\epsilon.
  \end{equation}
\end{theorem}

The quantity $\kappa$ is the probability that $\phi'$ holds at
$\rho(CZ)$. Again, the theorem does not require $F(CZ)$ to be 
a \QEF for a specific model.

\begin{proof}
  Without loss of generality, let $\kappa>0$.
  Define $p(cz)=1/(\epsilon F(cz))^{1/\beta}$.
  For $cz$ satisfying $\phi'(cz)$, we have $p\geq p(cz)$.
  By Thm.~\ref{thm:qef_suppproj} 
  \begin{align}
    \epsilon &\geq \sum_{cz} \tr(\pospart{\rho(cz)-p(cz)\rho(z)})
    \notag\\
    &\geq \sum_{cz:\phi'(cz)}
    \tr(\pospart{\rho(cz)-p(cz)\rho(z)})\notag\\
    &= \sum_{cz:\phi'(cz)}
    \tr(\pospart{\rho(cz)\knuth{\phi'(cz)}-p(cz)\rho(z)})\notag\\
    &= \kappa \sum_{cz:\phi'(cz)} \tr(\frac{1}{\kappa}
    \pospart{\rho(cz)\knuth{\phi'(cz)}-p(cz)\rho(z)})\notag\\
    &= \kappa \sum_{cz:\phi'(cz)} \tr(
    \pospart{\frac{1}{\kappa}\rho(cz)\knuth{\phi'(cz)}-
      \frac{p(cz)}{\kappa}\rho(z)})\notag\\
    &= \kappa \sum_{cz:\phi'(cz)} \tr( \pospart{\rho(cz|\phi')-
      \frac{p(cz)}{\kappa}\rho(z)})\notag\\
    &\geq \kappa \sum_{cz:\phi'(cz)} \tr( \pospart{\rho(cz|\phi')-
      \frac{p}{\kappa}\rho(z)}),
  \end{align}
  since $\tr(\pospart{\chi})$ is monotone in $\chi$.
\end{proof}

We can obtain a conditional min-entropy bound from
Thm.~\ref{thm:qef_epmax} after applying Lem.~6.1, Pg.~94 of
Ref.~\cite{tomamichel:qc2012a} and Lem.~\ref{lem:renyimaxprob}, in the
spirit of Prop.~6.2, Pg.~95 of the same reference.  This proposition was
extended to sandwiched R\'enyi entropies by Lem.~B.4 of
Ref.~\cite{dupuis:qc2016a}. The statement of Lem.~B.4 contains an 
unnecessary restriction $\alpha\leq 2$: The data processing inequality 
for sandwiched R\'enyi entropy applies for all $\alpha>1$. The 
same result for all $\alpha>1$ 
is a consequence of Prop.~6.5, Pg.~99 of Ref.~\cite{tomamichel:qc2015a}. 
Instead of deriving a conditional min-entropy bound from Thm.~\ref{thm:qef_epmax}, 
we apply this Prop.~6.5  to the conditional R\'eny power bound in the first 
part of the next theorem, in order to obtain the conditional max-prob bound 
in the second part. 

\begin{theorem}\label{thm:bnds_from_qef}
  Let $\rho(CZ)\in \cS_{1}(CZ\Pfnt{E})$ and suppose that $F(CZ)\geq 0$
  satisfies the \QEF inequality with power $\beta$ at $\rho(CZ)$ for
  $C|Z$.  Fix $\delta,q\in(0,1]$, and set
  $p=q/\delta^{1/\beta}$. Write $\phi(CZ)=\left(F(CZ)\geq
    1/(q^{\beta})\right)$.  Let $\phi'(CZ)$ satisfy
  $\{\phi'(CZ)\}\subseteq\{\phi(CZ)\}$, and define
  $\kappa=\tr(\rho(\phi'))$. Then
  \begin{equation}
    \sum_{cz}\Rpow{\alpha}{\rho(cz|\phi')}{\rho(z)}\leq \frac{q^{\beta}}{\kappa^{\alpha}}
    \label{eq:thm:from_qef:ralpha}
  \end{equation}
  and
  \begin{equation}
    P^{\sqrt{2\delta}}_{\max}(\rho(cz|\phi')|Z\Pfnt{E})\leq
    P^{\sqrt{2\delta-\delta^{2}}}_{\max}(\rho(cz|\phi')|Z\Pfnt{E})
    \leq \frac{p}{\kappa^{\alpha/\beta}}.
    \label{eq:thm:from_qef:pmax}
  \end{equation}
\end{theorem} 
Again, the theorem does not require $F(CZ)$ to be a \QEF for a specific model. 

\begin{proof}
  For the first part, it suffices to rewrite the QEF inequality and drop terms:
  \begin{align}
    1 &\geq \sum_{cz}F(cz)\Rpow{\alpha}{\rho(cz)}{\rho(z)}\notag\\
    &\geq \sum_{cz}F(cz)\knuth{\phi'(cz)}\Rpow{\alpha}{\rho(cz)}{\rho(z)}\notag\\
    &\geq \sum_{cz}\frac{1}{q^{\beta}}\knuth{\phi'(cz)}
    \Rpow{\alpha}{\rho(cz)}{\rho(z)}\notag\\
    &=\sum_{cz}\frac{1}{q^{\beta}}
    \Rpow{\alpha}{\knuth{\phi'(cz)}\rho(cz)}{\rho(z)}\notag\\
    &=\sum_{cz}\frac{\kappa^{\alpha}}{q^{\beta}}
    \Rpow{\alpha}{\knuth{\phi'(cz)}\rho(cz)/\kappa}{\rho(z)}\notag\\
    &=\sum_{cz}\frac{\kappa^{\alpha}}{q^{\beta}}
    \Rpow{\alpha}{\rho(cz|\phi')}{\rho(z)}.
  \end{align}
  The claimed inequality is obtained by multiplying both sides
  by $q^{\beta}/\kappa^{\alpha}$.

  For the second part, we interpret Eq.~\ref{eq:thm:from_qef:ralpha}
  as a sandwiched $\alpha$-R\'enyi relative entropy bound. According 
    to Def.~\ref{def:sandwiched_Renyi} we have 
  \begin{align} \label{eq:renyi_power_decomposition}
    & \Rpow{\alpha}{\sum_{cz}\hat{c}\otimes\hat{z}\otimes\rho(cz|\phi')}{\sum_{z}\one\otimes\hat{z}\otimes\rho(z)} \notag \\
    &=\tr( \bigg( \Big(\sum_{z}\one\otimes\hat{z}\otimes\rho(z)^{-\beta/(2\alpha)}\Big) \Big(\sum_{cz}\hat{c}\otimes\hat{z}\otimes\rho(cz|\phi')\Big) \Big(\sum_{z}\one\otimes\hat{z}\otimes\rho(z)^{-\beta/(2\alpha)}\Big) \bigg)^{\alpha} ) \notag \\
    &= \tr( \sum_{cz}\hat{c}\otimes\hat{z}\otimes \Big(\rho(z)^{-\beta/(2\alpha)}\rho(cz|\phi')\rho(z)^{-\beta/(2\alpha)}\Big)^\alpha) \notag \\
    &= \sum_{cz} \tr(\Big(\rho(z)^{-\beta/(2\alpha)}\rho(cz|\phi')\rho(z)^{-\beta/(2\alpha)}\Big)^\alpha) \notag \\
    &=\sum_{cz} \Rpow{\alpha}{\rho(cz|\phi')}{\rho(z)}.
  \end{align} 
  We can now apply Prop.~6.5, Pg. 99 of Ref.~\cite{tomamichel:qc2015a}.  
  We convert to our notation, and substitute for $\epsilon$ in the reference 
  according to $\delta=1-\sqrt{1-\epsilon^{2}}$ (equivalently, $\epsilon =
  \sqrt{2\delta-\delta^{2}}$), the operator $\rho$ there by
  $\sum_{cz}\hat{c}\otimes\hat{z}\otimes\rho(cz|\phi')$ here, and
  $\sigma$ there by $\sum_{z}\one\otimes\hat{z}\otimes\rho(z)$ here.
  This gives 
  \begin{align} \label{eq:renyi_power_bound}
    & \inf_{\rho'}\inf\{p': \rho'(CZ)\leq p'\rho(Z), \rho'(CZ)\in\cS_{\leq 1}(CZ\Pfnt{E}),
    \purdist{\rho'(CZ)}{\rho(CZ|\phi')}\leq\sqrt{2\delta-\delta^{2}}\} \notag \\
    & \leq \left(\frac{1}{\delta}\Rpow{\alpha}{\sum_{cz}\hat{c}\otimes\hat{z}\otimes\rho(cz|\phi')}{\sum_{z}\one\otimes\hat{z}\otimes\rho(z)}\right)^{1/\beta}.
  \end{align}
  Taking note of the definition of $P_{\max}^{\epsilon}$ in Def.~\ref{def:smooth_max_prob}, 
  we get 
  \begin{align} \label{eq:smooth_entropy_bound}
    & P_{\max}^{\sqrt{2\delta-\delta^{2}}}(\rho(CZ|\phi')|Z\Pfnt{E}) \notag \\
    & \leq \inf_{\rho'}\inf\{p': \rho'(CZ)\leq p'\rho(Z), \rho'(CZ)\in\cS_{\leq 1}(CZ\Pfnt{E}),
    \purdist{\rho'(CZ)}{\rho(CZ|\phi')}\leq\sqrt{2\delta-\delta^{2}}\}. 
  \end{align}
  Combining Eqs.~\ref{eq:renyi_power_decomposition},~\ref{eq:renyi_power_bound}, and~\ref{eq:smooth_entropy_bound},
  we get 
  \begin{equation}
    P_{\max}^{\sqrt{2\delta-\delta^{2}}}(\rho(CZ|\phi')|Z\Pfnt{E})
    \leq \left(\frac{1}{\delta}\sum_{cz}\Rpow{\alpha}{\rho(cz|\phi')}{\rho(z)}\right)^{1/\beta}.
  \end{equation}
  Continuing from the right-hand side and applying
  Eq.~\ref{eq:thm:from_qef:ralpha}, we get
  \begin{align}
    P_{\max}^{\sqrt{2\delta-\delta^{2}}}(\rho(CZ|\phi')|Z\Pfnt{E})
    &\leq \left(\frac{q^{\beta}}{\delta \kappa^{\alpha}}\right)^{1/\beta}\notag\\
    &= \frac{p}{\kappa^{\alpha/\beta}}.
  \end{align}
  Since $P_{\max}^{\epsilon}$ is monotonic in the smoothness parameter
  $\epsilon$, the proof of the second part of the theorem is complete.
\end{proof}

We also use a simplified version of Thm.~\ref{thm:bnds_from_qef} where
$F(CZ)$ has a uniform lower bound:
\begin{corollary}\label{cor:qef_epmax2}
  Fix $\delta,p\in(0,1]$.  Let $\rho(CZ)\in \cS_{1}(CZ\Pfnt{E})$ and
  suppose that $F(CZ)\geq 1/(p^{\beta}\delta)$ satisfies the \QEF
  inequality with power $\beta$ at $\rho(CZ)$ for $C|Z$.  Then
  $P_{\max}^{\sqrt{2\delta}}(\rho(CZ)|Z\Pfnt{E})\leq p$.
\end{corollary}

\begin{proof}
  It suffices to apply Eq.~\ref{eq:thm:from_qef:pmax} with $\kappa=1$.
\end{proof}

\section{\QEF-based Randomness Generation Protocols}
\label{sec:qef_protocols}

\subsection{Protocol Soundness and Completeness}
\label{subsec:soundness}

A generic randomness generation protocol $\cG$ produces three outputs:
a bit string of length $k_{o}$, a length $k_{u}$ bit string consisting
of potentially reusable random bits and a ``flag'' indicating failure
or success.  We write $\cG = (\cG_{X},\cG_{S},\cG_{P})$ accordingly
where $\cG_{X}$ is the bit string of length $k_{o}$, $\cG_{S}$ the bit
string of length $k_{u}$ and $\Rng(\cG_{P})=\{0,1\}$.  The values $0$
and $1$ of $\cG_{P}$ indicate failure and success, respectively.  The
outputs $\cG_{X}$, $\cG_{S}$ and $\cG_{P}$ are determined by CVs
associated with a sequence of trials involving the devices of the
protocols and a seed bit-string CV. Parameters of $\cG$ include
$k_{o}$, $k_{u}$, the length $k_{s}$ of the seed CV, and a target
error bound $\epsilon$. Other parameters may be relevant before the
protocol is invoked, such as the maximum number of trials $N$ and, after it
has executed, the number $n$ of trials actually performed and the
number of bits $k_{z}$ of input randomness used.

Informally, a protocol is $\epsilon$-sound if its output is within
$\epsilon$ of an ideal protocol. The distance measure used determines
the protocol's composability properties.  There is some variation in
the soundness definitions for randomness generation protocols in the
literature. We prove soundness with respect to purified distance,
which is stronger than other definitions. It implies soundness with
respect to TV distance including the devices, which is better
behaved for composability analyses.

\begin{definition}
  Let $\Sfnt{CZ}$ and $S$ be CVs, where $S$ is a length $k_{s}$ bit
  string, and let $\sigma(S)$ be the uniform distribution, 
  that is  $\sigma(S)=\Unif(S)$. 
  A randomness generation protocol $\cG=(\cG_{X},\cG_{S},\cG_{P})$ determined by
  $\Sfnt{CZ}S$ is \emph{$\epsilon$-sound for $\Sfnt{C}|\Sfnt{Z}$ at
    $\rho(\Sfnt{CZ})\in\cS_{1}(\Sfnt{CZ}\Pfnt{E})$} if there exists
  $\tau(Z)\in\cS_{1}(\Sfnt{Z}\Pfnt{E})$ such that
  \begin{equation}
    \PD\left(\vphantom{\big|}(\rho\otimes \sigma)(\cG_{X}\cG_{S}\Sfnt{Z}|\cG_{P}=1),
      \Unif(\cG_{X}\cG_{S})\otimes\tau(\Sfnt{Z})\right)
    \tr(\vphantom{\big|}(\rho\otimes\sigma)(\cG_{P}=1))
    \leq\epsilon.\label{eq:soundnessdef}
  \end{equation}
  $\cG$ is \emph{$\epsilon$-sound for $\Sfnt{C}|\Sfnt{Z}$ and model
    $\cC(\Sfnt{CZ})$} if it is $\epsilon$-sound for
  $\Sfnt{C}|\Sfnt{Z}$ at all $\rho(\Sfnt{CZ})\in\cC(\Sfnt{CZ})$.
  $\cG$ is \emph{$\kappa$-complete for $\Sfnt{C}|\Sfnt{Z}$ and model
    $\cC(\Sfnt{CZ})$} if there exists
  $\rho(\Sfnt{CZ})\in\cC(\Sfnt{CZ})$ such that
  $\tr((\rho\otimes\sigma)(\cG_{P}=1))\geq \kappa$.  
\end{definition}
If required for clarity, we may refer to the soundness in this
definition as \emph{PD soundness}.  Completeness is important to
ensure that protocols can be usefully realized.  For our protocols and
models with extractable randomness, completeness is readily achieved
with an exponentially good completeness parameter. In practice,
completeness parameters cannot be relied on to be exponentially good.
Further, the idea of device-independent protocols is that the devices
are minimally trusted, so regardless of completeness or other
expectations of the experimental configuration, provisions for failure
must be made to mitigate denial-of-service and mundane device
faults. Soundness makes sure that any randomness produced has
guaranteed performance even in the context of probabilities of success
that are temporarily or permanently far from $1$.

It is possible to consider soundness statements involving seed CVs
whose distributions are not uniform, but this requires extractors
satisfying stronger conditions than the quantum-proof strong
extractors considered here. See~\cite{kessler:qc2017a} and the
references therein for recent work with less-than-perfect seeds.

It may be desirable to have the purified distance conditional on success
be bounded by $\delta$ given that the success probability is larger
than some small threshold $\kappa$. For this it suffices to choose the
soundness error $\epsilon$ as $\epsilon\leq \delta\kappa$.  If one
wishes to be equally conservative for both $\delta$ and $\kappa$, it
makes sense to set $\epsilon=\delta^{2}$.

The purified distance allows for extension to the devices to enable
analysis of protocol composition involving the same devices, where the
devices may have memory. This kind of composition can introduce the
possibility of memory attacks, whereby the devices leak information
about past results through leakage channels enabled by later
protocols~\cite{barrett_j:qc2012a}.  For our randomness generation
protocols, such a leakage channel is introduced by the success
variable $\cG_{P}$: The devices can modify their future behavior so
that the variables $\cG_{P}$ in later protocols depend on the past
results.  This favors protocols with no possibility of failure such as
Protocol~\ref{prot:condgen_banked} below. A detailed discussion of
memory attacks for randomness generation is in the 
supplemental material of
Ref.~\cite{barrett_j:qc2012a}. We note that our protocols have fixed
length outputs, which avoids leakage channels based on the length of
the output but does not eliminate implementation-dependent leakage
channels such as variations in timing or side-effects of using
randomness.

We do not formally analyze composition of randomness-generation
protocols with the same devices, and unrestricted composability is not
assured.  But to support such composition, we require that the devices
are permanently isolated from $\Pfnt{E}$ and that they never gain
knowledge of seeds used for randomness extraction.  The latter
supports the following strategy to mitigate $\cG_{P}$-based leakage
channels: Anticipate the number of future instances of the protocol
and reduce the number of bits extracted from the current protocol
accordingly, similarly to how settings entropy is eliminated in
Protocol~\ref{prot:condimplicit} below. The requirements may be
difficult to guarantee in a practical setting but can be weakened once
the randomness generated is used, see the discussion in
Ref.~\cite{barrett_j:qc2012a}.

PD soundness implies a strong TV distance-based soundness.  Let
$\Pfnt{D}$ be the system of the devices and let
$\rho'(\cG_{X}\cG_{S}\Sfnt{Z}\cG_{P})\in
\cS_{1}(\cG_{X}\cG_{S}\Sfnt{Z}\cG_{P}\Pfnt{DE})$ be the final state of
the protocol. Thus $\rho(\cG_{X}\cG_{S}\Sfnt{Z}\cG_{P}) =
\tr_{\Pfnt{D}}\rho'(\cG_{X}\cG_{S}\Sfnt{Z}\cG_{P})$.  Here,
information about $\Sfnt{C}$ may be contained in the quantum part of
the state carried by $\Pfnt{D}$.  If the protocol is PD $\epsilon$-sound,
then the extension property of purified distance
(Ref.~\cite{tomamichel:qc2012a}, Cor.~3.6, Pg.~52) and the
relationship between $\TV$ and purified distances
(Lem.~\ref{lem:purdistprops}) imply that there exists a state
$\tau'(\cG_{X}\cG_{S}\Sfnt{Z})
\in\cS_{1}(\cG_{X}\cG_{S}\Sfnt{Z}\Pfnt{DE})$ such that
\begin{equation} \TV\left(\rho'(\cG_{X}\cG_{S}\Sfnt{Z}|\cG_{P}=1),
    \tau'(\cG_{X}\cG_{S}\Sfnt{Z})\right)
  \tr(\vphantom{\big|}(\rho\otimes\sigma)(\cG_{P}=1)) \leq\epsilon
  \label{eq:tvsoundness}
\end{equation} and
\begin{equation} \tr_{\Pfnt{D}}\tau'(\cG_{X}\cG_{S}\Sfnt{Z}) =
  \Unif(\cG_{X}\cG_{S})\otimes\tau(\Pfnt{Z}),
  \label{eq:tvsoundnesstr}
\end{equation} with $\tau(\Sfnt{Z})$ witnessing PD $\epsilon$-soundness.  
This construction can be used to justify the informal idea that
$\epsilon$-soundness relates to how close the protocol endstate
is from that of an ideal protocol.
From $\tau'(\cG_{X}\cG_{S}\Sfnt{Z})$, we
can construct an ideal protocol endstate
$\xi(\cG_{X}\cG_{S}\Sfnt{Z}\cG_{P})\in
\cS_{1}(\cG_{X}\cG_{S}\Sfnt{Z}\cG_{P}\Pfnt{DE})$ that
includes the devices and is $\epsilon$-close in $\TV$ distance to the
actual state:
\begin{equation} \xi(\cG_{X}\cG_{S}\Sfnt{Z}\cG_{P})=
\tr(\vphantom{\big|}(\rho\otimes\sigma)(\cG_{P}=1))\tau'(\cG_{X}\cG_{S}\Sfnt{Z})
+ \rho'(\cG_{X}\cG_{S}\Sfnt{Z},\cG_{P}=0).
\end{equation} This state satisfies $\epsilon$-soundness with
$\epsilon=0$ and agrees with the protocol endstate conditionally
on failure. The existence of this state motivates the definition of
soundness and our use of purified distance. But for composability
analysis, we use TV $\epsilon$-soundness including the devices, which
we define as existence of the state $\tau'$ satisfying 
Eqs.~\ref{eq:tvsoundness} and~\ref{eq:tvsoundnesstr}.

The protocols below are proved to be PD $\epsilon$-sound regardless of
the incoming state and dependence on initial classical variables that
may be public and on variables determined from initial information. TV
soundness extends to such initial variables without changing the error
bound.  From the previous paragraph, the TV soundness error is
uniformly bounded by $\epsilon$ given these initial variables, as is
the distance from an ideal protocol conditional on the initial
variables.  We can define an unconditional ideal protocol by having it
act as the conditional ideal protocol given the initial variables.
The probability distribution for the initial variables is the same for
the actual and the ideal protocol. The TV distance between two states
classical on $R$ with identical marginal distribution on $R$ is the
expected $R$-conditional TV distance.  It follows that the distance
between the two is the expected distance conditional on the initial
variables, which is less than $\epsilon$.

\Ac{It would be helpful to replace $\TV$ and purified distance with
  a distance measure that respects classical variables in that it
  agrees with $\TV$ distance on classical states and purified
  distance on fully quantum states. Perhaps one could define the
  distance between normalized $\rho(U)$ and $\sigma(U)$ by
  $\sum_{u}\min\{\TV(\phi,\psi): \textrm{$\phi,\psi$ are
    purifications of $\rho(u),\sigma(u)$}\}$.  For this definition,
  take note of the identity $\PD(\xi,\chi)=\min\{\TV(\phi,\psi):
  \textrm{$\phi,\psi$ are purifications of $\xi,\chi$}\}$ for
  normalized states $\xi,\chi$, see Ref.~\cite{tomamichel:qc2012a},
  Pg.~49.  One should check the that this expression has the
  properties required of a distance in information processing.}

\subsection{Protocols with \QEFs}

We define three sound randomness generation protocols given a \QEF.
Whether they are complete depends on the model and the \QEF.  The
results established later show that if a trial model permits proper
randomness generation in principle, then completeness with
exponentially good completeness parameter is readily achieved for
sequences of independent and identical (i.i.d.) trials, each constrained by
the trial model. This generally follows from large deviation results
applied to sums of i.i.d. RVs. 
For our protocols, these RVs are the logarithms of the 
\QEFs.  We do not explore the relevant arguments further here.

In this section we consider monolithic \QEFs and $CZ$, meaning
that we do not explicitly subdivide the results into a sequence of trials.
Thus, $CZ$ stands for all results, whether or not they were obtained
in a sequence of trials, and the \QEFs are the final \QEFs, obtained
by chaining if necessary.
Protocol-related issues when the \QEFs are determined by chaining are
discussed in Sect.~\ref{subsect:trialwise}. Anticipating the amount of
conditional min-entropy that can be certified is the topic of
Sect.~\ref{subsect:logprobrates}.

The first protocol directly composes Thm.~\ref{thm:bnds_from_qef} on the
relationship between \QEFs and smooth max-prob with a quantum-proof
strong extractor. The protocol is displayed in
Protocol~\ref{prot:condgen_direct}. We use the notation $a^{\conc k}$
to denote the $k$-fold concatenation of $a$ with itself.

\vspace*{\baselineskip}
\begin{algorithm}[H]
  \caption{Input-conditional randomness
    generation.}\label{prot:condgen_direct}
  
  \Input{Number of bits of randomness $k_{o}$ to be generated. Error
    bound $\epsilon\in(0,1]$. }
  
  \Given{Access to CVs ${CZ}$ and $S$, where $S$ is uniformly distributed
  and independent of all other systems. All CVs are represented by bit strings. 
  A \QEF $F({CZ})$ with  power $\beta$ for ${C}|{Z}$ and model
    $\cC({CZ})$.  A quantum-proof strong extractor $\cE$.
    }  
  
  \Output{Length $k_{o}$ bit string $\cG_{X}$,
    $\cG_{S}=S$, $\cG_{P}\in\{0,1\}$.}  
  
  \BlankLine 

  Define $n=|C|$,  $k_{s}=|S|$\; 
  
  Define $\cX=\{(k_{i},\epsx):\textrm{$(n,k_{s},k_{o},k_{i},\epsx<\epsilon)$
    satisfies the extractor 
    constraints for $\cE$}\}$
  \tcp*{See the paragraph after Def.~\ref{def:extractor}.}

  Get an instance $s$ of $S$\;

  \eIf{$\cX$ is empty}
  { Return $\cG_{P}=0$, $\cG_{X}=0^{\conc k_{o}}$, $\cG_{S}=s$
    \tcp*{Protocol failed.} }
  { Choose $(k_{i},\epsx)\in\cX$\;
    
    Set $\epse=(\epsilon-\epsx)$\; 
    If $\alpha>2$, then set $p= 2^{-k_{i}}\epsilon^{(\alpha-2)/\beta}$, otherwise set $p=2^{-k_{i}}$\;
    
    Set $f_{\min}=1/(p^{\beta}(\epse^{2}/2))$ \tcp*{Choose $(k_{i},\epsx)$ to
      minimize $f_{\min}$.} 

    Get an instance ${cz}$ of  ${CZ}$\; 
    
    Compute $f=F({cz})$.  
    
    \eIf{$f<f_{\min}$}
    {
      Return $\cG_{P}=0$, $\cG_{X}=0^{\conc k_{o}}$, $\cG_{S}=s$
      \tcp*{Protocol failed.} 
    }
    { Return $\cG_{P}=1$,
      $\cG_{X}=\cE(c,s;n,k_{s},k_{o},k_{i},\eps_{x})$, $\cG_{S}=s$   
      \tcp*{Protocol succeeded.} 
    } 
  
  }
\end{algorithm}

\begin{theorem}\label{thm:condgen_direct}
  Protocol~\ref{prot:condgen_direct} is an $\epsilon$-sound randomness generation protocol
  for ${C}|{Z}$ and model $\cC({CZ})$.
\end{theorem}

\begin{proof}
  According to our modeling assumptions, the model applies
  conditionally on the past, which includes the protocol inputs
  $k_{o}$ and $\epsilon$ and the specific choice for $(k_{i}, k_s, \epsx)$
  made in the protocol, as these parameters are determined before
  ${CZ}$ is instantiated. Let $\rho({CZ})\in\cC({CZ})$ be the specific
  state from which ${CZ}$ is instantiated to ${cz}$ in the protocol.
  Let $\phi({CZ})=\left(F({CZ})\geq
    f_{\min}\right)=\left(\cG_{P}=1\right)$.  Define
  $\kappa=\tr(\rho(\phi))$. First consider the case $\kappa\in[\epsilon,1]$.
 In Thm.~\ref{thm:bnds_from_qef}, set
  $\delta=(\epse/\kappa)^{2}/2$ and $p$ there to
  $p\kappa^{2/\beta}$ here. With these substitutions, $q$ there
  satisfies
  $q=(p\kappa^{2/\beta})((\epse/\kappa)^{2}/2)^{1/\beta}=
  p(\epse^{2}/2)^{1/\beta}= 1/f_{\min}^{1/\beta}$ so that the lower
  bound on $F(CZ)$ in the definition of $\phi(CZ)$ there is $f_{\min}$,
  which is the lower bound on $F(CZ)$ in the protocol
  required for success, that is for $\cG_{P}=1$. 
  Applying Thm.~\ref{thm:bnds_from_qef} therefore gives
  \begin{equation}
    P^{\epse/\kappa}_{\max}(\rho({CZ}|\phi)|{Z}\Pfnt{E})
    \leq p\kappa^{2/\beta}/\kappa^{\alpha/\beta} = p \kappa^{(2-\alpha)/\beta},
    \label{eq:thm:condgen_direct:1}
  \end{equation}
  where $p$ is defined in the protocol so that for $\epsilon\leq
  \kappa \leq 1$, we have $p\kappa^{(2-\alpha)/\beta}\leq 2^{-k_{i}}$.
  Specifically, if $\alpha\leq 2$, then
  $p\kappa^{(2-\alpha)/\beta}\leq p=2^{-k_{i}}$, and if $\alpha>2$, then
  $p\kappa^{(2-\alpha)/\beta}\leq
  p\epsilon^{(2-\alpha)/\beta}=2^{-k_{i}}$.  Hence, when $\epsilon\leq
  \kappa \leq 1$ we have
  $P^{\epse/\kappa}_{\max}(\rho({CZ}|\phi)|{Z}\Pfnt{E})\leq
  2^{-k_{i}}$.  That is, there exists $\rho'(CZ)\in
  \cS_{1}(CZ\Pfnt{E})$ such that $P_{\max}(\rho'(CZ)|Z\Pfnt{E})\leq
  2^{-k_{i}}$ and $\PD(\rho'(CZ), \rho({CZ}|\phi))\leq \epse/\kappa$
  (see Lem.~\ref{lem:purdist_normalized}, where the extractor
  constraints ensure that $2^{-k_{i}}\leq 2^{n}=|\Rng(C)|$).  As in
  the definition of soundness, let $\sigma(S)=\Unif(S)$. Because the
  parameters $n,k_{s},k_{o},k_{i},\epsx$ satisfy the extractor
  constraints, we get
  \begin{equation}\label{eq:ex_witness_1}
    \PD\left((\rho'\otimes \sigma)(\cG_{X}\cG_{S}{Z}),
      \Unif(\cG_{X}\cG_{S})\otimes\rho'({Z})\right)\leq \epsx.
  \end{equation}
  Since  $\PD(\rho'(CZ), \rho({CZ}|\phi))\leq \epse/\kappa$ and the  
  purified distance
  satisfies the data-processing inequality, 
  \begin{equation}\label{eq:ex_witness_2}
    \PD\left((\rho\otimes \sigma)(\cG_{X}\cG_{S}{Z}|\cG_{P}=1),
    (\rho'\otimes \sigma)(\cG_{X}\cG_{S}{Z})\right)\leq \epse/\kappa.
  \end{equation}
  The triangle inequality for the purified distance together with   
  Eqs.~\ref{eq:ex_witness_1} and~\ref{eq:ex_witness_2} yield
  \begin{equation}
  \PD\left((\rho\otimes \sigma)(\cG_{X}\cG_{S}{Z}|\cG_{P}=1),  \Unif(\cG_{X}\cG_{S})\otimes\rho'({Z})\right)\leq \epsx+\epse/\kappa.
  \end{equation}
  We multiply both sides by $\kappa$ for
  \begin{equation}
    \PD\left((\rho\otimes \sigma)(\cG_{X}\cG_{S}{Z}|\cG_{P}=1),
      \Unif(\cG_{X}\cG_{S})\otimes\rho'({Z})\right)\kappa
    \leq\epsx\kappa +\epse\leq\epsx+\epse=\epsilon.
  \end{equation}
  For $\kappa<\epsilon$, since the purified distance cannot be larger than one,
  \begin{equation}
    \PD\left((\rho\otimes \sigma)(\cG_{X}\cG_{S}{Z}|\cG_{P}=1),
      \Unif(\cG_{X}\cG_{S})\otimes\rho({Z}|\cG_{P}=1)\right)\kappa
    \leq\kappa<\epsilon,
  \end{equation}
  so the condition for $\epsilon$-soundness is satisfied 
  for the full range of values of $\kappa$. 
\end{proof}

Next we define a protocol that avoids failure by taking advantage of
banked randomness. It has the advantage of simplicity at the cost of
occasionally producing randomness that is not entirely fresh, which
adds effective latency. Of course, in situations where we can
experimentally ensure completeness, it is possible to make the
probability of requiring banked randomness extremely small.  The
protocol is displayed in Protocol~\ref{prot:condgen_banked}.

\vspace*{\baselineskip}
\begin{algorithm}[H]
  \caption{Input-conditional randomness
    generation with banked randomness.}\label{prot:condgen_banked}
  \Input{Number of bits of randomness $k_{o}$ to be generated. Error
    bound $\epsilon\in(0,1]$. }
  
  \Given{Access to CVs ${CZ}$, $S$ and $B$, where $|B|=k_{o}$ and $SB$ is uniformly
    distributed and independent of all other systems.  All CVs are represented by bit strings. 
    A \QEF $F({CZ})$ with power $\beta$ for ${C}|{Z}$ and
    model $\cC({CZ})$.  A quantum-proof strong extractor $\cE$.
    }
  
  \Output{Length $k_{o}$ bit string $\cG_{X}$,
    $\cG_{S}=S$, $\cG_{P}\in\{0,1\}$.}  
  
  \BlankLine 
  
  Define $n=|C|$,  $k_{s}=|S|$\; 

  Define $\cX=\{(k_{i},\epsx):\textrm{$(n+k_{o},k_{s},k_{o},k_{i},\epsx<\epsilon)$
    satisfies the extractor constraints for $\cE$}\}$\;

  Get an instance $s$ of $S$\;

  \eIf{$\cX$ is empty}
  { Get an instance $b_{\leq k_{o}}$ of $B_{\leq k_{o}}$\;

    Return $\cG_{P}=1$,
      $\cG_{X}=b_{\leq k_{o}}$, $\cG_{S}=s$
      \tcp*{Return only banked randomness.}
  }
  { Choose $(k_{i},\epsx)\in\cX$\;
    
    Set $\epse=(\epsilon-\epsx)$\; Set $p= 2^{-k_{i}}$\;
    
    Set $f_{\min}=1/(p^{\beta}(\epse^{2}/2))$ \tcp*{Choose $(k_{i},\epsx)$ to
      minimize $f_{\min}$} 

    Get an instance ${cz}$ of  ${CZ}$\; 
    
    Compute $f=F({cz})$.  
    
    \eIf{$f\geq f_{\min}$}
    { Return $\cG_{P}=1$,
      $\cG_{X}=\cE(c0^{\conc k_{o}},s;n+k_{o},k_{s},k_{o},k_{i},\eps_{x})$, $\cG_{S}=s$
      \tcp*{No banked randomness needed.}
    } 
    {
      Set $k_{b}=\lceil\log_{2}(f_{\min}/f)/\beta\rceil$\;

      Get an instance $b_{\leq k_{b}}$ of $B_{\leq k_{b}}$\;

      Return $\cG_{P}=1$,
      $\cG_{X}=\cE(cb_{\leq k_{b}}0^{\conc k_{o}-k_{b}},s;
      n+k_{o},k_{s},k_{o},k_{i},\eps_{x})$, $\cG_{S}=s$
      \tcp*{Needed $k_{b}$ bits of banked randomness.}

    }
  }
\end{algorithm}

\begin{theorem}
  Protocol~\ref{prot:condgen_banked} is a complete and $\epsilon$-sound
  randomness generation protocol
  for ${C}|{Z}$ and model $\cC({CZ})$.
\end{theorem}

\begin{proof}
  If $f\geq f_{\min}$ in the protocol, set $k_{b}=0$.  The protocol
  can be thought of as one that adds a final trial conditionally on
  $F(cz)<f_{\min}$, where the final trial has output $B'=B_{\leq
    k_{b}}0^{\conc k_{o}-k_{b}}$, which is a bit string of length
  $k_{o}$ and model $\{\Unif(B'_{\leq
    k_{b}})\rho:\rho\in\cS(\Pfnt{E})\}$.  We can define $G_{cz}(b') =
  2^{\beta k_{b}}$, which is a \QEF with power $\beta$ for the last
  trial, and chain $F$ with $G_{cz}$ to get a \QEF
  $F'(CZB')=F(CZ)G_{CZ}(B')$ with power $\beta $ for $CB'|Z$ and the
  chained model.  By construction, $F'(CZB') \geq f_{\min}$, so we can
  apply Cor.~\ref{cor:qef_epmax2} to show that for any $\rho(CZB')$
  in the chained model,
  $P^{\epsilon_{h}}_{\max}(\rho(CZB')|Z\Pfnt{E})\leq p$.  The theorem
  follows because $\cE$ is a quantum-proof strong extractor, its
  parameters satisfy the extractor constraints, the incoming smooth
  max-prob is less than $2^{-k_{i}}$, and the data-processing and
  triangle inequalities for the purified distance.
\end{proof}

The third protocol conditions on inputs indirectly by exploiting the
privacy amplification capabilities of extractors. We give a version
not relying on banked randomness. The only difference to the first
protocol is that the conditional min-entropy certified internally
needs to also account for the maximum number of bits that contribute
to the inputs.  An advantage is that the models for which this
protocol works need not involve chaining with explicitly conditional
inputs.  The protocol is displayed in Protocol~\ref{prot:condimplicit}.

\vspace*{\baselineskip}
\begin{algorithm}[H]
  \caption{Randomness
    generation with implicit input conditioning.}\label{prot:condimplicit}
  \Input{Number of bits of randomness $k_{o}$ to be generated. Error
    bound $\epsilon\in(0,1]$. }
  
  \Given{Access to CVs $CZ$ and $S$, where $Z=Z(H)$ is determined by a CV $H$ and $S$
    is uniformly distributed and independent of all other systems.
    All CVs are represented by bit strings.  A \QEF $F(CZ)$ with power $\beta$
    for $CZ$ and model $\cC(CZ)$.  A quantum-proof
    strong extractor $\cE$.}
  
  \Output{Length $k_{o}$ bit string $\cG_{X}$,
    $\cG_{S}=S$, $\cG_{P}\in\{0,1\}$.}  
  
  \BlankLine 
  
  Define $n=|CZ|$,  $k_{s}=|S|$, $k_{z}=|H|$\; 

  Define $\cX=\{(k_{i},\epsx):\textrm{$(n,k_{s},k_{o},k_{i},\epsx<\epsilon)$
    satisfies the extractor constraints for $\cE$}\}$\;

  Get an instance $s$ of $S$\;

  \eIf{$\cX$ is empty}
  { Return $\cG_{P}=0$, $\cG_{X}=0^{\conc k_{o}}$, $\cG_{S}=s$ \tcp*{Protocol failed.}
   }
  { Choose $(k_{i},\epsx)\in\cX$\;
    
    Set $\epse=(\epsilon-\epsx)$\; If $\alpha>2$, set $p=
    2^{-k_{i}-k_{z}}\epsilon^{(\alpha-2)/\beta}$, otherwise set
    $p=2^{-k_{i}-k_{z}}$\;
    
    Set $f_{\min}=1/(p^{\beta}(\epse^{2}/2))$ \tcp*{Choose $(k_{i},\epsx)$ to
      minimize $f_{\min}$.} 

    Get an instance $cz$ of  $CZ$\; 
    
    Compute $f=F(cz)$.  
    
    \eIf{$f<f_{\min}$}
    {
      Return $\cG_{P}=0$, $\cG_{X}=0^{\conc k_{o}}$, $\cG_{S}=s$ \tcp*{Protocol failed.}
    }
    { Return $\cG_{P}=1$,
      $\cG_{X}=\cE(cz,s;n,k_{s},k_{o},k_{i},\eps_{x})$, $\cG_{S}=s$
      \tcp*{Protocol succeeded.}
    } 
  
  }
\end{algorithm}

\begin{theorem}
  Protocol~\ref{prot:condimplicit} is an $\epsilon$-sound
  randomness generation protocol
  for ${C}|{Z}$ and model $\cC({CZ})$.
\end{theorem}

\begin{proof}
  The proof follows that of Protocol~\ref{prot:condgen_direct}. For the
  initial part, $CZ$ are both considered output and there is no explicit input.
  For the case $\kappa\in[\epsilon,1]$, the max-prob established for this protocol is
  \begin{equation}
    P^{\epse/\kappa}_{\max}(\rho({CZ}|\phi)|\Pfnt{E})
    \leq p\kappa^{2/\beta}/\kappa^{\alpha/\beta} \leq 2^{-k_{i}-k_{z}}.
  \end{equation}
  Since $Z$ is determined by $H$ and invoking Lem.~\ref{lem:addcondition}
  and Lem.~\ref{lem:pmaxdetermined} we get
  \begin{align}
    P^{\epse/\kappa}_{\max}(\rho({CZ}|\phi)|{Z}\Pfnt{E})
    &\leq    P^{\epse/\kappa}_{\max}(\rho({CH}|\phi)|{H}\Pfnt{E})\notag\\
    &\leq 2^{k_{z}}    P^{\epse/\kappa}_{\max}(\rho({CH}|\phi)|\Pfnt{E})\notag\\
    &\leq 2^{k_{z}}P^{\epse/\kappa}_{\max}(\rho({CZ}|\phi)|\Pfnt{E})\notag\\
    &\leq 2^{-k_{i}}.
  \end{align}
  The rest of the proof of Protocol~\ref{prot:condgen_direct} now applies
  without change.
\end{proof}

\subsection{Trial-Wise \QEF Computation for Protocols}
\label{subsect:trialwise}

For the applications we have in mind, the \QEFs $F(\Sfnt{CZ})$ used by the protocols
arise by chaining trial-wise \QEFs $F_i(C_iZ_i)$ for a sequence of trials, where 
the final model is an appropriate chaining of the trial models. An advantage of
\QEFs is that they can be adapted while the trials are acquired.  A
consequence is that one can stop acquiring trials as soon as the
chained \QEF witnesses sufficiently small R\'enyi power.  For
definiteness, we let $k$ be the number of trials performed (or
analyzed) so far. According to \QEF chaining, the next trial's \QEF
$F_{k+1}(C_{k+1}Z_{k+1})$ can depend arbitrarily on $(\Sfnt{cz})_{\leq k}$, the results from
trials so far. In particular, one can check the statistics of recent
trials to see whether the observed probability distribution of $CZ$
changed and if so, adapt the next trial's \QEFs accordingly.  Further,
if the chained \QEF so far, $\prod_{i=1}^{k}F_{i}(c_{i}z_{i})$, already exceeds the threshold for the
protocol, then one can set all future \QEFs $F_i(C_iZ_i)$ with $i>k$ to $1$. Since this
eliminates any contribution from future trials to the final chained
\QEF value, it is not necessary to perform the future trials at this
point.  Since the trial models can also depend on the past, one can
change the configuration between trials. If there is a change in trial
model, it must also be determined by $(\Sfnt{cz})_{\leq k}$ and the next
trial's \QEF needs to take the change into account.  Changes that do
not affect the model are not so restricted. For example, there are no
restrictions on device recalibration between trials.

Unlike \QEFs, soft PEFs as defined in Ref.~\cite{knill:qc2017a} can
directly use available information not determined by
$(\Sfnt{cz})_{\leq k}$ to choose the next trial's model and PEF.  We
have not implemented softening for \QEFs.  However, this is not a
fundamental obstacle. A feature of the CV $\Sfnt{CZ}$ as used in the
randomness generation protocols above is that $\Sfnt{C}$ must be
provided to the extractor, while $\Sfnt{Z}$ must be conditioned on.  A
simple method to enable use of information obtained during an
experiment besides $(\Sfnt{cz})_{\leq k}$ is the following:
Periodically, at predictable intervals, insert special trials with
output consisting of the information that one wishes to use in future
trials, but no input. These trials' outputs are ultimately included in
the extractor input or conditioned on via the method in
Protocol~\ref{prot:condimplicit}, which can add a moderate amount of
complexity to the extractor calculation. The \QEFs for the special
trials are set to $1$, so these trials contribute no
conditional min-entropy. Future trial's models and \QEFs can then depend on the
special trials' outputs in addition to the normal trial results.

We remark that when computing chained \QEF given by $\prod_{i=1}^{n}F_i(C_{i}Z_{i})$ 
with floating point numbers, to avoid overflow of the mantissa, it is good practice to
work with the logarithm of the \QEF and add the logarithms of the
trial-wise \QEFs $F_i(C_{i}Z_{i})$.

\subsection{\QEF Rates and Optimization}
\label{subsect:logprobrates}

Consider a trial model $\cC(CZ)$, a non-negative function $F(CZ)$, a
probability distribution $\nu(CZ)$ and a \QEF power $\beta$.
We treat $\nu(CZ)$ as the design or the predicted probability distribution
for $CZ$. 
\begin{definition}
  The \emph{log-prob rate of $F(CZ)$ at $\nu(CZ)$} is
  $\sum_{cz}\nu(cz)\log(F(cz))/\beta=
  \Exp_{\nu(CZ)}\left(\vphantom{\big|}\log(F(CZ))\right)/\beta$.
\end{definition}
If $F(CZ)$ is a \QEF with power $\beta$ for $C|Z$ and $\cC(CZ)$, then
from Thm.~\ref{thm:bnds_from_qef} we can see that the log-prob rate
can be interpreted as the expected conditional min-entropy of
$C|Z\Pfnt{E}$ witnessed by $F(CZ)$ without adjusting for the error
bound or for probability of success. It is a useful predictor of the
smooth conditional min-entropy witnessed in a sequence of trials with
trial models $\cC(C_{i}Z_{i})$ identical to $\cC(CZ)$ except for the
change of CVs, where the experiment is configured so that the marginal
trial distributions are $\nu(C_{i}Z_{i})$, or at least close to
$\nu(C_{i}Z_{i})$ conditionally on the past.  In such a sequence of
trials, $\log(F(C_{i}Z_{i}))/\beta$ are approximately i.i.d. RVs and
their mean is typically close to the log-prob rate at $\nu(CZ)$. If
the error bound and lower bound on probability of success are
constant, then the asymptotic smooth conditional min-entropy rate
according to Thm.~\ref{thm:bnds_from_qef} for the chained \QEF
$\prod_{i=1}^{n}F(C_{i}Z_{i})$ is the log-prob rate of the trial-wise
\QEF $F(CZ)$.  We emphasize that the assumption on the trial
distributions is a completeness assumption and not required for sound
conditional min-entropy estimation with \QEFs.  If the experiment does
not perform according to expectation, the worst that can happen is
that we do not witness the expected amount of conditional min-entropy.

The log-prob rate neglects the reduction of conditional
min-entropy due to the error bound, which is a problem for finite data
or when the error bound grows with number of trials.
\begin{definition}\label{def:netlogprob}
  Given an error bound $\epsilon$ and $n$ trials, the \emph{error
    bound rate} of $\epsilon$ is $r=|\log(\epsilon)/n|$.  Let
  $\bar\kappa\geq\epsilon$ be the smallest probability of success that we
  need to protect against.  The \emph{expected quantum net log-prob of
    $F(CZ)$ at $\nu(CZ)$} is
  \begin{equation}\label{eq:netlogprob}
    n\Exp_{\nu(CZ)}\left(\log(F(CZ))\right)/\beta +
    \log(\epsilon^{2}\bar\kappa^{(\beta-1)\knuth{\beta>1}}/2)/\beta.  
  \end{equation}
  The \emph{quantum net
    log-prob rate of $F(CZ)$ at $\nu(CZ)$} is
  \begin{equation}
    \Exp_{\nu(CZ)}\left(\vphantom{\big|}\log(F(CZ))\right)/\beta-2r/\beta.
  \end{equation}
\end{definition}
The expected quantum net log-prob reflects the smooth conditional
entropy one can aim for if the experiment is designed for trials with
i.i.d. observable distributions $\nu(CZ)$ for each trial.  The
dependence on $\bar\kappa$ is motivated by the reference protocol
Protocol~\ref{prot:condgen_direct} but accounts for $\bar\kappa$ and
neglects the extractor constraints: Let $\cO_{F}$ be the log-prob rate
of $F(CZ)$ at $\nu(CZ)$.  Let $\kappa$ be the probability of success
of the protocol, assume $\kappa\geq\bar\kappa$ and consider the proof
of Thm.~\ref{thm:condgen_direct}. To motivate the definition of
expected quantum net log-prob, we neglect the extractor constraints
and the error $\epsilon_{x}$, set $\epsilon_{h}=\epsilon$, and choose
$f_{\min}=e^{n\beta\cO_{F}}$, which is the maximum $f_{\min}$ at which
we can hope to have a reasonable probability of success for
completeness.  For soundness, we set $\delta=(\epsilon/\kappa)^{2}/2$.
When applying Thm.~\ref{thm:bnds_from_qef}, we determine $p$ and $q$
by $f_{\min}= q^{-\beta}$ and $p=q\delta^{-1/\beta}$, so
$p=(f_{\min}\delta)^{-1/\beta}$. On success, the
$(\epsilon/\kappa)$-smooth conditional min-entropy is given by the
negative logarithm of the right-hand side of
Eq.~\ref{eq:thm:from_qef:pmax}, which evaluates to
\begin{equation}
  n\cO_{F}+\log(\delta\kappa^{\alpha})/\beta
  = n\cO_{F}+\log(\epsilon^{2}\kappa^{\alpha-2}/2)/\beta
  \geq n\cO_{F}+\log(\epsilon^{2}\bar\kappa^{(\beta-1)\knuth{\beta>1}}/2)/\beta,
\end{equation}
which is the expected quantum net log-prob.  If we set
  $\bar\kappa=\epsilon$, the right-hand side is the amount of
  randomness that would be obtained in
  Protocol~\ref{prot:condgen_direct} if the extractor constraints are
  neglected, $\epsilon_{x}=0$, $f_{\min}$ is chosen as above and
  $k_{i}=k_{o}$. The proof of Thm.~\ref{thm:condgen_direct} makes it
clear that there is nothing to be gained by considering
$\bar\kappa<\epsilon$: For success probabilities smaller than
$\epsilon$, $\epsilon$-soundness is automatically satisfied.

The quantum net log-prob rate does not take into account the bound on
the probability of success, effectively assuming that this bound is
constant. The quantum net log-prob rate accounts for the asymptotic
contribution of the error bound to the conditional min-entropy
witnessed by $F(CZ)$ according to Thm.~\ref{thm:bnds_from_qef}, where
the error bound $\epsilon$ for $n$ trials is determined by the error
bound rate $r$ according to $\epsilon = e^{-rn}$.  It is distinguished
from the net log-prob rate as defined for PEFs in
Ref.~\cite{knill:qc2017a} by the factor of $2$ multiplying $r$, which
originates in Thm.~\ref{thm:bnds_from_qef}. It reflects a doubling of
the number of trials required to satisfy error bounds for quantum side
information compared to what is required for classical side
information in the PE and QPE frameworks.

Given an experimental configuration with target $\nu(CZ)$, a
first goal is to maximize the log-prob rate subject to $F(CZ)$
being a \QEF with power $\beta$ for $C|Z$ and $\cC(CZ)$. 
The power $\beta$ can then be varied to maximize the expected quantum net log-prob.
Define
\begin{equation}
  Q_{\alpha}(F(CZ),\rho(CZ))=\sum_{cz}F(cz)\Rpow{\alpha}{\rho(cz)}{\rho(z)}.
\end{equation}
The power-$\beta$ \QEF condition for $C|Z$ and $\cC(CZ)$ is $Q_{\alpha}(F(CZ),\rho(CZ))\leq 1$
for all $\rho(CZ)\in\cN(\cC)$.  If the probability distribution of $Z$
is fixed, given by $\mu(Z)$, then for $\rho(CZ)\in\cC(CZ)$,
$\rho(z)=\mu(z)\rho$ and according to Eq.~\ref{eq:qefdef1_nu} the expression 
for $Q_{\alpha}$ simplifies to
\begin{equation}
  Q_{\alpha}(F(CZ),\rho(CZ)) = 
    \sum_{cz}\mu(z)F(cz)
    \tr((\rho^{-\beta/(2\alpha)}\rho(c|z)\rho^{-\beta/(2\alpha)})^{\alpha}).
\end{equation}
\QEFs are optimized by maximizing the log-prob rate. Instead of requiring
$F(CZ)$ to be a \QEF with power $\beta$ for $C|Z$ and $\cC(CZ)$, we formulate 
the \QEF optimization problem as follows:
\begin{alignat}{3}
    \textrm{Maximize:\ }&\sum_{cz}\nu(cz)\log(F(cz))-\log(f_{\max})\notag\\
    \textrm{Variables:\ }& F(CZ), f_{\max}\notag\\
    \textrm{Subject to:\ }&  F(CZ)\geq 0,\sum_{cz}F(cz)=1,\notag\\
    & f_{\max}=\max\{Q_{
      \alpha}(F(CZ),\rho(CZ)): \rho(CZ)\in\cN(\cC)\}.\label{prob:qefopt}
\end{alignat}
Every feasible solution $(F(CZ),f_{\max})$ determines the \QEF $F(CZ)/f_{\max}$
with power $\beta$ for $C|Z$ and $\cC(CZ)$
whose log-prob rate 
is the objective function divided by $\beta$.

\section{\QEFs and Entropy Estimators}
\label{sec:qefsandee}

\subsection{Entropy Estimators from \QEFs}

\begin{definition}
  The function $K(CZ)$ is an \emph{entropy estimator
    for $C|Z$ and $\cC(CZ)$} if for all $\rho(CZ)\in\cC(CZ)$,
  \begin{equation}
    \sum_{cz}K(cz)\tr(\rho(cz))
    \leq
    -\sum_{cz}\tr(\rho(cz)
    \left(\vphantom{\big|}\log(\rho(cz))-\log(\rho(z))\right)).
    \label{eq:eestdef}
  \end{equation}
  The \emph{entropy estimate} of $K(CZ)$ at $\rho(CZ)$
  is $\sum_{cz}K(cz)\tr(\rho(cz))$.
\end{definition}
Both sides of Eq.~\ref{eq:eestdef} are positive homogeneous of degree
$1$ in $\rho(CZ)$, which implies that $K(CZ)$ is an entropy estimator
for $\cC(CZ)$ iff it is an entropy estimator for $\cN(\cC(CZ))$.  For
normalized states, the right-hand side of Eq.~\ref{eq:eestdef} is the
conditional entropy $H_{1}(\rho(CZ)|Z\Pfnt{E})$ of $C|Z\Pfnt{E}$ with
respect to $\rho(CZ)$.

\begin{theorem}\label{thm:qef_to_ee}
  Let $F(CZ)^{\beta}$ be a \QEF with power $\beta$ for $C|Z$ and $\cC(CZ)$.
  Then $K(CZ)=\log(F(CZ))$ is an entropy estimator for $C|Z$ and $\cC(CZ)$.
\end{theorem}

\begin{proof}
  Without loss of generality, consider $\rho(CZ)\in\cN(\cC(CZ))$.
  By power reduction, $F(CZ)^{\gamma}$ is a \QEF with power $\gamma$
  for $C|Z$ and $\cC(CZ)$ for all $0<\gamma\leq\beta$.   Hence
  \begin{align}
    1 &\geq \sum_{cz} F(cz)^{\gamma}\Rpow{1+\gamma}{\rho(cz)}{\rho(z)}\notag\\
    &=\sum_{cz} \tr(\rho(cz))
    F(cz)^{\gamma}\hatRpow{1+\gamma}{\rho(cz)}{\rho(z)}\notag\\
    &=\sum_{cz} \tr(\rho(cz)) \exp(\gamma\log(F(cz))+\gamma
    \tildeDrel{1+\gamma}{\rho(cz)}{\rho(z)} )\notag\\
    &\geq \sum_{cz}\tr(\rho(cz))
    \left(1+\gamma\left(\log(F(cz))
      +\tildeDrel{1+\gamma}{\rho(cz)}{\rho(z)}\right)\right)\notag\\
    &=1+\gamma\left(\sum_{cz}\tr(\rho(cz))\log(F(cz)) +
      \tr(\rho(cz))\tildeDrel{1+\gamma}{\rho(cz)}{\rho(z)}\right).
  \end{align}
  Subtracting $1$ on both sides and dropping the positive quantity
  $\gamma$ gives 
  \begin{align}
    \sum_{cz}\log(F(cz))\tr(\rho(cz))
    &\leq \sum_{cz}
    -\tr(\rho(cz))\tildeDrel{1+\gamma}{\rho(cz)}{\rho(z)},
  \end{align}
  where the right-hand side converges to
  $-\sum_{cz}\tr(\rho(cz)(\log(\rho(cz))-\log(\rho(z))))$ as
  $\gamma\searrow 0$ (Eq.~\ref{eq:relentlim1}), so $F(CZ)$ satisfies
  the entropy-estimator inequality Eq.~\ref{eq:eestdef}.
\end{proof}

\subsection{\QEFs from Entropy Estimators}

\begin{theorem}\label{thm:ee_to_qef}
  Let $K(CZ)$ be an entropy estimator for $C|Z$ and $\cC(CZ)$. Define
  $F(CZ)^{\beta}=e^{\beta K(CZ)}$ and 
  \begin{equation}
    c_{P}(\beta) = c_{P}(\beta;K(CZ)) = 
    \sup\left\{\sum_{cz}F(cz)\Ppow{1+\beta}{\tau(cz)}{\tau(z)}:
      \tau(CZ)\in\cN(\cC(CZ))
    \right\}-1.
  \end{equation}
  Then for $\beta\leq 1$, $F(CZ)^{\beta}/(1+c_{P}(\beta))$ is a \QEFP
  with power $\beta$ for $C|Z$ and $\cC(CZ)$. The function
  $c_{P}(\beta)$ can be extended to $\beta=0$ by taking the limit
  $\beta\searrow 0$ and satisfies $c_{P}(0)=0$ and $c_{P}$ is convex.
  Let $\iota_{0}\approx 2.065339$ be the positive solution $x$ to
  $2\coth(x)=x$.  Define $\llceil x\rrceil = \max(\iota_{0},x)$,
  $N=|\Rng(C)|$, $k_{\max}(z)=\max_{c}K(cz)$,
   and $\bar
  w_{\gamma}(z)=(1-\gamma)\max_{c}\left(\vphantom{\big|}
    \max\left(\vphantom{\big|}\log(N)-K(cz),K(cz)\right)\right)+\log(2)$.
  For $\beta< 1/2$, an upper bound on $c_{P}$ is given by
  $c_{P}(\beta) \leq \frac{\beta^{2}}{2}
  \sup\left\{c(\beta,\nu(Z)):\nu(CZ)\in\tr(\vphantom{\big|}\cN(\cC(CZ)))\right\}$, where 
  \begin{align}
     c(\beta,\nu(Z)) \defeq
    \frac{1}{3}\sum_{z}\nu(z)&
    \left(\vphantom{\frac{e^{k_{\max}(z)\beta}}{(1-\beta)^{2}}}
       2 \llceil\bar w_{0}(z)\rrceil\left(\llceil\bar w_{0}(z)\rrceil 
           + 2\coth(\llceil\bar w_{0}(z)\rrceil)\right)
        \right.\notag\\
      &\left.\hphantom{()} + 
       \frac{e^{k_{\max}(z)\beta}}{(1-\beta)^{2}}
       \llceil\bar w_{\beta}(z)\rrceil\left(\llceil\bar w_{\beta}(z)\rrceil 
           + 2\coth(\llceil\bar w_{\beta}(z)\rrceil)\right)
         \right).
  \end{align}
\end{theorem}
Note that the quantity $c(\beta,\nu(Z))$ is continuous, and it is well defined even if $\beta\in [1/2,1)$.
The definition of $c(\beta,\nu(Z))$ when $\beta\in [1/2,1)$ is used in the proof of 
Thm.~\ref{thm:eat_from_qef}. 

We demonstrate by example in Sect.~\ref{subsec:examples} that direct
constructions of \QEFs have much better performance than constructions
from entropy estimators. Direct constructions for $(k,2,2)$ Bell-test
configurations are given in Sect.~\ref{subsec:k22config}.  If it is
necessary to construct \QEFs from entropy estimators by applying
Thm.~\ref{thm:ee_to_qef}, the bound can be improved according to
expressions obtained in the proof, where we develop bounds suitable
for numerical implementation. Beyond taking advantage of input
probability constraints, the bounds are agnostic with regard to
specific properties of $\cC(CZ)$ and are therefore necessarily suboptimal.

The proof of Thm.~\ref{thm:ee_to_qef} is an elaboration on the
techniques for bounding R\'enyi entropies in
Ref.~\cite{tomamichel:qc2009a}, see the proof of Lem.~8 in this
reference. The same techniques also contribute to the proof of the
entropy accumulation theorem in~\cite{dupuis:qc2016a}, with similar
results for estimating conditional min-entropy. See the comparison in
the next section. Much of the complexity of the proof below arises
from squeezing out the best bounds possible given the constraints of
written text. The proof is presented to enable numerical
improvements and to provide information on limitations of the
technique. Improved bounds are readily obtained but matter primarily
when $\beta$ is not small. See relevant remarks in the proof.

\begin{proof}
  By definition of $c_{P}(\beta)$, $F(CZ)^{\beta}/(1+c_{P}(\beta))$
  satisfies the \QEFP inequality with power $\beta$
  at all $\tau(CZ)\in\cN(\cC(CZ))$, so the first claim is immediate.

  To determine an upper bound on $c_{P}(\beta)$, consider any
  $\tau(CZ)\in\cN(\cC(CZ))$.  The left-hand side of the \QEFP
  inequality with power $\beta$  (see Eq.~\eqref{eq:pefdef1}) at $\tau(CZ)$ for $C|Z$ and
  $F(CZ)^{\beta}$ is equivalent to
  \begin{equation}  
    h(\beta)=
    \sum_{cz}F(cz)^{\beta}
    \tr(\tau(cz)^{1+\beta}\tau(z)^{-\beta}).
    \label{eq:thm:ee_to_qef:qef_exp}
  \end{equation}
  The goal is to determine an upper 
  bound on $h(\beta)$ that depends on $\beta$ and the values of
  $K(CZ)$. In general, we may also take advantage of constraints on
  the probability distribution $\tr(\tau(CZ))$.  

  The $cz$-term in the sum for $h(\beta)$ is of the form
  \begin{equation}
    g(\beta)= g(\beta;a,\rho|\sigma)=
    \tr(\rho (e^{a}\rho)^{\beta}\sigma^{-\beta}),
  \end{equation}
  where $a=K(cz)$, $\rho=\tau(cz)$ and $\rho\ll\sigma=\tau(z)$. We
  bound $g(\beta)$ by Taylor expansion with a second-order
  remainder.  For this, it is convenient to express
  \begin{equation}
    \tr(\rho(e^{a}\rho)^{\beta}\sigma^{-\beta}) =
    \tr(\xi \left((\rho(e^{a}\rho)^{\beta})\otimes (\sigma^{T})^{-\beta}\right))
    = \tr(\xi (\rho\otimes \one)\left((e^{a}\rho\otimes\sigma^{-T})^{\beta}\right)),
  \end{equation}
  where $\xi=\dyad{\phi}$ with $\phi=\sum_{i}\ket{i}\otimes\ket{i}$ for
  some orthonormal basis $(\ket{i})_{i}$. 
  Write $g^{(k)}(\beta)$ for the $k$'th derivative of $g(\beta)$.
  For the Taylor expansion, we compute for $0\leq\gamma\leq \beta$
  \begin{align}
    g(0)&=\tr(\rho),\notag\\
    g^{(k)}(\gamma)&=
       \tr(\xi (\rho\otimes \one)\log(e^{a}\rho\otimes\sigma^{-T})^{k}
       ((e^{a}\rho\otimes\sigma^{-T})^{\gamma})).      
    \label{eq:thm:ee_to_qef:g''}
  \end{align}
  The factors after $\xi$ in the trace commute and multiply to a
  positive semidefinite operator for even $k$, so
  $g^{(2l)}(\gamma)\geq 0$ for all $l\in\nats$.  In
  particular, the second derivative is non-negative and convex.  We have
  \begin{equation}
    g(\beta)= g(0)+\beta g^{(1)}(0)
    + \int_{0}^{\beta}d\gamma \;(\beta-\gamma)g^{(2)}(\gamma),
    \label{eq:thm:ee_to_qef:taylor1}
  \end{equation}
  where by convexity we can replace $g^{(2)}(\gamma)$ by
  $(1-\gamma/\beta)g^{(2)}(0)+(\gamma/\beta) g^{(2)}(\beta)
  = g^{(2)}(0)+(\gamma/\beta)(g^{(2)}(\beta)-g^{(2)}(0))$ for an upper bound.
  Since $\int_{0}^{\beta}d\gamma\;(\beta-\gamma)\gamma/\beta =
  \beta^{2}/6$, we have the bound
  \begin{equation}
    g(\beta)\leq g(0)+\beta g^{(1)}(0)
    + \frac{\beta^{2}}{2}\left(\frac{2}{3}
      g^{(2)}(0)+\frac{1}{3}g^{(2)}(\beta)\right).
  \end{equation}
  To expand Eq.~\ref{eq:thm:ee_to_qef:qef_exp} in orders of $\beta$,
  we substitute $\rho$ by $\tau(cz)$, $\sigma$ by $\tau(z)$, $a$ by
  $K(cz)$ and replace the corresponding terms of
  Eq.~\ref{eq:thm:ee_to_qef:qef_exp} to obtain the bound
  \begin{align}
    h(\beta) &= h(0)+\beta h^{(1)}(0) +
    \int_{0}^{\beta}d\gamma(\beta-\gamma)h^{(2)}(\gamma)\notag\\
    &=
    h(0) + \beta \sum_{cz}\left(\vphantom{|_{1}^{1}}
      K(cz)\tr(\tau(cz))
      + \tr(\vphantom{\big|}\tau(cz)(\log(\tau(cz))-\log(\tau(z))))\right)\notag\\
    &\hphantom{=\;\;} 
    + \sum_{cz}\int_{0}^{\beta}d\gamma\;(\beta-\gamma)
    g^{(2)}(\gamma;K(cz),\tau(cz)|\tau(z))\notag\\
    &\leq 1 +\sum_{cz}\int_{0}^{\beta}d\gamma\;
    (\beta-\gamma)g^{(2)}(\gamma;K(cz),\tau(cz)|\tau(z))\notag\\
    &\leq 1+
    \sum_{cz}\frac{\beta^{2}}{2}\left(\frac{2}{3}
      g^{(2)}(0;K(cz),\tau(cz)|\tau(z))
      +\frac{1}{3}g^{(2)}(\beta; K(cz), \tau(cz)|\tau(z))\right),\label{eq:thm:ee_to_qef:h2betadef}
  \end{align}
  since $K(CZ)$ is an entropy estimator, which implies that
  $h^{(1)}(0)\leq 0$.  The results so far also establish that
  $h^{(2)}(\gamma)\geq 0$, so $h(\gamma)$ is convex. Since the
  supremum of convex functions is convex, so is $c_{P}$. 
  
  Write $h_{2}(\beta;K(cz),\tau(cz)|\tau(z))$ for the coefficient of
  $\beta^{2}/2$ of the $cz$-summand in the last line of
  Eq.~\ref{eq:thm:ee_to_qef:h2betadef}.  To prove the theorem, we
  determine a bound $b(\beta)\geq
  \sum_{cz}h_{2}(\beta;K(cz),\tau(cz)|\tau(z))$ expressed as an
  expectation over the probability distribution $\tr(\tau(Z))$ with
  no other dependence on $\tau(CZ)$. Then
  \begin{equation}
    h(\beta) \leq 1 + b(\beta) \beta^{2}/2.
    \label{eq:thm:ee_to_qef:hbeta}
  \end{equation}
  Since $h(0)=1$, the claim $c_{P}(\beta)\searrow 0$ follows once we
  establish that $b(\beta)$ is finite. To determine $b(\beta)$, we apply the 
  following lemma with $\beta\in (0,1]$. 

  \begin{lemma}\label{lem:meas_to_powerbnd}
    Fix $\beta>0$. For each $a\in\rls$, let $\mu_{a}$ be a positive measure on $[-1,1]$
    such that for all $y\in(0,\infty)$
    \begin{equation}
      \int_{[-1,1]}y^{\chi}d\mu_{a}(\chi) \geq
      \log(e^{a}y)^{2}\left(2/3+(1/3)(e^{a}y)^{\beta}\right).
    \end{equation}
    Let $k_{\max}(z)=\max_{c}K(cz)$, $k_{\min}(z)=\min_{c}K(cz)$,
    and $\bar k(z) = \sum_{c}K(cz)/N$ where $N=|\Rng(C)|$. Given $z$, let $p(a)=
    (a-k_{\min}(z))/(k_{\max}(z)-k_{\min}(z))$ so that $a =
    (1-p(a))k_{\min}(z) + p(a)k_{\max}(z)$.  Write
    $\bar\mu_{z,a}=p(a)\mu_{k_{\max}(z)}+(1-p(a))\mu_{k_{\min}(z)}$. Then for each $z$, 
     \begin{align}
      \sum_{c}h_{2}(\beta;K(cz),\tau(cz)|\tau(z))
       &\leq \tr(\tau(z))\Bigg(
       N \bar\mu_{z,\bar k(z)}(\{-1\})\notag\\
       &\hphantom{\leq \leq \tr(\tau(z))\Bigg(\;} + \int_{(-1,0)}N^{-\chi}d(\mu_{k_{\min}(z)}\vee\mu_{k_{\max}(z)})(\chi)\notag\\
       &\hphantom{\leq \leq \tr(\tau(z))\Bigg(\;}+ \max\left(
         \int_{[0,1]}d\mu_a(\chi):a\in\{k_{\min}(z),k_{\max}(z)\}\right)
       \Bigg).
      \label{eq:lem:meas_to_powerbnd:main}
    \end{align}
  \end{lemma}

  \Pc{The join of two signed measures $\mu$ and $\nu$ is well-defined
    by $(\mu\vee\nu)(X)= \sup_{R\subseteq X}(\mu(R)+\nu(X{\setminus}
    R))$, where the variables in these expressions are restricted
    measurable sets.}

  When we apply this lemma, the measures are sums of point measures at
  values of $\chi$ that depend on $\beta$ but not on $K(cz)$. If
  $[\tau(cz),\tau(z)]=0$ for all $c$, then the lemma can be improved
  by restricting $y$ to $y\in (0,1]$ in the first inequality in the
  lemma. See the remark in the proof for the explanation.

  \begin{proof}
    For this proof, $z$ can be held fixed, so we omit it, writing
    $K(c)$ for $K(cz)$, $\tau(c)$ for $\tau(cz)$, $\tau$ for
    $\tau(z)$, and similarly for the measures to be found.  

    For the moment, we fix $c$ and write $a=K(c)$.  We express
    \begin{equation}
      h_{2}(\beta;a,\rho|\sigma)
      = \tr(\xi (\rho \otimes\one)\log(e^{a}\rho\otimes\sigma^{-T})^{2}
      \left((2/3)\one\otimes\one
        +(1/3)(e^{a}\rho\otimes\sigma^{-T})^{\beta}\right)).
    \end{equation}
    Since $\xi$ is positive semidefinite, an upper bound on
    $h_{2}(\beta;a,\rho|\sigma)$ 
    can be obtained by determining an operator upper bound on
    \begin{equation}
      X_{a}=(\rho\otimes\one) \log(e^{a}\rho\otimes \sigma^{-T})^{2}
      \left((2/3)\one\otimes\one
        +(1/3)(e^{a}\rho\otimes\sigma^{-T})^{\beta}\right).
    \end{equation}
    For this, we can
    work in a joint eigenbasis of the form $(\ket{i}_{1}\otimes\ket{j}_{2})_{ij}$ of
    $\rho\otimes \one$ and $\one\otimes\sigma^{T}$. We identify the corresponding
    eigenvalues as $\rho_{i}$ and $\sigma_{j}$, which are also the
    diagonal elements in this basis. The operator $X_{a}$ is also diagonal,
    with diagonal elements
    \begin{equation}
      x_{ij}=
      \rho_{i}\log(e^{a}\rho_{i}/\sigma_{j})^{2}
      ((2/3)+(1/3)(e^{a}\rho_{i}/\sigma_{j})^{\beta}).
    \end{equation}
    Write $y_{ij}=\rho_{i}/\sigma_{j}$. The terms where
    $y_{ij}=0$ do not contribute to the relevant sums because of the
    additional factor of $\rho_{i}$.  Remark: If $[\rho,\sigma]=0$,
    then we can choose a common eigenbasis for $\rho$ and $\sigma$ in
    the expression for $\xi$ to see that $y_{ij}=0$ for $i\not=j$, and
    if $\sigma\geq \rho$, $y_{ij}\in [0,1]$.
    
    The constraint on $\mu_{a}$ can be reexpressed with the change of variables
    $y=e^{t}$ in terms of $t\in\rls$ as
    \begin{equation}
      \int_{[-1,1]}e^{t\chi}d\mu_{a}(\chi) \geq
      (t+a)^{2}\left(2/3+(1/3)e^{(t+a)\beta}\right).
      \label{eq:lem:meas_to_powerbnd:tconstr}
    \end{equation}
    We prove that the right-hand side is convex in
    $a$. With the change of variables $x=(t+a)\beta$, this is
    equivalent to $v(x)=x^{2}(2/3 + (1/3)e^{x})$ being convex in $x$.
    Compute $v^{(2)}(x) = 4/3 + (1/3)(2+4x+x^{2})e^{x} =
    4/3+(1/3)((x+2)^{2}-2)e^{x}$.  For $x\geq 0$, $(x+2)^{2}-2>0$, so
    to show that $v^{(2)}>0$, it suffices to consider $x<0$. Then
    $e^{x}\in (0,1]$ and $(x+2)^{2}-2\geq -2$ so $v^{(2)}(x) \geq 4/3
    + (1/3)(-2)= 2/3>0$ as claimed.

    Applying the convexity established in the previous paragraph, for $a\in[k_{\min},k_{\max}]$,
    \begin{align}
      \int_{[-1,1]}e^{t\chi}d(p(a)\mu_{k_{\max}}+(1-p(a))\mu_{k_{\min}})(\chi)
      \hspace*{-1in}&\notag\\
      &\geq p(a)(t+k_{\max})^{2}\left(2/3+(1/3)e^{(t+k_{\max})\beta}\right)\notag\\
      &\hphantom{\geq\;}
      +(1-p(a))(t+k_{\min})^{2}\left(2/3+(1/3)e^{(t+k_{\min})\beta}\right)\notag\\
      &\geq
      (t+a)^{2}\left(2/3+(1/3)e^{(t+a)\beta}\right).
    \end{align}
    From this inequality and with $\bar\mu_{a} = p(a)\mu_{k_{\max}}+(1-p(a))\mu_{k_{\min}}$ as
    defined in the statement of the lemma, we get  
     \begin{align}
      x_{ij} &= \rho_{i}\log(e^{a}y_{ij})^{2}\left((2/3)+(1/3)(e^{a}y_{ij})^{\beta}\right)\notag\\
      &\leq \rho_{i}\int_{[-1,1]}
      y_{ij}^{\chi}d\bar\mu_{a}(\chi)\notag\\
      &= \int_{[-1,1]}\rho_{i}^{1+\chi}\sigma_{j}^{-\chi}d\bar\mu_{a}(\chi).
    \end{align}
    It follows that $X_{a}\leq \int_{[-1,1]}\rho^{1+\chi}\otimes
    (\sigma^{-T})^{\chi}d\bar\mu_{a}(\chi)$
    and 
    \begin{align}
      h_{2}(\beta;a,\rho|\sigma)
      &=\tr(\xi X_{a})\notag\\
      &\leq \int_{[-1,1]}\tr(\xi \left(\rho^{1+\chi}\otimes
    (\sigma^{-T})^{\chi}\right))d\bar\mu_{a}(\chi)\notag\\
      &=
      \int_{[-1,1]}\tr(\rho^{1+\chi}\sigma^{-\chi})d\bar\mu_{a}(\chi).
    \end{align}
    Substituting accordingly we get
    \begin{align}
      \sum_{c}h_{2}(\beta;K(c),\tau(c)|\tau) \hspace*{-1in}&\notag\\
      &\leq \sum_{c}
      \int_{[-1,1]}
      \tr(\tau(c)^{1+\chi}\tau^{-\chi})d\bar\mu_{K(c)}(\chi)\notag\\
      &=
      \sum_{c}\tr(\tau)\bar\mu_{K(c)}(\{-1\})+
      \int_{(-1,0)}\sum_{c}\tr(
      \tau(c)^{1+\chi}\tau^{-\chi})d\bar\mu_{K(c)}(\chi)\notag\\
      &\hphantom{\leq\;}
      +\int_{[0,1]}\sum_{c}
      \tr(\tau(c)^{1+\chi}\tau^{-\chi})d\bar\mu_{K(c)}(\chi).
      \label{eq:lem:meas_to_powerbnd:1}
    \end{align}

    For $0\leq \chi\leq 1$, we have
    $\tr(\tau(c)^{1+\chi}\tau^{-\chi})\leq\tr(\tau(c))$, as can be
    seen by applying Lem.~\ref{lem:conditionmonotone_rp} with $\rho$
    there replaced by $\tau(c)$ here, $\sigma$ there with $\tau(c)$
    here, $\sigma'$ there with $\tau$ here, and $\beta$ there with
    $\chi$ here. For $-1< \chi< 0$, the dimension bounds on R\'enyi
    powers imply that $\sum_{c}\tr(\tau(c)^{1+\chi}\tau^{-\chi})\leq
    \tr(\tau) N^{-\chi}$ (Ref.~\cite{tomamichel:qc2015a},
    Sect. 5.3.5).  \Pc{To prove this directly, we can apply the fact
      that the Petz R\'enyi power for $-1<\beta< 0$ is jointly concave
      (Ref.~\cite{tomamichel:qc2015a}, Prop.~4.8, Pg.~61): Identify
      $\Rng(C)$ with $\ints_{N}$, define $\tau_{k}(c)=\tau(c+k)$ with
      addition modulo $N$, and consider the convex combination
      $\bar\tau(c)=\sum_{k}\tau_{k}(c)/N = \tau/N$. Concavity with the
      conditioning operator fixed implies
      \begin{align}
        \tr(\tau) &= 
        \tr(\tau^{1+{\chi}}\tau^{-\chi})\notag\\
        &= N^{1+\chi}\tr(\left(\sum_{k}\tau_{k}(c)/N\right)^{1+{\chi}}\tau^{-\chi})\notag\\
        &\geq 
        N^{1+\chi}\sum_{k}\tr(\tau_{k}(c)^{1+{\chi}}\tau^{-\chi})/N\notag\\
        &= N^{\chi}\sum_{c}\tr(\tau(c)^{1+\chi}\tau^{-\chi}).
      \end{align}
    }
    We can now bound each summand at
    the end of Eq.~\ref{eq:lem:meas_to_powerbnd:1}.
    By linearity of $\bar\mu_{a}$ in $a$,
    \begin{align}
      \sum_{c}\tr(\tau)\bar\mu_{K(c)}(\{-1\})
      &= \tr(\tau)N\bar\mu_{\bar k}(\{-1\}).
    \end{align}
    For $a\in [k_{\min},k_{\max}]$, $\bar\mu_a\leq
    \mu_{k_{\min}}\vee \mu_{k_{\max}}$, where $\mu_{k_{\min}}\vee \mu_{k_{\max}}$
    is independent of $c$. Therefore
    \begin{align}
      \int_{(-1,0)}\sum_{c}\tr(
      \tau(c)^{1+\chi}\tau^{-\chi})d\bar\mu_{K(c)}(\chi)
      &\leq \int_{(-1,0)}\sum_{c}\tr(\tau(c)^{1+\chi}\tau^{-\chi})
      d(\mu_{k_{\min}}\vee\mu_{k_{\max}})(\chi)
      \notag\\
      &\leq \tr(\tau)\int_{(-1,0)}N^{-\chi}d(\mu_{k_{\min}}\vee\mu_{k_{\max}})(\chi).
    \end{align}
    Since $\tr(\tau(C))/\tr(\tau)$ is a probability distribution and
    for $a\in[k_{\min},k_{\max}]$, the integral $\int_{[0,1]}d\bar\mu_{a}(\chi)$ is
    between $\int_{[0,1]}d\bar\mu_{k_{\min}}(\chi)$ and
    $\int_{[0,1]}d\bar\mu_{k_{\max}}(\chi)$,
    \begin{align}
      \int_{[0,1]}\sum_{c}
      \tr(\tau(c)^{1+\chi}\tau^{-\chi})d\bar\mu_{K(c)}(\chi) &\leq
      \int_{[0,1]}\sum_{c}
      \tr(\tau(c))d\bar\mu_{K(c)}(\chi)\notag\\
      &= \tr(\tau)\sum_{c}
      \frac{\tr(\tau(c))}{\tr(\tau)}\int_{[0,1]}d\bar\mu_{K(c)}(\chi)\notag\\
      &\leq\tr(\tau)\max_{c}\int_{[0,1]}d\bar\mu_{K(c)}(\chi)\notag\\
      &\leq
      \tr(\tau)\max\left(\int_{[0,1]}
        d\bar\mu_{a}(\chi):a\in\{k_{\min},k_{\max}\}\right) \notag \\
       & \leq
      \tr(\tau)\max\left(\int_{[0,1]}
        d\mu_{a}(\chi):a\in\{k_{\min},k_{\max}\}\right).
    \end{align}
    Inserting these summands back into the right-hand side of
    Eq.~\ref{eq:lem:meas_to_powerbnd:1} gives the lemma.
  \end{proof}

  Motivated by the above lemma we consider the reparameterized constraint
  in Eq.~\ref{eq:lem:meas_to_powerbnd:tconstr}. 
  Reparameterizing a second time by replacing $t+a$ by $t$ gives
  \begin{equation}
    \int_{[-1,1]}e^{(t-a)\chi}d\mu_{a}(\chi) \geq t^{2}(2/3+(1/3)e^{t\beta}).
    \label{eq:thm:ee_to_qef:wmu}
  \end{equation}
  To simplify the problem,
  we express $\mu_{a}$ in terms of a weighted sum of measures $\nu$
  satisfying
  \begin{equation}
    \int_{[-1,1]}e^{(t-a)\chi}d\nu(\chi) \geq t^{2}e^{t\gamma},
  \end{equation}
  for $\gamma=0$ or $\gamma=\beta \in (0,1]$. See Eqs.~\eqref{eq:thm:ee_to_qef:solution1} and~\eqref{eq:thm:ee_to_qef:solution2} 
  below for our proposed solutions for $\mu_a$.  Replacing
  $d\nu(\chi)$ by $e^{a\chi}d\mu(\chi)$ and dividing
  both sides by $e^{t\gamma}$, we can equivalently determine $\mu$ such that
  for all $t\in\rls$,
  \begin{equation}
    \int_{[-1,1]}e^{t(\chi-\gamma)} d\mu(\chi)\geq t^{2}.
    \label{eq:thm:ee_to_qef:smu1}
  \end{equation}
  In view of
  the form of Eq.~\ref{eq:lem:meas_to_powerbnd:main} and in view of
  the reparameterization of measures, we wish to minimize
  \begin{equation}
    \int_{[-1,1]}N^{-\chi\knuth{\chi\leq 0}}e^{a\chi }d\mu(\chi).
    \label{eq:thm:ee_to_qef:smu2}
  \end{equation}
  Let $\delta_{x}$ denote the delta-function probability distribution
  defined by $\int f(y)d\delta_{x}(y) = f(x)$.  We converge on the
  choice
  \begin{equation}\mu= \lambda_{1}
    (\delta_{-1+2\gamma}+\delta_{1}) + \lambda_{0}\delta_{\gamma},
    \label{eq:thm:ee_to_qef:muchoice}
  \end{equation}
  for
  which the constraints in  Eq.~\eqref{eq:thm:ee_to_qef:smu1} become 
  \begin{equation}
    2\cosh((1-\gamma)t)\lambda_{1}  + \lambda_{0}\geq t^{2}
    \label{eq:thm:ee_to_qef:coshprob1}
  \end{equation}
  for all $t\in\rls$, where we determine $\lambda_{1}\geq 0$ and
  $\lambda_{0}\geq 0$ so that this inequality is tight.  We naturally
  arrived at this choice after considering more general forms that
  satisfy the constraints. See the comment after the proof for a
  discussion.  Subject to the constraints,  according to Eq.~\eqref{eq:thm:ee_to_qef:smu2} we minimize
  \begin{align}
    \lambda_{1} 
    \left(N^{(1-2\gamma)\knuth{\gamma\leq 1/2}}e^{-(1-2\gamma)a}
      + e^{a}\right)+\lambda_{0}e^{\gamma a}
    &= e^{\gamma a}\left(
      \lambda_{1} 
      \left(N^{(1-2\gamma)\knuth{\gamma\leq 1/2}}e^{-(1-\gamma)a}
        + e^{(1-\gamma)a}\right)+\lambda_{0}\right)\notag\\
    &=  e^{\gamma a}\left(2\cosh(w(a)) \lambda_{1}
      + \lambda_{0}\right),
    \label{eq:thm:ee_to_qef:coshprob2}
  \end{align}
  where $e^{w(a)}$ is the larger of the two solutions $x$ to the identity
  $x+1/x = e^{ l(a)}$ with $ l(a) =
  \log(N^{(1-2\gamma)\knuth{\gamma\leq 1/2}}e^{-(1-\gamma)a} +
  e^{(1-\gamma)a})\geq \log(2)$. Thus $e^{w(a)}= \left(e^{
      l(a)}+\sqrt{e^{2 l(a)}-4}\right)/2\leq e^{ l(a)}$,
  where the upper bound is a good approximation for large $
  l(a)$.  The function $a\mapsto e^{ l(a)}$ is convex and symmetric
  around its minimum at $a=\log(N^{(1-2\gamma)\knuth{\gamma\leq
      1/2}})/(2(1-\gamma))$, and as a result $w(a)$ is also minimized at and
  symmetric about this value, and monotone on each side.  From
  $\log(e^{x}+e^{-x})\leq |x|+\log(2)$, we obtain that for $1/2\leq\gamma\leq
  1$, $ l(a)\leq (1-\gamma)|a|+\log(2)$, and for
  $0\leq\gamma\leq 1/2$,
  \begin{align}
    w(a)\leq  l(a) &=
    \log(N^{1-2\gamma}e^{-(1-\gamma)a} +  e^{(1-\gamma)a})\notag\\
   &= \log(N^{(1-2\gamma)/2}\left(N^{(1-2\gamma)/2}e^{-(1-\gamma)a} 
    +  N^{-(1-2\gamma)/2}e^{(1-\gamma)a}\right))\notag\\ 
   &\leq(1-2\gamma)\log(N)/2
          + |(1-2\gamma)\log(N)/2-(1-\gamma)a| +\log(2)\notag\\
   &=\max\big( (1-2\gamma)\log(N)-(1-\gamma)a,(1-\gamma)a\big)+\log(2)\notag\\
   &\leq \max\big( (1-\gamma)\log(N)-(1-\gamma)a,(1-\gamma)a\big)\notag+\log(2)\\
   &= (1-\gamma)\max\big(\log(N)-a,
           a\big) + \log(2).
  \end{align}
  For the last inequality we opted for a simpler expression at the
  cost of worse bounds when $\gamma$ is not small. The better bound is
  readily taken into account by changing the next definitions and the
  corresponding ones in the theorem statement.  Let $\tilde w(a)=
  (1-\gamma)\max\left(\log(N)-a, a\right)+\log(2)$ so that $w(a)\leq
  \tilde w(a)$. We also define $\bar w(z)=\max_{c}\tilde w(K(cz)) =
  \max(\tilde w(k_{\max}(z)),\tilde w(k_{\min}(z)))$, consistent with
  the theorem statement, but suppressing the subscript $\gamma$ for
  the moment.
    
  The minimization problem defined by the constraints in
  Eq.~\ref{eq:thm:ee_to_qef:coshprob1} and the objective function in
  Eq.~\ref{eq:thm:ee_to_qef:coshprob2} can be transformed to an
  instance of
  \begin{alignat}{3}
    \textrm{Minimize:\ } & 2\cosh(v)a_{1}+a_{0}\notag\\
    \textrm{Variables:\ } & a_{1},a_{0}\notag\\
    \textrm{Subject to:\ } & 
    2\cosh(s)a_{1}+a_{0}\geq s^{2} \textrm{\ for all $s\in\rls$},\notag\\
    &a_{1}\geq 0, a_{0}\geq 0,
    \label{eq:thm:ee_to_qef:simplified_coshprob}
  \end{alignat}
  for a given $v\geq 0$; the transformation is described below, right after Eq.~\eqref{eq:thm:ee_to_qef:aim}.   To satisfy the
  constraint, we determine the minimum value $f(s_{0})$ of
  $f(s)=f(s;a_{1})=2\cosh(s)a_{1} -s^{2}$. Decreasing either $a_{1}$ or $a_{0}$
  reduces the objective function. To minimize the objective function, we can 
  set $a_{0}=-f(s_{0})$ if $f(s_{0})\leq 0$ and $a_{0}=0$ otherwise.  In the second case, when $f(s_{0}) >0$,
  it is possible to further reduce $a_{1}$ to decrease the objective function. Thus the optimal value for 
  $a_{1}$ is $0$, which is not possible as the first constraint in Eq.~\eqref{eq:thm:ee_to_qef:simplified_coshprob}  would be violated. 
  In this way we find that the minimum is achieved with $f(s_0)\leq 0$, and $a_{1}$ and $a_{0}$ 
  are both determined by the single parameter $s_0$. As a result, in the process of determining
  the minimum of $f(s)$,  we parametrize $a_{1}$ and $a_{0}$ in terms
  of $s_{0}$.

  The minimum of $f(s)$ is achieved at a critical point $s_{0}$
  satisfying $f^{(1)}(s_{0})= 2\sinh(s_{0})a_{1} -2s_{0} = 0$. One
  such critical point is $s_{0}=0$.  By the symmetry of $f(s)$ over $s=0$, 
  it suffices to consider $s_{0}\geq 0$.  Without loss of generality, we can consider
  only the case where $a_{1}<1$. The reason is as follows: Consider $f^{(2)}(s)=2\cosh(s)a_{1} -2$.
  This is positive for $a_{1} \geq 1$ and $s>0$, in which case there
  are no positive critical points as $f^{(1)}(s=0)=0$. Hence  the minimum of $f(s)$ is $f(0)=2
  a_{1}$. However, according to the argument in the previous paragraph, the minimum of the 
  objective function is achieved when $f(s_0)\leq 0$. In particular the minimum is not achieved for $a_{1} \ge 1$.  
   When $a_{1}<1$,  the slope $f^{(2)}$ of $f^{(1)}$ is increasing for $s\geq 0$, negative at $s=0$, and positive for $s$
  large enough. Consequently, $f^{(1)}$ first decreases from $f^{(1)}(s=0)=0$ and then monotonically increases, 
   from which it follows that there is exactly one critical point $s_{0}>0$ for $f$, which determines the minimum of
  $f$.  By making use of the critical-point equation to express
  \begin{equation}
    a_{1}=a_{1}(s_{0})=s_{0}/\sinh(s_{0}),
    \label{eq:thm:ee_to_qef:a1}
  \end{equation}
  we have $f(s_{0};a_{1}(s_{0}))=2s_{0}\coth(s_{0})-s_{0}^{2}$.  Note that
  because $\sinh(s_{0})>s_{0}$, we have $a_1=s_{0}/\sinh(s_{0})<1$. The
  function $x\in(0,\infty)\mapsto 2x\coth(x)-x^{2}$ approaches $2$ as
  $x\searrow 0$ and has derivative $2\coth(x) - 2x/\sinh(x)^{2} - 2 x  = 2\coth(x)(1-x \coth(x))$ which is negative for $x>0$. Negativity
  follows from $\sinh(x)=\int_{0}^{x}\cosh(t)dt \leq x\cosh(x)$.
  Therefore $f(s_{0};a_{1}(s_{0}))$ is decreasing in $s_{0}$, thus
  negative for $s_{0}> \iota_{0}$ where $\iota_{0}>0$ uniquely
  satisfies $2\coth(\iota_{0})=\iota_{0}$.  By numerical calculation,
  $\iota_{0}\in(2.065338,2.065339)$.   For $s_0 <\iota_0$, $f(s_0) >0$, but according to the argument 
  in the previous paragraph, we should have $f(s_0)\leq 0$ in order to achieve the minimum of the objective function. 
  We now constrain $s_{0}\geq \iota_{0}$ and parametrize $a_{1}$ and $a_{0}$ in terms of $s_{0}$,
  with $a_{1}$ given in Eq.~\ref{eq:thm:ee_to_qef:a1} and $a_{0}\geq
  0$ given by
  \begin{equation}
    a_{0}=a_{0}(s_{0})=-f(s_{0};a_{1}(s_{0})) =
    s_{0}(s_{0}-2\coth(s_{0})).
    \label{eq:thm:ee_to_qef:a2}
  \end{equation}
  Here $a_{0}$ is increasing and $a_{1}$ is decreasing in $s_{0}$ for
  $s_{0}>0$.  For the latter, the function $x\mapsto x/\sinh(x)$ has
  derivative $(\sinh(x)- x\cosh(x))/\sinh(x)^{2}\leq 0$.

  It remains to minimize $2\cosh(v)a_{1}+a_{0}$ over
  $s_{0}\geq \iota_{0}$.  Rewrite
  \begin{align}\label{eq:thm:ee_to_qef:obj_reform}
      2\cosh(v) a_{1}+ a_{0} &= 2\cosh(v) a_{1}+ s_{0}(s_{0}-2\coth(s_{0}))\notag\\
    &=  2\cosh(v) a_{1} +a_{1}\sinh(s_{0})(s_{0}-2\coth(s_{0}))\notag\\
    &= a_{1}(2\cosh(v)+ s_{0}\sinh(s_{0})-2\cosh(s_{0})),
  \end{align}
  and differentiate by $s_{0}$
  \begin{align}
    \frac{d}{ds_{0}}(2\cosh(v) a_{1}+ a_{0}) &=
    \left(\frac{d}{ds_{0}} a_{1}\right)(2\cosh(v)+ s_{0}\sinh(s_{0})-2\cosh(s_{0}))
    \notag\\
    &\hphantom{=\;}+ a_{1}(\sinh(s_{0}) + s_{0}\cosh(s_{0})-2\sinh(s_{0}))\notag\\
    &= \left(\frac{d}{ds_{0}} a_{1}\right)(2\cosh(v)+ s_{0}\sinh(s_{0})-2\cosh(s_{0}))\notag\\
    &\hphantom{=\;}
    + a_{1}(s_{0}\cosh(s_{0})-\sinh(s_{0})).
  \end{align}
  Since $\frac{d}{ds_{0}} a_{1} =
  (\sinh(s_{0})-s_{0}\cosh(s_{0}))/\sinh(s_{0})^{2}$, we can replace
  the second factor of the second summand by $- a_{1}\sinh(s_{0})^{2} \frac{d}{d
    s_{0}} a_{1}$ to get
  \begin{align}\label{eq:thm:ee_to_qef:obj_der_reform}
    \frac{d}{ds_{0}}(2\cosh(v) a_{1}+ a_{0}) &=
    \left(\frac{d}{ds_{0}} a_{1}\right)\left(2\cosh(v)+ s_{0}\sinh(s_{0})-2\cosh(s_{0})
    - a_{1}\sinh(s_{0})^{2}\right)\notag\\
    &= \left(\frac{d}{ds_{0}} a_{1}\right)
    (2\cosh(v)+s_{0}\sinh(s_{0})- 2\cosh(s_{0}) -s_{0}\sinh(s_{0}))\notag\\
    &=\left(\frac{d}{ds_{0}} a_{1}\right)
    (2\cosh(v)-2\cosh(s_{0})).
  \end{align}   
   Since $a_{1}$ is decreasing in $s_{0}$, that is,
  $\frac{d}{ds_{0}} a_{1}<0$, we need to consider the following two cases in order to find the minimum 
  of the function in Eq.~\eqref{eq:thm:ee_to_qef:obj_reform} over the region $s_{0}\geq \iota_{0}$. First,
 consider the case that $v>\iota_{0}$.  The derivative in Eq.~\eqref{eq:thm:ee_to_qef:obj_der_reform} is negative 
 when $\iota_{0}\leq s_{0}<v$, becomes zero when $s_{0}=v$, and is positive when $s_{0}>v$. Therefore, the 
 function in Eq.~\eqref{eq:thm:ee_to_qef:obj_reform} takes its minimum when $s_{0}=v$.  Second, in 
 the case that $v\leq \iota_{0}$ the derivative in Eq.~\eqref{eq:thm:ee_to_qef:obj_der_reform} is always non-negative
 when $s_{0}\geq \iota_{0}$. Hence, the minimum of the function in Eq.~\eqref{eq:thm:ee_to_qef:obj_reform} 
 is achieved when $s_{0}=\iota_{0}$.  Accordingly, we set
  $s_{0}=\max(\iota_{0},v)$.  Define $\llceil
  x\rrceil=\max(\iota_{0},x)$. Substituting for $a_{0}$ and $a_{1}$
  gives
  \begin{align}
    a_{0} 
     &= \llceil v\rrceil(\llceil v\rrceil - 2\coth(\llceil v\rrceil))\notag\\
     &=\max(0,v(v-2 \coth(v))), \notag\\
    a_{1}
      &= \llceil v\rrceil \csch(\llceil v\rrceil)\notag\\
      &= \min(\iota_{0}\csch(\iota_{0}),v\csch(v)),\notag\\
    2\cosh(v)a_{1}+a_{0} &\leq 2\cosh(\llceil v\rrceil)a_{1}+a_{0}\notag\\
    &= \llceil v\rrceil^{2}.
    \label{eq:thm:ee_to_qef:aim}
  \end{align}
  
  To return to Eqs.~\ref{eq:thm:ee_to_qef:coshprob1}
  and~\ref{eq:thm:ee_to_qef:coshprob2}, we identify
  $s=(1-\gamma)t$ to match constraints. In
  Eq.~\ref{eq:thm:ee_to_qef:coshprob1}, this requires multiplying
  both sides by $(1-\gamma)^{2}$ to match the constraint of
  Eq.~\ref{eq:thm:ee_to_qef:simplified_coshprob}, after which we
  must identify $\lambda_{1}(1-\gamma)^{2}=a_{1}$ and
  $\lambda_{0}(1-\gamma)^{2}=a_{0}$.  For the objective function, we
  consider Eq.~\ref{eq:thm:ee_to_qef:coshprob2} to identify
  $v=w(a)$, as the positive prefactor
  $e^{\gamma a}/(1-\gamma)^{2}$ does not affect the optimizing
  variables. Since $s_{0}=\llceil w(a)\rrceil$,  this yields
  \begin{align}
    \lambda_{0,\gamma}(a) &= \frac{1}{(1-\gamma)^{2}} \llceil
    w_{\gamma}(a)\rrceil (\llceil w_{\gamma}(a)\rrceil-2\coth(\llceil
    w_{\gamma}(a)\rrceil)),
    \notag\\
    \lambda_{1,\gamma}(a) &= \frac{1}{(1-\gamma)^{2}} \llceil
    w_{\gamma}(a)\rrceil \csch(\llceil w_{\gamma}(a)\rrceil),
    \label{eq:thm:ee_to_qef:lambdas}
  \end{align}
  where we now make the parameter $\gamma$ explicit with subscripts
  and make $a$ visible as an argument of the $\lambda_{i}$.
  To apply Lem.~\ref{lem:meas_to_powerbnd}, we expand 
  \begin{equation}
    \mu_{a,\gamma} =
    \lambda_{1,\gamma}(a)\delta_{-1+2\gamma}+\lambda_{0,\gamma}(a)\delta_{\gamma}
    + \lambda_{1,\gamma}(a)\delta_{1}
  \end{equation}
  according to Eq.~\ref{eq:thm:ee_to_qef:muchoice}, where we now make
  the dependence on $a$ visible as a subscript. We then apply the
  replacement $d\nu_{a,\gamma}(\chi)$ by
  $e^{a\chi}d\mu_{a,\gamma}(\chi)$ used to arrive at the constraint of
  Eq.~\ref{eq:thm:ee_to_qef:smu1}, and finally express the
  $d\mu_{a}(\chi)$ required for applying
  Lem.~\ref{lem:meas_to_powerbnd} as the weighted combination of
  $d\nu_{a,\gamma}(\chi)=e^{a\chi}d\mu_{a,\gamma}(\chi)$ with
  $\gamma=0$ and $\gamma=\beta$ suggested by the form of
  Eq.~\ref{eq:thm:ee_to_qef:wmu}.  This gives
  \begin{equation}  \label{eq:thm:ee_to_qef:solution1}
      d\mu_{a}(\chi) = \frac{1}{3}e^{a\chi}\left(2d\mu_{a,0}(\chi)
      +d\mu_{a,\beta}(\chi)\right).
  \end{equation}
  The construction above ensures that $\mu_{a}$ satisfies the
  condition in Lem.~\ref{lem:meas_to_powerbnd}.
  Expanding in terms of the parameters found we get
  \begin{align}\label{eq:thm:ee_to_qef:solution2}
      \mu_{a} &= \frac{2\lambda_{1,0}(a)}{3}e^{-a}\delta_{-1}\notag\\
    &\hphantom{=\;}
    +\frac{\lambda_{1,\beta}(a)}{3}e^{-a(1-2\beta)}\delta_{-1+2\beta}\notag\\
    &\hphantom{=\;}
    + \frac{2\lambda_{0,0}(a)}{3}\delta_{0}
    +\frac{\lambda_{0,\beta}(a)}{3}e^{a\beta}\delta_{\beta} +
    \frac{2\lambda_{1,0}(a)+\lambda_{1,\beta}(a)}{3}e^{a}\delta_{1}.
  \end{align}
  For $\beta< 1/2$,
  the terms of Lem.~\ref{lem:meas_to_powerbnd} behind $\tr(\tau(z))$ are
  \begin{align}
    &N \bar\mu_{z,\bar k(z)}(\{-1\}) &&=
    \frac{2N}{3}\left(
      \lambda_{1,0}(k_{\min}(z))e^{-k_{\min}(z)}
      \vphantom{\frac{\bar
          k(z)-k_{\min}(z)}{k_{\max}(z)-k_{\min}(z)}}\right.\notag\\
    &&&\hphantom{=\;}\left.
      + \frac{\bar k(z)-k_{\min}(z)}{k_{\max}(z)-k_{\min}(z)}
      \left(\lambda_{1,0}(k_{\max}(z))e^{-k_{\max}(z)}
        -\lambda_{1,0}(k_{\min}(z))e^{-k_{\min}(z)}\right)
      \right)
    ,\notag\\
    &\mathrlap{\int_{(-1,0)}N^{-\chi}d(\mu_{k_{\min}(z)}\vee\mu_{k_{\max}(z)})(\chi)} 
    \notag\\
    &&&= \frac{1}{3}\max\left(
      N^{1-2\beta}\lambda_{1,\beta}(a)e^{-a(1-2\beta)}
      :a\in\{k_{\min}(z),k_{\max(z)}\}
    \right)
    ,\notag\\
    \mathrlap{\max\left(
      \int_{[0,1]}d\mu_{a}(\chi)
      :a\in\{k_{\min}(z),k_{\max}(z)\}\right)}\notag\\
    &&&=
    \frac{1}{3}\max\left(
      \int_{[0,1]}2 e^{a\chi}d\mu_{a,0}(\chi)
      + 
      \int_{[0,1]} e^{a\chi}d\mu_{a,\beta}(\chi)
      :a\in\{k_{\min}(z),k_{\max}(z)\}\right)
    \notag\\
    &&&=
    \frac{1}{3}\max\left(2\lambda_{0,0}(a)
        +\lambda_{0,\beta}(a)e^{a\beta} +
        (2\lambda_{1,0}(a)+\lambda_{1,\beta}(a))e^{a}\right.\notag\\
      &&&\left.\hphantom{.=\frac{1}{3}\max()} :
      a\in\{k_{\min}(z),k_{\max(z)}\}\right)
   .
   \label{eq:thm:ee_to_qef:three1}
  \end{align}
  These expressions are ready to implement for specific applications.
  It remains to obtain the bound in the statement of the theorem.  For
  this, we use the bound $w_{\gamma}(a)\leq \tilde w_{\gamma}(a)$
  obtained earlier.  

  We first simplify the third expression in
  Eq.~\ref{eq:thm:ee_to_qef:three1} by means of the inequality
  \begin{align}
    \int_{[0,1]}e^{a\chi}d\mu_{a,\gamma}(\chi)
    &\leq \int_{[-1,1]}N^{-\chi\knuth{\chi\leq 0}}
    e^{a\chi}d\mu_{a,\gamma}(\chi).
  \end{align}
  The right-hand side is the quantity in
  Eq.~\ref{eq:thm:ee_to_qef:smu2} that was evaluated in
  Eq.~\ref{eq:thm:ee_to_qef:coshprob2} and then minimized.  It is
  related to the third quantity given and bounded in
  Eq.~\ref{eq:thm:ee_to_qef:aim} by the conversion from the $a_{i}$ to
  the $\lambda_{i}$ and a scale, namely by a factor of $e^{\gamma
    a}/(1-\gamma)^{2}$. This gives
  \begin{align}
    \int_{[0,1]}e^{a\chi}d\mu_{a,\gamma}(\chi)
    &\leq e^{\gamma
      a}(2\cosh(w_{\gamma}(a))\lambda_{1,\gamma}(a)+\lambda_{0,\gamma}(a))
    \notag\\
    &\leq\frac{e^{\gamma a}}{(1-\gamma)^{2}}\llceil w_{\gamma}(a)\rrceil^{2}\notag\\
    &\leq\frac{e^{\gamma k_{\max}(z)}}{(1-\gamma)^{2}}\llceil w_{\gamma}(a)\rrceil^{2}\notag\\
    &\leq\frac{e^{\gamma k_{\max}(z)}}{(1-\gamma)^{2}}\llceil \tilde w_{\gamma}(a)\rrceil^{2}.
  \end{align}
  The third expression is therefore bounded by
  \begin{align}
    \frac{1}{3}\max\left(
      2\llceil\tilde w_{0}(a)\rrceil ^{2}
     +
     \frac{e^{\beta k_{\max}(z)}}{(1-\beta)^{2}}\llceil \tilde w_{\beta}(a)\rrceil^{2}
      ,a\in\{k_{\min}(z),k_{\max}(z)\}\right)\hspace*{-2.5in}&\notag\\
    &= \frac{1}{3}\left(2\llceil\bar w_{0}(z)\rrceil ^{2}
     +
     \frac{e^{\beta k_{\max}(z)}}{(1-\beta)^{2}}\llceil \bar w_{\beta}(z)\rrceil^{2}\right),
  \end{align}
  where $\bar w_{\beta}(z)$ is as defined in the theorem statement.

  Next, the first expression of Eq.~\ref{eq:thm:ee_to_qef:three1} is bounded by
  \begin{align}
    N \bar\mu_{z,\bar k(z)}(\{-1\}) 
    &\leq
    \frac{2}{3}\max\left(
      Ne^{-a}\lambda_{1,0}(a):a\in\{k_{\min}(z),k_{\max}(z)\}
      \right),
  \end{align}
  which differs from the second expression of Eq.~\ref{eq:thm:ee_to_qef:three1} only in the initial factor
  and a replacement of $\beta$ by $0$.  In view of the definition of
  $w_{\gamma}(a)$ after Eq.~\ref{eq:thm:ee_to_qef:coshprob2} and the
  expression for $\lambda_{1,\gamma}(a)$ in Eq.~\ref{eq:thm:ee_to_qef:lambdas},
   \begin{align}
     N^{1-2\gamma}e^{-a(1-2\gamma)}\lambda_{1,\gamma}(a)
     &\leq
     e^{a\gamma}\left(N^{1-2\gamma}e^{-a(1-\gamma)}
       + e^{a(1-\gamma)}\right)\lambda_{1,\gamma}(a)\notag\\
     &= e^{a\gamma}2\cosh(w_{\gamma}(a))\lambda_{1,\gamma}(a)\notag\\
     &\leq e^{k_{\max}(z)\gamma}2\cosh(w_{\gamma}(a))\lambda_{1,\gamma}(a)\notag\\
     &= \frac{e^{k_{\max}(z)\gamma}}{(1-\gamma)^{2}}
     2\cosh(w_{\gamma}(a))\llceil w_{\gamma}(a)\rrceil\csch(\llceil w_{\gamma}(a)\rrceil)
     \notag\\
     &\leq\frac{e^{k_{\max}(z)\gamma}}{(1-\gamma)^{2}}
     2\cosh(\llceil w_{\gamma}(a)\rrceil)\llceil w_{\gamma}(a)\rrceil\csch(\llceil w_{\gamma}(a)\rrceil)
     \notag\\
     &=\frac{e^{k_{\max}(z)\gamma}}{(1-\gamma)^{2}}
     2 \llceil w_{\gamma}(a)\rrceil\coth(\llceil w_{\gamma}(a)\rrceil)\notag\\
     &\leq\frac{e^{k_{\max}(z)\gamma}}{(1-\gamma)^{2}}
     2 \llceil \tilde w_{\gamma}(a)\rrceil\coth(\llceil \tilde w_{\gamma}(a)\rrceil),
   \end{align}
   where the last inequality follows from monotonicity of $x\coth(x)$.  \Pc{The
     derivative of $x\coth(x)$ is $\coth(x)-x\csch(x)^{2}$.  After
     dividing by $\csch(x)^{2}$ we get $\cosh(x)\sinh(x)-x$, and
     $\sinh(x)\geq x$ for $x\geq 0$.}  With this we can combine the
   bounds for the first and second expressions to
   \begin{align}
     \mathrlap{\max\left(
       \frac{4}{3}\llceil\tilde w_{0}(a)\rrceil\coth(\llceil\tilde w_{0}(a)\rrceil):
       a\in\{k_{\min},k_{\max}\}\right)}\hspace*{1in}&\notag\\
      &+
     \max\left(
       \frac{2 e^{k_{\max}(z)\beta}}{3(1-\beta)^{2}}
       \llceil\tilde w_{\beta}(a)\rrceil\coth(\llceil\tilde w_{\beta}(a)\rrceil):
       a\in\{k_{\min},k_{\max}\}\right)\notag\\
     &=
     \frac{1}{3}\left(\vphantom{\frac{e^{k_{\max}(z)\beta}}{(1-\beta)^{2}}}
       4\llceil\bar w_{0}(z)\rrceil\coth(\llceil\bar w_{0}(z)\rrceil)
       \right.\notag\\
       &\left.\hphantom{=\frac{2}{3}\max()}
       +2\frac{e^{k_{\max}(z)\beta}}{(1-\beta)^{2}}
         \llceil\bar w_{\beta}(z)\rrceil\coth(\llceil\bar w_{\beta}(z)\rrceil)\right).
   \end{align}
   By combining the bounds on all three
   expressions we get
   \begin{align}
     \sum_{c}h_{2}(\beta;K(cz),\tau(cz)|\tau(z))
     &\leq
     \frac{\tr(\tau(z))}{3}\left(
       \vphantom{\frac{e^{k_{\max}(z)\beta}}{(1-\beta)^{2}}}
       2 \llceil\bar w_{0}(z)\rrceil\left(\llceil\bar w_{0}(z)\rrceil 
           + 2\coth(\llceil\bar w_{0}(z)\rrceil)\right)
        \right.\notag\\
      &\left.\hphantom{\leq\frac{\tau(z)}{3}()} + 
       \frac{e^{k_{\max}(z)\beta}}{(1-\beta)^{2}}
       \llceil\bar w_{\beta}(z)\rrceil\left(\llceil\bar w_{\beta}(z)\rrceil 
           + 2\coth(\llceil\bar w_{\beta}(z)\rrceil)\right)\right).
   \end{align}  
   The bound in the theorem statement follows.
\end{proof}

The form of the measure $\mu$ in Eq.~\ref{eq:thm:ee_to_qef:muchoice}
for the proof of Thm.~\ref{thm:ee_to_qef} is guided by its potential
for closed-form determination of optimal parameters.  It is not
optimal for minimizing Eq.~\ref{eq:thm:ee_to_qef:smu2} subject to
Eq.~\ref{eq:thm:ee_to_qef:smu1}, which is the intent at that point in
the proof. An optimal measure $\mu$ is of the form
$b_{-}\delta_{-1}+b_{0}\delta_{0}+b_{+}\delta_{1}$.  To see this,
reparameterize $\mu$ by $d\mu(\chi)= N^{\chi\knuth{\chi\leq 0}}e^{-a\chi
  }d\mu'(\chi)$. The problem in terms of $\mu'$ is to
\begin{alignat}{3}
  \textrm{Minimize:\ } & \int_{[-1,1]}d\mu'(\chi)\notag\\
  \textrm{Variable:\ } & \textrm{The positive measure $\mu'$}\notag\\
  \textrm{Subject to:\ } & 
  \int_{[-1,1]}N^{\chi\knuth{\chi\leq 0}}e^{t(\chi-\gamma)-a\chi }d\mu'(\chi)\geq t^{2}
   \textrm{\ for all $t\in\rls$} .
\end{alignat}
Let $\mu'$ be a feasible solution.
For any positive measure $\nu$, reals $c\leq d$, measurable
$I\subseteq [c,d]$, and real parameter $s$, convexity of $x\mapsto e^{sx}$
implies
\begin{align}
  e^{sc}\int_{I}\frac{d-\chi}{d-c}d\nu(\chi)
  +e^{sd}\int_{I}\frac{\chi-c}{d-c}d\nu(\chi)
  &= \int_{I}\left(\frac{d-\chi}{d-c}e^{sc}+\frac{\chi-c}{d-c}e^{sd}\right)d\nu(\chi)
  \notag\\
  &\geq
  \int_{I}e^{s(c(d-\chi)/(d-c)+d(\chi-c)/(d-c))}d\nu(\chi)\notag\\
  &=\int_{I}e^{s\chi}d\nu(\chi).
\end{align}
By applying this inequality with $\nu=\mu'$, first with $c=-1$, $d=0$,
$I=[-1,0)$ and $s=\log(N)+t-a$, then with $c=0$, $d=1$, $I=(0,1]$ and
$s=t-a$, we find that the measure
\begin{align}
  \mu'' &= \delta_{-1}\int_{[-1,0)}(-\chi)d\mu'(\chi)
             +\delta_{1}\int_{(0,1]}\chi d\mu'(\chi)\notag\\
           &\hphantom{=\;}  +
             \delta_{0}\left(\mu'(\{0\})+
               \int_{[-1,0)}(\chi+1)d\mu'(\chi) +
               \int_{(0,1]}(1-\chi)d\mu'(\chi)\right)
\end{align}
is a feasible solution with the same value for the objective function.
The measure $\mu''$ can be interpreted as a redistribution of $\mu'$
to point measures at $-1$, $0$ and $1$. It is possible to apply this
technique to improve the bound in Thm.~\ref{thm:ee_to_qef} by
redistributing the contribution of
$\lambda_{1,\gamma}\delta_{-1+2\beta}$ to $\chi=-1$ and $\chi=0$ and
of $\lambda_{0,\gamma}\delta_{\beta}$ to $\chi=0$ and $\chi=1$.  This
mostly helps when $\tilde w(a)$ is not large.

\subsection{Comparison to the EAT}
\label{subsec:handeat}

The entropy accumulation theorem (EAT) is the main result of
Ref.~\cite{dupuis:qc2016a} (Thm. 4.4). It uses a different framework
for describing models, where models are obtained from an explicit
quantum representation of the devices.  The EAT estimates conditional
min-entropy from min-tradeoff functions applied to the observed
frequencies of a CV. The estimate can be used with quantum-quantum
states.  Here we consider the case of classical-quantum states
matching our scenarios, a restriction also used in
Ref.~\cite{arnon-friedman:qc2016a} for the same reasons.  Models
$\cC(\Sfnt{CZ})$ in our framework that fit the conditions of the EAT
are EAT models as introduced and related to 
EAT channel chains in Sect.~\ref{sec:models:examples}. EAT models are
chained with conditionally independent inputs from models induced by
POVMs associated with a given class of quantum processes.  With our
notation, the following is an instance of the EAT:

\begin{theorem} \label{thm:eat} \emph{Entropy Accumulation Theorem for
    Conditional Min-Entropy~\cite{dupuis:qc2016a}:} Let $K(CZ)$ be an entropy
  estimator for $C|Z$ and $\cC(CZ)$, where $\cC(CZ)$ is the trial
  model for EAT model $\cC(\Sfnt{CZ})$ with $n$ trials. Fix
  $\epsilon\in(0,1)$ and an entropy goal $h$ per trial. Let
  $\phi(\Sfnt{CZ})=(\sum_{i=1}^{n}K(C_{i}Z_{i})\geq nh)$.  Suppose
  $\{\phi'(\Sfnt{CZ})\}\subseteq\{\phi(\Sfnt{CZ})\}$ and
  $\rho(\Sfnt{CZ})\in\cC(\Sfnt{CZ})$.  Define
  $\kappa=\tr(\rho(\phi'))$,  $k_{\infty}=\max_{cz}|K(cz)|$
  and $N=|\Rng(C)|$. Then
  \begin{equation}
    H_{\infty}^{\epsilon}(\Sfnt{C}|\Sfnt{Z}\Pfnt{E};\rho(\Sfnt{CZ}|\phi'))
    \geq nh -
    2\sqrt{\log_{2}(e)}\big(\log(1+2N)+\lceil k_{\infty}\rceil\big)
    \sqrt{|\log(\epsilon^{2}\kappa^{2}/2)|}
    \sqrt{n}.
  \end{equation}
\end{theorem}
The EAT in Ref.~\cite{dupuis:qc2016a} is expressed in terms of bits.
We convert terms on both sides of the inequality to nits and change to
logarithms base $e$, which requires a factor of $\sqrt{\log_{2}(e)}$
for the error term. The version of the EAT given here omits the possibility
that the trial models may vary according to a predetermined schedule, which
can be taken into account in the min-tradeoff functions that
substitute for entropy estimators in Ref.~\cite{dupuis:qc2016a}.  For
the purpose of this comparison, we consider only the case where each
trial is constrained by the same model.

The EAT is formulated for affine min-tradeoff functions, not entropy
estimators. For the models under consideration, affine min-tradeoff
functions correspond to entropy estimators. With our notation, an
affine min-tradeoff function for $C|Z$ and $\cC(CZ)$ can be written as
a linear function $f:\mu(CZ)\mapsto f(\mu(CZ))\in \rls$ such that for
all $\rho(CZ)\in\cN(\cC(CZ))$, $f(\tr(\rho(CZ)))\leq
H_{1}(\rho(CZ)|Z\Pfnt{E})$. Since $f$ is linear, for probability
distributions $\mu(CZ)$ we can write
$f(\mu(CZ))=\sum_{cz}a_{cz}\mu(cz)+a_{0} =
\sum_{cz}(a_{cz}+a_{0})\mu(cz)$, so $f(\mu(CZ))=\Exp_{\mu(CZ)}K(CZ)$ where
$K(CZ):cz\mapsto a_{cz}+a_{0}$ is an entropy estimator.

A version of the EAT with a better coefficient of the $\sqrt{n}$ term
can be obtained by combining Thms.~\ref{thm:bnds_from_qef}
and~\ref{thm:ee_to_qef}. 

\begin{theorem}\label{thm:eat_from_qef}
  Let $0<\beta_{\max}<1/2$.
  Suppose that $\tilde c(\beta)$ is a continuous, non-decreasing function 
  of $\beta\in[0,\beta_{\max}]$ satisfying 
  $\tilde c(\beta)\geq c(\beta)\defeq \sup\left\{c(\beta,\nu(Z)):\nu(CZ)\in\tr(\vphantom{\big|}\cN(\cC(CZ)))\right\}$ with $c(\beta,\nu(Z))$ as defined in
  Thm.~\ref{thm:ee_to_qef}.  Define 
  \begin{equation}
    \bar\beta= \frac{\sqrt{2|\log(\epsilon^{2}\kappa^{2}/2)|}}
      {\sqrt{n \tilde c(0)}}.
  \end{equation}
  For $\bar\beta\leq\beta_{\max}$ and
  with the notation and assumptions of Thm.~\ref{thm:eat}
  \begin{equation}
    H_{\infty}^{\epsilon}(\Sfnt{C}|\Sfnt{Z}\Pfnt{E};\rho(\Sfnt{CZ}|\phi'))
    \geq nh  - \sqrt{2} \sqrt{\tilde c(\bar\beta)}\sqrt{|\log(\epsilon^{2}\kappa^{2})/2|}
    \sqrt{n}.
    \label{eq:thm:eat_from_qef:main}
  \end{equation}
\end{theorem}
\begin{proof}
  Let $G(\Sfnt{CZ})$ be the \QEF with power $\beta$ for
  $\Sfnt{C}|\Sfnt{Z}$ and $\cC(\Sfnt{CZ})$ obtained from chaining the
  \QEFP given by $e^{\beta K(CZ)}/(1+c_{P}(\beta))$ in
  Thm.~\ref{thm:ee_to_qef}.  Then
  $G(\Sfnt{CZ})=\prod_{i=1}^{n}e^{\beta
    K(C_{i}Z_{i})}/(1+c_{P}(\beta))$ and
  \begin{equation}
    \log(G(\Sfnt{CZ}))/\beta
    = \sum_{i=1}^{n}K(C_{i}Z_{i}) - n\log(1+c_{P}(\beta))/\beta
    \geq \sum_{i=1}^{n}K(C_{i}Z_{i}) - n c_{P}(\beta)/\beta.\label{eq:thm:eat_from_qef:1}
  \end{equation}
  The targeted threshold is $\sum_{i}K(C_{i}Z_{i})\geq nh$. The threshold
  in Thm.~\ref{thm:bnds_from_qef} with $F(CZ)$ there replaced by $G(\Sfnt{CZ})$ here
  is equivalent to
  \begin{equation}
    \sum_{i=1}^{n}\log(G(C_{i}Z_{i}))/\beta \geq -\log(p)-\log(\delta)/\beta.
  \end{equation}
  We set $\delta = \epsilon^{2}/2$ to achieve the error bound and
  determine $p$ by $-\log(p) = nh-n
  c_{P}(\beta)/\beta+\log(\delta)/\beta$.  The event $\{\phi\}$
    here is defined as $\{\sum_{i=1}^{n}K(C_{i}Z_{i})\geq nh\}$, and
    from Eq.~\ref{eq:thm:eat_from_qef:1}, $\phi$ implies
  \begin{equation}
    \log(G(\Sfnt{CZ}))/\beta \geq nh - n c_{P}(\beta)/\beta =-\log(p)-\log(\delta)/\beta
        = \log(\frac{1}{p\delta^{1/\beta}}),
  \end{equation}
  which matches the expression for $\phi$  in Thm.~\ref{thm:bnds_from_qef}.  
  The event $\phi'$ thus satisfies the conditions of Thm.~\ref{thm:bnds_from_qef}.  
  The conditional min-entropy bound is 
  \begin{align}
    -\log(p/\kappa^{\alpha/\beta}) &=
     nh-n c_{P}(\beta)/\beta+\log(\delta)/\beta+\log(\kappa^{\alpha})/\beta\notag\\
    &= nh-n c_{P}(\beta)/\beta-|\log(\epsilon^{2}\kappa^{\alpha}/2)|/\beta\notag\\
     &\geq nh - n \beta c(\beta)/2 - |\log(\epsilon^{2}\kappa^{\alpha}/2)|/\beta\notag\\
     &\geq nh - n \beta \tilde c(\beta)/2 - |\log(\epsilon^{2}\kappa^{\alpha}/2)|/\beta\notag\\
     &\geq nh - n \beta \tilde c(\beta)/2 - |\log(\epsilon^{2}\kappa^{2}/2)|/\beta,
     \label{eq:thm:eat_from_qef:2}
  \end{align}
  provided that $\beta\leq\beta_{\max}$.

  For the next step we need to extend the validity of the inequality
  $\tilde c(\beta)\geq c(\beta)$ to all $\beta<1$.  For
  $\beta>\beta_{\max}$, we define $\tilde c(\beta)=\max(\tilde
  c(\beta_{\max}), \max_{\beta'\in[0,\beta]} c(\beta'))$, which is still continuous and
  non-decreasing.  The quantity $-\log(p/\kappa^{\alpha/\beta})$ is a
  lower bound on the left-hand side of
  Eq.~\ref{eq:thm:eat_from_qef:main}, so we could choose
  $\beta\leq\beta_{\max}$ to maximize the last expression in
  Eq.~\ref{eq:thm:eat_from_qef:2}. To simplify the problem and find suboptimal solutions, 
    we use the case where $\tilde c(\beta)$ is independent of $\beta$ as a template.
     Specifically, if we replace $\tilde c(\beta)$ be a constant $\tilde c$ and maximize
      the last expression in Eq.~\ref{eq:thm:eat_from_qef:2}, we obtain the identity
      $\beta=\sqrt{2|\log(\epsilon^2\kappa^2/2)|}/\sqrt{n \tilde c}$. Substituting
      back $\tilde c(\beta)$ for $\tilde c$, we obtain the identity $\beta=f(\beta)$, where $f(\beta)\defeq
  \sqrt{2|\log(\epsilon^2\kappa^2/2)|}/\sqrt{n \tilde c(\beta)}$, and we choose $\beta$ to satisfies this identity. 
    Since $c(\beta)$ diverges as
  $\beta\nearrow 1$, $c(0)>0$, and $\tilde c(\beta)$ is non-decreasing in $\beta$ and satisfies $\tilde c(\beta)\geq c(\beta)$, the
  function $f(\beta)$  is
  positive at $\beta=0$, non-increasing in $\beta$, and goes to $0$ as $\beta\nearrow 1$.
    Moreover, since $\tilde c$ is continuous, so
    is $f$.  Accordingly there is a solution $\beta_{0}<1$ to the
  fixed-point equation
  \begin{equation}
    \beta_{0}=f(\beta_0)=\frac{\sqrt{2|\log(\epsilon^2\kappa^2/2)|}}{\sqrt{n \tilde c(\beta_{0})}}.
  \end{equation}
  Since $\tilde c(\beta)$ is non-decreasing in $\beta$, we have
  $\beta_{0}\leq\bar\beta$. Thus from Eq.~\ref{eq:thm:eat_from_qef:2}
  we obtain
  \begin{align}
    -\log(p/\kappa^{\alpha/\beta})& \geq nh- n \beta_0 \tilde c(\beta_0)/2 - |\log(\epsilon^{2}\kappa^{2}/2)|/\beta_0 \notag \\
    &= nh - \sqrt{2} \sqrt{\tilde c(\beta_{0})}\sqrt{|\log(\epsilon^2\kappa^2/2)|}\sqrt{n}
    \notag\\
    &\geq nh - \sqrt{2} \sqrt{\tilde c(\bar\beta)}\sqrt{|\log(\epsilon^2\kappa^2/2)|}\sqrt{n}.
  \end{align}
  
  The condition on $\bar\beta$ in the statement of the theorem is
  required to stay within the domain of the unextended function $\tilde c$.
\end{proof}

For comparison to the EAT, we determine a bound $\tilde c(\beta)\geq
c(\beta)$ satisfying the conditions in Thm.~\ref{thm:eat_from_qef}.
For a handicapped but direct comparison, we make conservative
estimates in terms of parameters that occur in the EAT to obtain
moderate improvements over the EAT. The main advantage of
Thm.~\ref{thm:eat_from_qef} is that one can choose $\tilde c(\beta)$
less conservatively, taking advantage of the average over inputs
in the expression for $c(\beta, \nu(Z))$ in Thm~\ref{thm:ee_to_qef},
which enables effective use of estimators that are heavily weighted toward
rare inputs.  This enables the clean exponential-expansion results of
Sect.~\ref{sec:maxprobest:exex}.

Let $k_{\max}=\max_{cz}K(cz)$ and $\bar
w_{\gamma}=\max_{z}\bar w_{\gamma}(z) = (1-\gamma)\max_{cz}(\max(\log(N)-K(cz),K(cz)))+\log(2)
=(1-\gamma)\max_{cz}(\log(N)/2+|\log(N)/2-K(cz)|)+\log(2)$.  We may
assume that $k_{\max}\geq 0$ as the entropy estimator is otherwise
useless. We have
\begin{align}
  c(\beta) &= 
  \sum_{z}\tr(\tau(z))\frac{1}{3}
  \left(\vphantom{\frac{e^{k_{\max}(z)\beta}}{(1-\beta)^{2}}}
    2 \llceil\bar w_{0}(z)\rrceil\big(\llceil\bar w_{0}(z)\rrceil 
      + 2\coth(\llceil\bar w_{0}(z)\rrceil)\big)
  \right.\notag\\
  &\left.\hphantom{\sum_{z}\tau(z)\frac{1}{3}()} + 
    \frac{e^{k_{\max}(z)\beta}}{(1-\beta)^{2}}
    \llceil\bar w_{\beta}(z)\rrceil\big(\llceil\bar w_{\beta}(z)\rrceil 
      + 2\coth(\llceil\bar w_{\beta}(z)\rrceil)\big)
  \right)\notag\\
  &\leq\frac{1}{3}\left(
    2 \llceil\bar w_{0}\rrceil\big(\llceil\bar w_{0}\rrceil 
      + 2\coth(\llceil\bar w_{0}\rrceil)\big)
    + 
    \frac{e^{k_{\max}\beta}}{(1-\beta)^{2}}
    \llceil\bar w_{\beta}\rrceil\big(\llceil\bar w_{\beta}\rrceil 
      + 2\coth(\llceil\bar w_{\beta}\rrceil)\big)\right),
  \label{eq:eatcomp1}
\end{align} 
where in the last step we used the facts that the functions 
$f(x)=x^2$ and $g(x)=x\coth(x)$ are monotonically increasing in $x$ when $x\geq0$. 
For a more specific comparison based on the parameters of
Thm.~\ref{thm:eat}, namely $k_{\infty}$, $N$, $\epsilon$, $\kappa$ and
$n$, we use $\bar w'_{\gamma}= (1-\gamma)(\log(N)+k_{\infty})+\log(2)
\geq \bar w_{\gamma}$.  Define $\tilde c(\beta)$ as the last
expression of Eq.~\ref{eq:eatcomp1} with $k_{\max}$, $\bar w_{0}$ and
$\bar w_{\beta}$ replaced by $k_{\infty}$, $\bar w'_{0}$ and $\bar
w'_{\beta}$, respectively. Then $\tilde c(\beta)$ is non-decreasing in
$\beta$ for $\beta<1$ 
and we can apply Thm.~\ref{thm:eat_from_qef} with any $\beta_{\max}<1/2$.  \Pc{To check
  monotonicity, the case where $\bar w'_{\beta}\leq \iota_{0}$ follows
  by inspection. For $\bar w'_{\beta}\geq\iota_{0}$, cancel prefactors of
  $1-\beta$ as needed and take into account the fact that $x\mapsto
  \coth(x)$ is decreasing for $x>0$, so $\beta \mapsto
  \coth((1-\beta)c)$ is increasing in $\beta$ for $c>0$.}

We first consider the asymptotic behavior as $n\rightarrow \infty$.
For simplicity, assume that $N\geq 4$, so that
$\log(2N)\geq\iota_{0}$.  We compare the coefficients 
$u_{\mathrm{EAT}}$ and $u_{\mathrm{QEF}}$ of the
$\sqrt{|\log(\epsilon^{2}\kappa^{2})/2|}\sqrt{n}$ terms in the
conditional min-entropy bounds. In Thm.~\ref{thm:eat_from_qef}, $\bar
\beta=O(1/\sqrt{n})$, so for large $n$ and with $\tilde c(\beta)$
as defined in the previous paragraph, we can
set $\bar \beta=0$. This gives
\begin{equation}
  u_{\mathrm{QEF}} = -\sqrt{2}
  \sqrt{(\log(2N)+k_{\infty})(\log(2N)+k_{\infty}+2\coth(\log(2N)+k_{\infty}))},
  \label{eq:eat_from_qef:sqrtn_prefix}
\end{equation}
where $2\leq 2\coth(\log(2N)+k_{\infty})\leq 2\coth(\log(8)) \approx 2.0635$.
This may be compared to 
\begin{equation}
   u_{\mathrm{EAT}}=-2\sqrt{\log_{2}(e)}\left(\log(1+2N)+\lceil k_{\infty}\rceil\right).
   \label{eq:eat:sqrtn_prefix}
\end{equation}
The terms involving $N$ and $k_{\infty}$ are similar and approach each
other for large $N$ or $k_{\infty}$. The constant initial factors in
Eq.~\ref{eq:eat_from_qef:sqrtn_prefix} and
Eq.~\ref{eq:eat:sqrtn_prefix} are $\sqrt{2}$ and $2\sqrt{\log_{2}(e)}$
respectively, which implies that $u_{\mathrm{EAT}}/u_{\mathrm{QEF}}$
approaches $\sqrt{2\log_{2}(e)}\approx 1.699$.  Of course, for large
$n$, the relative difference in conditional min-entropy witnessed
disappears.

For applications such as low-latency generation of a block of random
bits, optimal randomness expansion, or randomness with exponentially
small error, the above asymptotic regime is not relevant.  For the
next comparison, we parameterize the error term with
$l_{\epsilon}=|\log(\epsilon^{2}\kappa^{2}/2)|$.  We consider the
problem of determining the smallest $n$ for which there is positive
conditional min-entropy given $l_{\epsilon}$ and the threshold rate
$h$ for the entropy estimators in Thms.~\ref{thm:eat}
and~\ref{thm:eat_from_qef}.  This problem is closely related to the
problem where given an error bound rate $r$, we wish to determine the
infimum of the threshold rates $h$ such that if $l_{\epsilon}=rn$, the
asymptotic conditional min-entropy is positive. For the EAT, given
$l_{\epsilon}$ and $h$, the smallest value of $n$ for which the
conditional min-entropy lower-bound is positive is at least
\begin{equation}
  n_{\min,\mathrm{EAT}}(h,l_{\epsilon}) \defeq
  4\log_{2}(e)(\log(1+2N)+k_{\infty})^{2}l_{\epsilon}/h^{2}.
\end{equation}
If we set $l_{\epsilon}=rn$, then the smallest $h$ for which the
entropy lower-bound is non-negative is at least
\begin{equation}
  h_{\min,\mathrm{EAT}}(r) \defeq \left(4\log_{2}(e)(\log(1+2N)+k_{\infty})^{2} r\right)^{1/2}.
\end{equation}
The two expressions are related by $n_{\min,\mathrm{EAT}}(h,l_{\epsilon})h^{2}/l_{\epsilon}
= h_{\min,\mathrm{EAT}}(r)^{2}/r$.  In general, suppose we are given a
function $h_{\min}:r\mapsto h_{\min}(r)$ such that for all
$h>h_{\min}(r)$, the asymptotic conditional min-entropy with error
bound $l_{\epsilon}=rn$ is positive.  Then we can estimate the minimum
$n$ required for positive entropy given $l_{\epsilon}$ and $h$ from
$r_{\max}(h)=\sup\{r:h_{\min}(r)\leq h\}$ by computing $n$ according
to $n= l_{\epsilon}/r_{\max}(h)$. The estimate may be off because an
asymptotic computation of $h_{\min}(r)$ neglects lower-order terms,
but in the case of the EAT, it gives a valid answer.  In view of these
considerations, we compare the EAT and QEF constructions by
determining which has larger $r_{\max}(h)$. For this, we determine
$h_{\min,\mathrm{QEF}}(r)$ according to Thm.~\ref{thm:eat_from_qef}:
\begin{equation}
  h_{\min,\mathrm{QEF}}(r) = \left(2\tilde c(\bar \beta) r\right)^{1/2},
  \label{eq:hminqef}
\end{equation}
where we now use the function $\tilde c$ introduced after
Eq.~\ref{eq:eatcomp1} and $\bar\beta$ is
given in terms of $r$ by
\begin{equation}
  \bar\beta = \frac{\sqrt{2r}}{\sqrt{\tilde c(0)}}.
\end{equation}
Eq.~\ref{eq:hminqef} requires $\bar\beta\leq \beta_{\max}$, where
$\beta_{\max}< 1/2$, so we restrict $r$ accordingly.
An analytic comparison of the two expressions for $r_{\max}$ derived
from $h_{\min}$ is not simple, but we can plot specific examples for a
visual comparison. For this we consider relevant values of $N=2,4,8$
and $k_{\infty}=1,\log(N)$ and plot $r_{\max}$ as a function of
$h\in(0,\log(N))$, see Fig.~\ref{fig:eat_to_qef}. \Pc{\texttt{Octave}
  code to generate the data for these plots is embedded in the \TeX
  file.}  The values of $r_{\max}$ for \QEFs are up to a factor of
$2$ larger than those for the EAT. Such improvements in rates can be
significant in resource-limited applications.

\ignore{
function hmin = hmineat(r,N,kinf);
   # Vectorized.
   hmin = sqrt((4*log2(e)*(log(1+2.*N)+kinf).^2) .* r);
end;

global iota0;
function [f,fp]  = iotafn(x);
  f = 2*coth(x)-x;
  fp = -2/(sinh(x)^2) - 1;
end;
iota0 = fsolve(@iotafn,2);
function hmin = hminqef(r,N,kinf);
  # Vectorized.
  global iota0;
  w0 = max(iota0,log(N)+kinf+log(2));
  tildec0 = w0.*(w0+2*coth(w0));
  barb = sqrt(2.*r./tildec0);
  wb = max(iota0,(1-barb).*(log(N)+kinf)+log(2));
  tildecb = wb.*(wb+2*coth(wb));
  tildecbbar = (2/3)*tildec0 + (1/3)*(tildecb.*exp(kinf.*barb)./((1-barb).^2));
  hmin = sqrt(2*(tildecbbar.*r));
end;

function mr = mreat(h,N,kinf); 
   # Can use matched column vectors.
   mr = zeros(size(h));
   for k=(1:size(mr)(1));
     mr(k) = fminbnd(@(x)((hmineat(x,N(k),kinf(k))-h(k))^2), 0,h(k));
   end;
end;
function mr = mrqef(h,N,kinf); 
   # Can use matched column vectors.
   mr = zeros(size(h));
   for k=(1:size(mr)(1));
     mr(k) = fminbnd(@(x)((hminqef(x,N(k),kinf(k))-h(k))^2), 0,h(k));
   end;
end;

vals = [[2,1];[4,1];[8,1];[2,log(2)];[4,log(4)];[8,log(8)]];
data = {};
for j=(1:size(vals)(1));
  N=vals(j,1); kinf=vals(j,2);
  hpts = (0.01:0.05:log(N))';
  vpts = ones(size(hpts));
  mreats = mreat(hpts,N*vpts,kinf*vpts);
  mrqefs = mrqef(hpts,N*vpts,kinf*vpts);
  data{j} = [hpts,mreats,mrqefs];
end;

save('rmax_data.oct','data');

load('rmax_data.oct');
graphics_toolkit('gnuplot'); 

hold off;
h=figure(1);
set(h,'defaultaxesfontname','Helvetica') 
set(h,'papersize',[3.5,3.5]);
set(h,'paperposition',[0,0,3.5,3.5]);
plot(data{1}(:,1),data{1}(:,2), 'color', 'k', 'linewidth', 2, 'linestyle', '--');
xlabel('Entropy rate (nits/trial)','fontsize',10,'interpreter','tex');
ylabel('Maximum error bound rate (nits/trial)','fontsize',10,'interpreter','tex');
set(gca(),'fontsize',10);
xlim([0,2]);
ylim([0,.07]);
hold on;
plot(data{1}(:,1),data{1}(:,3), 'color', 'k', 'linewidth', 2, 'linestyle', '-');

plot(data{2}(:,1),data{2}(:,2), 'color', [0.3, 0.6, 1], 'linewidth', 2, 'linestyle', '--');
plot(data{2}(:,1),data{2}(:,3), 'color', [0.3, 0.6, 1], 'linewidth', 2, 'linestyle', '-');

plot(data{3}(:,1),data{3}(:,2), 'color', [1, 0.3, 0.0], 'linewidth', 2, 'linestyle', '--');
plot(data{3}(:,1),data{3}(:,3), 'color', [1, 0.3, 0.0], 'linewidth', 2, 'linestyle', '-');

print -dpdf eat_to_qef1.pdf

hold off;
h=figure(2);
set(h,'defaultaxesfontname','Helvetica') 
set(h,'papersize',[3.5,3.5]);
set(h,'paperposition',[0,0,3.5,3.5]);
plot(data{4}(:,1),data{4}(:,2), 'color', 'k', 'linewidth', 2, 'linestyle', '--');
xlabel('Entropy rate (nits/trial)','fontsize',10,'interpreter','tex');
set(gca(),'fontsize',10);
xlim([0,2]);
ylim([0,.07]);
hold on;
plot(data{4}(:,1),data{4}(:,3), 'color', 'k', 'linewidth', 2, 'linestyle', '-');

plot(data{5}(:,1),data{5}(:,2), 'color', [0.3, 0.6, 1], 'linewidth', 2, 'linestyle', '--');
plot(data{5}(:,1),data{5}(:,3), 'color', [0.3, 0.6, 1], 'linewidth', 2, 'linestyle', '-');

plot(data{6}(:,1),data{6}(:,2), 'color', [1, 0.3, 0.0], 'linewidth', 2, 'linestyle', '--');
plot(data{6}(:,1),data{6}(:,3), 'color', [1, 0.3, 0.0], 'linewidth', 2, 'linestyle', '-');

print -dpdf eat_to_qef2.pdf

}

\begin{figure}
  \begin{center}
    \begin{picture}(7,3.5)(.25,0)
      \put(0,0){\includegraphics[]{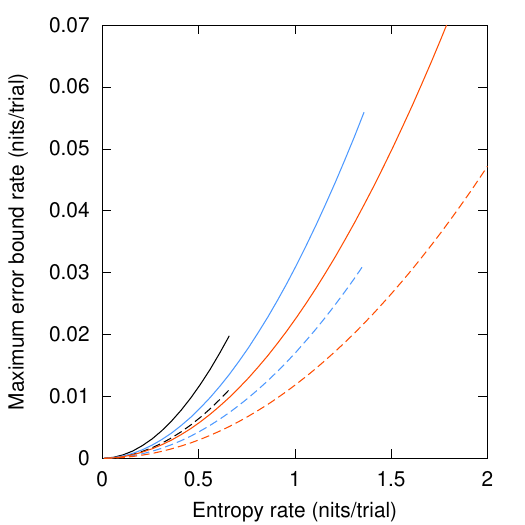}}
      \put(1,3.1){\makebox(0,0)[tl]{\large$k_{\infty}=1$:}}
      \put(0.9,1){\makebox(0,0)[bl]{$N=2$}}
      \put(1.6,2){\makebox(0,0)[bl]{$N=4$}}
      \put(2.4,3.1){\makebox(0,0)[bl]{$N=8$}}
      \put(3.5,0){\includegraphics[]{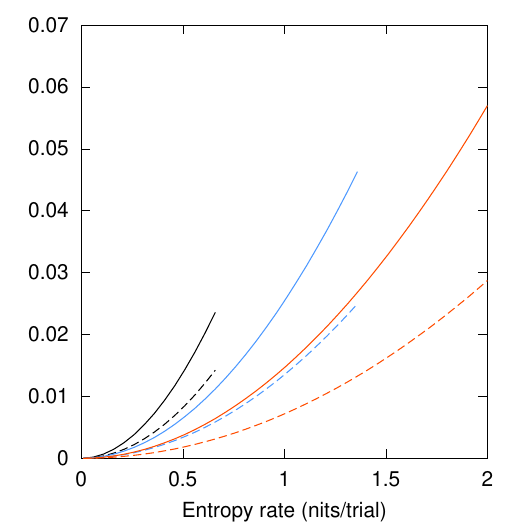}}
      \put(4.5,3.1){\makebox(0,0)[tl]{\large$k_{\infty}=\log(N)$:}}
      \put(4.2,1){\makebox(0,0)[bl]{$N=2$}}
      \put(5.2,2){\makebox(0,0)[bl]{$N=4$}}
      \put(6.2,2.8){\makebox(0,0)[bl]{$N=8$}}
    \end{picture}
  \end{center}
  \caption{Maximum error bound rates versus entropy threshold rates.
    The left plot has $k_{\infty}=1$, the right has
    $k_{\infty}=\log(N)$, where $N=|\Rng(C)|$.  Three pairs of curves are shown in each
    plot, for $N=2,4,8$ as labeled.  The dashed lines show the EAT
    curves, and the solid lines show the \QEFP curves according to the
    handicapped calculations in the text. From the maximum error bound
    rate $r_{\max}$ one can estimate the minimum number $n_{\min}$ of
    trials required for positive smooth conditional min-entropy
    with an error bound of $\epsilon$ at
    probability of success $\kappa=1$. The estimate is given by
    $n_{\min}=|\log(\epsilon^{2}/2)|/r_{\max}$.  The higher \QEFP
    curves imply about half the number of trials are required.
    Further improvements are possible by taking full advantage of
    Thm.~\ref{thm:ee_to_qef} and its proof.
    Achievable entropy threshold rates are determined by the entropy
    estimator and the trial probability distribution.
    \label{fig:eat_to_qef}}
\end{figure}

The values of $h$ occurring in the comparison have not been
constrained.  But since they play the role of a threshold rate for an
entropy estimator, the probability that the entropy estimate exceeds
$nh$ must be sufficiently large. Values of $h$ for which this is not
the case in a given situation are not relevant. For a given trial
distribution, this normally requires that $h$ is below the expected
value of the entropy estimator.

To finish this section, we remove the handicap to demonstrate the
broad applicability and finite-data efficiency of \QEFs. Let
$p\in(0,1)$ and consider the trial model $\cC(C)$ with
$\Rng(C)=\{0,1\}$, no inputs, and no quantum correlations, defined by
$\cC(C)=\Cvx\left(\{\mu(C)\rho: \mu(1)\leq p\}\right)$.  The
extremal states of $\cN(\cC(C))$ are of the form $\knuth{C=0}\hat
\psi$ and $((1-p)\knuth{C=0}+p\knuth{C=1})\hat\psi$.  This model is
equivalent to a classical-side-information model and may be relevant
for semi-device-dependent randomness generation.  For the extremal
states, if the number of times that $C=1$ is observed in $n$ trials is
$k$, then the probability of the experiment's output is at most
$p^{k}$ from $\Pfnt{E}$'s point of view. Converting this information
to a conditional min-entropy estimate without using \QEFs or the EAT
requires taking into account the probability that $k$ exceeds some
threshold. We do not attempt this conversion, but it suggests that it
is natural to analyze this model directly rather than to use \QEFs or
invoke the EAT.  However, \QEFs and the EAT are applicable and,
according to the optimality theorem Thm.~\ref{thm:optimality}, achieve
the asymptotically optimal rate for randomness generation.

\QEFs for $\cC(C)$ can be written in the form $F(C):c\mapsto
(\knuth{c=0}+f \knuth{c=1})/m$ where $f$ and $m$ are constrained so
that the \QEF inequality with power $\beta$ is satisfied.  The \QEF
inequalities for the two extremal states are
\begin{align}
  \frac{1}{m} &\leq 1\notag\\
  \frac{(1-p)^{\alpha}+f p^{\alpha}}{m} &\leq 1.
\end{align}
Thus $m\geq 1$, and given $m$, we choose $f$ as large as possible, which
gives $f=(m-(1-p)^{\alpha})/p^{\alpha}$.  The log-prob rate of $F(C)$
at $\mu(C):c\mapsto (1-q)\knuth{c=0}+q\knuth{c=1}$ with $q\in[0,p]$ is
\begin{equation}
  \cL_{q,\beta}(m)=
  \left(\vphantom{\big|}
    q\log((m-(1-p)^{\alpha})/p^{\alpha})-\log(m)\right)/\beta.
\end{equation}
To maximize the log-prob rate with respect to $m$, compute
\begin{equation}
  \beta\frac{d}{dm}\cL_{q,\beta}(m) = 
  \frac{q}{m-(1-p)^{\alpha}}-\frac{1}{m}
  = \frac{-(1-q)m +(1-p)^{\alpha}}{(m-(1-p)^{\alpha})m}.
\end{equation}
Since $1-q\geq 1-p$ and $\alpha>1$, $\frac{d}{dm}\cL_{q}(m)\leq 0$ for
$m\geq 1$, so the maximum is achieved at $m=1$.  

To illustrate the asymptotic optimality of QEFs established in
Sect.~\ref{subsec:qefoptimality}, we compute the limit $\beta\searrow
0$ of the log-prob rate.  Rearranging terms and the estimates
$(1-p)^{\beta}=1+\beta\log(1-p)+O(\beta^{2})$ and $\log(1+\beta
d+O(\beta^{2}))=\beta d + O(\beta^{2})$ give
\begin{align}
  \cL_{q,\beta}(1) &= \frac{q}{\beta}
  \left(\log(1-(1-p)^{1+\beta})-(1+\beta)\log(p)\right)\notag\\
  &=\frac{q}{\beta}\left(\log(p+(1-p)(1-(1-p)^{\beta})) -(1+\beta)\log(p)\right)
  \notag\\
  &=\frac{q}{\beta}\left(\log(p+(1-p)(-\beta\log(1-p)+O(\beta^{2})))
    -(1+\beta)\log(p)\right)\notag\\
  &=\frac{q}{\beta}\left(\log(p)+\log(1+((1-p)/p)(-\beta\log(1-p)+O(\beta^{2})))
    -(1+\beta)\log(p)\right)\notag\\
  &=\frac{q}{\beta}\left(\log(1-\beta(1-p)\log(1-p)/p + O(\beta^{2}))
    -\beta\log(p)\right)\notag\\
  &=\frac{q}{\beta}\left(-\beta(1-p)\log(1-p)/p + O(\beta^{2})
    -\beta\log(p)\right)\notag\\
  &=-\frac{q}{p}\left((1-p)\log(1-p)+
    p\log(p)\right)+O(\beta),
\end{align}
so $\cL_{q,0_{+}}(1) = (q/p)H(p)$, where $H(p)$ is the Shannon entropy
of the distribution $(1-p)\knuth{C=0}+p\knuth{C=1}$ in nits.  The
log-prob rate $\cL_{q,0_{+}}(1)$ can be recognized as the minimum
conditional entropy for states whose output distribution is
$(1-q)\knuth{C=0}+q\knuth{C=1}$ given the model $\cC(C)$, see
Sect.~\ref{subsec:qefoptimality}.

For comparing to the EAT, we fix $q\in(0,p]$ and consider the
simplified \QEF $F_{\beta}(C):c\mapsto
\knuth{c=0}+p^{-\beta}\knuth{c=1}$.  Because for $m=1$,
$f=(1-(1-p)^{\alpha})/p^{\alpha}\geq (1-(1-p))/p^{\alpha}=p^{-\beta}$,
this \QEF satisfies the \QEF inequalities.  Given $q$, from the
previous paragraph, the optimal log-prob rate is $h_{s}=(q/p)H(p)$.  The
log-prob rate of $F_{\beta}(C)$ is $h_{F}=q|\log(p)|\leq h_{s}$. For small
$p$, the ratio of the two rates approaches $1$. We determine the
minimum $n$ such that positive conditional min-entropy can be
certified. Let $\epsilon$ be the error bound and $\kappa$ the minimum
probability of success that we need to protect against.  For the EAT
with  entropy goal $h$ per trial,
\begin{align}
  n_{\min,\mathrm{EAT}} &\geq
  4\log_{2}(e)(\log(1+2N)+k_{\infty})^{2}|\log(\epsilon^{2}\kappa^{2}/2)|\frac{1}{h^{2}}
  \notag\\
  &> 4\log_{2}(e)\log(5)^{2}|\log(\epsilon^{2}\kappa^{2}/2)|\frac{1}{h^{2}},
\end{align}
where $4\log_{2}(e)\log(5)^{2}\approx 14.95$ and we set $k_{\infty}=0$ for
a lower bound. For the \QEF $F_{\beta}(C)$ with power $\beta$,
we apply Thm.~\ref{thm:bnds_from_qef} to get 
\begin{equation}
  n_{\min,\mathrm{QEF}} = 
  \left(|\log(\epsilon^{2}\kappa/2)|/\beta + |\log(\kappa)|\right)
                                           \frac{1}{h},
\end{equation}
where $\beta$ can be chosen arbitrarily large.  (To obtain this minimum $n_{\min,\mathrm{QEF}}$, 
we set $\delta=\epsilon^{2}/2$ and $q=e^{-nh}$ in Thm.~\ref{thm:bnds_from_qef} such that the $\epsilon$-smooth 
conditional min-entropy certified according to this theorem is bounded below by 
$n h +\log(\epsilon^{2}/2)/\beta+\log(\kappa^{\alpha/\beta})$.) For the explicit \QEF $F_{\beta}(C)$,
one can choose any $h\leq h_{F}$, but to satisfy completeness with
reasonable probabilities of success given the anticipated probability $q$ of $C=1$ 
and the QEF power $\beta$, the number
of trials needs to be at least some multiple of $1/q$.  For both the
EAT and \QEFs, useful values of $h$ are bounded by the optimal
log-prob rate $(q/p)H(p)$.  It is therefore clear that
$n_{\min,\mathrm{EAT}}$ has quadratically worse dependence on $q$ for
small $q$, and always depends on $\epsilon$ with significantly larger
prefactors. In contrast, $n_{\min,\mathrm{QEF}}$'s dependence on
$\epsilon$ can be suppressed by choosing large $\beta$. The effect of
the term $|\log(\kappa)|$ depends on what is considered the minimum
safe probability of success and the protocol.

We find similar, practical advantages of \QEF for the $(2,2,2)$
Bell-test configuration in Sect.~\ref{subsec:examples}, with clear
advantages for all useful probability distributions. 
As in the example above, the advantages can be particularly large 
at  probability distributions with low conditional entropy.

\subsection{Entropy Estimator Optimization Problem}

According to the above results, we can construct \QEFs from entropy
estimators, but the construction does not lead to a simple objective
function for entropy estimators.  However, one can seek optimal
entropy estimates.  Consider candidates for entropy estimators
$K(CZ)$.  The entropy estimate at the anticipated probability distribution $\nu$ is
$\cE(K;\nu)=\sum_{cz}\nu(cz) K(cz)$. The entropy estimator condition
is
\begin{equation}
  \sum_{cz}\tr(\rho(cz))K(cz)\leq
    -\sum_{cz}\tr(\vphantom{\big|}\rho(cz)(\log(\rho(cz))-\log(\rho(z))))
\end{equation}
for all $\rho(CZ)\in\cN(\cC(CZ))$.  When the probability distribution of
$Z$ is fixed at $\mu$, the constraint becomes
\begin{equation}
  \sum_{cz}\mu(z)\tr(\rho(c|z))K(cz)\leq
    -\sum_{cz}\mu(z)\tr(\vphantom{\big|}\rho(c|z)(\log(\rho(c|z))-\log(\rho))).
\end{equation}
The problem is then to
\begin{alignat}{3}
  \textrm{Maximize:\ }&\sum_{cz}\nu(cz) K(cz)\notag\\
  \textrm{Variables:\ }& K(CZ)\notag\\
  \textrm{Subject to:\ }&\sum_{cz}\mu(z)\tr(\rho(c|z))K(cz)\notag\\
  &\;\;\;\leq
  -\sum_{cz}\mu(z)\tr(\vphantom{\big|}\rho(c|z)(\log(\rho(c|z))-\log(\rho))) \textrm{\
    for all $\rho(CZ)\in\cN(\cC)$} .
\end{alignat}

\subsection{Optimality of \QEFs}
\label{subsec:qefoptimality}

In this section, we show that given the model $\cC(CZ)$ and a
probability distribution $\mu(CZ)\in\tr(\cN(\cC(CZ)))$ consistent with
the model, entropy estimators witness the maximum possible entropy
rates at $\mu(CZ)$.  By Thm.~\ref{thm:ee_to_qef} with $\beta\searrow
0$, these entropy rates are asymptotically achieved by log-prob rates
of \QEFs.  Thus QEFs are asymptotically optimal.
Optimality for min-tradeoff functions is mentioned in
Ref.~\cite{arnon-friedman:qc2016a}. For classical side information a
proof is in Ref.~\cite{knill:qc2017a}. Here we generalize this proof
for quantum side information. 

The optimality statement concerns the
experimentally desirable situation where the observed statistics are
i.i.d. for each trial.  \QEFs and entropy
estimators are designed for $\mu(CZ)$ but must be valid regardless of
how the observed statistics arise in the model.  The proof of
optimality requires relating information theoretic upper bounds on
achievable rates to lower bounds achieved by entropy estimators.  Both
are for i.i.d. states of the form $\rho(\Sfnt{CZ}) =
\bigotimes_{i=1}^{n}\rho(C_{i}Z_{i})$ in the chained model determined
by the fixed trial model $\cC(CZ)$, where the $\rho(C_{i}Z_{i})$ are
obtained from a fixed $\rho(CZ)\in\cC(CZ)$ by substitution of CVs.  We
abbreviate the expression for such states $\rho(\Sfnt{CZ})$ as
$\rho(CZ)^{\otimes n}$.  Because \QEFs remain valid under $\CPTP$ maps and
convex closure, we may assume that $\cC(CZ)$ is closed in both
respects. This ensures that the states $\rho(CZ)$ and their tensor
products are rich enough to witness the upper bounds without appealing
to ``mixed strategies'' for $\Pfnt{E}$.

\begin{theorem}\label{thm:optimality}
  Let $\cC(CZ)$ be a CPTP-closed and convex closed model and $\mu(CZ)$
  a distribution in the relative interior of $\tr(\cN(\cC(CZ)))$.
  Define
  \begin{align}
    g_{\mathrm{QEF}}(\mu(CZ)) &=
    \sup\Big\{\Exp_{\mu(CZ)}(K(CZ)):\notag\\
    &\hphantom{=\sup\Big\{}
    \textrm{$K(CZ)$ is an entropy estimator for $C|Z$ and $\cC(CZ)$}\Big\},
  \end{align}
  and
  \begin{align}
    s_{\infty}(\mu(CZ)) &= 
    \inf\Big\{\lim_{n\rightarrow \infty}\frac{1}{n}H^{\epsilon}_{\infty}(\Sfnt{C}|\Sfnt{Z}\Pfnt{E};
    \rho(CZ)^{\otimes n}):\notag\\
    &\hphantom{=\inf\Big\{}\rho(CZ)\in\cN(\cC(CZ)),\epsilon >0,\tr(\rho(CZ))=\mu(CZ)\Big\}.
  \end{align}
  Then $g_{\mathrm{QEF}}(\mu(CZ))=s_{\infty}(\mu(CZ))$.
\end{theorem}

The definition of $s_{\infty}(\mu(CZ))$ assumes constant error
bound \(\epsilon\) in taking the limit with respect to \(n\).
But it follows from Ref.~\cite{tomamichel:qc2009a} that the asymptotic dependence on \(\epsilon\) 
is such that error bounds decreasing sub-exponentially in $n$ can be
used.  To make the connection to Ref.~\cite{tomamichel:qc2009a},
  because \(H^{\epsilon}_{\infty}\) is monotonically decreasing in \(\epsilon\),
  we can replace the infimum in the definition of \(s_{\infty}(\mu(CZ))\) with an infimum over states
  of a limit as follows:
  \begin{align}
    s_{\infty}(\mu(CZ)) &= 
    \inf\Big\{\lim_{\epsilon\searrow 0}\lim_{n\rightarrow \infty}\frac{1}{n}H^{\epsilon}_{\infty}(\Sfnt{C}|\Sfnt{Z}\Pfnt{E};
    \rho(CZ)^{\otimes n}):\notag\\
    &\hphantom{=\inf\Big\{}\rho(CZ)\in\cN(\cC(CZ)),\tr(\rho(CZ))=\mu(CZ)\Big\}.
  \end{align} 
We use the condition that $\mu(CZ)$ is in the relative interior
to avoid issues that can arise at the boundary in the absence of
compactness of $\cC(CZ)$. Since every $\mu(CZ)\in\tr(\cN(\cC(CZ)))$ is
arbitrarily close to distributions in the relative interior, the
restriction does not have practical significance.

\begin{proof}
  According to the quantum asymptotic equipartition property, Thm.~1 of
  Ref.~\cite{tomamichel:qc2009a}, 
  \begin{equation}
    s_{\infty}(\mu(CZ)) =
    \inf\{H_{1}(\rho(CZ)|Z\Pfnt{E}):\rho(CZ)\in\cN(\cC(CZ)),\tr(\rho(CZ))=\mu(CZ)\}.
  \end{equation}
  Therefore, by the definition of entropy estimators, $g_{\mathrm{QEF}}\leq s_{\infty}$.

  We claim that $s_{\infty}(\nu(CZ))$ is a convex function of
  $\nu(CZ)\in\tr(\cN(\cC(CZ)))$.  Suppose that
  $\nu(CZ)=\lambda\nu_{1}(CZ)+(1-\lambda)\nu_{2}(CZ)$ with
  $\nu_{i}(CZ)\in\cN(\cC(CZ))$ and $\lambda\in[0,1]$.  Let
  $\rho_{i}(CZ)\in\cN(\cC(CZ))$ satisfy
  $\tr(\rho_{i}(CZ))=\nu_{i}(CZ)$ and
  $H_{1}(\rho_{i}(CZ)|Z\Pfnt{E})\leq s_{\infty}(\nu_{i}(CZ))+\delta$ with
  $\delta>0$ arbitrarily small.  Define
  $\rho(CZ)=\lambda\rho_{1}(CZ)\oplus (1-\lambda)\rho_{2}(CZ)$.
  Because of the closure properties of $\cC(CZ)$, we have
  $\rho(CZ)\in\cN(\cC(CZ))$. By additivity of $H_{1}$ over direct
  sums, $H_{1}(\rho(CZ)|Z\Pfnt{E})=\lambda
  H_{1}(\rho_{1}(CZ)|Z\Pfnt{E})+ (1-\lambda)
  H_{1}(\rho_{2}(CZ)|Z\Pfnt{E})$. It follows that
  $s_{\infty}(\nu(CZ))\leq \lambda s_{\infty}(\nu_{1}(CZ))
  +(1-\lambda) s_{\infty}(\nu_{2}(CZ)) + \delta$. Letting
  $\delta\searrow 0$ proves the claim.

  For concepts and properties used next, see~Ref.~\cite{boyd:qc2004a},
  particularly Ch.~3 on the convex conjugate of convex functions.  For
  $\mu(CZ)$ in the relative interior of $\tr(\cN(\cC(CZ)))$, for every
  $\delta>0$ there exists an ``affine underestimator''
  $\nu'(CZ)\mapsto \sum_{cz}K(cz) \nu'(cz)$ of $s_{\infty}$
  satisfying $\sum_{cz}K(cz) \nu'(cz) \leq s_{\infty}(\nu'(CZ))$
  for all $\nu'(CZ)\in\tr(\cN(\cC(CZ)))$ and $\sum_{cz}K(cz)\mu(cz)
  \geq s_{\infty}(\mu(CZ))-\delta$. This observation follows from
  Exercise~3.28 of Ref.~\cite{boyd:qc2004a}. Since
  $\sum_{cz}K(cz)\mu(CZ) = \Exp_{\mu(CZ)}(K(CZ)) \leq
  g_{\mathrm{QEF}}(\mu(CZ))$ and we can let $\delta\searrow 0$, this
  completes the proof of the theorem.
\end{proof}

\section{\QEFs and Max-Prob Estimators}
\label{sec:qef_constructs}

\subsection{Max-Prob Estimators}

So far we have shown how to determine \QEFs from entropy estimators.
There are presently few explicitly computable entropy estimators.
Examples can be obtained from the affine min-tradeoff functions
given in Refs.~\cite{arnon-friedman:qc2018a,kessler:qc2017a}.  In
this section, we assume that the inputs $Z$ are coming from a separate
and well-characterized source.

\begin{definition}
  $B(CZ)$ is a \emph{max-prob estimator for $C|Z$ and $\cC(CZ)$} if
  for all $\nu(CZ)\in\tr(\cC(CZ))$,
  \begin{equation}
    \Exp_{\nu(CZ)}(B(CZ)) \geq \max_{cz}\nu(c|z)
    \textrm{\ for all
      $\nu(CZ)\in\tr(\cC(CZ))$}.\label{eq:thm:ee_uniformbnd:maxprobbnd}
  \end{equation}
\end{definition}

Like entropy estimators, max-prob estimators for a model are max-prob
estimators for any submodel. Because the definition of max-prob
estimators depends only on $\tr(\cC(CZ))$, they are also max-prob
estimators of maximal extensions obtained from $\tr(\cC(CZ))$ provided
that the input distribution is fixed.

\begin{lemma}\label{lem:mpe_extension}
  Let $\mu(Z)$ be a probability distribution of $Z$ and $\cC(C|Z)$ a
  model for $(C|Z)\Pfnt{E}$.  If $B(CZ)$ is a max-prob estimator for
  $C|Z$ and $\mu(Z)\ltimes\cC(C|Z)$, then $B(CZ)$ is a max-prob estimator for
  $C|Z$ and $\cM(\Cvx(\mu(Z)\ltimes\tr(\cC(C|Z)));\Pfnt{E})$, the maximal extension of
  the model $\Cvx(\mu(Z)\ltimes\tr(\cC(C|Z)))$ for $CZ\Pfnt{E}$.
\end{lemma}
  
  We remark that $\cM(\Cvx(\mu(Z)\ltimes\tr(\cC(C|Z)));\Pfnt{E})=\mu(Z)\ltimes\cM(\Cvx(\tr(\cC(C|Z)));\Pfnt{E})$
 because the input distribution $\mu(Z)$ is fixed and $\Cvx(\mu(Z)\ltimes\tr(\cC(C|Z)))=\mu(Z)\ltimes\Cvx(\tr(\cC(C|Z)))$. 

\begin{proof}
  Let $\nu(CZ)$ be a probability distribution in
  $\tr(\cM(\Cvx(\mu(Z)\ltimes\tr(\cC(C|Z)));\Pfnt{E}))$.  By
  definition of maximal extensions,
  $\nu(CZ)\in\Cvx(\mu(Z)\ltimes\tr(\cC(C|Z)))$.  We have
  $\Cvx(\mu(Z)\ltimes\tr(\cC(C|Z)))=\mu(Z)\ltimes\Cvx(\tr(\cC(C|Z)))=
  \mu(Z)\ltimes\tr(\Cvx(\cC(C|Z)))$. We can therefore express
  $\nu(CZ)$ as a convex combination
  $\nu(CZ)=\sum_{i}\lambda_{i}\mu(Z)\tr(\tau_{i}(C|Z))$ with
  $\tau_{i}(C|Z)\in\cC(C|Z)$. Write
  $\nu_{i}(CZ)=\mu(Z)\tr(\tau_{i}(C|Z))\in
  \tr(\mu(Z)\ltimes\cC(C|Z))$.  Then
  $\nu(C|Z)=\sum_{i}\lambda_{i}\nu_{i}(C|Z)$.  Since $B(CZ)$ is a
  max-prob estimator and maxima are subadditive,
  \begin{align}
    \Exp_{\nu(CZ)}(B(CZ)) &= \sum_{i}\lambda_{i}(\Exp_{\nu_{i}(CZ)}B(CZ))\notag\\
       &\geq \sum_{i}\lambda_{i}\max_{cz}\nu_{i}(c|z)\notag\\
       &=\sum_{i}\max_{cz}(\lambda_{i}\nu_{i}(c|z))\notag\\
       &\geq \max_{cz}\left(\sum_{i}\lambda_{i}\nu_{i}(c|z)\right)\notag\\
       &=\max_{cz}(\nu(c|z)),
  \end{align}
  as required for the lemma.
\end{proof}

We remark that when the input distribution $\mu(Z)$ is fixed and the conditional distributions 
$\tr(\cC(C|Z))$ according to the model $\cC(C|Z)$ for $(C|Z)\Pfnt{E}$ are characterized by 
semidefinite constraints, max-prob estimators can be constructed by semidefinite programming.
See Sect.~VI.~A of Ref.~\cite{knill:qc2017a} for details. 

\subsection{Entropy Estimators From One-Trial Max-Prob}

\begin{theorem}\label{thm:ee_uniformbnd}
  Let $B(CZ)$ be a max-prob estimator for $CZ$ and $\cC(CZ)$.  For
  $0<\bar b$,  define
  \begin{equation}
    K(CZ)=-\log(\bar b)+1-B(CZ)/\bar b.
    \label{eq:approxbeta}
  \end{equation}
  Then $K(CZ)$ is an entropy estimator for $CZ$ and
  $\cC'(CZ)=\cM(\Cvx(\tr(\cC(CZ)));\Pfnt{E})$.
\end{theorem}

For this theorem, $Z$ is considered as part of the output, not input. 
We can construct \QEFPs for $CZ$ and $\cC'(CZ)$ from the entropy 
estimators obtained above according to Thm.~\ref{thm:ee_to_qef} with 
$C$ there replaced by $CZ$ here and $Z$ there set to be trivial. If it
is necessary to condition on $Z$ later, this can be done according to
Protocol~\ref{prot:condimplicit}.  Thm.~\ref{thm:ee_conduniformbnd} below
considers the case where $Z$ is input to be conditioned on explicitly
with a known probability distribution. An advantage of
Thm.~\ref{thm:ee_uniformbnd} is that it is applicable even when the
distribution of $Z$ is not predetermined at each trial.  As discussed
in Ref.~\cite{knill:qc2017a}, it is a good idea to choose $\bar
b=\Exp_{\nu(CZ)}(B(CZ))$ where $\nu(CZ)$ is the anticipated trial probability
distribution of $CZ$.

\begin{proof}
  According to Lem.~\ref{lem:maxextend_is_CPclosed}, $\cC'(CZ)$ is
  $\pCP$-closed. Lem.~\ref{lem:mpe_extension} with $C$ there replaced
  by $CZ$ here and $Z$ there set to the trivial CV here implies that
  $B(CZ)$ is a maximum probability estimator for $CZ$ and $\cC'(CZ)$.
  Let $\tau(CZ)\in\cN(\cC'(CZ))$ and $\nu(CZ)=\tr(\tau(CZ))$.  Below,
  we show that $\Exp_{\nu(CZ)}(B(CZ))\geq
  P_{\max}(\tau(CZ)|\Pfnt{E})$. Given this inequality and since $-\log
  P_{\max}(\tau(CZ)|\Pfnt{E})\leq H_{1}(\tau(CZ)|\Pfnt{E})$
  (Lem.~\ref{lem:minent_ent} with the above replacements),
  for the theorem it suffices to show that $\Exp_{\nu(CZ)}(K(CZ))\leq
  -\log(\Exp_{\nu(CZ)}(B(CZ)))$.  From
  \begin{equation}
    -\log(x)=-\log(\bar b)-\log(x/\bar b)\geq -\log(\bar b) - (x/\bar b-1)
    = -\log(\bar b) + 1-x/\bar b,
  \end{equation}
  we get
  \begin{align}
    -\log(\Exp_{\nu(CZ)}(B(CZ))) 
    &\geq -\log(\bar b)+1-\Exp_{\nu(CZ)}(B(CZ))/\bar b\notag\\
    &= \Exp_{\nu(CZ)}(-\log(\bar b)+1-B(CZ)/\bar b)\notag\\
    &= \Exp_{\nu(CZ)}(K(CZ)).
  \end{align}

  To prove that $\Exp_{\nu(CZ)}(B(CZ))\geq
  P_{\max}(\tau(CZ)|\Pfnt{E})$, we apply the relationship between
  maximum guessing probability and $P_{\max}$,
  Lem.~\ref{lem:maxprob_guess}.  Let $(\Pi_{cz})_{cz}$ be a POVM with
  guessing probability $p= \sum_{cz}\tr(\Pi_{cz}\tau(cz))$.  Let
  $\nu_{c'z'}(CZ)=\tr(\Pi_{c'z'}\tau(CZ))/\tr(\Pi_{c'z'}\tau)$.  By
  $\pCP$-closure, $\nu_{c'z'}(CZ)\in\tr(\cC'(CZ))$.  For each $c'z'$
  we have $\Exp_{\nu_{c'z'}(CZ)}(B(CZ))\geq
  \max_{cz}(\nu_{c'z'}(cz))\geq \nu_{c'z'}(c'z')$.  Consequently
  \begin{align}
    p&=\sum_{c'z'}\nu_{c'z'}(c'z')\tr(\Pi_{c'z'}\tau)\notag\\
     &\leq\sum_{c'z'}\Exp_{\nu_{c'z'}(CZ)}(B(CZ))\tr(\Pi_{c'z'}\tau)\notag\\
     &=\sum_{c'z'}\sum_{cz}B(cz)\nu_{c'z'}(cz)\tr(\Pi_{c'z'}\tau)\notag\\
     &=\sum_{c'z'}\sum_{cz}B(cz)\tr(\Pi_{c'z'}\tau(cz))\notag\\
     &=\sum_{cz}B(cz)\tr(\sum_{c'z'}\Pi_{c'z'}\tau(cz))\notag\\
     &=\sum_{cz}B(cz)\tr(\tau(cz))\notag\\
     &=\Exp_{\nu(CZ)}(B(CZ)).
  \end{align}
  The claim follows because the POVM can be chosen so that $p$ is
  arbitrarily close to $P_{\max}(\tau(CZ)|\Pfnt{E})$.
\end{proof}

\begin{theorem}\label{thm:ee_conduniformbnd}
  Let $\cC(C|Z)$ be a \(\pCP\)-closed model for $(C|Z)\Pfnt{E}$, $\mu$ a
  probability distribution of $Z$ and $B(CZ)$ a max-prob estimator for
  $C|Z$ and $\mu(Z)\ltimes\cC(C|Z)$.  
  For $0<\bar b$, 
    define
  \begin{equation}
    K(CZ)=-\log(\bar b)+1-B(CZ)/\bar b.
    \label{eq:approxbetacond}
  \end{equation}
  Then $K(CZ)$ is an entropy estimator for $C|Z$ and 
  $\mu(Z)\ltimes\cC(C|Z)$.
\end{theorem}

\begin{proof}
  Let $\tau(C|Z)\in\cC(C|Z)$ with $\tr(\tau(|z))=1$. Define
  $\nu(CZ)=\tr(\mu(Z)\tau(C|Z))$ and $\tau(CZ)=\mu(Z)\tau(C|Z)$.
  We have $\tau=\tau(|z)$, independent of $z$.   Below, we show that $\Exp_{\nu(CZ)}(B(CZ))\geq
  P_{\max}(\tau(CZ)|Z\Pfnt{E})$.  Given this inequality and since $-\log
  P_{\max}(\tau(CZ)|Z\Pfnt{E})\leq H_{1}(\tau(CZ)|Z\Pfnt{E})$
  (Lem.~\ref{lem:minent_ent}), for the theorem it suffices to show that
  $\Exp_{\nu(CZ)}(K(CZ))\leq -\log(\Exp_{\nu(CZ)}(B(CZ)))$, which follows from
  the same calculation as that given in the proof of
  Thm.~\ref{thm:ee_uniformbnd}.

  To prove that $\Exp_{\nu(CZ)}(B(CZ))\geq
  P_{\max}(\mu(Z)\tau(CZ)|Z\Pfnt{E})$, we again apply the relationship
  between maximum guessing probability and $P_{\max}$,
  Lem.~\ref{lem:maxprob_guess}.  For each $z$, let $(\Pi_{c|z})_{c}$ be
  a POVM with guessing probability
  $p_{z}=\sum_{c}\tr(\Pi_{c|z}\tau(c|z))$ and overall guessing
  probability $p=\sum_{z}\mu(z)p_{z}$.  Let
  $\nu_{c'|z'}(C|Z)=\tr(\Pi_{c'|z'}\tau(C|Z))/\tr(\Pi_{c'|z'}\tau)$.
  By $\pCP$-closure, $\nu_{c'|z'}\in\tr(\cC(C|Z))$.  For each $c'z'$,
  \begin{equation}
    \Exp_{\mu(Z)\nu_{c'|z'}(C|Z)}(B(CZ))\geq \max_{cz}(\nu_{c'|z'}(c|z))\geq
    \nu_{c'|z'}(c'|z').
  \end{equation}
  Consequently
  \begin{align}
    p=\sum_{z'}\mu(z')p_{z'}&
    =\sum_{c'z'}\mu(z')\nu_{c'|z'}(c'|z')\tr(\Pi_{c'|z'}\tau)\notag\\
    &\leq\sum_{c'z'}\mu(z')
    \Exp_{\mu(Z)\nu_{c'|z'}(C|Z)}(B(CZ))\tr(\Pi_{c'|z'}\tau)\notag\\
    &=\sum_{c'z'}\mu(z')
    \sum_{cz}B(cz)\mu(z)\nu_{c'|z'}(c|z)\tr(\Pi_{c'|z'}\tau)\notag\\
    &=\sum_{c'z'}\mu(z')
    \sum_{cz}B(cz)\mu(z)\tr(\Pi_{c'|z'}\tau(c|z))\notag\\
    &=\sum_{cz}B(cz)\mu(z)\tr(\sum_{z'}\mu(z')\sum_{c'}\Pi_{c'|z'}\tau(c|z))\notag\\
    &=\sum_{cz}B(cz)\mu(z)\tr(\sum_{z'}\mu(z')\tau(c|z))\notag\\
    &=\sum_{cz}B(cz)\mu(z)\tr(\tau(c|z))\notag\\
    &=\Exp_{\nu(CZ)}(B(cz)).
  \end{align}
  The claim follows because the POVMs can be chosen so
  that $p$ is arbitrarily close to
  $P_{\max}(\tau(CZ)|Z\Pfnt{E})$.
\end{proof}

\subsection{Exponential Expansion by Spot-Checking}
\label{sec:maxprobest:exex}

Let $\cC(C|Z)$ be a model for $(C|Z)\Pfnt{E}$ and $B(CZ)$ a max-prob
estimator for $C|Z$ and $\Unif(Z)\ltimes\cC(C|Z)$.  For
this section, we fix $\rho(C|Z)\in\cC(C|Z)$ with $\tr(\rho(|z))=1$ and
$\nu(CZ)=\Unif(Z)\tr(\rho(C|Z))$. Define
$\bar b=\Exp_{\nu(CZ)}(B(CZ))$.  We assume that $\bar b<1$ as
otherwise $\Unif(Z)\rho(C|Z)$ has only trivial max-prob witnessed
by $B(CZ)$. By definition, $\bar b> 0$. We also assume that the
  model $\cC(C|Z)$ is closed under $\pCP$ maps.  Both induced models
  and maximal extensions considered in this work are $\CP$-closed and
  $\pCP$-closed, and can therefore be used here.

The following repeats the treatment of spot-checking input
distributions in Ref.~\cite{knill:qc2017a}.  To simplify the analysis,
we take advantage of the fact that for configurations such as those of
Bell tests, we can hide the choice of whether or not to apply a test
trial from the devices. This corresponds to appending a test bit $T$
to $Z$, where $T=1$ indicates a test trial and $T=0$ indicates a fixed
one, with $Z=z_{0}$.  The model $\cC(C|ZT)$ is obtained from
$\cC(C|Z)$ by constraining $\sigma(C|Z0)=\sigma(C|Z1)$ and
$\sigma(C|Z0)\in\cC(C|Z)$. For any $\sigma(C|Z)\in\cC(C|Z)$ there is a
corresponding $\tilde\sigma(C|ZT)\in\cC(C|ZT)$ defined by
$\tilde\sigma(C|Zt)=\sigma(C|Z)$ for $t\in\{0,1\}$.  The map $\sigma(C|Z)\mapsto
\tilde\sigma(C|ZT)$ is a bijection between $\cC(C|Z)$ and $\cC(C|ZT)$.

Let $q=1/|\Rng(Z)|$.  Let $\mu_{r}$ be the probability distribution of
$ZT$ defined by $\mu_{r}(z1) = rq$ and
$\mu_{r}(z0)=(1-r)\knuth{z=z_{0}}$ for some value $z_{0}$ of
$Z$. Since we are interested in the case where $r$ is small, we assume $0<r<
1/2$.  The entropy of the distribution $\mu_{r}$ is given by
$S(\mu_{r})=H(r)+r\log(1/q)$, where $H(r)=-r\log(r)-(1-r)\log(1-r)$.
Let $\nu_{r}(CZT)=\mu_{r}(ZT)\tr(\rho(C|Z))$.  Define $B_{r}(CZT)$ by
\begin{align}
  B_{r}(CZ0) &= 1,\notag\\
  B_{r}(CZ1) &= 1+(B(CZ)-1)\frac{1}{r}.
  \label{eq:def_B_r}
\end{align}
Setting $B_{r}(CZT)$ to $1$ when $T=0$ is convenient, we did not
explore optimality of this choice. 

\begin{lemma}\label{lem:br_maxprobest}
  $B_{r}(CZT)$ is a max-prob estimator for $C|ZT$ and
  $\mu_{r}(ZT)\ltimes\cC(C|ZT)$.
\end{lemma}

\begin{proof}
  Let $\sigma(C|Z)\in\cC(C|Z)$ with $\tr(\sigma(|z))=1$.  For the
  duration of this proof, define
  $\nu(CZ)=\tr(\Unif(Z)\sigma(C|Z))=q\tr(\sigma(C|Z))$ and
  $\nu_{r}(CZT)=\tr(\mu_{r}(ZT)\sigma(C|Z))$.  All normalized members
  of $\mu_{r}(ZT)\ltimes\cC(C|ZT)$ are of the form
  $\mu_{r}(ZT)\sigma(C|Z)$, so it suffices to confirm that
  $\Exp_{\nu_{r}(CZT)}B_{r}(CZT)\geq \max_{czt}\nu_{r}(c|zt)$.  We
  have
  \begin{align}
    \Exp_{\nu_{r}(CZT)}(B_{r}(CZT)) &=
      \sum_{czt}\nu_{r}(czt) B_{r}(czt)\notag\\
      &=\sum_{cz}\nu_{r}(cz0)B_{r}(cz0)
        +\sum_{cz}\nu_{r}(cz1)B_{r}(cz1)\notag\\
      &=\sum_{cz}\mu_{r}(z0)\tr(\sigma(c|z)) +
           \sum_{cz}\mu_{r}(z1)\tr(\sigma(c|z))(1+(B(cz)-1)/r)\notag\\
      &=\sum_{z}\mu_{r}(z0) +
           \sum_{cz}rq\tr(\sigma(c|z))(1+(B(cz)-1)/r)\notag\\
      &=(1-r) +
           \sum_{cz}r\Unif(z)\tr(\sigma(c|z))(1+(B(cz)-1)/r)\notag\\
      &=(1-r) +
           r\sum_{cz}\nu(cz)(1+(B(cz)-1)/r)\notag\\
      &=(1-r) + r  -1 + \Exp_{\nu(CZ)}(B(CZ))\notag\\
      &= \Exp_{\nu(CZ)}(B(CZ))\label{eq:lem:br_maxprobest:1}\\
      &\geq \max_{cz}\nu(c|z) = \max_{czt}\nu_{r}(c|zt),\notag 
  \end{align}
  since $\nu_{r}(c|z1)=\nu(c|z)$ and $\nu_{r}(c|z0)=\knuth{z=z_{0}}\nu(c|z)$,
  according to our convention that zero-probability conditionals are $0$.
\end{proof}

Let
\begin{equation}
  K_{r}(CZT) = -\log(\bar b)+1-B_{r}(CZT)/\bar b,
  \label{eq:ee_r_beta}
\end{equation}
where $\bar b$ is introduced in the first paragraph of this section and $B_{r}(CZT)$
is defined in Eq.~\eqref{eq:def_B_r}. In view of Thm.~\ref{thm:ee_conduniformbnd} 
and Lem.~\ref{lem:br_maxprobest}, $K_{r}(CZT)$ is an entropy estimator for $C|ZT$ 
and  $\mu_{r}(ZT)\ltimes\cC(C|ZT)$ provided that the model $\cC(C|ZT)$ is
$\pCP$-closed, which is assumed in this section.

\begin{lemma}\label{lem:kr_rate}
  For $\nu_{r}(CZT)=\tr(\mu_{r}(ZT)\rho(C|Z))$,
  $\Exp_{\nu_r(CZT)}(K_r(CZT)) = -\log(\bar b)$.
\end{lemma}

\begin{proof}
  \begin{align}
    \Exp_{\nu_r(CZT)}(K_r(CZT)) &= -\log(\bar b)+1
    -\Exp_{\nu_{r}(CZT)}(B_{r}(CZT))/\bar b\notag\\
    = -\log(\bar b),
  \end{align}
  where the last equality follows from the more general
  Eq.~\ref{eq:lem:br_maxprobest:1} and the definition of $\bar b$.
\end{proof}

\begin{theorem}\label{thm:spotchecking_rate}
  With the notation of this section, there exist constants $d$ and
  $d'$ independent of $r$ such that for $0<\beta\leq dr$,
  \begin{equation}
    F_{r,\beta}(CZT) = \frac{e^{\beta K_{r}(CZT)}}{1+d'\beta^{2}/r}
    \label{eq:thm:spotchecking_rate:main}
  \end{equation}
  is a \QEFP with power $\beta$ for $C|ZT$ and $\mu_{r}(ZT)\ltimes\cC(C|ZT)$.
  The log-prob rate $g_{r,\beta}$ of $F_{r,\beta}$ at $\mu_{r}(ZT)\rho(C|Z)$ satisfies
  \begin{equation}
    g_{r,\beta}
      \geq -\log(\bar b)-d'\beta/r.
  \end{equation}
\end{theorem}

This theorem extends Thm.~50 from Ref.~\cite{knill:qc2017a} to \QEFPs
constructed from max-prob estimators.  We do not intend the constants
obtained in the proof to be used in practice. If necessary in an application, the \QEFPs and the
values of $\beta$ obtained according to the strategy here can be
optimized with numerical methods with the expressions obtained in
Thm.~\ref{thm:ee_to_qef} and its proof.

\begin{proof}
  The bound on the log-prob rate follows from Eq.~\ref{eq:thm:spotchecking_rate:main}
  by direct computation: With $\nu_{r}(CZT)=\tr(\mu_{r}(ZT)\rho(C|Z))$,
  \begin{align}
   g_{r,\beta}
    &=\Exp_{\nu_{r}(CZT)}(\log(F_{r,\beta}(CZT))/\beta)\notag\\
    &=\Exp_{\nu_{r}(CZT)}(K_{r}(CZT)) -\log(1+d'\beta^{2}/r)/\beta\notag\\
    &= -\log(\bar b)-\log(1+d'\beta^{2}/r)/\beta\notag\\
    &\geq-\log(\bar b)-d'\beta/r,
  \end{align}
  where we applied Lem.~\ref{lem:kr_rate} in the second-last step.

  For the main statement of the theorem, we apply
  Thm.~\ref{thm:ee_to_qef}, where $Z$ there becomes $ZT$ here and
    $K(CZ)$ there becomes $K_{r}(CZT)$ here.  Consider the upper
  bound $c(\beta)$ on $c_{P}(\beta,K_{r}(CZT))$ from Thm.~\ref{thm:ee_to_qef}.  
  Let $k_{\infty}(zt)=\max_{c}|K_{r}(czt)|$.   By the
    definition of $K_{r}(CZT)$ in Eq.~\eqref{eq:ee_r_beta} and in view
    of the assumption that $0<\bar b<1$, we have
    $k_{\infty}(zt)=\max_{c}|K_{r}(czt)|\leq -\log(\bar b)+1+
    \max_{c}|B_{r}(czt)|/\bar{b}$.  The expression for $B_{r}(czt)$
    with $r\in (0,1)$ implies that $\max_{c}|B_{r}(cz0)|=1$ and
    $\max_{c}|B_{r}(cz1)|=\max_{c}|B(cz)/r+(r-1)/r|\leq
    (\max_{c}|B(cz)|+1)/r$.  Therefore,
    $k_{\infty}(z0) \leq -\log(\bar b)+1+1/\bar b$ and
    $k_{\infty}(z1) \leq -\log(\bar b) + 1+(\max_{c} |B(cz)|+1)/\bar b
    r \leq (-\log(\bar b) + 1+(\max_{c} |B(cz)|+1)/\bar b )/r$.
  Define
  $d=\big(2\max_{cz}(-\log(\bar b)+1+(|B(cz)|+1)/\bar b)\big)^{-1}$ so
  that $k_{\infty}(z1)\leq 1/(2dr)$ and $k_{\infty}(z0)\leq
  1/(2d)$. Note that $d$ is independent of $r$ as stated in the
    theorem and $d\in (0,1/2)$.  In order to
    simplify the upper bound \(c(\beta)\) on $c_{P}(\beta,K_{r}(CZT))$ in
    Thm.~\ref{thm:ee_to_qef}, we increase the bound by
    replacing the quantities \(\bar w_{\gamma}(zt)\) for \(\gamma=0\) and \(\gamma=\beta\) by the larger quantity $v(zt)=k_{\infty}(zt)+\log(2N)+\iota_{0}$, and similarly, the quantity  $e^{k_{\max}(zt)\beta}$ by
    $e^{k_{\infty}(zt)\beta }$.
  With these replacements, the $\llceil\ldots\rrceil$ operation can be
  omitted and the terms combined for
  \begin{equation}
    c_{P}(\beta,K_{r}(CZT))\leq 
    \frac{\beta^{2}}{6}\sum_{zt}\mu_{r}(zt)
    \left(
      \left(2 + \frac{e^{ 
            k_{\infty}(zt)\beta}}{(1-\beta)^{2}}\right)v(zt)(v(zt)+\iota_{0})
    \right),
  \end{equation}
  noting that $\nu(Z)$ in Thm.~\ref{thm:ee_to_qef} becomes
  $\mu_{r}(ZT)$ here and $2\coth(v(zt))\leq 2\coth(\iota_{0})= \iota_{0}$
 in view of the monotonicity of the function $\coth(x)$ and  the definition of $\iota_{0}$.  
 The above bound on $c_{P}(\beta,K_{r}(CZT))$ is valid when $\beta<1/2$ according to Thm.~\ref{thm:ee_to_qef}.
 Since $0<d<1/2$ and $0<r<1$ (we actually assume that $0<r<1/2$ in this section), $0<dr< 1/2$. 
 For $\beta\leq dr$, we have $e^{k_{\infty}(zt)\beta }\leq
  e^{k_{\infty}(zt) dr}\leq e^{1/2}\leq 2$, and $(1-\beta)\geq 1/2$,
  so we can weaken the bound to
  \begin{align}
  c_{P}(\beta,K_{r}(CZT))&\leq 
    \frac{\beta^{2}}{6}\sum_{zt}\mu_{r}(zt)\,
      10\,v(zt)(v(zt)+\iota_{0})\notag \\
      &\leq \frac{5\beta^{2}}{3}\sum_{zt}\mu_{r}(zt)(v(zt)+\iota_{0})^2.
  \end{align}
  We have that $v(z0)+\iota_{0}\leq 1/(2d)+\log(2N)+2\iota_{0}$ and
  $v(z1)+\iota_{0}\leq (1/(2dr))+\log(2N)+2\iota_{0}\leq
  (1/(2d)+\log(2N)+2\iota_{0})/r$.  
  After separating the sum over $zt$ for
  $t=0$ and $t=1$, the bound weakens further to
  \begin{align}
    c_{P}(\beta,K_{r}(CZT))&\leq 
    \frac{5\beta^{2}}{3} \left((1-r)(1/(2d)+\log(2N)+2\iota_{0})^{2}
      + r(1/(2d)+\log(2N)+2\iota_{0})^{2}/r^{2}
      \right)\notag\\
      &\leq
      \frac{\beta^{2}}{r}\frac{5\times
        2}{3}(1/(2d)+\log(2N)+2\iota_{0})^{2}.
  \end{align}
  It now suffices to set $d' =
  10\,(1/(2d)+\log(2N)+2\iota_{0})^{2}/3$, which is independent of $r$ as stated in the theorem.
\end{proof}

Thm.~\ref{thm:spotchecking_rate} implies exponential expansion via the
argument used to prove exponential expansion in
Ref.~\cite{knill:qc2017a}, Thm 52. We formulate the theorem with
power-law error-bound rates to match the conclusion of Cor.~1.5 in
Ref.~\cite{miller_c:qc2014a}. Standard exponential expansion is
obtained by setting the parameter $\gamma$ in the next theorem to $\gamma=1$.

\begin{theorem}\label{thm:expexpansion}
  Let $F_{r,\beta}(CZT)$ be the family of \QEFPs of
  Thm.~\ref{thm:spotchecking_rate} for model
  $\mu_{r}(ZT)\ltimes\cC(C|ZT)$. Suppose that $\cC(\Sfnt{CZT})$ is
  obtained by chaining $\mu_{r}(ZT)\ltimes \cC(C|ZT)$ $n$ times, and
  $\rho(\Sfnt{CZT})\in\cC(\Sfnt{CZT})$ satisfies that
  $\tr(\rho(\Sfnt{CZT}))$ is i.i.d. with trial distribution
  $\tr(\rho(CZT))$,  with respect to which $\bar b<1$. Then given $\gamma\in(0,1]$,
  $l_{\epsilon}>0$, and error bound \ $\epsilon(n)$ defined by
  $\log(2/\epsilon(n)^{2}) = l_{\epsilon} n^{1-\gamma}$ for
  $\gamma\in[0,1]$, the expected quantum net log-prob
  $g_{\mathrm{net}}(n)$ at $\rho(\Sfnt{CZT})$ satisfies
  \begin{equation}
    g_{\mathrm{net}}(n) 
    = e^{\Omega(n^{\gamma-1} S_{\mathrm{net}}(n))},
  \end{equation}
  where $S_{\mathrm{net}}(n)$ is the net input entropy for $n$ trials.
\end{theorem}
The expansion in the theorem is from input entropy to output
conditional min-entropy.  To recover uniformly random bits still
requires randomness extraction, and we do not include the seed
requirements in our accounting here.  The theorem includes a
completeness statement via the assumed state $\rho(CZT)$ with respect
to which $\bar b$ is defined and $\bar b<1$.  Of course, it is not
necessary for the distribution to be i.i.d., this just makes sure that
the probability of witnessing exponentially large output smooth
min-entropy is $\Omega(1)$.  It suffices that with sufficiently high
probability, the observed frequencies are typical of such an i.i.d.
distribution. In Ref.~\cite{knill:qc2017a}, we discussed the
distribution of $\log(F_{r,\beta}(CZT))$ for the i.i.d. scenario in the 
presence of classical side information, establishing that the probability 
of success for protocols based on Thm.~52 there 
(a version of the theorem here in the presence of classical side information)  approaches $1$ with
sufficiently conservative choices of thresholds.  We expect the same property
  to extend for the i.i.d. scenario in the presence of quantum side information.

\begin{proof}
  We repeat the proof Thm.~52 in Ref.~\cite{knill:qc2017a} with minor
  modifications to obtain the more general statement of
  Thm.~\ref{thm:expexpansion}.  We determine constants $0<c<1$ and
  $0<c'$ for which the testing rate $r_{n}=c'/n^{\gamma}$ and the
  power $\beta_{n}=cr_{n} = cc'/n^{\gamma}$ achieve the goal of the
  theorem.  It suffices to prove the theorem assuming that $n$ is
  sufficiently large.  Let $d$ and $d'$ be the constants in
  Thm.~\ref{thm:spotchecking_rate}.  We require that $c\leq d$ to 
  ensure the statement that $\beta_{n}=cr_{n}\leq dr_{n}$ according to Thm.~\ref{thm:spotchecking_rate}.  
  From this theorem, Def.~\ref{def:netlogprob}, and since $\beta_{n}<1$ for
  sufficiently large $n$, the expected quantum net log-prob for power
  $\beta_{n}$ is bounded by
  \begin{align}
    g_{\mathrm{net}} &= n g_{r_{n},\beta_{n}}-\frac{\log(2/\epsilon(n)^{2})}{\beta_{n}}\notag\\
    &\geq n\left(-\log(\bar
      b)-\frac{d'\beta_{n}}{r_{n}}\right)-\frac{l_{\epsilon}
      n^{1-\gamma}}{\beta_{n}}.
    \label{eq:thm:expexp_gtotal}
  \end{align}
  The input entropy per trial $S(\nu_{r})=H(r)+r\log(1/q)$ is bounded
  from above by $-2r\log(r)$, provided we take $r\leq q/e$.  For this, note
  that $r\log(1/q)\leq r\log(1/(re))=-r\log(r)-r$ and
  $-(1-r)\log(1-r)\leq r$ since $-\log(1-r)=\log(1+r/(1-r))\leq r/(1-r)$.  
  For $r_{n}=c'/n^{\gamma}$, which is less than $q/e$ for sufficiently large $n$, 
  the expected number of test trials is $c'n^{1-\gamma}$ and the total
  input entropy satisfies $nS(\nu_{r_{n}}) \leq -2nr_n\log(r_n)=2c'
  n^{1-\gamma}(\gamma\log(n)-\log(c'))$, or equivalently $n\geq
  (c')^{1/\gamma}e^{n^{\gamma}S(\nu_{r_{n}})/(2\gamma c')}$.  Write
  $g_{0}=-\log(\bar b)$.  Substituting the expressions for $\beta_{n}$
  and $r_{n}$ in Eq.~\ref{eq:thm:expexp_gtotal},
  \begin{align}
    g_{\mathrm{net}}
    &\geq n g_{0}\left( 1-\frac{d' c}{g_{0}} - \frac{l_{\epsilon} }{cc' g_{0}}\right).
    \label{eq:thm:expexp_lowerbnd}
  \end{align}
  We first set $c=\min(d,g_{0}/(3d'))$, which ensures that
  $d'c/g_{0}\leq 1/3$. We then set $c'=3l_{\epsilon}/(cg_{0})$.  This
  gives the inequality
  \begin{equation}
    g_{\mathrm{net}} \geq n \frac{g_{0}}{3} \geq        
     (c')^{1/\gamma}e^{n^{\gamma}S(\nu_{r_{n}})/(2\gamma c')}\frac{g_{0}}{3},
    \label{eq:thm:expexp_explicit}
  \end{equation}
  which implies the theorem.
\end{proof}

\section{\QEFs for $(k,2,2)$-Bell-Test Configurations}
\label{sec:k22configs}

\subsection{$(k,2,2)$-Bell-Test Configurations}
\label{subsec:k22config}

We consider models induced by POVMs that are physically achievable on
the device side of $(k,2,2)$-Bell-test configurations.  A
$(k,2,2)$-Bell-test configuration involves $k$ stations (or parties or
devices), where each applies one of two binary-outcome measurements in
each trial. The trial CVs $C$ and $Z$ are both $k$-bit strings.  For
this section, we do not need to consider sequences of trials directly,
so $C_{i}$ and $Z_{i}$ refer to the $i$'th bits of these strings.  For
optimizing the log-prob, we assume a fixed input 
distribution given by $\mu(Z)$.  To define the induced model, the
total device Hilbert space is
$\cH(\Pfnt{D})=\bigotimes_{i=1}^{k}\cV^{(i)}$. The set of POVMs $\fkP$
consists of the families of positive semidefinite operators
$(\mu(z)P_{c|z})_{cz}$ with $P_{c|z}$ of the form
$P_{c|z}=\bigotimes_{i=1}^{k}P^{(i)}_{c_{i}|z_{i}}$ where
$P^{(i)}_{0|z_i}+P^{(i)}_{1|z_i} = \one_{\cV^{(i)}}$. Let
$\cC_{k22}(CZ)$ be the union of the $\cM(\fkP;\Pfnt{E})$
over choices for $\cV^{(i)}$.

For qubits (Hilbert space of dimension $2$),
let $\phi\in (-\pi, \pi]$ and  
\begin{align} \label{eq:qubit_povm}
  Q_{c|0;\phi}&=\frac{1}{2}(\one+(-1)^{c}\sigma_{z}), \notag\\
  Q_{c|1;\phi}
  &=\frac{1}{2}(\one+(-1)^{c}(\cos(\phi)\sigma_{z}+\sin(\phi)\sigma_{x})).
\end{align}

For $\theta$ a vector of length $k$, we define $Z$-indexed POVM operators
$P_{C|Z;\theta}$ by
$P_{c|z;\theta}=\bigotimes_{i=1}^{k}Q_{c_{i}|z_{i};\theta_{i}}$ for each 
$c$ and $z$, where $\theta_i\in (-\pi, \pi]$.

\begin{theorem}\label{thm:stdk22}
  Let $\mu(Z)$ be a fixed input distribution.
  The model $\cC_{k22}(CZ)$ consists of positive combinations
  of members $\rho(CZ)$ expressible in the form
  \begin{equation}
    \rho(CZ) = \mu(Z) U \tau^{1/2}P_{C|Z;\theta}\tau^{1/2}U^{\dagger}
  \end{equation}
  for an operator $\tau\geq 0$ and an isometry $U$ from
  $(\cmplx^{2})^{\otimes k}$ into $\cH(\Pfnt{E})$.
\end{theorem}

The theorem follows from a well-known analysis of this situation for
$k=2$ going back to Ref.~\cite{tsirelson:qc1993a} and
Ref.~\cite{masanes:qc2006a}; a nice version of this analysis is in
Ref.~\cite{pironio:qc2009a}, Sect.~2.4.1.

\begin{proof}
  According to the definition of induced models, we consider an
  initial state $\chi$ of $\Pfnt{DE}$ and a POVM
  $\mu(Z)P_{C|Z}\in\fkP$. This gives the generic state
  $\rho(CZ)=\mu(Z)\tr_{\Pfnt{D}}((P_{C|Z}\otimes\one)\chi)$ in
  $\cC_{k22}(CZ)$.  The usual dilation argument shows that we can
  extend $\cV^{(i)}$ and $P^{(i)}_{c_{i}|z_{i}}$ so that
  $P^{(i)}_{0|z_{i}}$ and $P^{(i)}_{1|z_{i}}$ are pairs of orthogonal
  and complete projectors. For $\mu(Z)P_{C|Z}$ replaced by the
  extended POVM, and $\cV^{(i)}$ replaced by its dilation, we still
  have $\rho(CZ)=\mu(Z)\tr_{\Pfnt{D}}((P_{C|Z}\otimes\one)\chi)$.
  With this,
  $A^{(i)}_{z_{i}}=\left(P^{(i)}_{0|z_{i}}-P^{(i)}_{1|z_{i}}\right)$
  are observables with eigenvalues in $\{-1,1\}$. Since there are two such
  observables for each $\cV^{(i)}$, Lem. 2 of
  Ref.~\cite{pironio:qc2009a} now applies so that
  $\cV^{(i)}=\oplus_{j}\cV^{(i)}_{j}$ with $\cV^{(i)}_{j}$ of
  dimension one or two and
  $A^{(i)}_{z_{i}}=\oplus_{j}A^{(i)}_{z_{i},j}$.  On the
  one-dimensional summands, $A^{(i)}_{b,j}=\pm \one$. We can add a
  second dimension on which the state has no support and extend
  $A^{(i)}_{b,j}$ to the added dimension so that
  $A^{(i)}_{b,j}=\pm\sigma_{z}$. We also extend the POVM operators so
  that their relationship to the $\pm 1$ eigenspaces of the
  $A^{(i)}_{b,j}$ is unchanged.  According to the proof of the
  referenced lemma, we may assume that on the two-dimensional
  summands, $A^{(i)}_{z_{i},j}$ act as conjugated Pauli
  matrices.  Thus, after
  extending the one-dimensional summands as described, for all $j$ we
  may choose logical bases such that $A^{(i)}_{0,j}=\sigma_{z}$ and
  $A^{(i)}_{1,j}=
  \cos(\theta_{j})\sigma_{z}+\sin(\theta_{j})\sigma_{x}$, for some
  $\theta_{j}$.

  The reasoning so far shows that $P_{C|Z}=\oplus_{l} P_{C|Z,l}$ with
  $P_{C|Z,l}$ acting on tensor products of two-dimensional subspaces
  of the subsystems $\cV^{(i)}$. The direct sum is over $l$ defined as
  sequences of $k$ indices, with each index labeling a direct summand 
  of the
  corresponding subsystem.  The transition elements of $\chi$ between the
  direct summands tensored with $\cH(\Pfnt{E})$ do not contribute to
  $\rho(CZ)$, so by zeroing these transition elements with the
  appropriate decoherence superoperator we may assume
  $\chi=\oplus_{l}\chi_{l}$.  Now $\rho(CZ)$ is a positive combination
  of the
  $\rho_{l}(CZ)=\mu(Z)\tr_{\Pfnt{D}}((P_{C|Z,l}\otimes\one)\chi_{l})$, which
  is in $\cC_{k22}(CZ)$. With this we have reduced the problem to one
  where the $\cV^{(i)}$ are two-dimensional and
  $P_{C|Z}=P_{C|Z;\theta}$ for some $\theta$ as defined after Eq.~\ref{eq:qubit_povm}.

  Next, we may assume that $\chi$ is pure. If not we purify $\chi$ with the
  addition of another system $\Pfnt{E'}$. Then $\chi=\sum_{m}\chi_{m}$
  where the $\chi_{m}$ are the unnormalized pure states obtained from the
  purification of $\chi$ by projecting onto the $m$'th basis state of
  $\cH(\Pfnt{E'})$ for some choice of orthonormal basis and tracing out the 
  system $\Pfnt{E'}$, and $\rho(CZ)$ is the
  sum of the $\rho_{m}(CZ)=\mu(Z)\tr_{\Pfnt{D}}((P_{C|Z;\theta}\otimes\one)\chi_{m})$ 
  with $\chi_{m}$ pure, which again are in $\cC_{k22}(CZ)$.
  
  Consider $\cU\otimes\cW$ with $\vdim(\cW)\geq\vdim(\cU)=d$ and a
  given orthonormal basis $\{\ket{y}_{\Pfnt{U}}\}_{y\in I}$ of $\cU$.
  Every pure state $\ket{\psi}$ of $\cU\otimes\cW$ can be written in
  the form $(\one\otimes\tau^{1/2})\sum_{y}\ket{y}_{\Pfnt{U}}
  \otimes\ket{y}_{\Pfnt{W}}/\sqrt{d}$, where the
    $\ket{y}_{\Pfnt{W}}$ are orthonormal and we can choose
  $\tau^{1/2}$ to be positive semidefinite and preserve the subspace
  spanned by the $\ket{y}_{\Pfnt{W}}$: One way to determine the
  $\ket{y}_{\Pfnt{W}}$ and $\tau$ is to let
  $\{\ket{y'}_{\Pfnt{U}}\}_{y'\in I'}$ be a Schmidt basis for the pure
  state $\ket{\psi}$ and $\lambda_{y'}\geq 0$ the corresponding
  Schmidt amplitudes, where the label set $I'$ is disjoint from $I$
  but $|I|=|I'|=d$.  With $\ket{y'}_{\Pfnt{W}}$ the corresponding
  partial Schmidt basis of $\cW$, define $\tau^{1/2}$ by
  $\tau^{1/2}\ket{y'}_{\Pfnt{W}}=\sqrt{d}\lambda_{y'}\ket{y'}_{\Pfnt{W}}$
  and $\tau^{1/2}\ket{\varphi}=0$ for $\ket{\varphi}$ orthogonal
    to the $\ket{y'}_{\Pfnt{W}}$.  With
  this, $\ket{\psi}= (\one\otimes\tau^{1/2})\sum_{y'\in
    I'}\ket{y'}_{\Pfnt{U}} \otimes\ket{y'}_{\Pfnt{W}}/\sqrt{d}$.  By
  the properties of maximally entangled states, there exists a partial
  orthonormal basis $\{\ket{y}_{\Pfnt{W}}\}_{y\in I}$ of $\cW$ such
  that $\sum_{y}\ket{y}_{\Pfnt{U}}\otimes\ket{y}_{\Pfnt{W}}/\sqrt{d}
  =\sum_{y'}\ket{y'}_{\Pfnt{U}}\otimes\ket{y'}_{\Pfnt{W}}/\sqrt{d}$.

  For $x$ a $k$-bit string, let $\ket{x}_{\Pfnt{D}}$ be the
  corresponding logical basis element of $\bigotimes_{i=1}^{k}\cV_{i}$
  considered as $k$ qubits.  Applying the
  observation of the previous paragraph and the reduction to
  $k$-qubits and pure states from before, define
  $\ket{\psi}_{\Pfnt{DE}}=\sum_{x}\ket{x}_{\Pfnt{D}}\otimes\ket{x}_{\Pfnt{E}}/2^{k/2}$
  so that
  $\chi=(\one\otimes\tau^{1/2})\dyad{\psi}(\one\otimes\tau^{1/2})$ for
  some positive semidefinite $\tau^{1/2}$.  Now
  \begin{align}
    \rho(cz) &= \mu(z)\tr_{\Pfnt{D}}\left((P_{c|z;\theta}\otimes\one)
    (\one\otimes\tau^{1/2})\dyad{\psi}(\one\otimes\tau^{1/2})\right)\notag\\
    &=
    \mu(z)\tau^{1/2}
    \tr_{\Pfnt{D}}\left((P_{c|z;\theta}
      \otimes\one)\dyad{\psi}\right)\tau^{1/2}\notag\\
   &=\mu(z) U\tilde \tau^{1/2}
    P_{c|z;\theta}^{T}\tilde \tau^{1/2}U^{\dagger},
  \end{align}
  where $U$ is the isometry that maps $\ket{x}_{\Pfnt{D}}$ to
  $\ket{x}_{\Pfnt{E}}$, $\tilde \tau^{1/2}=U^{\dagger
  }\tau^{1/2}U$, and the transpose is taken with respect to the basis
  $\ket{x}_{\Pfnt{D}}$.  To complete the proof, since $P_{c|z;\theta}$
  is real and symmetric in this basis,
  $P_{c|z;\theta}^{T}=P_{c|z;\theta}$.
\end{proof}

With Thm.~\ref{thm:stdk22} and Lem.~\ref{lem:extremal_qef_constraints}, 
the \QEF optimization problem Prob.~\ref{prob:qefopt} for $\cC_{k22}(CZ)$ and anticipated
probability distribution $\nu(CZ)\in\tr(\cN(\cC_{k22}(CZ)))$ simplifies to
a finite-dimensional problem. Let $\rho(CZ)\in\cC_{k22}(CZ)$ have the
form given in Thm.~\ref{thm:stdk22} with $\tr(\rho)=\tr(\tau)=1$.  In
view of the simplification of Eq.~\ref{eq:qefdef1_nu} for the fixed input
distribution $\mu(Z)$, the \QEF inequality with power $\beta$ for
$F(CZ)$ and $\cC_{k22}(CZ)$ at $\rho(CZ)$ is
\begin{align}
  1&\geq \sum_{cz}F(cz)\mu(z)\Rpow{\alpha}{\rho(c|z)}{\rho} \notag\\
  &= \sum_{cz}F(cz)\mu(z)\tr( \left(\tau^{-\beta/(2\alpha)}\tau^{1/2}
  P_{c|z;\theta} \tau^{1/2}\tau^{-\beta/(2\alpha)}\right)^{\alpha})
  \notag\\
  &=\sum_{cz}F(cz)\mu(z)\tr( \left(\tau^{1/(2\alpha)} P_{c|z;\theta}
  \tau^{1/(2\alpha)}\right)^{\alpha})\notag\\
  &=\sum_{cz}F(cz)\mu(z)\left(\tr( \tau^{1/(2\alpha)} P_{c|z;\theta}
  \tau^{1/(2\alpha)})\right)^{\alpha}\notag\\
  &=\sum_{cz}F(cz)\mu(z)\left(\tr( P_{c|z;\theta}\tau^{1/\alpha}
  P_{c|z;\theta})\right)^{\alpha},
\end{align}
where we used the fact that for a rank $1$ projector $\Pi$,
$\tr((\chi^{1/2}\Pi\chi^{1/2})^{\alpha})=\left(\tr(\chi^{1/2}\Pi\chi^{1/2})\right)^{\alpha}$
and cyclicity of the trace. In the last expression, one of the
projectors in the argument of the trace can be omitted.  
The \QEF optimization problem of Eq.~\ref{prob:qefopt} now reduces to the following:

\begin{alignat}{3}
  \textrm{Maximize:\ }&\sum_{cz}\nu(cz)\log(F(cz))-\log(f_{\max})\notag\\
  \textrm{Variables:\ }& F(CZ), f_{\max}\notag\\
  \textrm{Subject to:\ }& F(CZ)\geq 0,\sum_{cz}F(cz)=1,\notag\\
  & f_{\max}\geq \sum_{cz}\mu(z)F(cz)
  \left(\tr(P_{c|z;\theta}\tau^{1/\alpha}P_{c|z;\theta})\right)^{\alpha}
  \textrm{\ for all $\theta$ and $\tau\geq 0$ with $\tr(\tau)=1$}. \label{prob:qefoptk22}
\end{alignat}
As in Eq.~\ref{prob:qefopt}, the variable
$F(CZ)$ in this optimization problem is not a
\QEF, but every feasible solution $(F(CZ), f_{\max})$ determines
the \QEF $F(CZ)/f_{\max}$ with power $\beta$.

Define $Q_{\alpha}(F(CZ),\theta,\tau)=
  \sum_{cz}\mu(z)F(cz)\left(\tr(P_{c|z;\theta}\tau^{1/\alpha}P_{c|z;\theta})\right)^{\alpha}$.

\begin{lemma}\label{lem:qefoptk22_real}
  In Prob.~\ref{prob:qefoptk22}, $Q_{\alpha}(F(CZ),\theta,\tau)$ is concave in
  the density operator $\tau$, the operator $\tau$ may be restricted to
  be real, and it suffices to consider $\theta$ with
  $\theta_{i}\in[0,\pi]$.
\end{lemma}

\begin{proof}
  For the first claim, we apply the general fact that $A\mapsto
  \tr((K^{\dagger} A^{1/\alpha} K)^{\alpha})$ is a concave function in
  $A\geq 0$ given $\alpha\geq 1$, see Ref.~\cite{carlen:qc2009a},
  Thm.~7.2.  The concavity of $\left(\tr(P_{c|z;\theta}\tau^{1/\alpha}P_{c|z;\theta})\right)^{\alpha}$ 
  is obtained with $K=P_{c|z;\theta}$ and $A=\tau$, and since $K$ is now rank $1$, $\tr((K^{\dagger}
  A^{1/\alpha} K)^{\alpha})=(\tr(K^{\dagger} A^{1/\alpha}
  K))^{\alpha}$. It follows that $Q_{\alpha}(F(CZ), \theta,\tau)$ is a positive linear combination
  of concave functions and is therefore itself concave. 
   Concavity implies that the set of $\tau$ over which
  $Q_{\alpha}(F(CZ),\theta,\tau)$ needs to be maximized can be restricted to
  real matrices.  This follows from
  $Q_{\alpha}(F(CZ),\theta,\tau)=Q_{\alpha}(F(CZ),\theta,\bar\tau)$, which is a
  consequence of $P_{c|z;\theta}$ being real, so by concavity
  $Q_{\alpha}(F(CZ),\theta,(\tau+\bar\tau)/2)\geq Q_{\alpha}(F(CZ),\theta,\tau)$.  (Here
  we used mathematics conventions to denote conjugates of complex
  quantities by an overline).  For the last claim, for each $i$, let
  $\sigma_{z}^{(i)}$ be $\sigma_{z}$ acting on the $i$'th subsystem.  By
  periodicity, we may assume $\theta_{i}\in [-\pi,\pi]$.  Fix $i$ and
  define $\theta'$ by $\theta'_{i}=-\theta_{i}$ and
  $\theta'_{l}=\theta_{l}$ for $l\not=i$.  Then
  \begin{align}
    \tr(P_{c|z;\theta}\tau^{1/\alpha}P_{c|z;\theta})
    &= \tr(\sigma_{z}^{(i)}P_{c|z;\theta}\tau^{1/\alpha}P_{c|z;\theta}\sigma_{z}^{(i)}) \notag\\
    &= \tr(P_{c|z;\theta'}\sigma_{z}^{(i)}\tau^{1/\alpha}\sigma_{z}^{(i)}P_{c|z;\theta'})\notag\\
    &=\tr(P_{c|z;\theta'}(\sigma_{z}^{(i)}\tau\sigma_{z}^{(i)})^{1/\alpha}P_{c|z;\theta'}).
  \end{align}
  Since $\tau\mapsto\sigma_{z}^{(i)}\tau\sigma_{z}^{(i)}$ is a bijection
  of density matrices, the maximum over $\tau$ of the above expression
  does not change when $\theta$ is changed to $\theta'$.
  Therefore, if any $\theta_{i}\in [-\pi,0)$, we can replace it with
  $\theta'_{i}=-\theta_{i}$.
\end{proof}

\subsection{Schemas for \QEF Optimization}
\label{subsec:optschemas}

With the help of Lem.~\ref{lem:qefoptk22_real},
Prob.~\ref{prob:qefoptk22} can be attacked by numerical methods.  An
algorithm for solving Prob.~\ref{prob:qefoptk22} needs to certify that
$f_{\max}$ exceeds $Q(F(CZ),\theta,\tau)$ for all $\theta$ and density
operators $\tau\geq 0$. By concavity, given $\theta$, the maximum in
$\tau$ is unique, but the dependence of this maximum on $\theta$ is
less well behaved. We give a strategy for ensuring that $f_{\max}$
satisfies its constraint for all $\tau$ and $\theta$ with arbitrarily
small slack.

Let $H_{1}$ denote the displaced half unit circle in $\rls^{3}$
consisting of the points of the form $(\cos(\theta),\sin(\theta),1)$
with $\theta\in[0,\pi]$, and let $R_{1}$ be the set of semidefinite
operators operators $\chi$ on $k$ qubits that are real with respect to
the logical basis and satisfy $\tr(\chi^{\alpha})=1$.  For the purpose
of distinguishing factors in tensor products, for each $i$ let
$H_{1}^{(i)}$ be an identified copy of $H_{1}$.  Define
$\cR_{1}=R_{1}\otimes (\bigotimes_{i=1}^{k}H_{1}^{(i)})$. Let $\cR$ be
the tensor product of the vector spaces containing $R_{1}$ and the
$H_{1}^{(i)}$.  Write $r_{i}=(u_{i},v_{i},w_{i})$ for a point in
linear span of $H_{1}^{(i)}$.  For each $cz$ the map
\begin{align}
  L_{cz}:\chi,r_{1},\ldots,r_{k}  &\mapsto
  2^{-k}\tr\Big(\chi \notag\\
  &\hphantom{\mapsto\;}\times
  \bigotimes_{i} \big((1+(-1)^{c_{i}}\sigma_{z})\knuth{z_{i}=0} +
  (w_{i} +(-1)^{c_{i}}(u_{i}\sigma_{z}+v_{i}\sigma_{x})\knuth{z_{i}=1})\big)\Big)
\end{align}
is multilinear with respect to $\chi$ and each of the $r_{i}$.
It therefore lifts to a linear map $\tilde L_{cz}$ on $\cR$ so that
$\tilde L_{cz}(\chi\otimes r_{1}\otimes\ldots\otimes r_{k}) =
L_{cz}(\chi,r_{1},\ldots,r_{k})$.
This map satisfies
\begin{equation}
  \tilde L_{cz}(\chi\otimes (\cos(\theta_{1}),\sin(\theta_{1}),1)\otimes\ldots\otimes
  (\cos(\theta_{k}),\sin(\theta_{k}),1))
  = \tr(\chi P_{c|z;\theta}).
\end{equation}
Since $x\mapsto
|x|^{\alpha}$ is convex and the compositions of linear and convex maps
are convex, the map $|\tilde L_{cz}|^{\alpha}$ is convex.
Since positive linear combinations of convex maps are convex, the
map 
\begin{equation}\label{eq:def_tilde_Q}
  \tilde Q_{\alpha}: F(CZ),u\in\cR \mapsto \sum_{cz}\mu(z)F(cz)|\tilde L_{cz}(u)|^{\alpha}
\end{equation}
is convex. 

In Prob.~\ref{prob:qefoptk22} with $F(CZ)$ fixed, we can set
$f_{\max}$ to $f_{\max}(F(CZ))=\max_{u\in\cR_{1}}\tilde
Q(F(CZ),u)$. Given an algorithm to determine $f_{\max}(F(CZ))$, any
generic local search algorithm can be used to optimize $F(CZ)$, so we
focus on algorithms for $f_{\max}$. A certified upper bound
  on $f_{\max}$ suffices, and such a bound can be obtained by maximizing
$\tilde Q_{\alpha}$ over any convex set $\cR'\supseteq\Cvx(\cR_{1})$.
For example, if $\cP^{(i)}$ are convex polygons
satisfying $\Cvx(H_{1}^{(i)})\subseteq\cP^{(i)}$, then we can let
$\cR'=\Cvx(R_{1})\otimes \bigotimes_{i=1}^{k}\cP^{(i)}$.  
Because $\tilde Q_{\alpha}(F(CZ),u)$ is convex in $u$, the maximum
  is achieved on an extreme point of $\cR'$ and the upper
  bound becomes tight in the limit where the $\cP^{(i)}$ converge to
  $\Cvx\left(H_{1}^{(i)}\right)$.  The extreme points of $\cR'$ are tensor
products of some $\chi\in \Cvx(R_{1})$  with members
of the finite sets $\Xtrm(\cP^{(i)})$. 
\Pc{In fact, for extreme points, $\chi\in R_{1}$: Since $R_{1}$ is closed and bounded, $R_{1}$ contains the extreme points of $\Cvx(R_{1})$ (Thm 1.4.5 of Ref.~\cite{kadison:qf1997a}).}
 Provided we can effectively maximize over $\chi
\in \Cvx(R_{1})$, there are finitely many tensor products of
extreme points of $\Xtrm(\cP)$ to check.  Let
$r=\bigotimes_{i=1}^{k}r_{i}$ be in
$\bigotimes_{i=1}^{k}\Xtrm(\cP^{(i)})$, where
$r_{i}=(u_{i},v_{i},1)=((1+\epsilon_{i})\cos(\theta_{i}),(1+\epsilon_{i})\sin(\theta_{i}),1)$.
Then
\begin{equation}
  \tilde L_{cz}(\chi\otimes r) = \tr(\chi\otimes \bigotimes_{i=1}^{k} P_{i}),
\end{equation}
where for $z_{i}=0$, $P_{i}=(\one + (-1)^{c_{i}}\sigma_{z})/2$ and for
$z_{i}=1$, $P_{i}=(\one + (1+\epsilon_{i})(-1)^{c_{i}}\sigma_{\hat u_{i}})/2$
with
$\sigma_{\hat u_{i}}=\cos(\theta_{i})\sigma_{z}+\sin(\theta_{i})\sigma_{x}$.
An issue is that $P_{i}$ is not positive semidefinite, so the concavity
property with respect to $\tau$ with $\tau^{1/\alpha}=\chi$ does not
apply and maximizing over $\chi\in \Cvx(R_{1})$ is more difficult.  \Pc{The
  expression we have is of the form
  $|\tr(\tau^{1/\alpha}P)|^{\alpha}$.  For $P$ positive this is the
  same as $\tr(P^{1/2}\tau^{1/\alpha}P^{1/2})^{\alpha}$ which is
  concave in $\tau$. When $P$ is not positive, this relationship
  fails.}  To avoid this difficulty we give an algorithm that uses
inner approximations of $\Cvx(H_{1}^{(i)})$ instead.

For the simplest algorithm, let $\cX=(j\pi/m)_{j=0}^{m}$ evenly divide
$[0,\pi]$ with $m\geq 2$. Write
$r(\theta)=\bigotimes_{i=1}^{k}(\cos(\theta_{i}),\sin(\theta_{i}),1)$.
Let $\cX^{l}$ denote the $l$-fold cartesian product of $\cX$ with
  itself.  For each $r\in r\left(\cX^{k}\right)$, compute
$f_{\max}(r)=\max\{\tilde Q_{\alpha}(F(CZ),\tau^{1/\alpha}\otimes
r):\tau^{1/\alpha}\in R_{1}\}$, where the maximization is concave over
real $2^{k}\times 2^{k}$ density matrices $\tau$.  How to perform this
maximization will be explained later.  Given that $f_{\max}(r)$ has been determined for all $r\in
  r\left(\cX^{k}\right)$, a lower bound on $f_{\max}$ is given by
$f_{\max}\geq \max\left\{f_{\max}(r):r\in
  r\left(\cX^{k}\right)\right\}$. An upper bound can be obtained
  by recursively applying the next lemma.
\begin{lemma}\label{lem:fmax_intervalbnd}
  Consider $\theta,\theta'$ so that $\theta'-\theta=\phi e_{i}$ where
  $\phi\in(0,\pi/2]$ and $e_{i}=(\knuth{j=i})_{j=1}^{k}$.  Let $f =
  f_{\max}(r(\theta))$ and $f'= f_{\max}(r(\theta'))$. For $\varphi\in[0,\phi]$
  and $\theta''=\theta+\varphi e_{i}$,
  \begin{align}
    f_{\max}(r(\theta''))  &\leq u(\varphi)\defeq
    \frac{(\sin(\phi-\varphi)+\sin(\varphi))^{\beta}(\sin(\phi-\varphi) f +
      \sin(\varphi)f')}{\sin(\phi)^{\alpha}}.
    \label{eq:lem:fmax_intervalbnd:1}
  \end{align}
  The bound $u(\varphi)$ is log-concave in $\varphi$ and satisfies
  \begin{align}
    u(\varphi) &\leq \left(\frac{\phi}{\sin(\phi)}\right)^{\alpha}\max(f,f').
    \label{eq:lem:fmax_intervalbnd:2}
  \end{align}
\end{lemma}
If only upper bounds $u$ and $u'$ respectively on $f$ and $f'$ are known, 
then upper bounds on $f_{\max}(r(\theta''))$ can be obtained from Eqs.~\ref{eq:lem:fmax_intervalbnd:1}
and~\ref{eq:lem:fmax_intervalbnd:2} with the replacement of $f$ and $f'$ by 
their upper bounds $u$ and $u'$.

\begin{proof}
  Write $f'' = f_{\max}(r(\theta''))$.  Let $\chi$ witness $f''$ in
  the sense that $f'' = \tilde Q_{\alpha}(F(CZ),\chi\otimes r(\theta''))$.  For
  each $cz$, consider the contribution $f''(cz)=\mu(z)F(cz)\tilde
  L_{cz}(\chi\otimes r(\theta''))^{\alpha}$ to $f''$.  If $z_{i}=0$,
  then
  \begin{equation}
    f''(cz) = \mu(z)F(cz)\tilde L_{cz}(\chi\otimes r(\theta))^{\alpha} = 
    \mu(z)F(cz)\tilde L_{cz}(\chi\otimes r(\theta'))^{\alpha},
  \end{equation}
  since for $z_{i}=0$, the $i$'th factor $P^{(i)}_{c_i|z_i,\psi_i}$ of $P_{c|z;\psi}$ does not
  depend on $\psi_{i}$.  For $z_{i}=1$, the $i$'th factor of
  $P_{c|z;\theta''}$ is $(\one + \cos(\theta_{i}+\varphi)\sigma_{z}+
  \sin(\theta_{i}+\varphi)\sigma_{x})/2$.  Let
  $a=(\cos(\theta_{i}),\sin(\theta_{i}))$,
  $a'=(\cos(\theta_{i}+\phi),\sin(\theta_{i}+\phi))$ and
  $a''=(\cos(\theta_{i}+\varphi),\sin(\theta_{i}+\varphi))$.  Then there
  exist $\lambda\in[0,1]$ and $b\in(0,1]$ such that $\lambda a+
  (1-\lambda)a' = b a''$. The values of $\lambda$ and $b$ 
  will be determined later.
  Given such $\lambda$ and $b$, we have
  \begin{equation}
    P^{(i)}_{c_{i}|z_{i};\theta_{i}+\varphi} \leq
    P^{(i)}_{c_{i}|z_{i};\theta_{i}+\varphi} + (1/b-1) \one= 
    (\lambda P^{(i)}_{c_{i}|z_{i};\theta_{i}}
    + (1-\lambda) P^{(i)}_{c_{i}|z_{i};\theta_{i}+\phi})/b.
    \label{eq:after_a'''}
  \end{equation}
  The operator inequality extends to
  \begin{equation}
    \chi\otimes P_{c|z;\theta''}
    \leq \left(\lambda (\chi\otimes P_{c|z;\theta})+(1-\lambda)(\chi\otimes P_{c|z;\theta'})\right)/b.
  \end{equation}
  By operator monotonicity, homogeneity and convexity it follows that
  \begin{equation}
    f''(cz) \leq \mu(z)F(cz)\left(\lambda\tilde L_{cz}(\chi\otimes r(\theta))^{\alpha} + (1-\lambda) 
    \tilde L_{cz}(\chi\otimes r(\theta'))^{\alpha}\right)/b^{\alpha}.
  \end{equation}
  Since $b<1$, this inequality is also satisfied for $z_{i}=0$.  Since
  $f \geq \tilde Q_{\alpha}(F(CZ),\chi\otimes r(\theta))$ and similarly for
  $f'$, after summing over $cz$ to add the contributions to $f''$, we
  conclude that
  \begin{equation}
    f'' \leq (\lambda f + (1-\lambda) f')/b^{\alpha}.  
  \end{equation}

  To determine $\lambda$ and $b$ in terms of $\phi$ and $\varphi$, we
  solve a geometrical problem involving chords. For this paragraph we
  use notational conventions from plane geometry.  Let $O$ be the
  center of a unit circle and $A$, $B$ and $C$ points on the
  circumference with $C$ between $A$ and $B$.  Write $\angle AOB
  =\phi$ and $\angle AOC = \varphi$. Let $M$ be the intersection of
  the lines $\overline{OC}$ and $\overline{AB}$. Let $x=AM$, $y=MB$
  and $b=OM$ be the lengths of the respective line segments.  Then $b
  \sin(\varphi)+b\sin(\phi-\varphi) = \sin(\phi)$ since the
  $\sin(\phi)/2$ is the area of $\triangle OAB$, $b\sin(\varphi)/2$ the
  area of $\triangle OAM$ and $b\sin(\phi-\varphi)/2$ the area of
  $\triangle OMB$. Thus
  $b=\sin(\phi)/(\sin(\varphi)+\sin(\phi-\varphi))$.  Since $\angle
  OAB=(\pi/2-\phi/2) $, $x \sin(\pi/2-\phi/2)=b\sin(\varphi)$ and
  $y\sin(\pi/2-\phi/2)=b\sin(\phi-\varphi)$. 
  From this we determine $\lambda=y/(x+y)= \sin(\phi-\varphi)/(\sin(\varphi)+\sin(\phi-\varphi))$.
  Summarizing, we
  have
  \begin{align}
    b &= \frac{\sin(\phi)}{\sin(\varphi)+\sin(\phi-\varphi)}\in(0, 1],\notag\\
    \lambda &= \frac{\sin(\phi-\varphi)}{\sin(\varphi)+\sin(\phi-\varphi)}\in[0,1].
  \end{align}
  By rotational symmetry, the desired identity
  $\lambda a+(1-\lambda)a' = b a''$ is satisfied with
  $a$, $a'$ and $a''$ as defined before Eq.~\ref{eq:after_a'''}.
  \Pc{Since
    $\sin(\phi)=
    \sin((\phi-\varphi)+\varphi)=
    \sin(\varphi)\cos(\phi-\varphi)+\sin(\phi-\varphi)\cos(\varphi)$,
    we also have
    \begin{align*}
      b&= \frac{\sin(\varphi)\cos(\phi-\varphi)+\sin(\phi-\varphi)\cos(\varphi)}
      {\sin(\varphi)+\sin(\phi-\varphi)}\\
      &= (1-\lambda)\cos(\phi-\varphi)+\lambda\cos(\varphi).
    \end{align*}
  }
  It is possible to maximize the upper bound $(\lambda f
  +(1-\lambda)f')/b^{\alpha}$ on $f''$ over $\varphi\in [0,\phi]$.
  In terms of $\varphi$, the bound is
  \begin{align}
    u(\varphi) &= \frac{\lambda f  +(1-\lambda)f'}{b^{\alpha}}\notag\\
    &= 
    \frac{(\sin(\phi-\varphi)+\sin(\varphi))^{\beta}(\sin(\phi-\varphi) f + \sin(\varphi)f')}{\sin(\phi)^{\alpha}}.
  \end{align}
  To show that the function $u(\varphi)$ has a unique maximum we
  prove log-concavity in $\varphi$. Consider
  \begin{equation}
    v(\varphi)=\log(\sin(\phi)^{\alpha}u(\varphi))
    =\beta\log(\sin(\phi-\varphi)+\sin(\varphi))+\log(\sin(\phi-\varphi) f + \sin(\varphi)f').\label{eq:lem:fmax_intervalbnd:3}
  \end{equation}
  As a functions of $\varphi$, both $\sin(\phi-\varphi)$ and $\sin(\varphi)$ are concave
  for the values of $\phi$ and $\varphi$ under consideration. Therefore,
  any linear combination $g(\varphi)=c\sin(\phi-\varphi)+c'\sin(\varphi)$ with $c,c'\geq 0$
  is concave. Since $\log$ is monotone increasing and concave,  
  $\log(g(\varphi))$ 
  is concave for any concave $g(\varphi)$. Consequently, $v(\varphi)$ is the sum of 
  two concave functions and therefore also concave. 

  We use the small angle approximation to upper bound
  $u(\varphi)$. Applying the inequalities $\sin(\phi-\varphi)\leq
  (\phi-\varphi)$ and $\sin(\varphi)\leq\varphi$ gives
  \begin{equation}
    u(\varphi) \leq \left(\frac{\phi}{\sin(\phi)}\right)^{\alpha}\max(f,f').
  \end{equation}
\end{proof}

The maximum of the bound $u(\varphi)$ defined in 
Eq.~\ref{eq:lem:fmax_intervalbnd:1} can be found as follows: With
$v(\varphi)$ as defined in Eq.~\ref{eq:lem:fmax_intervalbnd:3} and
considering the concavity of $v(\varphi)$, if the derivative
$v^{(1)}(0)\leq 0$ the maximum of $u(\varphi)$ is $f$, 
if $v^{(1)}(\phi)\geq 0$, the maximum is $f'$, and otherwise there is a 
unique critical point $\varphi_0$ between $0$ and $\phi$ for $v(\varphi)$, 
and the maximum of $u(\varphi)$ is $u(\varphi_0)$. The critical point 
is found by solving $v^{(1)}(\varphi_0)=0$.

We can now determine an upper bound on $f_{\max}$ from the values of
$f_{\max}(r)$ for $r\in r\left(\cX^{k}\right)$.  Write
$\theta_{>l}=(\theta_{l+i})_{i=1}^{k-l}$ and $\theta_{\leq
  l}=(\theta_{i})_{i=1}^{l}$ so that $\theta=\theta_{\leq
  l}\theta_{>l}$ with our concatenation conventions.  For any $l$
define
\begin{equation}
  f_{\max}(\theta_{>l})=\max_{\chi,\theta_{\leq l}}
  \tilde Q_{\alpha}(F(CZ),\chi\otimes r(\theta_{\leq l}\theta_{>l})),
\end{equation}
where we are overloading the symbol $f_{\max}$ by making it depend on
the type and length of the argument. 
The upper bound on $f_{\max}$ can be obtained recursively, where
at the $l$'th step we obtain upper bound $v(\theta_{>l})$ on
$f_{\max}(\theta_{>l})$, so that the $k$'th step yields an upper
bound on $f_{\max}$.  To initialize the procedure (the $0$'th step),
we determine $f_{\max}(\theta)$ for all $\theta=\theta_{>0}\in
\cX^{k}$. This requires a method for maximizing $\tau\in
S_{1}(\cH)\mapsto \tilde
Q_{\alpha}(F(CZ),\tau^{1/\alpha}\otimes r)$ for given $r$, and
such a method is given later in this section.  Let
$v(\theta)=f_{\max}(\theta)$.  For the $l$'th step, fix
$\theta_{>l}\in \cX^{k-l}$. From the previous steps, for all
$\theta_{l}\in \cX$, we have determined upper bounds
$v(\theta_{l}\theta_{>l})\geq f_{\max}(\theta_{l}\theta_{>l})$.  For
any pair of successive $\psi,\psi'\in\cX$,  we can
apply Lem.~\ref{lem:fmax_intervalbnd} to obtain a bound
$u(\psi,\psi') \geq f_{\max}(\psi''\theta_{>l})$ for all
$\psi''\in[\psi,\psi']$.  The maximum of these bounds is an upper
bound on $f_{\max}(\theta_{>l})$.  After having determined
  $u(\psi,\psi')$, we can set $v(\theta_{>l})=\max\{u(\psi,\psi'):
  \textrm{$\psi,\psi'$ are successive pairs in $\cX$}\}$. 

The upper and lower bounds on $f_{\max}$ obtained converge
with the resolution $m$ used for $\cX$. It is possible
to start at low resolution, and refine the subdivision $\cX$ if the
gap between lower and upper bounds is too large. However, not all
intervals need refinement and we can significantly reduce the work
required by selectively refining a cubical grid in
$\bigotimes_{i=1}^{k}H_{1}^{(i)}$. The grid-refinement algorithm's
state contains two data structures. 
Let $[0,\pi]^{l}$ denote the $l$-fold cartesian product of $[0,\pi]$ with
itself. The first data structure is $\cT$ and contains
the pairs of
$\theta\in [0,\pi]^{k}$ and the corresponding values
$f_{\max}(\theta)$
for which $f_{\max}(\theta)$ has
been determined.
The second is $\cK$ and consists of cuboidal
regions in $[0,\pi]^{k}$, where each region $K$ is
specified by its $2^{k}$ vertices. The region $K$ comes with an upper
bound $f_{\max}(K)\geq \max_{\theta\in K}f_{\max}(\theta)$. The
structure $\cK$ may be organized as a priority heap, 
where the priority of the region $K$ is determined by $f_{\max}(K)$. The region $K$'s vertices can
be given in the form $\theta+\sum_{i\in I}\varphi_{i} e_{i}$ for
subsets $I$ of $[k]$, and $K$ consists of the convex closure of the
set of these vertices. We require that 1) $\cT$ contains 
the vertices of regions in $\cK$, and 2) the union of the closed
cubical regions of $\cK$ is $[0,\pi]^{k}$.  We can also
ensure that the cubical regions have disjoint interiors.  The current
overall upper bound $f_{\max}$ is the maximum of $f_{\max}(K)$ over
regions $K$ in $\cK$. A lower bound is given by the maximum of
$f_{\max}(\theta)$ over the $\theta$ in $\cT$.  The algorithm is
initialized with a grid $\cX$ for some resolution $m\geq 2$.  For
this, it computes $f_{\max}(\theta)$ for each $\theta\in\cX$ and adds
$(\theta, f_{\max}(\theta))$ to $\cT$. It then iterates over the
cubical regions $K$ defined by $\cX$, computes $f_{\max}(K)$ and adds
$(K,f_{\max}(K))$ to $\cK$. We can compute $f_{\max}(K)$
for $K$ consisting of the convex closure of $\{\theta+\sum_{i\in
  I}\varphi_{i} e_{i}:I\subseteq [k]\}$ according to the strategy for
computing the global $f_{\max}$ given $\cX$.  For this, we replace
$\cX$ by $\prod_{i}\{\theta_{i},\theta_{i}+\varphi_{i}\}$, which is
the cartesian product of the sets $\{\theta_{i},\theta_{i}+\varphi_{i}\}$. 
The
strategy gives the value of $f_{\max}(K)$ for the region $K$ covered
by the convex closure of 
$\prod_{i}\{\theta_{i},\theta_{i}+\varphi_{i}\}$.  After
initialization, the algorithm updates the structures in each step by
refining the top region $K$ in $\cK$. If $K$ is the convex closure of
$\{\theta+\sum_{i\in I}\varphi_{i} e_{i}:I\subseteq [k]\}$, 
 a possible refinement strategy  is to divide each of $K$'s edges in 
 two for $2^{k}$ subregions defined as the convex closures $K_{J}$ of
$\{\theta+\sum_{i\in J}\varphi_{i} e_{i}/2+\sum_{i\in I}\varphi_{i}
e_{i}/2:I\subseteq [k]\}$
for $J\subseteq [k]$. 
For each new vertex
$\theta'$, if the vertex is not in $\cT$, the algorithm computes 
$f_{\max}(\theta')$ and adds $(\theta', f_{\max}(\theta'))$ to $\cT$.  
For each $K_{J}$ the algorithm computes $f_{\max}(K_{J})$ and adds 
$(K_{J}, f_{\max}(K_{J}))$ to $\cK$. The
original region $K$ is removed from $\cK$ at the beginning of the
refinement cycle.

To complete the schema for determining $f_{\max}$, we return to the
problem of maximizing the concave, homogeneous-of-degree-1 function
$g:\tau\in S_{1}(\cH)\mapsto \tilde Q_{\alpha}(F(CZ),\tau^{1/\alpha}\otimes r)$
for fixed $r\in \bigotimes_{i=1}^{k}H_{1}^{(i)}$.  It can in principle be
maximized by any method for concave maximization over a domain defined
by semi-definite constraints. Here we have a special domain and we can
take advantage of this. Further, $g$ is differentiable at full rank
$\tau$. Write $g$ in the form
\begin{equation}
  g(\tau) = \sum_{cz}\left(\tr(\tau^{1/\alpha}Q_{cz})\right)^{\alpha}
\end{equation}
for a family of positive semidefinite operators $Q_{cz}$.  Each
$Q_{cz}$ is a product of $(\mu(z)F(cz))^{1/\alpha}$ and a rank-1
projector $P_{c|z;\theta}$.  We begin by reducing the problem to the
case where it suffices to consider operators $\tau$ with full support on one
of the irreducible subspaces generated by the $Q_{cz}$.
Let $\Pi_{0}$ be the null-space projector for $\tau$.  Suppose that
$\Pi_{0}\ne 0$, and consider changing $\tau$ to $\tau'=(1-\epsilon)\tau+\epsilon
\Pi_{0}/\tr(\Pi_{0})$.  Then
\begin{equation}
  \tau'^{1/\alpha}=
  (1-\epsilon)^{1/\alpha}\tau^{1/\alpha}+(\epsilon/\tr(\Pi_{0}))^{1/\alpha}\Pi_{0}
  = \tau^{1/\alpha}+\gamma \epsilon^{1/\alpha}\Pi_{0} + O(\epsilon),
\end{equation}
with $\gamma=(\tr(\Pi_{0}))^{-1/\alpha}$. Consider the set $I$ of
$cz$ such that $\tr(\tau Q_{cz})> 0$ and $\tr(Q_{cz}\Pi_{0})>
0$. For $cz\in I$,  $\tr(\tau^{1/\alpha}Q_{cz})>0$ and
\begin{align}
  \left(\tr(\tau'^{1/\alpha}Q_{cz})\right)^{\alpha}
  &= \left(\tr(\tau^{1/\alpha}Q_{cz}) + \gamma\epsilon^{1/\alpha}\tr(\Pi_{0}
  Q_{cz})+ O(\epsilon)\right)^{\alpha} 
  \notag\\
  &= \left(\tr(\tau^{1/\alpha}Q_{cz})\right)^{\alpha}
     +\alpha (\tr(\tau^{1/\alpha}Q_{cz}))^{\beta}
     \gamma\epsilon^{1/\alpha}\tr(\Pi_{0}Q_{cz}) + o(\epsilon^{1/\alpha}).
\end{align}
If $\tr(Q_{cz}\tau)=0$ or $\tr(Q_{cz}\Pi_{0})=0$, then
$(\tr(\tau'^{1/\alpha}Q_{cz}))^{\alpha} =
(\tr(\tau^{1/\alpha}Q_{cz}))^{\alpha}+O(\epsilon)$.  It follows that
if $I$ is not empty, for small enough $\epsilon>0$, $g(\tau')-g(\tau)$
is dominated by positive terms of order $\epsilon^{1/\alpha}$ and,
unless $I$ is empty, $\tau$ does not maximize $g$.  The set $I$ is
empty iff for all $cz$ either $\tr(Q_{cz}\tau)=0$ or
$\tr(Q_{cz}\Pi_{0})=0$, which implies that every $Q_{cz}$ is supported
in $\one-\Pi_{0}$ or in $\Pi_{0}$. In other words, the $Q_{cz}$ can be
block-diagonalized with respect to $\Pi_{0}$.  Let $\{\Pi_{i}\}_{i}$
be a maximal complete set of projectors for which the $Q_{cz}$ are
block-diagonal. Equivalently, the $\Pi_{i}$ 
project onto the irreducible subspaces of the algebra generated by the
$Q_{cz}$ and generate the center of this algebra.  For an orthogonal
$U$ that commutes with all $Q_{cz}$, $(\tr((U\tau
U^{T})^{1/\alpha}Q_{cz}))^{\alpha}=
(\tr(\tau)^{1/\alpha}Q_{cz})^{\alpha}$ for all $cz$. Since averaging
over such $U$ is decoherence of $\tau$ with respect to the center of
the algebra generated by the $Q_{cz}$ and by concavity, the maximum of
$g$ is achieved for $\tau$ block-diagonal with respect to the
$\Pi_{i}$. We can then write $\tau$ as a mixture
  $\tau=\bigoplus_{i}\mu(i)\tau_{i}$ where the $\tau_{i}$ are density
  matrices supported in the $i$'th irreducible subspace and $\mu$ is a
  probability distribution. With this,
  $g(\tau)=\sum_{i}\mu(i)g(\tau_{i})$, so $g(\tau)\leq
  \max_{i}g(\tau_{i})$, and the problem reduces to the case where
  $\tau$ has full support in one of the irreducible subspaces.  
  We remark that for determining $f_{\max}$ it may be
necessary to check for  reducability of the $Q_{cz}$. In particular,
for the cases where $F(cz)$ has zeros or if any of the angles defining
the $Q_{cz}$ are $0$ or $\pi$, the algebra generated by the $Q_{cz}$
may not be complete, in which case the $Q_{cz}$ can be jointly block
diagonalized.

The previous paragraph implies that it suffices to consider the
general problem of maximizing a concave, homogeneous- of-degree-1 and
differentiable function $g:\tau\in S_{1}(\cH) \mapsto g(\tau)$ over
real positive density operators.  Let $\grad g$ be the derivative
expressed as a Hermitian operator so that for positive semidefinite
$\tau+\epsilon\Delta$,
$g(\tau+\epsilon\Delta)=g(\tau)+\epsilon\tr(\Delta\grad g
)+o(\epsilon)$. An iterative maximization algorithm updates
$\tau$ to $\tau'$ to approach the maximum.  For this problem, given a
density operator $\Delta$, we can update $\tau' = (1-\epsilon)\tau +
\epsilon\Delta$ to satisfy the constraints. By degree-1 homogeneity,
$\tr(\tau\grad g(\tau) )=g(\tau)$. \Pc{Homogeneity implies that
  $g(\tau+\epsilon \tau)=(1+\epsilon)g(\tau)= g(\tau)+\epsilon
  g(\tau)$. The statement follows because this has to match the
  first-order expansion in terms of $\grad g(\tau)$.}  Thus $g(\tau') =
(1-\epsilon)g(\tau)+ \epsilon\tr(\Delta\grad g(\tau))+o(\epsilon)$.
Write $\grad g(\tau)=\sum_{i=1}^{d}\lambda_{i}\Pi_{i}$ with $\Pi_{i}$
a complete family of orthogonal projectors onto the distinct
eigenvalue eigenspaces of $\grad g(\tau)$.  We order the eigenvalues
so that $\lambda_{1}$ is the maximum eigenvalue.  Then we have
$\tr(\Delta\grad g(\tau))\leq\lambda_{1}$, so it is natural to choose
directions $\Delta$ supported in $\Pi_{1}$.  The maximum is achieved
if $\lambda_{1}=g(\tau)$, in which case necessarily $\tau$ is
supported in $\Pi_{1}$, and $\Pi_{1}=\one$ since $\tau$ has full
support.  That is, $\Pi_{1}=\one$ is a necessary and sufficient
condition for maximum $g(\tau)$.  If this condition is not satisfied,
an update option is to set $\Delta=\Pi_{1}/\tr(\Pi_{1})$.  An
alternative is to set $\Delta=\knuth{\grad
  g(\tau)>g(\tau)}/\tr(\knuth{\grad g(\tau)>g(\tau)})$. One can choose
$\epsilon$ according to a schedule such as one of 
those used in the Frank-Wolfe algorithm~\cite{jaggi:qc2013a}, or
one can choose $\epsilon$ by performing a one-dimensional maximization
in the direction $\Delta$.  Concave maximization over density matrices
is also a task for maximum-likelihood state tomography, where a common
strategy is the $R\rho R$ algorithm~\cite{hradil:qc2004a}.  A diluted
version of this algorithm~\cite{rehacek:qc2006a} could be used here
also.  However, the methods discussed so far do not have good
convergence properties, so some exploration may be required to
determine the best update strategy. Convergence issues can be
mitigated by taking advantage of the fact that $\lambda_{1}$ is also
an upper bound on the maximum value of $g$, so $\lambda_{1}-g(\tau)$
is the gap and can be used as a stopping criterion, noting that we
often do not require extremely small gaps between upper and lower
bounds in our applications.

For computing $\grad g$, it suffices to consider the coefficients of
the form $g_{P}(\tau)=\tr(\tau^{1/\alpha}P)^{\alpha}$ of $\mu(z)F(cz)$ in
the sum for $\tilde Q_{\alpha}$. Here $P$ is a projector. We can write
the gradient in the form
\begin{equation} \label{eq:gradient_g_fun}
  \grad_{\tau} g_{P}(\tau) = \alpha\tr(\tau^{1/\alpha}P)^{\beta} X,
\end{equation}
where $X\defeq \grad_{\tau} \tr(\tau^{1/\alpha}P)$. To compute 
$X$ requires perturbation techniques. Write
$\tau'=\tau+\epsilon\Delta$ and express
$\tau=\sum_{i}\lambda_{i}\Pi_{i}$ in terms of its eigenspace
projectors, where the $\lambda_{i}$ are positive.  This enables a
unique decomposition of $\Delta$ in the form
$\Delta=\sum_{i}\Delta_{i}+[S,\tau]$, where the support of
$\Delta_{i}$ is in $\Pi_{i}$ and $S$ is skew-symmetric with
$\Pi_{i}S\Pi_{i}=0$ for each $i$.  To compute $\Delta_{i}$ and $S$ in
terms of $\Delta$, define $\Delta_{ij}=\Pi_{i} \Delta\Pi_{j}$.  Then
$\Delta_{i}=\Delta_{ii}$ and $S=\sum_{i\not=j}S_{ij}$ with $S_{ij} =
\Delta_{ij}/(\lambda_{j}-\lambda_{i})$.  For orthogonal $U$, $(U\tau
U^{T})^{1/\alpha}=U\tau^{1/\alpha}U^{T}$.  With $U=e^{\epsilon S}$,
$\gamma>0$ and $Y$ commuting with $\tau$, we have $U(\tau+\epsilon
Y)^{\gamma}U^{T} = \tau^{\gamma} + \epsilon
\gamma\tau^{\gamma-1}Y+\epsilon[S,\tau^{\gamma}] + O(\epsilon^{2})$,
where we used the assumption that $\tau$ is positive.  For
sufficiently small $\epsilon$, we can expand
\begin{align}
  (\tau+\epsilon\Delta)^{1/\alpha} & = \left(\tau
    + \sum_{i}\epsilon\Delta_{i} + \epsilon [S,\tau]\right)^{1/\alpha} \notag\\
  &= \left(U(\tau+
    \sum_{i}\epsilon\Delta_{i})U^{T}+O(\epsilon^{2})\right)^{1/\alpha}
  \notag\\
  &= \left(U(\tau+
    \sum_{i}\epsilon\Delta_{i}+O(\epsilon^{2}))U^{T}\right)^{1/\alpha}
  \notag\\
  &= U\left(\tau+
    \sum_{i}\epsilon\Delta_{i}+O(\epsilon^{2}))\right)^{1/\alpha}U^{T}
  \notag\\
  &= U\left(\left(\tau+
      \sum_{i}\epsilon\Delta_{i}\right)^{1/\alpha}+O(\epsilon^{2})\right)U^{T}
  \notag\\
  &= \tau^{1/\alpha} + \epsilon\frac{1}{\alpha}\tau^{-\beta/\alpha}\sum_{i}\Delta_{i}
  +\epsilon [S,\tau^{1/\alpha}]
  + O(\epsilon^{2})\notag\\
  &= \tau^{1/\alpha} + \epsilon\left(
    \frac{1}{\alpha}\sum_{i}\lambda_{i}^{-\beta/\alpha}\Delta_{i} +
    [S,\tau^{1/\alpha}]\right) +O(\epsilon^{2}).
\end{align}
Expressed with the $\Delta_{ij}$ this is
\begin{align}
  (\tau+\epsilon\Delta)^{1/\alpha}
  &= \tau^{1/\alpha}
  + \epsilon\left(
    \sum_{i}\frac{1}{\alpha}\lambda_{i}^{-\beta/\alpha}\Delta_{ii}
    +
    \sum_{i\not=j} \frac{1}{\lambda_{j}-\lambda_{i}}
    (\lambda_{j}^{1/\alpha}-\lambda_{i}^{1/\alpha}) \Delta_{ij}
  \right) + O(\epsilon^{2})\notag\\
  &= \tau^{1/\alpha}
  + \epsilon\left(
    \sum_{i}\frac{1}{\alpha}\lambda_{i}^{-\beta/\alpha}\Pi_{i}\Delta\Pi_{i}
    +
    \sum_{i\not=j} \frac{1}{\lambda_{j}-\lambda_{i}}
    (\lambda_{j}^{1/\alpha}-\lambda_{i}^{1/\alpha}) \Pi_{i}\Delta\Pi_{j}
  \right) + O(\epsilon^{2}).
\end{align}
With this,
\begin{align}
  g_{P}(\tau+\epsilon\Delta)^{1/\alpha} &=
  \tr((\tau+\epsilon\Delta)^{1/\alpha}P)\notag\\
  &= \tr(\tau^{1/\alpha}P)
   + \epsilon\left(\tr(\sum_{i }\frac{\lambda_{i}^{-\beta/\alpha}}{\alpha}
     \Pi_{i}P\Pi_{i}\Delta)
   + \tr(\sum_{i\not=j}\frac{\lambda_{j}^{1/\alpha}-\lambda_{i}^{1/\alpha}}
   {\lambda_{j}-\lambda_{i}}\Pi_{j}P\Pi_{i}\Delta)\right)\notag\\
 &\hphantom{=\;}
   + o(\epsilon)\notag\\
  &=
   \tr(\tau^{1/\alpha}P)
   + \epsilon
   \tr(\left(\sum_{i}\frac{\lambda_{i}^{-\beta/\alpha}}{\alpha} P_{ii} +
       \sum_{i\not=j} \frac{\lambda_{j}^{1/\alpha}-\lambda_{i}^{1/\alpha}}
   {\lambda_{j}-\lambda_{i}} P_{ji}
     \right)\Delta) + o(\epsilon),
\end{align}
where $P_{ij} \defeq \Pi_{i}P\Pi_{j}$.
With this equation and the definition of the gradient, 
we can determine that $X$ in Eq.~\ref{eq:gradient_g_fun} is given by
\begin{align}
  X = \sum_{i}\frac{\lambda_{i}^{-\beta/\alpha}}{\alpha} P_{ii} +
       \sum_{i\not=j} \frac{\lambda_{j}^{1/\alpha}-\lambda_{i}^{1/\alpha}}
   {\lambda_{j}-\lambda_{i}} P_{ji}.
\end{align}
Note that the limit of
$(\lambda_{j}^{1/\alpha}-\lambda_{i}^{1/\alpha})/(\lambda_{j}-\lambda_{i})$
as $\lambda_{j}\rightarrow \lambda_{i}$ is
$\lambda_{i}^{-\beta/\alpha}/\alpha$, so the potentially problematic
term for near-degenerate eigenvalues can be stably computed.  The
simplest way to avoid precision problems with this expression is to
always collapse nearby eigenvalues of $\tau$, where $\lambda_{i}$ and
$\lambda_{j}$ should be considered nearby if
$\left|\lambda_{i}^{1/\alpha}-\lambda_{j}^{1/\alpha}\right|\leq
\sqrt{\delta}$ with $\delta$ the machine precision. This limits
numerical errors in the computation of $X$ to approximately
$\sqrt{\delta}$.  However, the numerical error has less effect on the
validity of the upper bound on $g$ if we replace $\tau$ by $\tilde
\tau$ where $\tilde \tau$ is $\tau$ with nearby eigenvalues collapsed
and rescaled to satisfy the constraint $\tr(\tilde \tau) = 1$ before
determining the upper bound from the maximum eigenvalue of the
gradient.

A protocol-style outline of \QEF optimization is given in
Protocol~\ref{prot:qefopt_outline}.

\vspace*{\baselineskip}
\begin{algorithm}[H]
  \caption{Schema for \QEF optimization for the $(k,2,2)$-Bell-test configuration with
  known input distribution $\mu(Z)$.}\label{prot:qefopt_outline}
  \Input{The targeted trial probability distribution $\nu(CZ)$ and an initial candidate
    $F_{0}(CZ)\geq 0$, $\sum_{cz}F_{0}(cz)=1$ with its $f_{0,\max}$.}  

  \tcp{The input distribution is $\mu(Z)=\nu(Z)$.}

  \tcp{Recommendation: $F_{0}(CZ)$  can be obtained by rescaling a good PEF with power $\beta$ at $\nu(CZ)$.}

  \Output{Best $F(CZ),f_{\max}(F(CZ))$ found and its log-prob rate $r_{F(CZ)}$.}

  \BlankLine

  Initialize an empty list $L$ of triples of candidates $F(CZ)$, $f_{\max}(F(CZ))$ and
  their log-prob rates $r_{F(CZ)}$\;

  \While{stopping criteria are not satisfied}{

    \tcp{Stopping criteria may be satisfied if resource limits are reached or log-prob rates
      are not improving sufficiently anymore.}

    \eIf{$L$ is empty}{

      Set $F(CZ)=F_{0}(CZ)$\;

    }{

      Determine the next candidate $F(CZ)\geq 0$, $\sum_{cz}F(cz)=1$ by
      using the triples in $L$ as a discrete sample of the \QEF landscape\;

    }

    Compute $f_{\max}(F(CZ))$
    \tcp*{Strategies are given in the text.}

    Compute $r_{F(CZ)}$ and add $(F(CZ), f_{\max}(F(CZ)), r_{F(CZ)})$ to $L$\;

  }

\end{algorithm}

\subsection{Optimal PEFs for Comparison}

In Protocol~\ref{prot:qefopt_outline}, we suggested starting \QEF
optimization with a good PEF previously determined for the
$(k,2,2)$-Bell-test configuration at trial probability distribution
$\nu(CZ)$. In Ref.~\cite{knill:qc2017a}, we gave algorithms for
determining such PEFs with respect to polytope envelopes of the
classical-side-information models. The simplest such polytope is the
non-signaling polytope, which can be restricted with Tsirelson's bounds
or other linear inequalities obtained from the hierarchy of
semidefinite programs in Ref.~\cite{navascues:2007}.
The schema for \QEF optimization suggests optimizing PEFs directly
using the reduction enabled by Thm.~\ref{thm:stdk22}.
The PEF optimization problem then reduces to an analog of the
\QEF optimization problem Prob.~\ref{prob:qefoptk22} 
as follows:

\begin{alignat}{3}
    \textrm{Maximize:\ }&\sum_{cz}\nu(cz)\log(F'(cz))-\log(f'_{\max})\notag\\
    \textrm{Variables:\ }& F'(CZ), f'_{\max}\notag\\
    \textrm{Subject to:\ }& F'(CZ)\geq 0,\sum_{cz}F'(cz)=1,\notag\\
    & 
    f'_{\max}\geq \sum_{cz}\mu(z)F'(cz)
      \tr(\tau P_{c|z;\theta})^{\alpha} \textrm{\ for all
    $\tau\geq 0$ with $\tr(\tau)=1$ and $\theta$}. \label{prob:pefoptk22}
\end{alignat}
The PEF constraint is obtained since
$\nu'(cz)=\mu(z)\tr(\tau P_{c|z;\theta})$ defines the trial probability
distribution for the model state under consideration.  The coefficient
of $F'(CZ)$ is $\nu'(cz)\nu'(c|z)^{\beta}$. The PEF
constraint on $f'_{\max}$ is convex in $\tau$, so we cannot use the
same argument to restrict $\tau$ to real density operators.  However,
convexity implies that $\tau$ can be restricted to pure states.  In
solving Prob.~\ref{prob:pefoptk22}, we can set $f'_{\max}$ to the
maximum value of $Q'_{\alpha}(F'(CZ),\theta,\tau) \defeq \sum_{cz}\mu(z)F'(cz)
\tr(\tau P_{c|z;\theta})^{\alpha}$ over $\tau$ and $\theta$.

\begin{lemma}
  In Prob.~\ref{prob:pefoptk22}, the operator $\tau$ may be restricted
  to pure states $\hat\psi$ with $\ket{\psi}$ real, and it suffices to
  consider $\theta$ with $\theta_{i}\in[0,\pi]$.
\end{lemma}

\begin{proof}
  We noted before the lemma that $\tau$ may be assumed to be pure.
  That we only need to consider $\theta_{i}\in[0,\pi]$ follows by the
  same argument as that used to prove the corresponding statement of
  Lem.~\ref{lem:qefoptk22_real}. Suppose $\tau$ is not real.  Then the
  conditional probabilities $\nu'(c|z)=\tr(\tau P_{c|z;\theta})$
  contributing to $Q'_{\alpha}$ satisfy 
  \begin{equation}
    \tr(\tau P_{c|z;\theta}) = \tr(\overline\tau P_{c|z;\theta})
    =\tr(\frac{1}{2}(\tau+\overline\tau)P_{c|z;\theta}),
  \end{equation}
  so the set of constraints on $f'_{\max}$ is unchanged if we restrict
  $\tau$ to real density matrices. Since real density matrices can
  be diagonalized over the reals, they are mixtures of real pure
  states and by convexity we can further restrict to real pure states.
\end{proof}

While we cannot take advantage of concavity to simplify maximizing
$Q'_{\alpha}(F'(CZ),\theta,\tau)$ with respect to $\tau$, we can take
advantage of convexity as before, but need to extend the strategy used
to optimize over $\theta$ to also include $\tau$.  With the notation
of Sect.~\ref{subsec:optschemas}, $Q'_{\alpha}(F'(CZ),\theta,\hat\psi)=\tilde 
Q_{\alpha}(F'(CZ),\hat\psi\otimes r(\theta))$ (see Eq.~\ref{eq:def_tilde_Q}), 
and $\tilde Q_{\alpha}(F'(CZ),u)$ is convex in $u$. 
If we can maximize over real $\ket{\psi}$ for given $\theta$,
then the schemas for maximizing over $\theta$ in
Sect.~\ref{subsec:optschemas} can also be used here. To perform the
maximization over $\ket{\psi}$, we describe an inner approximation
generalizing the one used to maximize over the $\theta_{i}\in[0,\pi]$.
The real pure states $\ket{\psi}$ can be identified with points in the
sphere $S_{2^{k}-1}$.  We reduce the inner-most maximization problem
to one of maximizing over $\ket{\psi}$ contained in convex cones
spanned by small sets of points on the sphere with large overlaps as
vectors.  Refinement involves subdividing the cones.  In the case of
$k=2$, we suggest sets of points defining the eight corners of a
cuboid.  For describing the technique, we fix $\theta$ and $F'(CZ)$, and omit 
them from expressions. 
In particular,
we abbreviate $\tilde Q_{\alpha}(F'(CZ),\hat\psi\otimes r(\theta))$ as 
$\tilde Q_{\alpha}(\hat\psi)$. The general goal is to upper bound a non-negative,
convex function $\tilde Q_{\alpha}(\hat\psi)$ homogeneous of degree $\alpha$ in
$\hat\psi$ over $\ket{\psi}\in S_{2^{k}-1}$, where the function
$\tilde Q_{\alpha}(\tau)$ is operator monotone in $\tau$.  We switch to
mathematical notation for real vectors, omitting kets and bras.

\begin{lemma}\label{lem:tightreals}
  Fix $\epsilon\in (0,1)$.
  Let  $I$ be a finite index set and for $i\in I$, let $x_{i}$ be real unit vectors with
  $x_{i}^{T}x_{j}\geq 1-\epsilon$ for all $j\in I$.  If $y$ is a unit vector
  that is a positive
  combination of the $x_{i}$, then 
  there is a convex combination
  $\rho$ of the $x_{i}x_{i}^{T}$ such that $yy^{T}\leq
  \rho/(1-\epsilon)$.
\end{lemma}

\begin{proof}
  Write $y$ as an explicit positive combination $y=\sum_{i}\lambda_{i}x_{i}$.
  Define 
  \begin{equation}
     \rho'=\sum_{i}\lambda_{i} \frac{x_{i}x_{i}^{T}}{x_{i}^{T}y}.
  \end{equation}
  Then for any real vector $z$, $z^{T}\rho' z\geq 0$, that is 
  $\rho'\geq 0$. Moreover, $\rho' y = \sum_{i}\lambda_{i}x_{i}=y$ 
  so that $y$ is a unit eigenvector with eigenvalue $1$ of $\rho'$.
  Therefore $\rho'\geq yy^{T}$.  
  Let $\lambda = \sum_{i}\lambda_{i}$.  Compute
  \begin{equation}
    x_{i}^{T}y = \sum_{j}\lambda_{j}x_{i}^{T}x_{j}\geq \sum_{j}\lambda_{j}(1-\epsilon)
    = \lambda (1-\epsilon),
  \end{equation}
  which gives
  \begin{equation}
    \tr(\rho') = \sum_{i}\frac{\lambda_{i}}{x_{i}^{T}y}
    \leq \sum_{i}\frac{\lambda_{i}}{\lambda(1-\epsilon)}
    = \frac{1}{1-\epsilon}.
  \end{equation}
  To complete the proof of the lemma, we set $\rho=\rho'/\tr(\rho')$.
\end{proof}

\begin{lemma}\label{lem:tightrealsbnd}
  Fix $\epsilon\in (0,1)$.
  Let $I$ be a finite index set and for $i\in I$, let $x_{i}$ be real unit vectors with
  $x_{i}^{T}x_{j}\geq 1-\epsilon$ for all $j\in I$.  Let $q_{i}=\tilde
  Q_{\alpha}(x_{i}x_{i}^{T})$. Then for all unit vectors $y$ in the positive convex cone
  generated by the $x_{i}$, $\tilde Q_{\alpha}(yy^{T})\leq \max_{i}q_{i}/(1-\epsilon)^{\alpha}$.
\end{lemma}

\begin{proof}
  Let $\rho=\sum_{i}\lambda_{i}x_{i}x_{i}^{T}$ be a
  convex combination of $x_{i}x_{i}^{T}$ with $y y^{T}\leq \rho/(1-\epsilon)$
  according to Lem.~\ref{lem:tightreals}.
  Then by monotonicity, homogeneity of degree $\alpha$ and convexity of $\tilde Q_{\alpha}$, we have
  \begin{align}
    \tilde Q_{\alpha}(yy^{T})&\leq \frac{1}{(1-\epsilon)^{\alpha}}\tilde Q_{\alpha}(\rho)\notag\\
     &\leq \frac{1}{(1-\epsilon)^{\alpha}}\sum_{i}\lambda_{i}\tilde Q_{\alpha}(x_{i}x_{i}^{T})\notag\\
     &\leq \frac{1}{(1-\epsilon)^{\alpha}}\max_{i}q_{i}.
  \end{align}
\end{proof}

We describe the $\hat\psi$-maximization strategy for the case $k=2$,
so that $\ket{\psi}\in S_{3}\subset\rls^{4}$.  We parametrize
$x\in S_{3}$ with angles $\phi_{1}\in[0,\pi/2]$, $\phi_{2}\in[0,2\pi]$
and $\phi_{3}\in[0,2\pi]$ according to
\begin{equation}
  x(\phi_{1},\phi_{2},\phi_{3})=
  \sin(\phi_{1})(\sin(\phi_{2}),\cos(\phi_{2}),0,0)^{T} 
           + \cos(\phi_{1})(0,0,\sin(\phi_{3}),\cos(\phi_{3}))^{T}.
\end{equation}
Because $x$ and $-x$ correspond to the same density matrix, we can
restrict $\phi_{2}$ to $[0,\pi]$.  To start the maximization, we can
choose points according to a cubical grid on $[0,\pi/2] \times [0,\pi]
\times [0,2\pi]$. For this, fix $m\ge 2$ and let $x_{i,j,k}=x(i
\pi/(2m),j\pi/(2m),k\pi/(2m))$ for $i\in\{0,\ldots,m\}$,
$j\in\{0,\ldots,2m\}$ and $k\in\{0,\ldots,4m\}$.  We identify a set of
facets, where each facet is defined by the eight corners of the cubes
in the cubical grid.  The facets may be identified with the sets of
points defined by
$f_{i,j,k}=\big\{x_{i+b_{1},j+b_{2},k+b_{3}}:b_{1},b_{2},b_{3}\in\{0,1\}\big\}$
for $i\in\{0,\ldots,m-1\}$, $j\in\{0,\ldots,2m-1\}$ and
$k\in\{0,\ldots,4m-1\}$. The positive convex cones generated by the
$f_{i,j,k}$ cover the half space of $\rls^{4}$ with non-negative first
coordinate. Thus we can first compute $\tilde Q_{\alpha}$ for all
$x_{i,j,k}x_{i,j,k}^{T}$ to get a lower bound and then compute an upper
bound for each facet according to Lem.~\ref{lem:tightrealsbnd}. Facets
whose upper bounds are below one of the values of $\tilde Q_{\alpha}$ obtained
can be abandoned. Facets for which the upper bound exceeds the maximum
value of $\tilde Q_{\alpha}$ over all vertices by more than the tolerance can
be refined by dividing the angle intervals determining the facet's
cube in half. This determines $19$ new points and $8$ subfacets.

The strategy of the previous paragraph can be combined with that for
maximizing over the $\theta$ by covering $S_{3}\times
[0,\pi]^{2}$ with an initial cubical grid and refining
cuboids as described in Sect.~\ref{subsec:optschemas}. In this
case the cuboids are five-dimensional.

\subsection{Examples}
\label{subsec:examples}

In Ref.~\cite{knill:qc2017a} we analyzed PEF performance on photonic
and atomic experimental data from published experiments, and in
Ref.~\cite{zhang:qc2018a} we determined PEF finite-data performance in
comparison to other methods, in particular trial-wise guessing
probability~\cite{pironio:qc2010a,fehr:qc2013a,pironio:qc2011a,acin:qc2012a,nieto-silleras:qc2014a,bancal:qc2014a,nieto-silleras:qc2016a}
and entropy accumulation~\cite{dupuis:qc2016a,arnon-friedman:qc2018a}.
Here we repeat some of these analyses and perform comparisons with QEFs instead.  For
this, we do not optimize QEFs. Instead, we compute optimal PEFs
$F'(CZ)$ for $C|Z$ with appropriate parameters, determine an upper
bound on $f_{\max}$ for each $F'(CZ)$ according to the methods in
Sect.~\ref{subsec:optschemas}, and obtain a QEF $F(CZ)$ by dividing the PEF by
$f_{\max}$, that is $F(CZ)=F'(CZ)/f_{\max}$. Throughout, we assume that the 
PEFs are for the classical trial model $\cT$ where the 
input distribution is uniform and the input-conditional output 
distributions satisfy non-signaling and Tsirelson's bounds, see Ref.~\cite{knill:qc2017a},
Sect.~VIII for details.  This classical trial model includes $\tr(\cC_{222}(CZ))$ 
with the uniform input distribution. In each case, we optimize the 
expected net log$_{2}$-prob for $\cT$ at a trial distribution $\nu(CZ)$, where 
the expected net log$_{2}$-prob is computed according to Eq.~\eqref{eq:netlogprob} with $\bar{\kappa}=1$.  When
obtaining a bound on $f_{\max}$, we stopped refining the evaluation
grid when the difference between lower and upper bounds on $f_{\max}$
was smaller than a stopping criterion determined by the
application. We set the stopping criterion so that the difference
between the upper and lower bounds on $f_{\max}$ has negligible impact on the
QEF's performance.  For all PEFs checked, we found that $f_{\max}$ was
indistinguishable from $1$ at numerical precision.  We conjecture that these 
PEFs are QEFs with the same power $\beta$ for $C|Z$ and $\cC_{222}(CZ)$ 
with the uniform input distribution.

We first reconsider the results from the first experiment to
demonstrate certified conditional min-entropy with a Bell
test~\cite{pironio:qc2010a}. The experiment established entangled
states of two ions in two separate ion-traps by entanglement swapping
with photons as intermediaries.  From the results of the experiment,
the authors claimed $42$ bits of conditional min-entropy at a
smoothness error bounded by $0.01$.  That the claim did not take into
account probability of success or quantum side information was
clarified in subsequent papers~\cite{fehr:qc2013a,pironio:qc2011a}. A
question is whether the experiment could have certified positive
conditional min-entropy with respect to quantum side information. To
answer this question we repeated the analysis of
Ref.~\cite{knill:qc2017a}, Sect.~VIII.E with modifications for quantum
side information. The experiment consisted of $3016$ trials, of which
we used the first $1000$ for training.  We optimized a PEF on the
training set by maximizing the expected net log$_{2}$-prob in the
remaining $2016$ trials, where the expected net log$_{2}$-prob is
computed according to Eq.~\eqref{eq:netlogprob} with $\bar{\kappa}=1$.
For this we also optimized the power $\beta$. The PEF is designed for
the trial model $\cT$. After training, we determined that $f_{\max}$
for the PEF found satisfies $f_{\max}\in [1,1+9.56\times 10^{-6}]$.
The upper bound was computed at numerical precision with Matlab, then
verified with Mathematica at a precision of $10^{-32}$.  We then
divided the PEF used by the upper bound on $f_{\max}$ to construct a
valid QEF.  After applying this QEF to the remaining $2016$ trials, we
found that it witnesses $127.86$ bits of quantum net log-prob at
smoothness error $\epsilon=0.01$ and presumed lower bound $\kappa=1$
of the success probability.  For the observed frequencies in this
experiment, entropy accumulation requires $54688$ trials to certify
any random bits at $\epsilon=0.01$ and $\kappa=1$ with the
min-tradeoff functions given in Ref.~\cite{arnon-friedman:qc2018a}.
Here, the assignment of $\kappa=1$ is purely formal for comparison
with respect to the soundness criteria implicit in
Ref.~\cite{pironio:qc2010a}. These soundness criteria are now
considered inadequate. With modern soundness criteria and at
$\epsilon=0.03$ and $\kappa=0.03$, the number of bits witnessed by the
QEF is $72.70$. This number is derived from the experimental QEF
value. In a protocol, the number of bits to be produced needs to be
decided before the experiment and would have been less to ensure
sufficiently high probability of success.

Next we compare the finite-data efficiency of QEFs to that of entropy
accumulation with the min-tradeoff functions given in the EAT references for computed trial results distributions with uniform
inputs. We consider the families of distributions,
$\cP_{E}=\{\nu_{E,\theta}\}_{0\leq\theta\leq\pi/4}$,
$\cP_{W}=\{\nu_{W,p}\}_{1/\sqrt{2}< p\leq 1}$ and $\cP_{P}=\{\nu_{P,\eta}\}_{2/3<\eta\leq 1}$
studied in Ref.~\cite{zhang:qc2018a}.  They are defined as follows:
For the first and third, the two-party device to be measured is
initially in the unbalanced Bell state defined by
$\ket{\Psi_{\theta}}=\cos(\theta)\ket{00}+\sin(\theta)\ket{11}$.  For
the second, the initial state is the Werner state
$p\dyad{\Psi_{\pi/4}}+(1-p)\one/4$.  To compute $\nu_{E,\theta}$ and
$\nu_{W,p}$, the input-dependent measurements are chosen so as to
maximize the expected CHSH value $\hat I$~\cite{clauser:qc1969a}
defined by $\hat I = \Exp(4(1-2XY)(-1)^{A+B})$ with $A, B, X, Y\in \{0,1\}$, where 
$X$ and $Y$ are the inputs and $A$ and $B$ are the 
outputs of Alice and Bob, respectively. For local realistic
distributions, $\hat I\leq 2$ and for quantum distributions, $\hat
I\leq 2\sqrt{2}$. To compute $\nu_{P,\eta}$, we use detectors 
of efficiency $\eta\in(2/3,1]$ and choose both 
the state $\ket{\Psi_{\theta}}$ and the input-dependent measurements 
such that the statistical strength for rejecting 
local realism~\cite{vanDam:2005, zhang:2010} is maximized.  
The value of $\hat I$ for each family is monotonic in the parameters. That is, for
$\nu_{E,\theta}$, $\hat I$ increases with $\theta$ for
$\theta\in[0,\pi/4]$, for $\nu_{W,p}$ it increases with $p$ for
$p\in(1/\sqrt{2},1]$, and for $\nu_{P,\eta}$ it increases with
$\eta\in(2/3,1]$. The family $\cP_{E}$ and $\cP_{W}$ represent the
best and worst cases for conditional min-entropy as a function of $\hat
I$, while $\cP_{P}$ is experimentally relevant, particularly for
photonic experiments.

Entropy accumulation is formulated to yield smooth min-entropy
estimates and we compare performances accordingly.  Specifically, we
consider protocols for certifying $\epsilon$-smooth min-entropy
conditional on success that satisfy the following: For specified
values of $\sigma$, $\epsilon$ and $\kappa$, for all states in the
model, if the probability of success is at least $\kappa$, then the
$\epsilon$-smooth min-entropy of the output conditional on success is
at least $\sigma$. A QEF protocol is determined by the application of
Thm.~\ref{thm:bnds_from_qef} to all states in the model for which the
probability of success is at least $\kappa$, and where $p$ and
$\delta$ satisfy $-\log_{2}(p/\kappa^{\alpha/\beta})\geq\sigma$ and
$\delta=\epsilon^{2}/2$.  Here, we refer to the quantity
$\log_{2}(F(CZ))/\beta+\log_{2}(\epsilon^2/2)/\beta+\alpha\log_{2}(\kappa)/\beta$
in such a protocol as its min-entropy estimate.  We remark that for
randomness generation, the quantum net log-prob has better dependence
on the probability of success parameter.  Both entropy accumulation
and QEFs give valid estimates regardless of the experimental
distributions provided that the model is satisfied.  But the
performances are determined by the actual trial distributions. EAT
protocols also have an associated min-entropy estimate determined from
an affine min-tradeoff function.

We assume that for the ``honest'' devices, namely the devices as
designed, the trials are i.i.d. with distribution $\nu$ in one of the
families $\cP_{E}$, $\cP_{W}$ and $\cP_{P}$. We are interested in the
minimum number of trials required for a protocol with parameters
$\sigma$, $\epsilon$ and $\kappa$ as described in the previous
paragraph. To be useful, such a protocol should have a large
probability of success greater than $\kappa$ for honest devices. For
QEFs, the probability of success is determined by the distribution of
the min-entropy estimate, which is obtained from a sum of
i.i.d. random variables for honest devices. In the absence of specific
information of the QEF defining these random variable, the probability
of success cannot be estimated. Instead, we set $\sigma$ to the
expectation of the min-entropy estimate.  Generically, this implies an
honest probability of success near $1/2$, at least for large enough
$n$.  For the EAT, we use the same strategy, setting $\sigma$ to the
expectation of the EAT min-entropy estimate.  For both QEFs and the
EAT, the probability of success can be made close to $1$ by reducing
$\sigma$, provided the number of trials is large enough.  For a
representative comparison, we formally set $\epsilon=10^{-6}$ and
$\kappa=1$ to determine the minimum number of trials required for
positive $\sigma$. The assignment $\kappa=1$ is singular but chosen as
a convenient reference point for values of $\kappa$ that are not
small.  The improvements obtained by QEFs are as significant for all
meaningful assignments with the same value for the product
$\epsilon\kappa$.

First consider \QEFs. Suppose that $F(CZ)$ is a trial-wise \QEF with power
$\beta$ and log$_2$-prob rate g.  According to
Thm.~\ref{thm:bnds_from_qef}, the expected $\epsilon$-smooth conditional min-entropy
estimate in bits for $n$ trials is
\begin{equation}
  ng+\frac{\log_{2}(\epsilon^{2}/2)}{\beta}
  +\frac{\alpha\log_{2}(\kappa)}{\beta},
\end{equation}
so the minimum number of trials required for positive $\epsilon$-smooth conditional min-entropy
is
\begin{equation}
  n_{\min,\textrm{QEF}}(F(CZ);\beta,\epsilon,\kappa) = \frac{1}{g \beta}|\log_{2}(\epsilon^{2}\kappa^{\alpha}/2)|.
  \label{eq:qef_mintrials}
\end{equation}
For simplicity we do not require that the number of trials is an
integer.
Except for the replacement of the error bound $\epsilon$ by
$\epsilon^{2}/2$, this agrees with the expressions in
Ref.~\cite{zhang:qc2018a}. 

For entropy accumulation, we can apply Thm.~\ref{thm:eat} with an
entropy estimator, where the entropy estimator can be derived either
from the QEF $F(CZ)$, or from the min-tradeoff function given in
Ref.~\cite{arnon-friedman:qc2018a}. With the QEF, from
Thm.~\ref{thm:eat} in terms of bits, with $h$ replaced by the
log$_{2}$-prob rate $g$ and $k_{\infty}=\lceil\max
|\log_{2}(F(CZ))/\beta|\rceil$, the expected $\epsilon$-smooth 
conditional min-entropy estimate is
\begin{equation}
  ng -2\left(\log_{2}(9)+\lceil k_{\infty}\rceil\right)
  \sqrt{1-2\log_{2}(\epsilon\kappa)}\sqrt{n},
\end{equation}
which implies that the minimum number of trials is
\begin{equation}
  n_{\min,\textrm{EAT}}(F(CZ);\beta,\epsilon,\kappa) = 
   \frac{4}{g^{2}}\left(\log_{2}(9)+\lceil k_{\infty}\rceil\right)^{2}
   (1-2\log_{2}(\epsilon\kappa)).
\end{equation}
We write $n_{\min,\textrm{EAT}}(T;\epsilon,\kappa)$ for the same 
quantity but computed for the min-tradeoff function $T$ given in
Ref.~\cite{arnon-friedman:qc2018a}.  An explicit but involved
expression for $n_{\min,\textrm{EAT}}(T;\epsilon,\kappa)$ is given in
Ref.~\cite{zhang:qc2018a}, which we do not repeat here. Its evaluation
involves optimizing over additional parameters.

For the comparison at a given distribution $\nu$, we first minimize
the expression for
$n_{\min,\textrm{QEF}}(F'(CZ);\beta,\epsilon,\kappa)$ over $\beta$ and
PEFs $F'(CZ)$ for $\cT$. The minimum found is witnessed by PEF
$F'(CZ)$ and $\beta$. We then compute $f_{\max}$ for $F'(CZ)$, which
determines a valid \QEF $F(CZ)=F'(CZ)/f_{\max}$ with the same power $\beta$.
This determines $n_{\nu, \textrm{QEF}} \defeq
n_{\min,\textrm{QEF}}(F(CZ);\beta,\epsilon,\kappa)$.  We then obtain
$n_{\nu,F,\textrm{EAT}} \doteq
n_{\min,\textrm{EAT}}(F(CZ);\beta,\epsilon,\kappa)$ according to the
above formula and $n_{\nu,T,\textrm{EAT}}\doteq
n_{\min,\textrm{EAT}}(T;\epsilon,\kappa)$ according to the
instructions in Ref.~\cite{zhang:qc2018a}.  The QEF advantages are
determined by the ratios
$f_{\nu,F}=n_{\nu,F,\textrm{EAT}}/n_{\nu,\textrm{QEF}}$ and
$f_{\nu,T}=n_{\nu,T,\textrm{EAT}}/n_{\nu,\textrm{QEF}}$.  For the
distributions $\nu_{W,p}$, the advantage $f_{\nu,T}$ depends weakly on
$\hat I$: $f_{\nu_{W,p},T}$ increases from $36.9$ at $\hat I =
2.008$ to $38.2$ at $\hat I=2\sqrt{2}$.
For the other
distributions, $f_{\nu,T}$ can be much larger, particularly at $\hat
I$ near $2$, as shown in Fig.~\ref{fig:finite_data_comparison}.  We
also find that $f_{\nu,F}$ is systematically larger than $f_{\nu,T}$
by factors of at least two near maximum $\hat I$ and growing
substantially toward minimum $\hat I$. Thus, determining the entropy
estimator from the \QEFs found and applying the EAT performs worse
than applying the EAT with the min-tradeoff function from
Ref.~\cite{arnon-friedman:qc2018a}.  This suggests that the
  problem of optimizing QEFs and that of optimizing entropy estimators
  or min-tradeoff functions are not well matched.  With entropy estimators
determined from \QEFs optimized for powers near zero, the EAT
performance improves substantially. In some cases, the performance is
better than the EAT with the min-tradeoff function
given in Ref.~\cite{arnon-friedman:qc2018a}. We remark that this
comparison does not take advantage of the improvements to the EAT
implied by Thm.~\ref{thm:ee_to_qef}.

\begin{figure}
  \begin{center}
    \resizebox{5in}{!}{\includegraphics[viewport=1.25in 3.25in 7in 7.5in,clip=true]{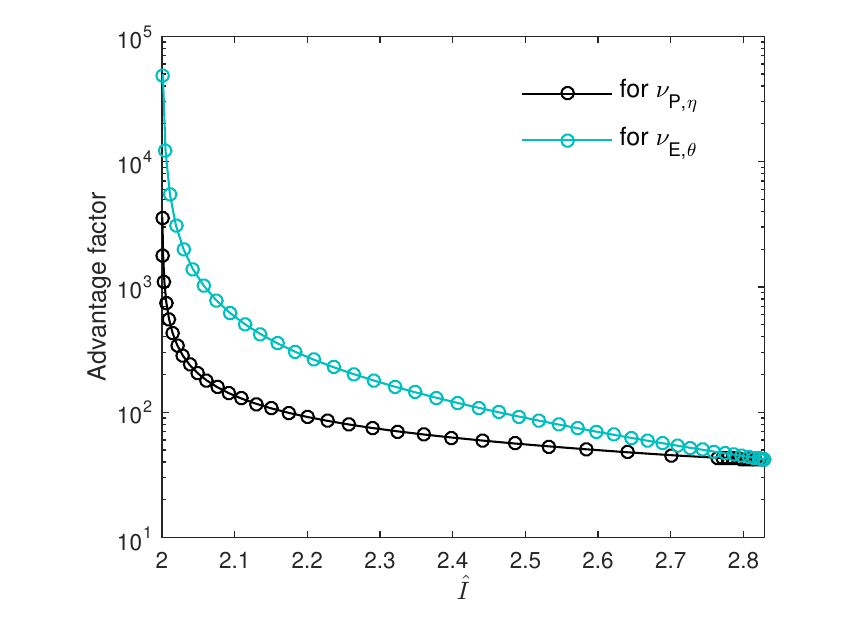}}
 \end{center}
 \caption{ QEF advantage factors for $\cP_{E}$ and $\cP_{P}$ as a
   function of $\hat I$.  Shown are values for $f_{\nu_{E,\theta},T}$
   and $f_{\nu_{P,\eta},T}$.  We verified that the quantity 
   $f_{\max}$ is indistinguishable from $1$ at high precision for each of the points indicated by 
   open circles.    
   \label{fig:finite_data_comparison}
 }
\end{figure}

For the last example, we consider the problem of producing $512$ bits
at smoothness error $\epsilon=2^{-64}$ and
probability of success parameter $\kappa=2^{-64}$ with trials whose
results distribution matches that observed in the photonic
loophole-free randomness generation experiment reported in
Ref.~\cite{bierhorst:qc2018a}.  For this, we do not consider the
overhead of extracting the random bits and ask for the minimum number
of trials for which $512$ bits of smooth conditional min-entropy can
be certified at the given $\epsilon,\kappa$.  We optimized the minimum
number of trials required according to Eq.~\ref{eq:qef_mintrials} over
PEFs and powers, assuming that the PEFs are QEFs.  We confirmed that
the best PEF found has $f_{\max}\leq 1+9.88\times 10^{-9}$, which we
verified with Mathematica at a precision of $10^{-32}$. The QEF thus
found requires $6.97\times 10^{7}$ trials on average.  For entropy
accumulation, $2.89\times 10^{11}$ trials are required, as reported in
Ref.~\cite{zhang:qc2018a}. Given the trial rate in the experiment of
Ref.~\cite{bierhorst:qc2018a}, this would require $11.62$ minutes of
experimental time with QEFs, and $802.1$ hours with entropy
accumulation.

\begin{acknowledgments}
  We thank Carl Miller and Peter Bierhorst for stimulating
  discussions, help with moving this project forward and editorial
  help.  This work includes contributions of the National Institute of
  Standards and Technology, which are not subject to U.S. copyright.
  The use of trade names is for informational purposes only and does
  not imply endorsement or recommendation by the U.S. government.
\end{acknowledgments}

\bibliography{qefs}

{\small
  \begin{description}[\compact]
  \item[]\textbf{arXiv revision notes:}
  \item[]
    \begin{description}[\compact]
    \item[V1.] Original submission.
    \item[V2.] First revision.
    \item[]
      \begin{description}[\compact]
      \item[1.] Clarified soundness definitions in Sect.~\ref{subsec:soundness}
        and corrected the discussion of dependence and extension to
        initial classical variables.
      \item[2.] Added remarks on and references to Dupuis and Fawzi's
        second-order improvement of the EAT~\cite{dupuis:qc2018a}.
      \item[3.] Clarified definition of expected quantum net log-prob.
      \item[4.] Miscellaneous clarifications and minor corrections.
      \end{description}
    \item[V3.] Second revision.
    \item[]
      \begin{description}[\compact]
      \item[1.] Clarified the comparison to EAT in Sect.~\ref{subsec:examples}
        and improved the treatment and discussion of the probability of success
        parameter $\kappa$.
      \item[2.] Added references to our published papers based on this work.
        Ref.~\cite{zhang_y:qc2020a} covers the basic theory of QEFs
        for randomness generation and Ref.~\cite{zhang_y:qc2018a} describes an
        experimental implementation for repeated and low-latency production
        of blocks of $512$ random bits.
      \end{description}
    \item[V4.] Third revision.
    \item[]
      \begin{description}[\compact]
      \item[1.] Corrected the definition of $n$ and the argument to the
        extractor in Protocol~\ref{prot:condimplicit}. This protocol
        does not require the Markov chain condition and works best if the settings distribution
        is explicitly generated from a random source that is designed to be uniform.
        Uniformity is not required for validity of the protocol.
      \item[2.] Fixed the definition of conditional states, it was
        unintentionally much too restrictive.
      \item[3.] Fixed a mistake in Thm.~\ref{thm:ee_conduniformbnd}: \(\pCP\) closure
        needs to be assumed rather than added.
      \item[4.] Edited Sects.~\ref{sec:qefsandee} and~\ref{sec:qef_constructs} for
        clarity and typos.
      \end{description}
    \end{description}
  \end{description}
}
\end{document}

%% file: qefs_defs.tex
\providecommand{\ignore}[1]{}

\newif\ifcmnt
\cmnttrue
\ifdefined\cmntsoff\cmntfalse\fi

\ifcmnt
    \providecommand{\aucmnt}[1]{#1}

    \providecommand{\Pcolor}{\color{gray}}
\else
    \providecommand{\aucmnt}[1]{}

    \providecommand{\Pcolor}{}
\fi

\newcommand{\Pc}[1]{\aucmnt{{\Pcolor \textbf{[}Permanent comment: {\color{gray}#1}\textbf{]}}}}
\newcommand{\Ac}[1]{\aucmnt{{\Pcolor \textbf{[}Authors' record: {\color{gray}#1}\textbf{]}}}}

\newcommand{\Prob}{\mathbb{P}}
\newcommand{\Exp}{\mathbb{E}}

\newcommand{\TV}{\mathrm{TV}}
\newcommand{\PD}{\mathrm{PD}}

\newcommand{\Unif}{\mathrm{Unif}}
\newcommand{\knuth}[1]{\left\llbracket #1 \right\rrbracket}

\newcommand{\eps}{\epsilon}
\newcommand{\epse}{\epsilon_{h}}
\newcommand{\epsx}{\epsilon_{x}}

\newcommand{\rls}{\mathbb{R}}
\newcommand{\cmplx}{\mathbb{C}}

\newcommand{\ints}{\mathbb{Z}}
\newcommand{\nats}{\mathbb{N}}

\newcommand{\vdim}{\mathrm{dim}}

\newcommand{\conc}{%
  \mathord{
    \mathchoice
    {\raisebox{1ex}{\scalebox{.7}{$\frown$}}}
    {\raisebox{1ex}{\scalebox{.7}{$\frown$}}}
    {\raisebox{.7ex}{\scalebox{.5}{$\frown$}}}
    {\raisebox{.7ex}{\scalebox{.5}{$\frown$}}}
  }
}

\newcommand{\defeq}{\doteq}
\newcommand{\Rng}{\mathrm{Rng}}
\newcommand{\Supp}{\mathrm{Supp}}
\newcommand{\Cvx}{\mathrm{Cvx}}
\newcommand{\Cone}{\mathrm{Cone}}

\newcommand{\Xtrm}{\mathrm{Extr}}
\newcommand{\CPTP}{\mathrm{CPTP}}
\newcommand{\pCP}{\mathrm{pCP}}
\newcommand{\CP}{\mathrm{CP}}

\newcommand{\pospart}[1]{\left[#1\right]_+}

\newcommand{\suppproj}[1]{\knuth{#1>0}}
\newcommand{\trdist}[2]{\TV(#1,#2)}
\newcommand{\purdist}[2]{\PD(#1,#2)}

\newcommand{\Rpow}[3]{\cR_{#1}\left(#2\middle|#3\right)}
\newcommand{\hatRpow}[3]{\hat\cR_{#1}\left(#2\middle|#3\right)}

\newcommand{\Ppow}[3]{\cP_{#1}\left(#2\middle|#3\right)}
\newcommand{\hatPpow}[3]{\hat\cP_{#1}\left(#2\middle|#3\right)}

\newcommand{\tildeDrel}[3]{\tilde D_{#1}\left({#2{\parallel} #3}\right)}
\newcommand{\Spec}{\mathrm{Spec}}
\newcommand{\QEF}{$\mathrm{QEF}$\xspace}
\newcommand{\QEFP}{$\mathrm{QEFP}$\xspace}
\newcommand{\QEFs}{$\mathrm{QEF}$s\xspace}
\newcommand{\QEFPs}{$\mathrm{QEFP}$s\xspace}

\newcommand{\Sfnt}[1]{\mathbf{#1}} 
\newcommand{\Pfnt}[1]{\mathsf{#1}} 

\newcommand{\cA}{\mathcal{A}}
\newcommand{\cB}{\mathcal{B}}
\newcommand{\cC}{\mathcal{C}}
\newcommand{\cD}{\mathcal{D}}
\newcommand{\cE}{\mathcal{E}}
\newcommand{\cF}{\mathcal{F}}
\newcommand{\cG}{\mathcal{G}}
\newcommand{\cH}{\mathcal{H}}

\newcommand{\cK}{\mathcal{K}}
\newcommand{\cL}{\mathcal{L}}
\newcommand{\cM}{\mathcal{M}}
\newcommand{\cN}{\mathcal{N}}
\newcommand{\cO}{\mathcal{O}}
\newcommand{\cP}{\mathcal{P}}

\newcommand{\cR}{\mathcal{R}}
\newcommand{\cS}{\mathcal{S}}
\newcommand{\cT}{\mathcal{T}}
\newcommand{\cU}{\mathcal{U}}
\newcommand{\cV}{\mathcal{V}}
\newcommand{\cW}{\mathcal{W}}
\newcommand{\cX}{\mathcal{X}}
\newcommand{\cY}{\mathcal{Y}}
\newcommand{\cZ}{\mathcal{Z}}

\newcommand{\fkC}{\mathfrak{C}}

\newcommand{\fkP}{\mathfrak{P}}

\newcommand{\one}{\mathds{1}}

\numberwithin{equation}{section}
\newtheorem{theorem}{Theorem}[section]
\newtheorem*{theorem*}{Theorem}
\newtheorem{lemma}[theorem]{Lemma}
\newtheorem*{lemma*}{Lemma}
\newtheorem{corollary}[theorem]{Corollary}
\newtheorem{definition}[theorem]{Definition}



%% file: qefs.bbl
\begin{thebibliography}{49}%
\makeatletter
\providecommand \@ifxundefined [1]{%
 \@ifx{#1\undefined}
}%
\providecommand \@ifnum [1]{%
 \ifnum #1\expandafter \@firstoftwo
 \else \expandafter \@secondoftwo
 \fi
}%
\providecommand \@ifx [1]{%
 \ifx #1\expandafter \@firstoftwo
 \else \expandafter \@secondoftwo
 \fi
}%
\providecommand \natexlab [1]{#1}%
\providecommand \enquote  [1]{``#1''}%
\providecommand \bibnamefont  [1]{#1}%
\providecommand \bibfnamefont [1]{#1}%
\providecommand \citenamefont [1]{#1}%
\providecommand \href@noop [0]{\@secondoftwo}%
\providecommand \href [0]{\begingroup \@sanitize@url \@href}%
\providecommand \@href[1]{\@@startlink{#1}\@@href}%
\providecommand \@@href[1]{\endgroup#1\@@endlink}%
\providecommand \@sanitize@url [0]{\catcode `\\12\catcode `\$12\catcode
  `\&12\catcode `\#12\catcode `\^12\catcode `\_12\catcode `\%12\relax}%
\providecommand \@@startlink[1]{}%
\providecommand \@@endlink[0]{}%
\providecommand \url  [0]{\begingroup\@sanitize@url \@url }%
\providecommand \@url [1]{\endgroup\@href {#1}{\urlprefix }}%
\providecommand \urlprefix  [0]{URL }%
\providecommand \Eprint [0]{\href }%
\providecommand \doibase [0]{https://doi.org/}%
\providecommand \selectlanguage [0]{\@gobble}%
\providecommand \bibinfo  [0]{\@secondoftwo}%
\providecommand \bibfield  [0]{\@secondoftwo}%
\providecommand \translation [1]{[#1]}%
\providecommand \BibitemOpen [0]{}%
\providecommand \bibitemStop [0]{}%
\providecommand \bibitemNoStop [0]{.\EOS\space}%
\providecommand \EOS [0]{\spacefactor3000\relax}%
\providecommand \BibitemShut  [1]{\csname bibitem#1\endcsname}%
\let\auto@bib@innerbib\@empty
\bibitem [{\citenamefont {Knill}\ \emph {et~al.}(2017)\citenamefont {Knill},
  \citenamefont {Zhang},\ and\ \citenamefont {Bierhorst}}]{knill:qc2017a}%
  \BibitemOpen
  \bibfield  {author} {\bibinfo {author} {\bibfnamefont {E.}~\bibnamefont
  {Knill}}, \bibinfo {author} {\bibfnamefont {Y.}~\bibnamefont {Zhang}},\ and\
  \bibinfo {author} {\bibfnamefont {P.}~\bibnamefont {Bierhorst}},\ }\bibfield
  {title} {\bibinfo {title} {Quantum randomness from probability estimation
  with classical side information}} (\bibinfo {year} {2017}),\ \bibinfo {note}
  {arXiv:1709.06159}\BibitemShut {NoStop}%
\bibitem [{\citenamefont {Miller}\ and\ \citenamefont
  {Shi}(2014)}]{miller_c:qc2014a}%
  \BibitemOpen
  \bibfield  {author} {\bibinfo {author} {\bibfnamefont {C.~A.}\ \bibnamefont
  {Miller}}\ and\ \bibinfo {author} {\bibfnamefont {Y.}~\bibnamefont {Shi}},\
  }\bibfield  {title} {\bibinfo {title} {Robust protocols for securely
  expanding randomness and distributing keys using untrusted quantum devices},\
  }in\ \href {https://doi.org/10.1145/2591796.2591843} {\emph {\bibinfo
  {booktitle} {STOC '14 Proceedings of the 46th Annual ACM Symposium on Theory
  of Computing}}}\ (\bibinfo {year} {2014})\ pp.\ \bibinfo {pages}
  {417--426}\BibitemShut {NoStop}%
\bibitem [{\citenamefont {Miller}\ and\ \citenamefont
  {Shi}(2016)}]{miller_c:qc2014b}%
  \BibitemOpen
  \bibfield  {author} {\bibinfo {author} {\bibfnamefont {C.~A.}\ \bibnamefont
  {Miller}}\ and\ \bibinfo {author} {\bibfnamefont {Y.}~\bibnamefont {Shi}},\
  }\bibfield  {title} {\bibinfo {title} {Universal security for randomness
  expansion from the spot-checking protocol},\ }\href@noop {} {\bibfield
  {journal} {\bibinfo  {journal} {J.~ACM}\ }\textbf {\bibinfo {volume} {63}},\
  \bibinfo {pages} {Art. No. 33} (\bibinfo {year} {2016})},\ \bibinfo {note}
  {arXiv:1411.6608}\BibitemShut {NoStop}%
\bibitem [{\citenamefont {Dupuis}\ \emph {et~al.}(2016)\citenamefont {Dupuis},
  \citenamefont {Fawzi},\ and\ \citenamefont {Renner}}]{dupuis:qc2016a}%
  \BibitemOpen
  \bibfield  {author} {\bibinfo {author} {\bibfnamefont {F.}~\bibnamefont
  {Dupuis}}, \bibinfo {author} {\bibfnamefont {O.}~\bibnamefont {Fawzi}},\ and\
  \bibinfo {author} {\bibfnamefont {R.}~\bibnamefont {Renner}},\ }\bibfield
  {title} {\bibinfo {title} {Entropy accumulation}} (\bibinfo {year} {2016}),\
  \bibinfo {note} {arXiv:1607.01796 (specific citations are for version
  1)}\BibitemShut {NoStop}%
\bibitem [{\citenamefont {Arnon-Friedman}\ \emph {et~al.}(2018)\citenamefont
  {Arnon-Friedman}, \citenamefont {Dupuis}, \citenamefont {Fawzi},
  \citenamefont {Renner},\ and\ \citenamefont
  {Vidick}}]{arnon-friedman:qc2018a}%
  \BibitemOpen
  \bibfield  {author} {\bibinfo {author} {\bibfnamefont {R.}~\bibnamefont
  {Arnon-Friedman}}, \bibinfo {author} {\bibfnamefont {F.}~\bibnamefont
  {Dupuis}}, \bibinfo {author} {\bibfnamefont {O.}~\bibnamefont {Fawzi}},
  \bibinfo {author} {\bibfnamefont {R.}~\bibnamefont {Renner}},\ and\ \bibinfo
  {author} {\bibfnamefont {T.}~\bibnamefont {Vidick}},\ }\bibfield  {title}
  {\bibinfo {title} {Practical device-independent quantum cryptography via
  entropy accumulation},\ }\href {https://doi.org/10.1038/s41467-017-02307-4}
  {\bibfield  {journal} {\bibinfo  {journal} {Nature Communications}\ }\textbf
  {\bibinfo {volume} {9}},\ \bibinfo {pages} {459} (\bibinfo {year}
  {2018})}\BibitemShut {NoStop}%
\bibitem [{\citenamefont {Dupuis}\ and\ \citenamefont
  {Fawzi}(2018)}]{dupuis:qc2018a}%
  \BibitemOpen
  \bibfield  {author} {\bibinfo {author} {\bibfnamefont {F.}~\bibnamefont
  {Dupuis}}\ and\ \bibinfo {author} {\bibfnamefont {O.}~\bibnamefont {Fawzi}},\
  }\bibfield  {title} {\bibinfo {title} {Entropy accumulation with improved
  second-order}} (\bibinfo {year} {2018}),\ \bibinfo {note}
  {arXiv:1805.11652}\BibitemShut {NoStop}%
\bibitem [{\citenamefont {Tomamichel}\ \emph {et~al.}(2009)\citenamefont
  {Tomamichel}, \citenamefont {Colbeck},\ and\ \citenamefont
  {Renner}}]{tomamichel:qc2009a}%
  \BibitemOpen
  \bibfield  {author} {\bibinfo {author} {\bibfnamefont {M.}~\bibnamefont
  {Tomamichel}}, \bibinfo {author} {\bibfnamefont {R.}~\bibnamefont
  {Colbeck}},\ and\ \bibinfo {author} {\bibfnamefont {R.}~\bibnamefont
  {Renner}},\ }\bibfield  {title} {\bibinfo {title} {A fully quantum asymptotic
  equipartition property},\ }\href {https://doi.org/10.1109/TIT.2009.2032797}
  {\bibfield  {journal} {\bibinfo  {journal} {IEEE Trans. Inf. Theory}\
  }\textbf {\bibinfo {volume} {55}},\ \bibinfo {pages} {5840} (\bibinfo {year}
  {2009})}\BibitemShut {NoStop}%
\bibitem [{\citenamefont {Fischer}(2011)}]{fischer:qc2011a}%
  \BibitemOpen
  \bibfield  {author} {\bibinfo {author} {\bibfnamefont {M.~J.}\ \bibnamefont
  {Fischer}},\ }\bibfield  {title} {\bibinfo {title} {A public randomness
  service},\ }in\ \href@noop {} {\emph {\bibinfo {booktitle} {SECRYPT 2011}}}\
  (\bibinfo {year} {2011})\ pp.\ \bibinfo {pages} {434--438}\BibitemShut
  {NoStop}%
\bibitem [{\citenamefont {Zhang}\ \emph
  {et~al.}(2020{\natexlab{a}})\citenamefont {Zhang}, \citenamefont {Fu},\ and\
  \citenamefont {Knill}}]{zhang_y:qc2020a}%
  \BibitemOpen
  \bibfield  {author} {\bibinfo {author} {\bibfnamefont {Y.}~\bibnamefont
  {Zhang}}, \bibinfo {author} {\bibfnamefont {H.}~\bibnamefont {Fu}},\ and\
  \bibinfo {author} {\bibfnamefont {E.}~\bibnamefont {Knill}},\ }\bibfield
  {title} {\bibinfo {title} {Efficient randomness certification by quantum
  probability estimation},\ }\href
  {https://doi.org/10.1103/PhysRevResearch.2.013016} {\bibfield  {journal}
  {\bibinfo  {journal} {Physical Review Research}\ }\textbf {\bibinfo {volume}
  {2}},\ \bibinfo {pages} {013016} (\bibinfo {year}
  {2020}{\natexlab{a}})}\BibitemShut {NoStop}%
\bibitem [{\citenamefont {Zhang}\ \emph
  {et~al.}(2020{\natexlab{b}})\citenamefont {Zhang}, \citenamefont {Shalm},
  \citenamefont {Bienfang}, \citenamefont {Stevens}, \citenamefont {Mazurek},
  \citenamefont {Nam}, \citenamefont {Abellán}, \citenamefont {Amaya},
  \citenamefont {Mitchell}, \citenamefont {Fu}, \citenamefont {Miller},
  \citenamefont {Mink},\ and\ \citenamefont {Knill}}]{zhang_y:qc2018a}%
  \BibitemOpen
  \bibfield  {author} {\bibinfo {author} {\bibfnamefont {Y.}~\bibnamefont
  {Zhang}}, \bibinfo {author} {\bibfnamefont {L.~K.}\ \bibnamefont {Shalm}},
  \bibinfo {author} {\bibfnamefont {J.~C.}\ \bibnamefont {Bienfang}}, \bibinfo
  {author} {\bibfnamefont {M.~J.}\ \bibnamefont {Stevens}}, \bibinfo {author}
  {\bibfnamefont {M.~D.}\ \bibnamefont {Mazurek}}, \bibinfo {author}
  {\bibfnamefont {S.~W.}\ \bibnamefont {Nam}}, \bibinfo {author} {\bibfnamefont
  {C.}~\bibnamefont {Abellán}}, \bibinfo {author} {\bibfnamefont
  {W.}~\bibnamefont {Amaya}}, \bibinfo {author} {\bibfnamefont {M.~W.}\
  \bibnamefont {Mitchell}}, \bibinfo {author} {\bibfnamefont {H.}~\bibnamefont
  {Fu}}, \bibinfo {author} {\bibfnamefont {C.~A.}\ \bibnamefont {Miller}},
  \bibinfo {author} {\bibfnamefont {A.}~\bibnamefont {Mink}},\ and\ \bibinfo
  {author} {\bibfnamefont {E.}~\bibnamefont {Knill}},\ }\bibfield  {title}
  {\bibinfo {title} {Experimental low-latency device-independent quantum
  randomness},\ }\href {https://doi.org/10.1103/PhysRevLett.124.010505}
  {\bibfield  {journal} {\bibinfo  {journal} {Phys. Rev. Lett.}\ }\textbf
  {\bibinfo {volume} {124}},\ \bibinfo {pages} {010505} (\bibinfo {year}
  {2020}{\natexlab{b}})},\ \bibinfo {note} {arXiv:1812.07786}\BibitemShut
  {NoStop}%
\bibitem [{\citenamefont {Mauerer}\ \emph {et~al.}(2012)\citenamefont
  {Mauerer}, \citenamefont {Portmann},\ and\ \citenamefont
  {Scholz}}]{mauerer:2012}%
  \BibitemOpen
  \bibfield  {author} {\bibinfo {author} {\bibfnamefont {W.}~\bibnamefont
  {Mauerer}}, \bibinfo {author} {\bibfnamefont {C.}~\bibnamefont {Portmann}},\
  and\ \bibinfo {author} {\bibfnamefont {V.~B.}\ \bibnamefont {Scholz}},\
  }\bibfield  {title} {\bibinfo {title} {A modular framework for randomness
  extraction based on {T}revisan's construction}} (\bibinfo {year} {2012}),\
  \bibinfo {note} {arXiv:1212.0520, code available on
  \texttt{github}.}\BibitemShut {Stop}%
\bibitem [{\citenamefont {Bierhorst}\ \emph {et~al.}(2017)\citenamefont
  {Bierhorst}, \citenamefont {Knill}, \citenamefont {Glancy}, \citenamefont
  {Mink}, \citenamefont {Jordan}, \citenamefont {Rommal}, \citenamefont {Liu},
  \citenamefont {Christensen}, \citenamefont {Nam},\ and\ \citenamefont
  {Shalm}}]{bierhorst:qc2017a}%
  \BibitemOpen
  \bibfield  {author} {\bibinfo {author} {\bibfnamefont {P.}~\bibnamefont
  {Bierhorst}}, \bibinfo {author} {\bibfnamefont {E.}~\bibnamefont {Knill}},
  \bibinfo {author} {\bibfnamefont {S.}~\bibnamefont {Glancy}}, \bibinfo
  {author} {\bibfnamefont {A.}~\bibnamefont {Mink}}, \bibinfo {author}
  {\bibfnamefont {S.}~\bibnamefont {Jordan}}, \bibinfo {author} {\bibfnamefont
  {A.}~\bibnamefont {Rommal}}, \bibinfo {author} {\bibfnamefont {Y.-K.}\
  \bibnamefont {Liu}}, \bibinfo {author} {\bibfnamefont {B.}~\bibnamefont
  {Christensen}}, \bibinfo {author} {\bibfnamefont {S.~W.}\ \bibnamefont
  {Nam}},\ and\ \bibinfo {author} {\bibfnamefont {L.~K.}\ \bibnamefont
  {Shalm}},\ }\bibfield  {title} {\bibinfo {title} {Experimentally generated
  random numbers certified by the impossibility of superluminal signaling
  (version 1)}} (\bibinfo {year} {2017}),\ \bibinfo {note}
  {arXiv:1702.05178v1}\BibitemShut {NoStop}%
\bibitem [{\citenamefont {Bierhorst}\ \emph {et~al.}(2018)\citenamefont
  {Bierhorst}, \citenamefont {Knill}, \citenamefont {Glancy}, \citenamefont
  {Zhang}, \citenamefont {Mink}, \citenamefont {Jordan}, \citenamefont
  {Rommal}, \citenamefont {Liu}, \citenamefont {Christensen}, \citenamefont
  {Nam}, , \citenamefont {Stevens},\ and\ \citenamefont
  {Shalm}}]{bierhorst:qc2018a}%
  \BibitemOpen
  \bibfield  {author} {\bibinfo {author} {\bibfnamefont {P.}~\bibnamefont
  {Bierhorst}}, \bibinfo {author} {\bibfnamefont {E.}~\bibnamefont {Knill}},
  \bibinfo {author} {\bibfnamefont {S.}~\bibnamefont {Glancy}}, \bibinfo
  {author} {\bibfnamefont {Y.}~\bibnamefont {Zhang}}, \bibinfo {author}
  {\bibfnamefont {A.}~\bibnamefont {Mink}}, \bibinfo {author} {\bibfnamefont
  {S.}~\bibnamefont {Jordan}}, \bibinfo {author} {\bibfnamefont
  {A.}~\bibnamefont {Rommal}}, \bibinfo {author} {\bibfnamefont {Y.-K.}\
  \bibnamefont {Liu}}, \bibinfo {author} {\bibfnamefont {B.}~\bibnamefont
  {Christensen}}, \bibinfo {author} {\bibfnamefont {S.~W.}\ \bibnamefont
  {Nam}}, , \bibinfo {author} {\bibfnamefont {M.~J.}\ \bibnamefont {Stevens}},\
  and\ \bibinfo {author} {\bibfnamefont {L.~K.}\ \bibnamefont {Shalm}},\
  }\bibfield  {title} {\bibinfo {title} {Experimentally generated random
  numbers certified by the impossibility of superluminal signaling},\
  }\href@noop {} {\bibfield  {journal} {\bibinfo  {journal} {Nature}\ }\textbf
  {\bibinfo {volume} {556}},\ \bibinfo {pages} {223} (\bibinfo {year}
  {2018})}\BibitemShut {NoStop}%
\bibitem [{\citenamefont {Shafer}\ \emph {et~al.}(2011)\citenamefont {Shafer},
  \citenamefont {Shen}, \citenamefont {Vereshchagin},\ and\ \citenamefont
  {Vovk}}]{shafer:qc2009a}%
  \BibitemOpen
  \bibfield  {author} {\bibinfo {author} {\bibfnamefont {G.}~\bibnamefont
  {Shafer}}, \bibinfo {author} {\bibfnamefont {A.}~\bibnamefont {Shen}},
  \bibinfo {author} {\bibfnamefont {N.}~\bibnamefont {Vereshchagin}},\ and\
  \bibinfo {author} {\bibfnamefont {V.}~\bibnamefont {Vovk}},\ }\bibfield
  {title} {\bibinfo {title} {Test martingales, {Bayes} factors and
  $p$-values},\ }\href {https://doi.org/10.1214/10-STS347} {\bibfield
  {journal} {\bibinfo  {journal} {Statistical Science}\ }\textbf {\bibinfo
  {volume} {26}},\ \bibinfo {pages} {84} (\bibinfo {year} {2011})}\BibitemShut
  {NoStop}%
\bibitem [{\citenamefont {Zhang}\ \emph {et~al.}(2011)\citenamefont {Zhang},
  \citenamefont {Glancy},\ and\ \citenamefont {Knill}}]{zhang:2011}%
  \BibitemOpen
  \bibfield  {author} {\bibinfo {author} {\bibfnamefont {Y.}~\bibnamefont
  {Zhang}}, \bibinfo {author} {\bibfnamefont {S.}~\bibnamefont {Glancy}},\ and\
  \bibinfo {author} {\bibfnamefont {E.}~\bibnamefont {Knill}},\ }\bibfield
  {title} {\bibinfo {title} {Asymptotically optimal data analysis for rejecting
  local realism},\ }\href@noop {} {\bibfield  {journal} {\bibinfo  {journal}
  {Phys. Rev. A}\ }\textbf {\bibinfo {volume} {84}},\ \bibinfo {pages} {062118}
  (\bibinfo {year} {2011})}\BibitemShut {NoStop}%
\bibitem [{\citenamefont {Arnon-Friedman}\ \emph {et~al.}(2016)\citenamefont
  {Arnon-Friedman}, \citenamefont {Renner},\ and\ \citenamefont
  {Vidick}}]{arnon-friedman:qc2016a}%
  \BibitemOpen
  \bibfield  {author} {\bibinfo {author} {\bibfnamefont {R.}~\bibnamefont
  {Arnon-Friedman}}, \bibinfo {author} {\bibfnamefont {R.}~\bibnamefont
  {Renner}},\ and\ \bibinfo {author} {\bibfnamefont {T.}~\bibnamefont
  {Vidick}},\ }\bibfield  {title} {\bibinfo {title} {Simple and tight
  device-independent security proofs}} (\bibinfo {year} {2016}),\ \bibinfo
  {note} {arXiv:1607.01797 (specific citations are for version 1)}\BibitemShut
  {NoStop}%
\bibitem [{\citenamefont {Pironio}\ \emph {et~al.}(2009)\citenamefont
  {Pironio}, \citenamefont {Acin}, \citenamefont {Brunner}, \citenamefont
  {Gisin}, \citenamefont {Massar},\ and\ \citenamefont
  {Scarani}}]{pironio:qc2009a}%
  \BibitemOpen
  \bibfield  {author} {\bibinfo {author} {\bibfnamefont {S.}~\bibnamefont
  {Pironio}}, \bibinfo {author} {\bibfnamefont {A.}~\bibnamefont {Acin}},
  \bibinfo {author} {\bibfnamefont {N.}~\bibnamefont {Brunner}}, \bibinfo
  {author} {\bibfnamefont {N.}~\bibnamefont {Gisin}}, \bibinfo {author}
  {\bibfnamefont {S.}~\bibnamefont {Massar}},\ and\ \bibinfo {author}
  {\bibfnamefont {V.}~\bibnamefont {Scarani}},\ }\bibfield  {title} {\bibinfo
  {title} {Device-independent quantum key distribution secure against
  collective attacks},\ }\href {https://doi.org/10.1088/1367-2630/11/4/045021}
  {\bibfield  {journal} {\bibinfo  {journal} {New Journal of Physics}\ }\textbf
  {\bibinfo {volume} {11}},\ \bibinfo {pages} {045021} (\bibinfo {year}
  {2009})}\BibitemShut {NoStop}%
\bibitem [{\citenamefont {Tsirelson}(1993)}]{tsirelson:qc1993a}%
  \BibitemOpen
  \bibfield  {author} {\bibinfo {author} {\bibfnamefont {B.}~\bibnamefont
  {Tsirelson}},\ }\bibfield  {title} {\bibinfo {title} {Some results and
  problems on quantum bell-type inequalities},\ }\href@noop {} {\bibfield
  {journal} {\bibinfo  {journal} {Hadronic J. Suppl.}\ }\textbf {\bibinfo
  {volume} {8}},\ \bibinfo {pages} {329} (\bibinfo {year} {1993})}\BibitemShut
  {NoStop}%
\bibitem [{\citenamefont {Masanes}(2006)}]{masanes:qc2006a}%
  \BibitemOpen
  \bibfield  {author} {\bibinfo {author} {\bibfnamefont {L.}~\bibnamefont
  {Masanes}},\ }\bibfield  {title} {\bibinfo {title} {Asymptotic violation of
  bell inequalities and distillability},\ }\href@noop {} {\bibfield  {journal}
  {\bibinfo  {journal} {Phys. Rev. Lett.}\ }\textbf {\bibinfo {volume} {97}},\
  \bibinfo {pages} {050503/1} (\bibinfo {year} {2006})}\BibitemShut {NoStop}%
\bibitem [{\citenamefont {Pironio}\ \emph {et~al.}(2010)\citenamefont
  {Pironio}, \citenamefont {Acin}, \citenamefont {Massar}, \citenamefont {de~la
  Giroday}, \citenamefont {Matsukevich}, \citenamefont {Maunz}, \citenamefont
  {Olmschenk}, \citenamefont {Hayes}, \citenamefont {Luo}, \citenamefont
  {Manning},\ and\ \citenamefont {Monroe}}]{pironio:qc2010a}%
  \BibitemOpen
  \bibfield  {author} {\bibinfo {author} {\bibfnamefont {S.}~\bibnamefont
  {Pironio}}, \bibinfo {author} {\bibfnamefont {A.}~\bibnamefont {Acin}},
  \bibinfo {author} {\bibfnamefont {S.}~\bibnamefont {Massar}}, \bibinfo
  {author} {\bibfnamefont {A.~B.}\ \bibnamefont {de~la Giroday}}, \bibinfo
  {author} {\bibfnamefont {D.~N.}\ \bibnamefont {Matsukevich}}, \bibinfo
  {author} {\bibfnamefont {P.}~\bibnamefont {Maunz}}, \bibinfo {author}
  {\bibfnamefont {S.}~\bibnamefont {Olmschenk}}, \bibinfo {author}
  {\bibfnamefont {D.}~\bibnamefont {Hayes}}, \bibinfo {author} {\bibfnamefont
  {L.}~\bibnamefont {Luo}}, \bibinfo {author} {\bibfnamefont {T.~A.}\
  \bibnamefont {Manning}},\ and\ \bibinfo {author} {\bibfnamefont
  {C.}~\bibnamefont {Monroe}},\ }\bibfield  {title} {\bibinfo {title} {Random
  numbers certified by bell's theorem},\ }\href@noop {} {\bibfield  {journal}
  {\bibinfo  {journal} {Nature}\ }\textbf {\bibinfo {volume} {464}},\ \bibinfo
  {pages} {1021} (\bibinfo {year} {2010})}\BibitemShut {NoStop}%
\bibitem [{\citenamefont {Carlen}(2010)}]{carlen:qc2009a}%
  \BibitemOpen
  \bibfield  {author} {\bibinfo {author} {\bibfnamefont {E.~A.}\ \bibnamefont
  {Carlen}},\ }\bibfield  {title} {\bibinfo {title} {Trace inequalities and
  quantum entropy: An introductory course},\ }in\ \href@noop {} {\emph
  {\bibinfo {booktitle} {Entropy and the Quantum}}},\ \bibinfo {series}
  {Contemporary Mathematics}, Vol.\ \bibinfo {volume} {529}\ (\bibinfo
  {publisher} {American Mathematical Society},\ \bibinfo {year} {2010})\ pp.\
  \bibinfo {pages} {73--140}\BibitemShut {NoStop}%
\bibitem [{\citenamefont {Hayden}\ \emph {et~al.}(2004)\citenamefont {Hayden},
  \citenamefont {Jozsa}, \citenamefont {Petz},\ and\ \citenamefont
  {Winter}}]{hayden:qc2004b}%
  \BibitemOpen
  \bibfield  {author} {\bibinfo {author} {\bibfnamefont {P.}~\bibnamefont
  {Hayden}}, \bibinfo {author} {\bibfnamefont {R.}~\bibnamefont {Jozsa}},
  \bibinfo {author} {\bibfnamefont {D.}~\bibnamefont {Petz}},\ and\ \bibinfo
  {author} {\bibfnamefont {A.}~\bibnamefont {Winter}},\ }\bibfield  {title}
  {\bibinfo {title} {Structure of states which satisfy strong subadditivity of
  quantum entropy with equality},\ }\href
  {https://doi.org/10.1007/s00220-004-1049-z} {\bibfield  {journal} {\bibinfo
  {journal} {Comm. Math. Phys.}\ }\textbf {\bibinfo {volume} {246}},\ \bibinfo
  {pages} {359} (\bibinfo {year} {2004})}\BibitemShut {NoStop}%
\bibitem [{\citenamefont {Tomamichel}(2012)}]{tomamichel:qc2012a}%
  \BibitemOpen
  \bibfield  {author} {\bibinfo {author} {\bibfnamefont {M.}~\bibnamefont
  {Tomamichel}},\ }\emph {\bibinfo {title} {A Framework for Non-Asymptotic
  Quantum Information Theory}},\ \href@noop {} {Ph.D. thesis},\ \bibinfo
  {school} {ETH}, \bibinfo {address} {Z\"urich, Switzerland} (\bibinfo {year}
  {2012}),\ \bibinfo {note} {(specific citations are for arXiv:1203.2142
  version 2, note that definitions, lemmas, propositions, etc. are
  independently numbered)}\BibitemShut {NoStop}%
\bibitem [{\citenamefont {Tomamichel}(2016)}]{tomamichel:qc2015a}%
  \BibitemOpen
  \bibfield  {author} {\bibinfo {author} {\bibfnamefont {M.}~\bibnamefont
  {Tomamichel}},\ }\href {https://doi.org/10.1007/978-3-319-21891-5} {\emph
  {\bibinfo {title} {Quantum Information Processing with Finite Resources -
  Mathematical Foundations}}},\ SpringerBriefs in Mathematical Physics\
  (\bibinfo  {publisher} {Springer Verlag},\ \bibinfo {year} {2016})\ \bibinfo
  {note} {(specific citations are for arXiv:1504.00233 version 3, note that
  definitions, lemmas, propositions, etc. are independently
  numbered)}\BibitemShut {NoStop}%
\bibitem [{\citenamefont {Nielsen}\ and\ \citenamefont
  {Chuang}(2001)}]{nielsen:qc2001a}%
  \BibitemOpen
  \bibfield  {author} {\bibinfo {author} {\bibfnamefont {M.~A.}\ \bibnamefont
  {Nielsen}}\ and\ \bibinfo {author} {\bibfnamefont {I.~L.}\ \bibnamefont
  {Chuang}},\ }\href@noop {} {\emph {\bibinfo {title} {Quantum Computation and
  Quantum Information}}}\ (\bibinfo  {publisher} {Cambridge University Press},\
  \bibinfo {address} {Cambridge, UK},\ \bibinfo {year} {2001})\BibitemShut
  {NoStop}%
\bibitem [{\citenamefont {Bhatia}(1997)}]{bhatia:qc1997a}%
  \BibitemOpen
  \bibfield  {author} {\bibinfo {author} {\bibfnamefont {R.}~\bibnamefont
  {Bhatia}},\ }\href@noop {} {\emph {\bibinfo {title} {Matrix Analysis}}}\
  (\bibinfo  {publisher} {Springer},\ \bibinfo {address} {New York},\ \bibinfo
  {year} {1997})\BibitemShut {NoStop}%
\bibitem [{\citenamefont {M\"uller-Lennert}\ \emph {et~al.}(2013)\citenamefont
  {M\"uller-Lennert}, \citenamefont {Dupuis}, \citenamefont {Szehr},
  \citenamefont {Fehr},\ and\ \citenamefont
  {Tomamichel}}]{mueller-lennert:qc2013a}%
  \BibitemOpen
  \bibfield  {author} {\bibinfo {author} {\bibfnamefont {M.}~\bibnamefont
  {M\"uller-Lennert}}, \bibinfo {author} {\bibfnamefont {F.}~\bibnamefont
  {Dupuis}}, \bibinfo {author} {\bibfnamefont {O.}~\bibnamefont {Szehr}},
  \bibinfo {author} {\bibfnamefont {S.}~\bibnamefont {Fehr}},\ and\ \bibinfo
  {author} {\bibfnamefont {M.}~\bibnamefont {Tomamichel}},\ }\bibfield  {title}
  {\bibinfo {title} {On quantum {R\'enyi} entropies: A new generalization and
  some properties},\ }\href {https://doi.org/10.1063/1.4838856} {\bibfield
  {journal} {\bibinfo  {journal} {J. Math. Phys.}\ }\textbf {\bibinfo {volume}
  {54}},\ \bibinfo {pages} {122203} (\bibinfo {year} {2013})}\BibitemShut
  {NoStop}%
\bibitem [{\citenamefont {Frank}\ and\ \citenamefont
  {Lieb}(2013)}]{frank:qc2013a}%
  \BibitemOpen
  \bibfield  {author} {\bibinfo {author} {\bibfnamefont {R.~L.}\ \bibnamefont
  {Frank}}\ and\ \bibinfo {author} {\bibfnamefont {E.~H.}\ \bibnamefont
  {Lieb}},\ }\bibfield  {title} {\bibinfo {title} {Monotonicity of a relative
  {R\'enyi} entropy},\ }\href {https://doi.org/10.1063/1.4838835} {\bibfield
  {journal} {\bibinfo  {journal} {J. Math. Phys.}\ }\textbf {\bibinfo {volume}
  {54}},\ \bibinfo {pages} {122201} (\bibinfo {year} {2013})}\BibitemShut
  {NoStop}%
\bibitem [{\citenamefont {Beigi}(2013)}]{beigi:qc2013a}%
  \BibitemOpen
  \bibfield  {author} {\bibinfo {author} {\bibfnamefont {S.}~\bibnamefont
  {Beigi}},\ }\bibfield  {title} {\bibinfo {title} {Sandwiche {R\'enyi}
  divergence satisfies data processing inequality},\ }\href
  {https://doi.org/10.1063/1.4838855} {\bibfield  {journal} {\bibinfo
  {journal} {J. Math. Phys}\ }\textbf {\bibinfo {volume} {54}},\ \bibinfo
  {pages} {122202} (\bibinfo {year} {2013})}\BibitemShut {NoStop}%
\bibitem [{\citenamefont {K\"onig}\ \emph {et~al.}(2009)\citenamefont
  {K\"onig}, \citenamefont {Renner},\ and\ \citenamefont
  {Schaffner}}]{koenig:qc2009a}%
  \BibitemOpen
  \bibfield  {author} {\bibinfo {author} {\bibfnamefont {R.}~\bibnamefont
  {K\"onig}}, \bibinfo {author} {\bibfnamefont {R.}~\bibnamefont {Renner}},\
  and\ \bibinfo {author} {\bibfnamefont {C.}~\bibnamefont {Schaffner}},\
  }\bibfield  {title} {\bibinfo {title} {The operational meaning of min- and
  max-entropy},\ }\href@noop {} {\bibfield  {journal} {\bibinfo  {journal}
  {IEEE Trans. Inf. Th.}\ }\textbf {\bibinfo {volume} {55}},\ \bibinfo {pages}
  {4337} (\bibinfo {year} {2009})}\BibitemShut {NoStop}%
\bibitem [{\citenamefont {Renner}(2005)}]{renner:qc2005}%
  \BibitemOpen
  \bibfield  {author} {\bibinfo {author} {\bibfnamefont {R.}~\bibnamefont
  {Renner}},\ }\emph {\bibinfo {title} {Security of Quantum Key
  Distribution}},\ \href@noop {} {Ph.D. thesis},\ \bibinfo  {school} {ETH},
  \bibinfo {address} {Z\"urich, Switzerland} (\bibinfo {year} {2005}),\
  \bibinfo {note} {(available as arXiv:quant-ph/0512258 version 2)}\BibitemShut
  {NoStop}%
\bibitem [{\citenamefont {Kessler}\ and\ \citenamefont
  {Arnon-Friedman}(2017)}]{kessler:qc2017a}%
  \BibitemOpen
  \bibfield  {author} {\bibinfo {author} {\bibfnamefont {M.}~\bibnamefont
  {Kessler}}\ and\ \bibinfo {author} {\bibfnamefont {R.}~\bibnamefont
  {Arnon-Friedman}},\ }\bibfield  {title} {\bibinfo {title} {Device-independent
  randomness amplification and privatization}} (\bibinfo {year} {2017}),\
  \bibinfo {note} {arXiv:1705.04148}\BibitemShut {NoStop}%
\bibitem [{\citenamefont {Barrett}\ \emph {et~al.}(2013)\citenamefont
  {Barrett}, \citenamefont {Colbeck},\ and\ \citenamefont
  {Kent}}]{barrett_j:qc2012a}%
  \BibitemOpen
  \bibfield  {author} {\bibinfo {author} {\bibfnamefont {J.}~\bibnamefont
  {Barrett}}, \bibinfo {author} {\bibfnamefont {R.}~\bibnamefont {Colbeck}},\
  and\ \bibinfo {author} {\bibfnamefont {A.}~\bibnamefont {Kent}},\ }\bibfield
  {title} {\bibinfo {title} {Memory attacks on device-independent quantum
  cryptography},\ }\href@noop {} {\bibfield  {journal} {\bibinfo  {journal}
  {Phys. Rev. Lett.}\ }\textbf {\bibinfo {volume} {110}},\ \bibinfo {pages}
  {010503} (\bibinfo {year} {2013})}\BibitemShut {NoStop}%
\bibitem [{\citenamefont {Boyd}\ and\ \citenamefont
  {Vandenberghe}(2004)}]{boyd:qc2004a}%
  \BibitemOpen
  \bibfield  {author} {\bibinfo {author} {\bibfnamefont {S.}~\bibnamefont
  {Boyd}}\ and\ \bibinfo {author} {\bibfnamefont {L.}~\bibnamefont
  {Vandenberghe}},\ }\href@noop {} {\emph {\bibinfo {title} {Convex
  Optimization}}}\ (\bibinfo  {publisher} {Cambridge University Press},\
  \bibinfo {address} {Cambridge, UK},\ \bibinfo {year} {2004})\BibitemShut
  {NoStop}%
\bibitem [{\citenamefont {Kadison}\ and\ \citenamefont
  {Ringrose}(1997)}]{kadison:qf1997a}%
  \BibitemOpen
  \bibfield  {author} {\bibinfo {author} {\bibfnamefont {R.~V.}\ \bibnamefont
  {Kadison}}\ and\ \bibinfo {author} {\bibfnamefont {J.~R.}\ \bibnamefont
  {Ringrose}},\ }\href@noop {} {\emph {\bibinfo {title} {Fundamentals of Theory
  of Operator Algebras. Vol. I: Elementary Theory}}},\ \bibinfo {series}
  {Graduate Studies in Mathematics}, Vol.~\bibinfo {volume} {15}\ (\bibinfo
  {publisher} {American Mathematical Socieity},\ \bibinfo {address}
  {Providence, RI},\ \bibinfo {year} {1997})\BibitemShut {NoStop}%
\bibitem [{\citenamefont {Jaggi}(2013)}]{jaggi:qc2013a}%
  \BibitemOpen
  \bibfield  {author} {\bibinfo {author} {\bibfnamefont {M.}~\bibnamefont
  {Jaggi}},\ }\bibfield  {title} {\bibinfo {title} {Revisiting frank-wolfe:
  Projection-free sparse convex optimization},\ }in\ \href@noop {} {\emph
  {\bibinfo {booktitle} {Proceedings of the 30th International Conference on
  Machine Learning}}},\ \bibinfo {series} {Proceedings of Machine Learning
  Research}, Vol.~\bibinfo {volume} {28}\ (\bibinfo {year} {2013})\ pp.\
  \bibinfo {pages} {427--435}\BibitemShut {NoStop}%
\bibitem [{\citenamefont {Hradil}\ \emph {et~al.}(2004)\citenamefont {Hradil},
  \citenamefont {Rehacek}, \citenamefont {Fiurasek},\ and\ \citenamefont
  {Jezek}}]{hradil:qc2004a}%
  \BibitemOpen
  \bibfield  {author} {\bibinfo {author} {\bibfnamefont {Z.}~\bibnamefont
  {Hradil}}, \bibinfo {author} {\bibfnamefont {J.}~\bibnamefont {Rehacek}},
  \bibinfo {author} {\bibfnamefont {J.}~\bibnamefont {Fiurasek}},\ and\
  \bibinfo {author} {\bibfnamefont {M.}~\bibnamefont {Jezek}},\ }\bibfield
  {title} {\bibinfo {title} {Maximum-likelihood methods in quantum mechanics},\
  }in\ \href@noop {} {\emph {\bibinfo {booktitle} {Quantum State Estimation}}}\
  (\bibinfo  {publisher} {Springer-Verlag},\ \bibinfo {address} {New York},\
  \bibinfo {year} {2004})\ pp.\ \bibinfo {pages} {163--172}\BibitemShut
  {NoStop}%
\bibitem [{\citenamefont {Rehacek}\ \emph {et~al.}(2006)\citenamefont
  {Rehacek}, \citenamefont {Hradil}, \citenamefont {Knill},\ and\ \citenamefont
  {Lvovsky}}]{rehacek:qc2006a}%
  \BibitemOpen
  \bibfield  {author} {\bibinfo {author} {\bibfnamefont {J.}~\bibnamefont
  {Rehacek}}, \bibinfo {author} {\bibfnamefont {Z.}~\bibnamefont {Hradil}},
  \bibinfo {author} {\bibfnamefont {E.}~\bibnamefont {Knill}},\ and\ \bibinfo
  {author} {\bibfnamefont {A.~I.}\ \bibnamefont {Lvovsky}},\ }\bibfield
  {title} {\bibinfo {title} {Diluted maximum-likelihood algorithm for quantum
  tomography},\ }\href {https://doi.org/10.1007/s10751-007-9571-y} {\bibfield
  {journal} {\bibinfo  {journal} {Phys. Rev. A}\ }\textbf {\bibinfo {volume}
  {75}},\ \bibinfo {pages} {042108/1} (\bibinfo {year} {2006})},\ \bibinfo
  {note} {arXiv:quant-ph/0611244}\BibitemShut {NoStop}%
\bibitem [{\citenamefont {Navascu\'es}\ \emph {et~al.}(2007)\citenamefont
  {Navascu\'es}, \citenamefont {Pironio},\ and\ \citenamefont
  {Ac\'{\i}n}}]{navascues:2007}%
  \BibitemOpen
  \bibfield  {author} {\bibinfo {author} {\bibfnamefont {M.}~\bibnamefont
  {Navascu\'es}}, \bibinfo {author} {\bibfnamefont {S.}~\bibnamefont
  {Pironio}},\ and\ \bibinfo {author} {\bibfnamefont {A.}~\bibnamefont
  {Ac\'{\i}n}},\ }\bibfield  {title} {\bibinfo {title} {Bounding the set of
  quantum correlations},\ }\href
  {https://doi.org/10.1103/PhysRevLett.98.010401} {\bibfield  {journal}
  {\bibinfo  {journal} {Phys. Rev. Lett.}\ }\textbf {\bibinfo {volume} {98}},\
  \bibinfo {pages} {010401} (\bibinfo {year} {2007})}\BibitemShut {NoStop}%
\bibitem [{\citenamefont {Zhang}\ \emph {et~al.}(2018)\citenamefont {Zhang},
  \citenamefont {Knill},\ and\ \citenamefont {Bierhorst}}]{zhang:qc2018a}%
  \BibitemOpen
  \bibfield  {author} {\bibinfo {author} {\bibfnamefont {Y.}~\bibnamefont
  {Zhang}}, \bibinfo {author} {\bibfnamefont {E.}~\bibnamefont {Knill}},\ and\
  \bibinfo {author} {\bibfnamefont {P.}~\bibnamefont {Bierhorst}},\ }\bibfield
  {title} {\bibinfo {title} {Certifying quantum randomness by probability
  estimation},\ }\href {https://doi.org/10.1103/PhysRevA.98.040304} {\bibfield
  {journal} {\bibinfo  {journal} {Phys. Rev. A}\ }\textbf {\bibinfo {volume}
  {98}},\ \bibinfo {pages} {040304(R)} (\bibinfo {year} {2018})}\BibitemShut
  {NoStop}%
\bibitem [{\citenamefont {Fehr}\ \emph {et~al.}(2013)\citenamefont {Fehr},
  \citenamefont {Gelles},\ and\ \citenamefont {Schaffner}}]{fehr:qc2013a}%
  \BibitemOpen
  \bibfield  {author} {\bibinfo {author} {\bibfnamefont {S.}~\bibnamefont
  {Fehr}}, \bibinfo {author} {\bibfnamefont {R.}~\bibnamefont {Gelles}},\ and\
  \bibinfo {author} {\bibfnamefont {C.}~\bibnamefont {Schaffner}},\ }\bibfield
  {title} {\bibinfo {title} {Security and composability of randomness expansion
  from {Bell} inequalities},\ }\href@noop {} {\bibfield  {journal} {\bibinfo
  {journal} {Phys. Rev. A}\ }\textbf {\bibinfo {volume} {87}},\ \bibinfo
  {pages} {012335} (\bibinfo {year} {2013})}\BibitemShut {NoStop}%
\bibitem [{\citenamefont {Pironio}\ and\ \citenamefont
  {Massar}(2013)}]{pironio:qc2011a}%
  \BibitemOpen
  \bibfield  {author} {\bibinfo {author} {\bibfnamefont {S.}~\bibnamefont
  {Pironio}}\ and\ \bibinfo {author} {\bibfnamefont {S.}~\bibnamefont
  {Massar}},\ }\bibfield  {title} {\bibinfo {title} {Security of practical
  private randomness generation},\ }\href
  {https://doi.org/10.1103/PhysRevA.87.012336} {\bibfield  {journal} {\bibinfo
  {journal} {Phys. Rev. A}\ }\textbf {\bibinfo {volume} {87}},\ \bibinfo
  {pages} {012336} (\bibinfo {year} {2013})},\ \bibinfo {note}
  {arXiv:1111.6056}\BibitemShut {NoStop}%
\bibitem [{\citenamefont {Acin}\ \emph {et~al.}(2012)\citenamefont {Acin},
  \citenamefont {Massar},\ and\ \citenamefont {Pironio}}]{acin:qc2012a}%
  \BibitemOpen
  \bibfield  {author} {\bibinfo {author} {\bibfnamefont {A.}~\bibnamefont
  {Acin}}, \bibinfo {author} {\bibfnamefont {S.}~\bibnamefont {Massar}},\ and\
  \bibinfo {author} {\bibfnamefont {S.}~\bibnamefont {Pironio}},\ }\bibfield
  {title} {\bibinfo {title} {Randomness versus nonlocality and entanglement},\
  }\href@noop {} {\bibfield  {journal} {\bibinfo  {journal} {Phys. Rev. Lett}\
  }\textbf {\bibinfo {volume} {108}},\ \bibinfo {pages} {100402/1} (\bibinfo
  {year} {2012})}\BibitemShut {NoStop}%
\bibitem [{\citenamefont {Nieto-Silleras}\ \emph {et~al.}(2014)\citenamefont
  {Nieto-Silleras}, \citenamefont {Pironio},\ and\ \citenamefont
  {Silman}}]{nieto-silleras:qc2014a}%
  \BibitemOpen
  \bibfield  {author} {\bibinfo {author} {\bibfnamefont {O.}~\bibnamefont
  {Nieto-Silleras}}, \bibinfo {author} {\bibfnamefont {S.}~\bibnamefont
  {Pironio}},\ and\ \bibinfo {author} {\bibfnamefont {J.}~\bibnamefont
  {Silman}},\ }\bibfield  {title} {\bibinfo {title} {Using complete measurement
  statistics for optimal device-independent randomness evaluation},\
  }\href@noop {} {\bibfield  {journal} {\bibinfo  {journal} {New Journal of
  Physics}\ }\textbf {\bibinfo {volume} {16}},\ \bibinfo {pages} {013035}
  (\bibinfo {year} {2014})}\BibitemShut {NoStop}%
\bibitem [{\citenamefont {Bancal}\ \emph {et~al.}(2014)\citenamefont {Bancal},
  \citenamefont {Sheridan},\ and\ \citenamefont {Scarani}}]{bancal:qc2014a}%
  \BibitemOpen
  \bibfield  {author} {\bibinfo {author} {\bibfnamefont {J.-D.}\ \bibnamefont
  {Bancal}}, \bibinfo {author} {\bibfnamefont {L.}~\bibnamefont {Sheridan}},\
  and\ \bibinfo {author} {\bibfnamefont {V.}~\bibnamefont {Scarani}},\
  }\bibfield  {title} {\bibinfo {title} {More randomness from the same data},\
  }\href@noop {} {\bibfield  {journal} {\bibinfo  {journal} {New Journal of
  Physics}\ }\textbf {\bibinfo {volume} {16}},\ \bibinfo {pages} {033011}
  (\bibinfo {year} {2014})}\BibitemShut {NoStop}%
\bibitem [{\citenamefont {Nieto-Silleras}\ \emph {et~al.}(2016)\citenamefont
  {Nieto-Silleras}, \citenamefont {Bamps}, \citenamefont {Silman},\ and\
  \citenamefont {Pironio}}]{nieto-silleras:qc2016a}%
  \BibitemOpen
  \bibfield  {author} {\bibinfo {author} {\bibfnamefont {O.}~\bibnamefont
  {Nieto-Silleras}}, \bibinfo {author} {\bibfnamefont {C.}~\bibnamefont
  {Bamps}}, \bibinfo {author} {\bibfnamefont {J.}~\bibnamefont {Silman}},\ and\
  \bibinfo {author} {\bibfnamefont {S.}~\bibnamefont {Pironio}},\ }\bibfield
  {title} {\bibinfo {title} {Device-independent randomness generation from
  several {Bell} estimators}} (\bibinfo {year} {2016}),\ \bibinfo {note}
  {arXiv:1611.00352}\BibitemShut {NoStop}%
\bibitem [{\citenamefont {Clauser}\ \emph {et~al.}(1969)\citenamefont
  {Clauser}, \citenamefont {Horne}, \citenamefont {Shimony},\ and\
  \citenamefont {Holt}}]{clauser:qc1969a}%
  \BibitemOpen
  \bibfield  {author} {\bibinfo {author} {\bibfnamefont {J.~F.}\ \bibnamefont
  {Clauser}}, \bibinfo {author} {\bibfnamefont {M.~A.}\ \bibnamefont {Horne}},
  \bibinfo {author} {\bibfnamefont {A.}~\bibnamefont {Shimony}},\ and\ \bibinfo
  {author} {\bibfnamefont {R.~A.}\ \bibnamefont {Holt}},\ }\bibfield  {title}
  {\bibinfo {title} {Proposed experiment to test local hidden-variable
  theories},\ }\href@noop {} {\bibfield  {journal} {\bibinfo  {journal} {Phys.
  Rev. Lett.}\ }\textbf {\bibinfo {volume} {23}},\ \bibinfo {pages} {880}
  (\bibinfo {year} {1969})}\BibitemShut {NoStop}%
\bibitem [{\citenamefont {van Dam}\ \emph {et~al.}(2005)\citenamefont {van
  Dam}, \citenamefont {Gill},\ and\ \citenamefont {Grunwald}}]{vanDam:2005}%
  \BibitemOpen
  \bibfield  {author} {\bibinfo {author} {\bibfnamefont {W.}~\bibnamefont {van
  Dam}}, \bibinfo {author} {\bibfnamefont {R.~D.}\ \bibnamefont {Gill}},\ and\
  \bibinfo {author} {\bibfnamefont {P.~D.}\ \bibnamefont {Grunwald}},\
  }\bibfield  {title} {\bibinfo {title} {The statistical strength of
  nonlocality proofs},\ }\href@noop {} {\bibfield  {journal} {\bibinfo
  {journal} {IEEE Trans. Inf. Theory.}\ }\textbf {\bibinfo {volume} {51}},\
  \bibinfo {pages} {2812} (\bibinfo {year} {2005})}\BibitemShut {NoStop}%
\bibitem [{\citenamefont {Zhang}\ \emph {et~al.}(2010)\citenamefont {Zhang},
  \citenamefont {Knill},\ and\ \citenamefont {Glancy}}]{zhang:2010}%
  \BibitemOpen
  \bibfield  {author} {\bibinfo {author} {\bibfnamefont {Y.}~\bibnamefont
  {Zhang}}, \bibinfo {author} {\bibfnamefont {E.}~\bibnamefont {Knill}},\ and\
  \bibinfo {author} {\bibfnamefont {S.}~\bibnamefont {Glancy}},\ }\bibfield
  {title} {\bibinfo {title} {Statistical strength of experiments to reject
  local realism with photon pairs and inefficient detectors},\ }\href
  {https://doi.org/10.1103/PhysRevA.81.032117} {\bibfield  {journal} {\bibinfo
  {journal} {Phys. Rev. A}\ }\textbf {\bibinfo {volume} {81}},\ \bibinfo
  {pages} {032117} (\bibinfo {year} {2010})}\BibitemShut {NoStop}%
\end{thebibliography}%
